\documentclass[letterpaper,11pt]{article}
\usepackage{fullpage}
\usepackage[ansinew]{inputenc}
\usepackage{verbatim}
\usepackage{graphicx}
\usepackage{tabularx}
\usepackage{array}
\usepackage{delarray}
\usepackage{dcolumn}
\usepackage{float}
\usepackage{amsmath}
\usepackage{psfrag}
\usepackage{booktabs}
\usepackage{multirow,hhline}
\usepackage{amstext}
\usepackage{amssymb}
\usepackage{fancyhdr}
\usepackage{setspace}
\usepackage{amsthm}
\usepackage{makeidx}
\usepackage{tikz}
\usepackage{lipsum}
\usepackage{url}
\usepackage{adjustbox}
\usepackage{enumitem}
\usepackage{commath}
\usepackage{hyperref}
\usepackage{threeparttable}
\usepackage{ragged2e}
\usepackage{pdflscape}

\RequirePackage[round,authoryear,longnamesfirst]{natbib}
\usetikzlibrary{arrows,calc}
\hypersetup{citecolor=true,colorlinks=true,allcolors=blue}
\makeindex
\allowdisplaybreaks
\numberwithin{equation}{section}
\newtheorem{lem}{Lemma}[section]
\newtheorem{thm}{Theorem}[section]

\newtheorem{ass}{Assumption}

\newtheorem{ex}{Example}
\renewcommand{\citep}[1]{\citeauthor{#1}, \citeyear{#1}}

\newcommand{\convP}{\stackrel{p}{\longrightarrow}}
\newcommand{\convD}{\rightsquigarrow}

\newcommand{\N}{\mathcal{N}}
\newcommand{\eps}{\varepsilon}

\renewcommand{\epsilon}{\varepsilon}
\DeclareMathOperator*{\argmax}{arg\,max}
\DeclareMathOperator*{\argmin}{arg\,min}
\DeclareMathOperator*{\var}{var}
\makeatletter
\newcommand*{\rom}
[1]{\expandafter\@slowromancap\romannumeral #1@}
\makeatother

\usepackage[customcolors]{hf-tikz} 
\usetikzlibrary{patterns}
\usetikzlibrary{matrix,decorations.pathreplacing}
\usepackage{pdfpages}
\usepackage{pgf}
\usepackage{tikz}
\usetikzlibrary{arrows,automata}
\pgfkeys{tikz/mymatrixenv/.style={decoration={brace},every left delimiter/.style={xshift=8pt},every right delimiter/.style={xshift=-8pt}}}
\pgfkeys{tikz/mymatrix/.style={matrix of math nodes,nodes in empty cells,left delimiter={[},right delimiter={]},inner sep=1pt,outer sep=1.5pt,column sep=2pt,row sep=2pt,nodes={minimum width=20pt,minimum height=10pt,anchor=center,inner sep=0pt,outer sep=0pt}}}
\pgfkeys{tikz/mymatrixbrace/.style={decorate,thick}}

\title{Improving Estimation Efficiency via Regression-Adjustment in Covariate-Adaptive Randomizations with Imperfect Compliance\thanks{ Yichong Zhang acknowledges the financial support from Singapore Ministry of Education Tier 2 grant under grant MOE2018-T2-2-169 and the Lee Kong Chian fellowship. Any and all errors are our own.\vspace{1.3mm}} 
	\\ \vspace{2mm}
}

\author{Liang Jiang\thanks{Fudan University.\ E-mail~address: jiangliang@fudan.edu.cn.} 
\and Oliver B. Linton\thanks{University of Cambridge.\ E-mail~address: obl20@cam.ac.uk.}
\and Haihan Tang\thanks{Fudan University.\ E-mail~address: hhtang@fudan.edu.cn.} 
\and Yichong Zhang\thanks{%
Singapore Management University.\ E-mail~address: yczhang@smu.edu.sg.}
\date{\today}
}

\begin{document}
\maketitle

\def \baselinestretch{1.05}

\begin{abstract}
	We investigate how to improve efficiency using regression adjustments with covariates in covariate-adaptive randomizations (CARs) with imperfect subject compliance. Our regression-adjusted estimators, which are based on the doubly robust moment for local average treatment effects, are consistent and asymptotically normal even with heterogeneous probability of assignment and misspecified regression adjustments. We propose an optimal but potentially misspecified linear adjustment and its further improvement via a nonlinear adjustment, both of which lead to more efficient estimators than the one without adjustments. We also provide conditions for nonparametric and regularized adjustments to achieve the semiparametric efficiency bound under CARs.

	\bigskip
	
	
	
\end{abstract}


\section{Introduction}

\label{sec:intro} 
Randomized experiments have become increasingly popular in economic research. One commonly used randomization method employed by economists to ensure balance between treatment and control is covariate-adaptive randomization (CAR) (\citep{B09}), in which subjects are randomly assigned to treatment and control within strata formed by a few key pretreatment variables. However, subject compliance with the random
assignment is usually imperfect. We survey all publications using randomized experiments in eight leading economics journals from 2015 to 2022 and identify eleven papers that used CARs with imperfect compliance.\footnote{See Section \ref{sec:survey} for more details.}


When subjects do not comply with the assignment in CARs, researchers usually estimate the local average treatment effects (LATEs) for the compliers using the two-stage least squares (TSLS) method with treatment assignment as an instrumental variable and covariates and strata fixed effects as exogenous controls. Actually, all eleven papers mentioned above estimate the LATE in this way. We simply denote this estimator as TSLS. Recently, \cite{anseletal2018} proposed an S estimator (denoted as S) which aggregates IV estimators for each stratum. \cite{BG21} proposed a fully saturated estimator with strata dummies, which we call the unadjusted estimator (NA) as it does not use covariates. The standard theory for the consistency of TSLS requires both correct specification of the conditional mean model and homogeneous treatment effect. In contrast, both S and NA estimators are consistent under CARs without requiring correct specifications, homogeneous treatment effect, or identical treatment assignment probability across strata. \cite{anseletal2018}  further show the S estimator is the most efficient among all the estimators discussed in their paper (Proposition 7). 

The existing literature lacks a systematic study and comparison of various LATE estimators under CARs. TSLS and S estimators impose different linear conditional mean models, which can be viewed as different types of linear regression adjustments. Then, under what conditions the TSLS estimator, like the S estimator, is consistent even when the regression adjustments are misspecified? How is the efficiency comparison among TSLS, S, and NA estimators when all of them are consistent? Is the S estimator the most efficient among all linearly adjusted LATE estimators? Can other potentially misspecified  \textit{nonlinear} regression adjustments lead to more efficient LATE estimators? Last, what is the semiparametric efficiency bound (SEB) for LATE estimation under (CARs) and how can we achieve it? 

In this paper, we provide answers to all these questions. Specifically, we follow the framework that was recently established by \cite{BCS17} to study causal inference under CARs, which allows for heterogeneous assignment probabilities and treatment effects. We first show that (1) TSLS with both strata dummies and covariates as exogenous controls is inconsistent if both the assignment probabilities and treatment effects are heterogeneous across strata; (2) even when TSLS is consistent (especially when the treatment assignment probabilities are homogeneous), its usual heteroskedasticity robust standard error is conservative due to the cross-sectional dependence introduced by CARs;\footnote{This point is consistent with the result in \cite{anseletal2018} for their estimator $\hat \beta_2$. However, $\hat \beta_2$ is computed by TSLS with only strata dummies under the assumption of homogeneous assignment probabilities, but no covariates as exogenous control variables.} (3) the correct asymptotic variance of the TSLS estimator may be greater than that of the NA estimator, which defeats the purpose of using covariates in the regression.

We then propose a general adjusted estimator using the doubly robust moment for LATE with a consistent estimator of the assignment probability and potentially misspecified regression adjustments based on covariates. The doubly robust moment for LATE has been derived by \cite{F07late} and used for estimating LATE by \cite{SUW22} and \cite{H22}. But we are the first to apply it under CARs and investigate the potential efficiency improvements when the regression adjustments are misspecified. We show that our inference method (1) achieves the exact asymptotic size under the null despite the cross-sectional dependence introduced by CARs, (2) is robust to
adjustment misspecification, and (3) achieves the SEB when the adjustments are correctly specified. The SEB for LATE under CARs is also new to the literature and complements those bounds derived by \cite{F07late} and \cite{A22}.\footnote{\cite{F07late} derived the SEB for LATE assuming i.i.d. data. However, CARs can introduce cross-sectional dependence, and thus, violate the independence assumption. \cite{A22} derived the SEB for average treatment effect under CARs but without covariates. The SEB for LATE under CARs but without covariates is a byproduct of our result by letting our covariates be an empty set.} 

Finally, we compare the efficiency of our LATE estimators with three specific forms of regression adjustments: (1) the optimal linear adjustment (denoted as L), which yields the most efficient estimator among all linearly adjusted estimators, (2) the nonlinear logistic adjustment (denoted as NL), and (3) a combination of linear and nonlinear adjustments (denoted as F) which is more efficient than both linear and nonlinear adjustments and new to the literature. We also extend \cite{anseletal2018} by showing that their S estimator is asymptotically equivalent to our estimator L, thus is optimal among the linearly adjusted estimators but less efficient than estimator F. We further give conditions under which estimators with nonparametric (denoted as NP) and regularized (denoted as R) regression adjustments achieve the SEB. Figure \ref{fig:order} visualizes the partial order of efficiency of these estimators.

\begin{figure}[t]
	\label{fig:order}  \centering
	\begin{tikzpicture}[->,>=stealth',shorten >=1pt,auto,node distance=2.8cm,
		semithick]
		\tikzstyle{state}=[circle,fill=white,draw=black,text=black]
		\tikzstyle{state1}=[circle,fill=white,draw=black,text=black,dashed]
		
		\node[state,minimum size=1.5cm]         (A)                    {NP,R};
		\node[state,minimum size=1.5cm]         (B) [right of=A] {F};
		\node[state,minimum size=1.5cm]         (C) [right of=B] {L=S};
		\node[state,minimum size=1.5cm]         (D) [right of=C] {NA};
		\node[state,minimum size=1.5cm]         (E) [below right of=B] {NL};
		\node[state1,minimum size=1.5cm]         (F) [below right of=C] {TSLS};
		\path (A) edge              node {} (B);
		\path (B) edge              node {} (C);
		\path (C) edge              node {} (D);
		\path (B) edge              node {} (E);
		\path (C) edge              node {} (F);
		
	\end{tikzpicture}
	\caption{Efficiency of Various LATE Estimators (from the most efficient to the
		least)\\
		\footnotesize{Note: The dashed circle around the TSLS indicates that it is not always consistent. There are no arrows between NA and TSLS because TSLS can be less efficient than NA even when it is consistent. \citeauthor{anseletal2018}'s (\citeyear{anseletal2018}) S estimator is asymptotically equivalent to our estimator L with the optimal linear adjustment. Since both NA and TSLS (when TSLS is consistent) have linear adjustments (NA has a linear adjustment with zero coefficient), they are less efficient than S and L. There is no clear winner between NL and L because even the optimal linear adjustment can be misspecified and thus potentially less efficient than some nonlinear adjustments. Theoretically, the logistic regression adjustment can be even less efficient than NA depending on how severe the misspecification is. However, the F estimator is guaranteed to be more efficient than both L and NL by construction. Last, as NP and R achieve the SEB, they are more efficient than F. Notice that all the comparisons, except for those with the TSLS, are made under the same set of assumptions (Assumptions \ref{ass:assignment1} and \ref{ass:Delta} later). As for those with TSLS, the comparisons are made when TSLS is consistent.}}%
	
\end{figure}


Our paper is related to several lines of research.
\cite{HH12,MHZ15,MQLH18,O21,SY13,ZZ20,Y18,YS20} studied inference of either the average treatment effect (ATE) or quantile treatment effect (QTE) under CARs without considering covariates.
\cite{BCS17,BCS18,BL16,F18,L13,L16,LD20,LiD20,LTM20,LY20,NW20,SYZ10,YYS20,ZD20}
studied the estimation and inference of ATEs using a variety of regression
methods under various randomization schemes. \cite{jiang2021b} examine
regression-adjusted estimation and inference of QTEs under CARs. Based on
pilot experiments, \cite{T18} and \cite{B19} devise optimal randomization
designs that may produce an ATE estimator with the lowest variance. \cite{BG21} further examine the optimal design  with imperfect compliance. All the
above works, except \cite{BG21}, assume perfect compliance, while we contribute to the literature by
studying the LATE estimators in the context of CARs and regression adjustment,
which allows imperfect compliance. \cite{renliu2021} study the
regression-adjusted LATE estimator in completely randomized experiments for a
\textit{binary} outcome using finite population asymptotics. We differ
from their work by considering the regression-adjusted estimator in
\textit{covariate-adaptive} randomizations for a \textit{general} outcome
using the \textit{superpopulation} asymptotics. Finally, our paper also connects to a
vast literature on estimation and inference in randomized experiments,
including \cite{HHK11, athey2017, abadie2018, T18, BRS19, B19, JL20}, among
many others.

\textbf{Acronyms}. In this paper, we refer to the optimally linearly adjusted, nonlinearly (logistic) adjusted, and nonparametrically adjusted estimators as L, NL, and NP, respectively. We also use NA and S to denote the fully saturated and S estimators proposed by \cite{BG21} and \cite{anseletal2018}, respectively. F denotes the estimator with adjustments that improve upon both optimal linear and nonlinear adjustments, while R denotes the estimator with regularized adjustments. We will provide more details about these estimators below. 

\section{Setting and Empirical Practice}
\subsection{Setup}
\label{sec setup}

Let $Y_{i}$ denote the observed outcome of interest for individual $i$; write $Y_{i} = Y_{i}(1)D_{i} + Y_{i}(0)(1-D_{i})$, where $Y_{i}(1)$ and $Y_{i}(0)$ are the potential treated and untreated outcomes for the individual $i$, respectively, and $D_{i}$ is a binary random variable indicating whether the individual $i$ received treatment ($D_{i}=1$) or not ($D_{i}=0$) in the actual study. One could link $D_{i}$ to the treatment assignment $A_{i}$ in the following way: $D_{i} = D_{i}(1)A_{i} + D_{i}(0)(1-A_{i})$, where $D_{i}(a)$ is the individual $i$'s treatment outcome upon receiving treatment
status $A_{i}=a$ for $a=0,1$; $D_{i}(a)$ is a binary random variable. Define
$Y_{i}(D_{i}(a)): = Y_{i}(1)D_{i}(a) + Y_{i}(0)(1-D_{i}(a))$, so we can write
$Y_{i}=Y_{i}(D_{i}(1))A_{i}+Y_{i}(D_{i}(0))(1-A_{i})$. 
Individual $i$ belongs
to stratum $S_{i}$ and possesses covariate vector $X_{i}$, where $X_{i}$
does not include the constant term. The support of
the vectors $\{X_{i}\}_{i=1}^{n}$ is denoted $\text{Supp}(X),$ while the
support of $\{S_{i}\}_{i=1}^{n}$ is $\mathcal{S}$, which is a finite set.

A researcher can
observe the data $\{Y_{i},D_{i},A_{i},S_{i},X_{i}\}_{i=1}^{n}$. 
Define $[n]:=\{1,2,...n\}$, $p(s):=\mathbb{P}(S_{i}=s)$, $n(s):=\sum
_{i\in\lbrack n]}1\{S_{i}=s\}$, $n_{1}(s):=\sum_{i\in\lbrack n]}A_{i}%
1\{S_{i}=s\}$, $n_{0}(s):=n(s)-n_{1}(s)$, $S^{(n)}:=(S_{1},\ldots,S_{n})$,
$X^{(n)}:=(X_{1},\ldots,X_{n})$, and $A^{(n)}:=(A_{1},\ldots,A_{n})$. We make
the following assumptions on the data generating process (DGP) and the
treatment assignment rule.

\begin{ass}
	\begin{enumerate}
		[label=(\roman*)]
		
		\item $\{Y_{i}(1),Y_{i}(0),D_{i}(0),D_{i}(1),S_{i},X_{i}\}_{i=1}^{n}$ is
		i.i.d. over $i$. For each $i$, we allow $X_{i}$ and $S_{i}$ to be dependent.
		
		\item $\{Y_{i}(1),Y_{i}(0),D_{i}(0),D_{i}(1),X_{i}\}^{n}_{i=1} \perp
		\!\!\!\perp A^{(n)}|S^{(n)}$.
		
		\item Suppose that $p(s)$ is fixed with respect to $n$ and positive for every $s
		\in\mathcal{S}$.
		
		\item Let $\pi(s)$ denote the propensity score for stratum $s$ (i.e., the
		targeted assignment probability for stratum $s$). Then, $c<\min_{s
			\in\mathcal{S}}\pi(s) \leq\max_{s \in\mathcal{S}}\pi(s)<1-c$ for some constant
		$c \in(0,0.5)$ and $\frac{B_{n}(s)}{n(s)} = o_{p}(1)$ for $s \in\mathcal{S}$,
		where $B_{n}(s) := \sum_{i =1}^{n} (A_{i}-\pi(s))1\{S_{i} = s\}$.
		
		\item Suppose $\mathbb{P}(D(1)=0, D(0)=1)=0$.
		
		\item $\max_{a=0,1,s\in\mathcal{S}}\mathbb{E}(|Y_{i}(a)|^{q}|S_{i}=s)\leq
		C<\infty$ for some $q\geq4$.
		
		
	\end{enumerate}
	
	\label{ass:assignment1}
\end{ass}

Several remarks are in order. First, Assumption \ref{ass:assignment1}(i) allows for the treatment assignment $A^{(n)}$, and thus, the observed outcome $\{Y_i\}_{i \in [n]}$ to be cross-sectionally dependent, which is usually the case for CARs. Second, Assumption \ref{ass:assignment1}(ii)
implies that the treatment assignment $A^{(n)}$ are generated only based on
strata indicators. Third, Assumption \ref{ass:assignment1}(iii) imposes that
the strata sizes are roughly balanced. Fourth, \cite{BCS17} show that
Assumption \ref{ass:assignment1}(iv) holds under several covariate-adaptive
treatment assignment rules such as simple random sampling (SRS), biased-coin
design (BCD), adaptive biased-coin design (WEI) and stratified block
randomization (SBR).\footnote{For completeness, we briefly repeat their descriptions in Appendix \ref{sec:car_descri}.} Note that we only require $B_{n}(s)/n(s)=o_{p}(1)$, which is weaker than the assumption imposed by \cite{BCS17} but the same as that imposed by \cite{BCS18} and \cite{ZZ20}. Fifth, Assumption \ref{ass:assignment1}(v) implies there are no defiers. Last, Assumption \ref{ass:assignment1}(vi) is a standard moment condition.

Throughout the paper, we are interested in estimating the \textit{local
	average treatment effect} (LATE), which is denoted by $\tau$ and defined as
\[
\tau:=\mathbb{E}\sbr[1]{Y(1) - Y(0)|D(1)>D(0)};
\]
that is, we are interested in the ATE for the compliers
(\citep{angristimbens1994}).

\subsection{Examples of Economics Datasets}
\label{sec:examples}

To motivate our work, we give three examples of prominent economic datasets that use CARs and have imperfect compliance. 
\begin{ex}
	\cite{atkin2017} conducted a randomized experiment with a CAR design to identify the impact of exporting on firm performance.\footnote{The dataset can be found at https://doi.org/10.7910/DVN/QOGMVI.} They had two samples of firms. In sample 1, they randomized firms into treatment or control with a target probability of 0.5 in each of the strata named: Goublan, Tups and Duble. In sample 2, they randomly select firms for the treatment group with a target probability of 0.25 in stratum Duble. They then combined the two samples together, which makes the probabilities of treatment assignment ($\pi(s)$) in their joint sample heterogeneous across strata. Firms with treatment assignment were offered an initial opportunity to sell to high-income markets, but only 62.16\% of them managed to secure large and lasting orders. 
\end{ex}

\begin{ex}
	\cite{dupasetal2018} studied how rural households benefit from free bank accounts.\footnote{The dataset is available at https://www.openicpsr.org/openicpsr/project/116346/version/V1/view.} They randomly assigned 2,160 households to treatment or control groups within each of the 41 strata. The targeted assignment probability for each stratum is 0.5. Households with treatment assignment received vouchers to open accounts, but only 41.87\% of them did so and deposited money within 2 years.     
\end{ex} 

\begin{ex}
	\cite{jha2019} examined how financial market participation affects political views and voting behavior.\footnote{The dataset can be found at https://onlinelibrary.wiley.com/doi/abs/10.3982/ECTA16385.} They used CAR to randomly assign 1345 participants to treatment or control groups within each stratum, with a target probability of 0.75. Participants with treatment assignment were offered to trade assets. But only 81.08\% of them made a trade.  
\end{ex}
\defcitealias{royer2015}{Royer~et~al.~(2015)}
\defcitealias{himmler2019}{Himmler~et~al.~(2019)}
\defcitealias{bolhaar2019}{Bolhaar~et~al.~(2019)}
\defcitealias{angrist2021}{Angrist~et~al.~(2021)}

\subsection{Survey of Empirical Practice}
\label{sec:survey}
\begin{table}[H]
	\caption{ Empirical Papers Using CARs with Imperfect Compliance }
	\centering
	\label{tab:survey}
	\begin{tabular}{lcccc}
		\toprule
		& Journal & Method & Covariates & Strata fixed effects \\ 
		\midrule
		\citetalias{royer2015} & AEJ: Applied & TSLS & Yes & Yes \\ 
		\cite{atkin2017} & QJE & TSLS & Yes & Yes \\ 
		\cite{dupasetal2018} & AEJ: Applied & TSLS & Yes & Yes \\
		\cite{marx2019} & AEJ: Applied & TSLS & Yes & Yes \\ 
		\cite{jha2019} & Ecnometrica & TSLS & Yes & Yes \\ 
		\citetalias{himmler2019} & AEJ: Applied & TSLS & Yes & Yes \\ 
		\citetalias{bolhaar2019} & AEJ: Applied & TSLS & Yes & Yes \\ 
		\cite{davis2020} & ReStat & TSLS & Yes & Yes \\ 
		\cite{beam2021} & ReStat & TSLS & Yes & Yes \\ 
		\citetalias{angrist2021} & AEJ: Applied & TSLS & Yes & Yes \\ 
		\cite{okunogbe2022} & AEJ: Policy & TSLS & Yes & Yes \\ 
		\bottomrule
	\end{tabular}
\end{table}

We survey the common practice for analyzing experiments in the empirical economics literature. Our survey is limited to articles that contain the term ``experiment" in their title or abstract and are published between January 2015 and December 2022 in eight journals: the \textit{American Economic Journal: Applied Economics} (AEJ: Applied), \textit{American Economic Journal: Economic Policy} (AEJ: Policy), \textit{American Economic Review}, \textit{Econometrica}, \textit{Journal of Political Economy}, \textit{Quarterly Journal of Economics} (QJE), \textit{Review of Economics and Statistics} (ReStat), and \textit{Review of Economic Studies}. We then manually select the articles that use CARs and report imperfect compliance. Table \ref{tab:survey} tabulates the articles found in our survey. It shows that all the papers in our sample use TSLS with covariates and strata fixed effects to estimate the LATE. This finding motivates us to study the statistical properties of this commonly used TSLS estimator in Section \ref{sec:TSLS} before proposing our new estimator.

\subsection{TSLS with Covariates and Strata Fixed Effects}
\label{sec:TSLS}

Our survey shows that empirical researchers using CARs usually estimate LATE via TSLS regressions with strata dummies and covariates. The first and second stages of the TSLS regression can be formed as
\begin{align}
	&  D_{i} \sim \gamma A_{i} + \sum_{s \in \mathcal{S}}a_{s}1\{S_i=s\} + X_{i}^{\top}\theta, \quad Y_{i} \sim \tau D_{i} + \sum_{s \in \mathcal{S}}\alpha_{s}1\{S_i=s\} + X_{i}^{\top}\delta, \label{eq:TSLS}
\end{align}
where $\{a_s\}_{s \in \mathcal S}$ and $\{\alpha_{s}\}_{s\in\mathcal{S}}$ are the strata fixed effects.

Denote the TSLS estimator of $\tau$ by $\hat{\tau}_{TSLS}$. To study the asymptotic properties of $\hat{\tau}_{TSLS}$, we follow \cite{BCS17} and \cite{anseletal2018} and make the following additional assumption on the
treatment assignment mechanism. 
\begin{ass}
	\label{ass:1iv} Suppose $\pi(s)\in(0,1)$ and
	\begin{align*}
		\left\{  \left\{ \frac{B_{n}(s)}{\sqrt{n}} \right\}_{s \in\mathcal{S}}
		\bigg|\{S_{i}\}_{i \in[n]} \right\}  \rightsquigarrow\mathcal{N}(0,\Sigma_{B}),
	\end{align*}
	where $B_{n}(s) = \sum_{i =1}^{n} (A_{i}-\pi(s))1\{S_{i} = s\}$, $\Sigma_{B} =
	\text{diag}(p(s)\gamma(s):s\in\mathcal{S})$, and $0 \leq\gamma(s) \leq
	\pi(s)(1-\pi(s))$.
\end{ass}

Three remarks are in order. First, Assumption \ref{ass:1iv} is used to analyze the TSLS estimator only and is not needed for all the analyses in later sections in the paper. Second, it implies Assumption \ref{ass:assignment1}(iv). Third, we have $\gamma(s) = \pi(s)(1-\pi(s))$ for SRS and $\gamma(s) < \pi(s)(1-\pi(s))$ for the other three randomization designs mentioned after Assumption \ref{ass:assignment1}. Specifically, for BCD and SBR, we have $\gamma(s) = 0$, which means the assignment rules achieve the strong balance. 


Following empirical researchers, we also consider the usual IV heteroskedasticity-robust standard error estimator for TSLS estimator $\hat{\tau}_{TSLS}$, which is denoted as $\hat{\sigma}_{TSLS,naive}/\sqrt{n}$.\footnote{The detailed definition of $\hat{\sigma}_{TSLS,naive}$ can be found in the proof of Theorem \ref{thm:TSLS}.} We compare $\hat \tau_{TSLS}$ with \citeauthor{BG21}'s (\citeyear{BG21}) fully saturated estimator (denoted as $\hat \tau_{NA}$) for $\tau$ under CAR, which does not use any covariates $X_i$. The asymptotic variance of $\hat \tau_{NA}$ is then denoted as $\sigma_{NA}^2$, which is given in \cite{BG21}. In Section \ref{sec:estimation}, we further show that $\hat \tau_{NA}$ is a special case of our general estimator whose asymptotic variance is derived in the proof of Theorem \ref{thm:est}. 


\begin{thm}
	Suppose Assumption \ref{ass:assignment1} holds. Then, we have
	\begin{align*}
		\hat{\tau}_{TSLS}  &  \overset{p}{\longrightarrow} \frac{\mathbb{E}\left(\pi(S_{i})(1-\pi(S_{i}))\left[ \mathbb{E}(Y_{i}(D_{i}(1))|S_{i})-\mathbb{E}(Y_{i}(D_{i}(0))|S_{i}) \right]  \right)}{\mathbb{E}\left(\pi(S_{i})(1-\pi(S_{i}))\left[
			\mathbb{E}(D_{i}(1)|S_{i})-\mathbb{E}(D_{i}(0)|S_{i}) \right]\right) },
	\end{align*}
	If $\pi(s)$ or $\frac{\mathbb{E}(Y_{i}(D_{i}(1))|S_{i}=s)-\mathbb{E}(Y_{i}(D_{i}(0))|S_{i}=s)}{
		\mathbb{E}(D_{i}(1)|S_{i}=s)-\mathbb{E}(D_{i}(0)|S_{i}=s)}$ is the same across $s \in\mathcal{S}$, then $\hat \tau_{TSLS} \convP \tau$. If $\pi(s) = \pi$ for all $s \in \mathcal{S}$ and Assumptions \ref{ass:assignment1} and
	\ref{ass:1iv} hold, then
	\begin{align*}
		\sqrt{n}(\hat{\tau}_{TSLS}-\tau) \rightsquigarrow\mathcal{N}(0,\sigma
		_{TSLS}^{2}) \quad \text{and} \quad \hat \sigma_{TSLS,naive}^2 \convP \sigma_{TSLS,naive}^{2},
	\end{align*}
	where the definitions of $\sigma_{TSLS}^{2}$ and $\sigma_{TSLS,naive}^{2}$ can be found in the proof, $\sigma_{TSLS}^{2} \leq \sigma_{TSLS,naive}^{2}$, and the inequality is strict if $\gamma(s) < \pi(1-\pi)$. Last, it is possible to have $\sigma_{TSLS}^2>\sigma_{NA}^2$.
	\label{thm:TSLS}
\end{thm}


Theorem \ref{thm:TSLS} highlights one advantage and three limitations of the commonly used TSLS estimator under CARs. The advantage is that the TSLS estimator can consistently estimate the LATE under certain conditions without assuming the linear regression in \eqref{eq:TSLS} being correctly specified. So the reason for incorporating covariates in the regression is to improve estimation efficiency. The first limitation is that the TSLS estimator is inconsistent when both the treatment effect and the probabilities of treatment assignment vary across strata. To ensure its consistency, economists should thus keep the target assignment probability ($\pi(s)$) equal across all strata in the experimental design stage, which may not be satisfied in reality (see, for example, the first dataset in Section \ref{sec:examples}). The second limitation is that the heteroskedasticity-robust standard error reported by standard software such as STATA is conservative and inconsistent unless $\gamma(s) = \pi(1-\pi)$. However, this condition is violated when treatment is not assigned independently, such as BCD and SBR, which are widely used in RCTs. With the cross-sectional dependence among treatment assignments, it is expected that the usual heteroskedasticity-robust standard error is inconsistent. The third limitation is that the asymptotic variance $\sigma^2_{TSLS}$ may not be smaller than that of the unadjusted estimator, which goes against the purpose of using covariates in the regression. In this paper, we develop estimators that have the same advantage but avoid all these limitations. Specifically, our proposed LATE estimators are (1) consistent even under misspecification of regression models, (2) consistent even when the probabilities of treatment assignment are heterogeneous across strata, and (3) guaranteed to be weakly more efficient than the unadjusted estimator. We also provide consistent estimators of the asymptotic variances for our LATE estimators.




\section{The General Estimator and its Asymptotic Properties}

\label{sec:estimation}
In this section, we propose a general regression-adjusted LATE estimator for
$\tau$. Define $\mu^{D}(a,s,x) := \mathbb{E}\sbr[1]{D(a)|S = s, X=x}$ and $\mu
^{Y}(a,s,x) := \mathbb{E}\sbr[1]{Y(D(a))|S=s, X=x}$ for $a=0,1$ as the true
specifications. In practice, these are unknown and empirical
researchers employ working models $\overline{\mu}^{D}(a,s,x)$ and
$\overline{\mu}^{Y}(a,s,x)$, which may differ from the true specifications. We then proceed to estimate the working models
with estimators $\hat{\mu}^{D}(a,s,x)$ and $\hat{\mu}^{Y}(a,s,x)$. As
the working models are potentially misspecified, their estimators are
potentially inconsistent for the true specifications. 

To further differentiate $\mu^b(\cdot)$, $\overline{\mu}^b(\cdot)$, and $\hat \mu^b(\cdot)$ for $b \in \{D,Y\}$, we consider an example that $\mu^D(a,s,x)$ follows a probit model, i.e., $\mu^D(a,s,x) = F_N(\tilde \alpha_{a,s}+x^\top\tilde \beta_{a,s})$, where $F_N(\cdot)$ is the standard normal CDF, and $\tilde \alpha_{a,s}$ and $\tilde \beta_{a,s}$ are the regression coefficients which are allowed to depend on assignment $a$ and stratum $s$. However, the researcher does not know the correct specification and instead uses a logit model $\overline{\mu}^D(a,s,x) = \lambda(\alpha_{a,s}+x^\top \beta_{a,s})$  as the working model, where $\lambda(\cdot)$ is the logistic CDF. Then $(\alpha_{a,s},\beta_{a,s})$ are the pseudo true values that depend on how they are estimated and can be defined as the probability limits of the chosen estimator $(\hat \alpha_{a,s},\hat \beta_{a,s})$. For instance, we can estimate the regression coefficients in the logistic model via logistic quasi MLE or nonlinear least squares. As the logistic model is misspecified, the two estimation methods lead to two different pseudo true values.  Suppose we estimate $(\alpha_{a,s},\beta_{a,s})$ by quasi MLE and denote their estimators as $(\hat \alpha_{a,s},\hat \beta_{a,s})$.  The estimator of the working model is then $\hat \mu^D(a,s,x) = \lambda(\hat \alpha_{a,s}+x^\top\hat \beta_{a,s})$.



In CAR, the targeted assignment probability for stratum $s$, $\pi(s)$, is usually known or can be consistently estimated
by $\hat{\pi}(s) := \frac{n_{1}(s)}{n(s)}$. Then our proposed estimator of LATE based on the doubly robust moments\footnote{For reference of doubly robust moments, see
	\cite{RobinsRotnitzkyZhao1994}, \cite{RR95},
	\cite{ScharfsteinRotnitzkyRobins1999}, \cite{RobinsRotnitzkyvanderLaan2000},
	\cite{hiranoimbens2001}, \cite{F07late}, \cite{wooldridge2007},
	\cite{rothefirpo2019} etc; see \cite{sloczynskiwooldridge2018} and
	\cite{seamanvansteelandt2018} for recent reviews.} is%
\begin{align}
	\hat{\tau}  &  := \del[3]{\frac{1}{n}\sum_{i \in [n]}\Xi_{H,i}}^{-1}%
	\del[3]{\frac{1}{n}\sum_{i \in [n]}\Xi_{G,i}}, \quad\text{where}%
	\label{eq:tau0}\\
	\Xi_{H,i}  & :=\frac{A_{i}(D_{i} - \hat{\mu}^{D}(1,S_{i},X_{i}))}{\hat{\pi
		}(S_{i})} - \frac{(1-A_{i})(D_{i}-\hat{\mu}^{D}(0,S_{i},X_{i}))}{1-\hat{\pi
		}(S_{i})} + \hat{\mu}^{D}(1,S_{i},X_{i})-\hat{\mu}^{D}(0,S_{i},X_{i}%
	),\label{eq:phi_H}\\
	\Xi_{G,i}  &  :=\frac{A_{i}(Y_{i} - \hat{\mu}^{Y}(1,S_{i},X_{i}))}{\hat{\pi
		}(S_{i})} - \frac{(1-A_{i})(Y_{i}-\hat{\mu}^{Y}(0,S_{i},X_{i}))}{1-\hat{\pi
		}(S_{i})} + \hat{\mu}^{Y}(1,S_{i},X_{i})-\hat{\mu}^{Y}(0,S_{i},X_{i}%
	).\label{eq:phi_G}%
\end{align}
%


Given the double robustness and the consistency of $\hat \pi(s)$, our estimator $\hat \tau$ is consistent even when the working models $(\hat \mu^D(\cdot),\hat \mu^Y(\cdot))$ are misspecified. Our analysis also takes into account the cross-sectional dependence of
the treatment statuses caused by the randomization and is therefore different
from the double robustness literature that mostly focuses on the observational
data with independent treatment statuses. Furthermore, our general adjusted estimator is numerically invariant to the stratum-specific location shift because 
\begin{align*}
	\sum_{i=1}^{n} \left( \frac{A_{i}}{\hat{\pi}(S_{i})}-1\right)  1\{S_{i}=s\} =
	0 \quad\text{and} \quad\sum_{i=1}^{n} \left( \frac{1-A_{i}}{1-\hat{\pi}%
		(S_{i})}-1\right)  1\{S_{i}=s\} = 0.
\end{align*}
Therefore, using adjustments $\hat \mu^b(a,S_i,X_i)$ and $\hat \mu^b(a,S_i,X_i) - \mathbb{E}(\mu^b(a,S_i,X_i)|S_i)$ for $b \in \{D,Y\}$ are numerically equivalent.

\begin{ass}
	\begin{enumerate}
		[label=(\roman*)]
		
		\item For $a =0,1$ and $s \in\mathcal{S}$, define $I_{a}(s) :=\cbr[1]{i\in [n]: A_i=a, S_i=s}$, 
		\begin{align*}
			\Delta^{Y}(a, s, X_{i})  & := \hat{\mu}^{Y}(a, s, X_{i})-\overline{\mu}^{Y}(a,
			s, X_{i}), \quad\text{and}\\
			\Delta^{D}(a, s, X_{i}) & := \hat{\mu}^{D}(a, s,
			X_{i})-\overline{\mu}^{D}(a, s, X_{i}).
		\end{align*}
		Then, for $a=0,1$, $b=D,Y$, we have
		\begin{align*}
			\max_{s\in\mathcal{S}}%
			\envert[3]{\frac{\sum_{i\in I_1(s)}\Delta^b(a,s,X_i)}{n_1(s)}-\frac{\sum_{i\in I_0(s)}\Delta^b(a,s,X_i)}{n_0(s)}}=o_{p}%
			(n^{-1/2}).
		\end{align*}

		\item For $a =0,1$ and $b = D,Y$, $\frac{1}{n}\sum_{i=1}^{n} (\Delta
		^{b}(a,S_{i},X_{i}))^2 = o_{p}(1)$.
		
		\item Suppose $\max_{a=0,1,s\in\mathcal{S}}\mathbb{E}([\overline{\mu}%
		^{b}(a,S_{i},X_{i})]^2|S_{i}=s) \leq C<\infty$ for $b = D,Y$ and some constant
		$C$.
	\end{enumerate}
	
	\label{ass:Delta}
\end{ass}

Assumption \ref{ass:Delta} requires $\hat \mu^b(\cdot)$ to be a consistent estimator of $\overline{\mu}^b(\cdot)$ for $b=D,Y$. For instance, we can consider a linear working model
$\overline{\mu}^{Y}(a,s,X_{i})=X_{i}^{\top}\beta_{a,s}$, where the pseudo true value 
$\beta_{a,s}$ is defined as the probability limit of the OLS estimator $\hat \beta_{a,s}$ from regressing $Y_i$ on $X_i$ using observations in $I_a(s)$. Then, the estimator $\hat{\mu}^{Y}(a,s,X_{i})$ can be written as $X_{i}^{\top}\hat{\beta}_{a,s},$
and Assumption
\ref{ass:Delta}(i) reduces to 
\begin{align}
	\max_{s \in \mathcal{S}, a=0,1} \envert[3]{ \del[3]{\frac{1}{n_1(s)}\sum_{i \in I_1(s)}X_i - \frac{1}{n_0(s)}\sum_{i \in I_0(s)}X_i}^\top (\hat{\beta}_{a,s} - \beta_{a,s}) } = o_p(n^{-1/2}),
	\label{eq:ex}
\end{align}
which holds automatically because by definition, $\hat{\beta}_{a,s}\overset{p}{\longrightarrow}\beta
_{a,s}$, and we will assume $\mathbb EX_i^2 <\infty$. This example shows that we do not need to assume the working model $\overline{\mu}^{Y}(a,s,X_{i})=X_{i}^{\top}\beta_{a,s}$ is correctly specified. A similar remark applies to Assumption \ref{ass:Delta}(ii) and nonlinear working models such as the logistic regression mentioned earlier. We verify Assumption \ref{ass:Delta} for general parametric adjustments in Section \ref{sec:par} below. 


To state our first main result below, we need to introduce extra notation. Let $\mathcal{D}_{i}: = \{Y_{i}(1), Y_{i}(0), D_{i}(1), D_{i}(0), X_{i}\}$, $W_{i}  := Y_{i}(D_{i}(1)),$ $Z_{i} :=Y_{i}(D_{i}(0))$, $\tilde{W}_{i} :=W_{i}-\mathbb{E}[W_{i}|S_{i}],$ $\tilde{Z}_{i}:=Z_{i}-\mathbb{E}[Z_{i}|S_{i}],$ $\tilde{X}_{i}:=X_{i}-\mathbb{E}[X_{i}|S_{i}]$, $\tilde{D}_{i}(a) :=D_{i}(a)-\mathbb{E}[D_{i}(a)|S_{i}]$ for $a=0,1,$ and 
\begin{align}
	\tilde{\mu}^{b}(a, S_{i}, X_{i}) &  :=\overline{\mu}^{b}(a, S_{i},
	X_{i})-\mathbb{E}\sbr[1]{\overline{\mu}^Y(a, S_i, X_i)|S_i}, \quad b \in \{D,Y\}. \label{eq:mutilde}%
\end{align}

\begin{thm}
	\begin{enumerate}
		[label=(\roman*)]
		
		\item Suppose Assumptions \ref{ass:assignment1} and \ref{ass:Delta} hold,
		then
		\begin{align}
			\sqrt{n}(\hat{\tau} - \tau) \rightsquigarrow\mathcal{N}(0,\sigma^{2}),
			\quad \text{where} \quad \sigma^{2} := \frac{\sigma_{1}^{2} + \sigma_{0}^{2} + \sigma_{2}^{2}%
			}{\mathbb{P}(D(1)>D(0))^{2}},\label{eq:sigma}%
		\end{align} 
		\begin{align*}
			\sigma_{1}^{2} := \mathbb{E}\sbr[1]{\pi(S_i)\Xi_1^2(\mathcal{D}_i,S_i)},
			\quad\sigma_{0}^{2} := \mathbb{E}%
			\sbr[1]{(1-\pi(S_i))\Xi_0^2(\mathcal{D}_i,S_i)}, \quad\sigma_{2}^{2} :=
			\mathbb{E}\sbr[1]{\Xi_{2}^2(S_i)},
		\end{align*}
		and $\Xi_{1}(\mathcal{D}_{i},S_{i})$, $\Xi_{0}(\mathcal{D}_{i},S_{i})$, and
		$\Xi_{2}(S_{i})$ are defined as
		\begin{align}
			& \Xi_{1}(\mathcal{D}_{i}, S_{i})  := \left[ \left( 1- \frac{1}{\pi(S_{i})}
			\right) \tilde{\mu}^{Y}(1,S_{i},X_{i}) - \tilde{\mu}^{Y}(0,S_{i},X_{i}) +
			\frac{\tilde{W}_{i}}{\pi(S_{i})}\right] \nonumber\\
			& \qquad- \tau\left[ \left( 1- \frac{1}{\pi(S_{i})} \right) \tilde{\mu}%
			^{D}(1,S_{i},X_{i}) - \tilde{\mu}^{D}(0,S_{i},X_{i}) + \frac{ \tilde{D}%
				_{i}(1)}{\pi(S_{i})}\right] ,\label{eq:Xi1}\\
			& \Xi_{0}(\mathcal{D}_{i}, S_{i})  := \left[ \left( \frac{1}{1-\pi(S_{i})}-1
			\right) \tilde{\mu}^{Y}(0,S_{i},X_{i}) + \tilde{\mu}^{Y}(1,S_{i},X_{i}) -
			\frac{\tilde{Z}_{i}}{1-\pi(S_{i})}\right] \nonumber\\
			& \qquad- \tau\left[ \left( \frac{1}{1-\pi(S_{i})}-1 \right) \tilde{\mu}%
			^{D}(0,S_{i},X_{i}) + \tilde{\mu}^{D}(1,S_{i},X_{i}) - \frac{\tilde{D}_{i}%
				(0)}{1-\pi(S_{i}) }\right] ,\label{eq:Xi0}\\
			& \Xi_{2}(S_{i})  := \del[1]{\mathbb{E}[W_i-Z_i|S_i]-\mathbb{E}[W_i-Z_i]} -
			\tau\left( \mathbb{E}[D_{i}(1)-D_{i}(0)|S_{i}]-\mathbb{E}[D_{i}(1)-D_{i}(0)]
			\right) .\label{eq:Xi2}%
		\end{align}
		
		\item Next, we define $\hat{\sigma}^{2}$ as
		\[
		\hat{\sigma}^{2} = \frac{\frac{1}{n}\sum_{i=1}^{n}\left[ A_{i}\hat{\Xi}%
			_{1}^{2}(\mathcal{D}_{i},S_{i}) + (1-A_{i})\hat{\Xi}_{0}^{2}(\mathcal{D}%
			_{i},S_{i})+\hat{\Xi}_{2}^{2}(S_{i})\right] }{\left( \frac{1}{n}\sum_{i=1}^{n}
			\Xi_{H,i}\right) ^{2}}, \quad \text{where}
		\]
		\begin{align*}
			\hat{\Xi}_{1}(\mathcal{D}_{i},s)  &  := \tilde{\Xi}_{1}(\mathcal{D}_{i},s) -
			\frac{1}{n_{1}(s)}\sum_{j \in I_{1}(s)} \tilde{\Xi}_{1}(\mathcal{D}_{j},s),\\
			\hat{\Xi}_{0}(\mathcal{D}_{i},s)  &  := \tilde{\Xi}_{0}(\mathcal{D}_{i},s) -
			\frac{1}{n_{0}(s)}\sum_{j \in I_{0}(s)} \tilde{\Xi}_{0}(\mathcal{D}_{j},s),\\
			\hat{\Xi}_{2}(s)  &  :=
			\del[3]{\frac{1}{n_1(s)}\sum_{i \in I_1(s)}(Y_i - \hat{\tau} D_i)} -
			\del[3]{\frac{1}{n_0(s)}\sum_{i \in I_0(s)}(Y_i - \hat{\tau} D_i)},\\
			\tilde{\Xi}_{1}(\mathcal{D}_{i},s)  &  :=
			\sbr[3]{\del[3]{1- \frac{1}{\hat{\pi}(s)} }\hat{\mu}^Y(1,s,X_i) - \hat{\mu}^Y(0,s,X_i) + \frac{Y_i}{\hat{\pi}(s)}}\\
			&  \qquad- \hat{\tau}
			\sbr[3]{\del[3]{1- \frac{1}{\hat{\pi}(s)} }\hat{\mu}^D(1,s,X_i) - \hat{\mu}^D(0,s,X_i) + \frac{D_i}{\hat{\pi}(s)}},
			\quad\text{and}\\
			\tilde{\Xi}_{0}(\mathcal{D}_{i}, s)  &  :=
			\sbr[3]{\del[3]{\frac{1}{1-\hat{\pi}(s)}-1 }\hat{\mu}^Y(0,s,X_i) + \hat{\mu}^Y(1,s,X_i) - \frac{Y_i}{1-\hat{\pi}(s)}}\\
			& \qquad- \hat{\tau}%
			\sbr[3]{\del[3]{\frac{1}{1-\hat{\pi}(s)}-1}\hat{\mu}^D(0,s,X_i) + \hat{\mu}^D(1,s,X_i) - \frac{D_i}{1-\hat{\pi}(s) }}.
		\end{align*}
		Then, we have $\hat \sigma^2 \convP \sigma^2.$
		
		\item If the working models are correctly specified, i.e.,
		$\overline{\mu}^{b}(a,s,x) = \mu^{b}(a,s,x)$ for all $(a,b,s,x) \in\{0,1\}
		\times\{D,Y\} \times\mathcal{SX}$, where $\mathcal{SX}$ is the joint support
		of $(S,X)$, then the asymptotic variance $\sigma^{2}$ achieves the
		SEB.
	\end{enumerate}
	
	\label{thm:est}
\end{thm}

Several remarks are in order. First, Theorem \ref{thm:est}(i) establishes the limiting distribution of our adjusted LATE estimator, which also implies its consistency. Our estimator inherits the advantage of the TSLS estimator because it remains consistent even when the adjustment $\overline{\mu}^b(\cdot)$ is misspecified, but avoids its limitation because our estimator remains consistent when $\pi(s)$ varies across strata. Additionally, the terms $\sigma_0^2$, $\sigma_1^2$, and $\sigma_2^2$ in the asymptotic variance of our regression-adjusted LATE estimator represent the sampling variations from the control units within each stratum, the treatment units within each stratum, and the strata itself, respectively.

Second, Theorem \ref{thm:est}(ii) gives a consistent estimator of this asymptotic variance, which depends on the working model $\overline{\mu}^{b}(a,s,x)$ for
$(a,b)\in\{0,1\}\times\{D,Y\}$. Different working models lead to different estimation efficiencies.  

Third, Theorem \ref{thm:est}(iii) further shows that
our general regression-adjusted estimator achieves the semiparametric
efficiency bound $\underline{\sigma}^{2}$ derived in Theorem \ref{thm:eff} below 
when the working models are correctly specified.

Fourth, when there are no adjustments so that $\overline{\mu}^{Y}(\cdot)$ and
$\overline{\mu}^{D}(\cdot)$ are zero, we obtain
\begin{align*}
	\sigma^{2} = \frac{\sum_{s \in S} \frac{p(s)}{\pi(s)}Var(W - \tau D(1)|S=s) +
		\sum_{s \in S}\frac{p(s)}{1-\pi(s)}Var(Z - \tau D(0)|S=s) + \sigma_{2}^{2}%
	}{\mathbb{P}(D(1)>D(0))^{2}}.
\end{align*}
In this case, our estimator coincides numerically with \citeauthor{BG21}'s
(\citeyear{BG21}) fully saturated estimator (i.e., NA). Indeed, we can verify that $\sigma^{2}$ defined above is the same as the
asymptotic variance of the fully saturated estimator derived by \cite{BG21}.

\section{Semiparametric Efficiency Bound}
\label{sec:bound}

\begin{thm} \label{thm:eff}
	Suppose that Assumption \ref{ass:assignment1} and the regularity conditions in Assumption \ref{ass:E} in the Online
	Supplement hold. For $a=0,1$, define $\underline{\Xi}_{1}(\mathcal{D}%
	_{i},S_{i})$, $\underline{\Xi}_{0}(\mathcal{D}_{i},S_{i})$ and $\underline{\Xi
	}_{2}(S_{i})$ as $\Xi_{1}(\mathcal{D}_{i},S_{i})$, $\Xi_{0}(\mathcal{D}%
	_{i},S_{i})$ and $\Xi_{2}(S_{i})$ in \eqref{eq:Xi1}--\eqref{eq:Xi2},
	respectively, with the researcher-specified working model $\overline{\mu}%
	^{b}(a,s,x)$ equal to the true specification $\mu^{b}(a,s,x)$ for all
	$(a,b,s,x) \in\{0,1\} \times\{D,Y\} \times\mathcal{SX}$, where $\mathcal{SX}$
	is the joint support of $(S,X)$. Then the SEB for
	$\tau$ is $\underline{\sigma}^{2} := \frac{\underline{\sigma}_{1}^{2} + \underline{\sigma
		}_{0}^{2} + \underline{\sigma}_{2}^{2}}{\mathbb{P}(D(1)>D(0))^{2}}$, where $\underline{\sigma}_{1}^{2} := \mathbb{E}%
	\sbr[1]{\pi(S_i)\underline{\Xi}_1^2(\mathcal{D}_i,S_i)}$, $\underline{\sigma}_{0}^{2} := \mathbb{E}%
	\sbr[1]{(1-\pi(S_i))\underline{\Xi}_0^2(\mathcal{D}_i,S_i)}$, and $\underline{\sigma}_{2}^{2} := \mathbb{E}\underline{\Xi}_{2}^{2}(S_{i})$. 
\end{thm}

Several remarks are in order. First, Theorem \ref{thm:eff} suggests that the asymptotic variance of any regular root-$n$ consistent and asymptotically normal semiparametric estimator of LATE is bounded from below
by $\underline{\sigma}^{2}$. Second, the proof of Theorem \ref{thm:eff}
follows the arguments of \cite{A22}, who accounted for the cross-sectional dependence of
$\{A_{i}\}_{i\in\lbrack n]}$. Third, the
efficiency bound here differs slightly from the one derived by \cite{F07late} under unconfoundedness for observational data because here the covariates $X_{i}$ only affect the conditional mean models (i.e.,
$\mu^{b}(a,s,x)$ for $a=0,1$, $b=\{D,Y\}$) but not the \textquotedblleft
propensity score" $\pi(\cdot)$. Fourth, Theorem \ref{thm:eff} implies that
various CARs (with or without achieving strong balance) lead to the same SEB for LATE estimation. Such a result is consistent with what \cite{A22}
found for ATE under general randomization schemes.



\section{Specific Adjustment Frameworks}

\subsection{Parametric Working Model}
\label{sec:par} In this section, we consider estimating $\overline{\mu}%
^{b}(a,s,x)$ for $a=0,1$, $s \in\mathcal{S}$, and $b = D,Y$ via parametric
regressions. Note that we do not require $\overline{\mu}^{b}(a,s,x)$ to be
correctly specified. Suppose that
\begin{align}
	\overline{\mu}^{Y}(a,S_{i},X_{i}) = \sum_{s \in\mathcal{S}} 1\{S_{i}=s\}
	\Lambda_{a,s}^{Y}(X_{i},\theta_{a,s})\quad\text{and} \quad\overline{\mu}%
	^{D}(a,S_{i},X_{i}) = \sum_{s \in\mathcal{S}}1\{S_{i}=s\} \Lambda_{a,s}%
	^{D}(X_{i},\beta_{a,s}),\label{eq:overlinemu_par}%
\end{align}
where $\Lambda_{a,s}^{b}(\cdot)$ for $(a,b,s) \in\{0,1\} \times\{D,Y\}\times
\mathcal{S}$ is a known function of $X_{i}$ up to some finite-dimensional
parameter (i.e., $\theta_{a,s}$ and $\beta_{a,s}$). The researchers have the
freedom to choose the functional forms of $\Lambda_{a,s}^{b}(\cdot)$, the
parameter values of $(\theta_{a,s},\beta_{a,s})$, and the methods of estimation. As mentioned above, because the parametric models are potentially misspecified,
different estimation methods of the same model can lead to distinctive pseudo
true values. We will discuss several detailed examples in Sections
\ref{sec:linear}, \ref{sec:OLS_MLE}, and \ref{sec:imp} below. Here, we first
focus on the general setup.

Define the estimators of $(\theta_{a,s},\beta_{a,s})$ as $(\hat{\theta}%
_{a,s},\hat{\beta}_{a,s})$, and hence the corresponding feasible parametric
regression adjustments as
\begin{align}
	\hat{\mu}^{Y}(a,s,X_{i}) = \Lambda_{a,s}^{Y}(X_{i},\hat{\theta}_{a,s})
	\quad\text{and} \quad\hat{\mu}^{D}(a,s,X_{i}) = \Lambda_{a,s}^{D}(X_{i}%
	,\hat{\beta}_{a,s}).\label{eq:hatmu_par}%
\end{align}

\bigskip

\begin{ass}
	\begin{enumerate}
		[label=(\roman*)]
		
		\item Suppose that $\max_{a=0,1,s\in\mathcal{S}}||\hat{\theta}_{a,s}%
		-\theta_{a,s}||_{2}\overset{p}{\longrightarrow}0$ and $\max_{a=0,1,s\in
			\mathcal{S}}||\hat{\beta}_{a,s}-\beta_{a,s}||_{2}\overset{p}{\longrightarrow
		}0$, where $\|\cdot\|_2$ is the Euclidean norm.
		
		\item There exist a positive random variable $L_{i}$ and a positive constant
		$C>0$ such that for all $a = 0,1$ and $s \in\mathcal{S}$,
		\begin{align*}
			&
			\enVert[3]{\frac{\partial\Lambda_{a,s}^Y(X_i,\theta_{a,s})}{\partial \theta_{a,s}}}_{2}
			\leq L_{i}, \quad||\Lambda_{a,s}^{Y}(X_{i},\theta_{a,s})||_{2} \leq L_{i}\\
			&
			\enVert[3]{\frac{\partial \Lambda_{a,s}^D(X_i,\beta_{a,s})}{\partial \beta_{a,s}}}_{2}
			\leq L_{i}, \quad||\Lambda_{a,s}^{D}(X_{i},\beta_{a,s})||_{2} \leq L_{i},
		\end{align*}
		almost surely and $\mathbb{E}(L_{i}^{q}|S_{i}=s) \leq C$ for some $q > 2$.
	\end{enumerate}
	
	\label{ass:par}
\end{ass}

Assumption \ref{ass:par}(i) means that $(\hat{\theta}%
_{a,s},\hat{\beta}_{a,s})$ are consistent estimators for $(\theta_{a,s},\beta_{a,s})$.
Assumption \ref{ass:par}(ii) means that the parametric models are smooth in
their parameters, which is true for many widely used regression models such as
linear, logit, and probit regressions. This restriction can be further relaxed
to allow for non-smoothness under less intuitive entropy conditions.

\bigskip

\begin{thm}
	Suppose Assumption \ref{ass:par} hold. Then $\overline{\mu}^{b}(a,s,X_{i})$
	and $\hat{\mu}^{b}(a,s,X_{i})$ defined in \eqref{eq:overlinemu_par} and
	\eqref{eq:hatmu_par}, respectively, satisfy Assumption \ref{ass:Delta}.
	\label{thm:par}
\end{thm}

\bigskip

Theorem \ref{thm:par} generalizes the intuition in \eqref{eq:ex} and shows that Assumption \ref{ass:Delta} holds for general 
parametric models as long as the parameters are consistently estimated.

\subsubsection{Optimal Linear Adjustments}
\label{sec:linear} 
In this section, we consider working models that are linear in $\Psi_{i,s}$ where $\Psi_{i,s}=\Psi_{s}(X_{i})$ is a function of $X_{i}$ and its functional form can vary across $s\in\mathcal{S}$. Specifically, suppose, for $a=0,1$ and $s\in\mathcal{S}$, that
$\overline{\mu}^{Y}(a,s,X)=\Psi_{i,s}^{\top}t_{a,s}$ and $\overline{\mu}%
^{D}(a,s,X)=\Psi_{i,s}^{\top}b_{a,s}$, where $t_{a,s}$ and $b_{a,s}$ are the regression coefficients whose values are freely chosen by the researchers. The
restriction that the function $\Psi_{s}(\cdot)$ does not depend on $a=0,1$ is
innocuous as, if it does, we can stack them up and denote $\Psi_{i,s}%
=(\Psi_{1,s}^{\top}(X_{i}),\Psi_{0,s}^{\top}(X_{i}))^{\top}$. Similarly, it is
also innocuous to impose that the function $\Psi_{s}(\cdot)$ is the same for
modeling $\overline{\mu}^{Y}(a,s,X)$ and $\overline{\mu}^{D}(a,s,X)$.

Given that all values of $t_{a,s}$ and $b_{a,s}$ lead to consistent estimators of LATE, a natural question to ask is what values give the most precise estimator. Let the asymptotic variance of the adjusted LATE estimator $\hat{\tau}$ be
as $\sigma^{2}$, which depends on $(\overline{\mu}^{Y}(a,s,X) ,\overline{\mu
}^{D}(a,s,X) )$, and thus, $(t_{a,s},b_{a,s})$. Let $\Theta^*$ be the collection of optimal linear coefficients that minimize the
asymptotic variance of $\hat{\tau}$ over all possible $(t_{a,s},b_{a,s})$, i.e., 
\begin{align*}
	\Theta^{*} :=
	\begin{pmatrix}
		& (\theta_{a,s}^{*},\beta_{a,s}^{*})_{a=0,1,s\in\mathcal{S}}:\\
		& (\theta_{a,s}^{*},\beta_{a,s}^{*})_{a=0,1,s\in\mathcal{S}}\in
		\argmin_{(t_{a,s},b_{a,s})_{a=0,1,s\in\mathcal{S}}} \sigma^{2}((t_{a,s}%
		,b_{a,s})_{a=0,1,s\in\mathcal{S}}).
	\end{pmatrix}
\end{align*}

\begin{ass}
	Suppose that $\mathbb{E}(||\Psi_{i,s}||_{2}^{q}|S_{i}=s)\leq C<\infty$ for
	constants $C$ and $q>2$. Denote $\tilde{\Psi}_{i,s} := \Psi_{i,s} -
	\mathbb{E}(\Psi_{i,s}|S_{i}=s)$ for $s \in\mathcal{S}$. Then there exist
	constants $0<c<C<\infty$ such that $c< \lambda_{\min}(\mathbb{E}(\tilde{\Psi}_{i,s}\tilde{\Psi}_{i,s}^{\top}))
	\leq\lambda_{\max}(\mathbb{E}(\tilde{\Psi}_{i,s}\tilde{\Psi}_{i,s}^{\top}))
	\leq C,$ 
	where for a generic symmetric matrix $A$, $\lambda_{\min}(A)$ and
	$\lambda_{\max}(A)$ denote the minimum and maximum eigenvalues of $A$,
	respectively. \label{ass:psi}
\end{ass}

Assumption \ref{ass:psi} requires that the regressor $\Psi_{i,s}$ does not contain a constant term. In fact, \eqref{eq:phi_H} and \eqref{eq:phi_G} imply that
our estimator is numerically invariant to a stratum-specific location shift.
The following theorem characterizes the set of optimal linear coefficients. 
\begin{thm}
	Suppose that Assumptions \ref{ass:assignment1} and \ref{ass:psi} hold. Then, we
	have
	\begin{align*}
		\Theta^{*} =
		\begin{pmatrix}
			& (\theta_{a,s}^{*},\beta_{a,s}^{*})_{a=0,1,s\in\mathcal{S}}:\\
			& \sqrt{\frac{1-\pi(s)}{\pi(s)}} (\theta_{1,s}^{*} -\tau\beta_{1,s}^{*}) +
			\sqrt{\frac{\pi(s)}{1-\pi(s)}}(\theta_{0,s}^{*} -\tau\beta_{0,s}^{*})\\
			& = \sqrt{\frac{1-\pi(s)}{\pi(s)}} (\theta_{1,s}^{L} -\tau\beta_{1,s}^{L}) +
			\sqrt{\frac{\pi(s)}{1-\pi(s)}}(\theta_{0,s}^{L} -\tau\beta_{0,s}^{L}).
		\end{pmatrix}
		, \quad \text{where}
	\end{align*}
	\begin{align}
		&  \theta_{a,s}^{L} = [\mathbb{E}(\tilde{\Psi}_{i,s}\tilde{\Psi}_{i,s}^{\top
		}|S_{i}=s)]^{-1}[\mathbb{E}(\tilde{\Psi}_{i,s}Y_{i}(D_{i}(a))|S_{i}%
		=s)]\nonumber\\
		&  \beta_{a,s}^{L} = [\mathbb{E}(\tilde{\Psi}_{i,s}\tilde{\Psi}_{i,s}^{\top
		}|S_{i}=s)]^{-1}[\mathbb{E}(\tilde{\Psi}_{i,s}D_{i}(a)|S_{i}%
		=s)].\label{eq:optimallinear}%
	\end{align}
	\label{thm:linear}
\end{thm}

The optimality result in Theorem \ref{thm:linear} relies on two key
restrictions: (1) the regressor $\Psi_{i,s}$ is the same for treated and
control units and (2) both the adjustments $\overline{\mu}^{Y}(a,s,X)$ and
$\overline{\mu}^{D}(a,s,X)$ are linear. It is possible to have nonlinear adjustments
that are more efficient. We will come back to this point in Sections
\ref{sec:OLS_MLE}, \ref{sec:imp}, and \ref{sec:np}.

In view of Theorem \ref{thm:linear}, the optimal linear coefficients are not
unique. In order to achieve the optimality, we only need to consistently
estimate one point in $\Theta^{*}$. For the rest of the section, we choose
$(\theta_{a,s}^{L},\beta_{a,s}^{L})$ with the corresponding optimal linear
adjustments
\begin{align}
	\overline{\mu}^{Y}(a,s,X_{i}) = \Psi_{i,s}^{\top}\theta_{a,s}^{L}
	\quad\text{and} \quad\overline{\mu}^{D}(a,s,X_{i}) = \Psi_{i,s}^{\top}%
	\beta_{a,s}^{L}.\label{eq:muoverline_LP}%
\end{align}
We estimate $(\theta_{a,s}^{L},\beta_{a,s}^{L})$ by $(\hat{\theta}%
_{a,s}^{L},\hat{\beta}_{a,s}^{L})$, where
\begin{align}
	\dot{\Psi}_{i,a,s}  &  := \Psi_{i,s} - \frac{1}{n_{a}(s)}\sum_{j\in I_{a}(s)}\Psi_{j,s}\nonumber\\
	\hat{\theta}_{a,s}^{L}  &  :=
	\del[3]{\frac{1}{n_a(s)}\sum_{i \in I_a(s)}\dot{\Psi}_{i,a,s}\dot{\Psi}_{i,a,s}^\top }^{-1}%
	\del[3]{\frac{1}{n_a(s)}\sum_{i \in I_a(s)}\dot{\Psi}_{i,a,s}Y_i }\nonumber\\
	\hat{\beta}_{a,s}^{L}  &  :=
	\del[3]{\frac{1}{n_a(s)}\sum_{i \in I_a(s)}\dot{\Psi}_{i,a,s}\dot{\Psi}_{i,a,s}^\top }^{-1}%
	\del[3]{\frac{1}{n_a(s)}\sum_{i \in I_a(s)}\dot{\Psi}_{i,a,s}D_i}.\label{rr1}%
\end{align}
Then, the feasible linear adjustments can be defined as
\begin{align}
	\hat{\mu}^{Y}(a,s,X_{i}) = \Psi_{i,s}^{\top}\hat{\theta}_{a,s}^{L}
	\quad\text{and} \quad\hat{\mu}^{D}(a,s,X_{i}) = \Psi_{i,s}^{\top}\hat{\beta
	}_{a,s}^{L}.\label{eq:muhat_LP}%
\end{align}

Suppose that $\mathcal{S}=\{1,\ldots,S\}$ for some integer $S>0$. It is clear
that $\hat{\theta}_{a,s}^{L}$ and $\hat{\beta}_{a,s}^{L}$ are the OLS-estimated
slopes of the following two linear regressions using observations in $I_{a}(s)$:
\begin{align}
	&  Y_{i} \sim \gamma_{a,s}^{Y}+\Psi_{i,s}^{\top}\theta
	_{a,s} \quad \text{and} \quad D_{i}\sim \gamma_{a,s}^{D}+\Psi_{i,s}^{\top}\beta
	_{a,s}.\label{eq:optimal_Y}
\end{align}


\begin{thm}
	Suppose that Assumptions \ref{ass:assignment1} and \ref{ass:psi} hold. Then,
	\[
	\{\overline{\mu}^{b}(a,s,X_{i})\}_{b=D,Y,a=0,1,s\in\mathcal{S}}\quad
	\text{and}\quad\{\hat{\mu}^{b}(a,s,X_{i})\}_{b=D,Y,a=0,1,s\in\mathcal{S}}%
	\]
	defined in \eqref{eq:muoverline_LP} and \eqref{eq:muhat_LP}, respectively,
	satisfy Assumption \ref{ass:Delta}. Denote the adjusted LATE estimator with
	adjustment $\{\overline{\mu}^{b}(a,s,X_{i})\}_{b=D,Y,a=0,1,s\in\mathcal{S}}$
	defined in \eqref{eq:muhat_LP} as $\hat{\tau}_{L}$. Then, all the results in
	Theorem \ref{thm:est}(i)-(ii) hold for $\hat{\tau}_{L}$. In addition, $\hat{\tau}_{L}$ is the most efficient among all linearly adjusted LATE estimators, and
	in particular, weakly more efficient than the LATE estimator with no
	adjustments. In the special case that $\pi(s)$ is homogeneous across strata and $\Psi_{i,s} = X_i$ so that the TSLS estimator $\hat \tau_{TSLS}$ is consistent, $\hat \tau_{L}$ is also weakly more efficient than $\hat \tau_{TSLS}$.  
	\label{thm:linear2}
\end{thm}

The asymptotic variance of the LATE estimator with the optimal linear
adjustments ($\hat{\tau}_{L}$) takes the form of \eqref{eq:sigma} with
$\{\overline{\mu}^{b}(a,s,X_{i})\}_{b = D,Y, a=0,1, s \in\mathcal{S}}$ in
\eqref{eq:Xi1}--\eqref{eq:Xi2} defined in \eqref{eq:muoverline_LP}. It is also guaranteed to be weakly smaller than that of both $\hat \tau_{NA}$ and $\hat \tau_{TSLS}$, which addresses the Freedman's critique (\citeauthor{F08b}, \citeyear{F08b}, \citeyear{F081}). When
$\Psi_{i,s} = X_{i}$, this asymptotic variance is the same as that of \cite{anseletal2018}'s  S estimator, derived in Section \ref{sec:s} of the Online Supplement. This implies the S estimator is the most efficient LATE estimator adjusted by linear
functions of $X_{i}$, and thus, more efficient than $\hat{\tau}_{TSLS}$ and $\hat \tau_{NA}$.

\subsubsection{Linear and Logistic Regressions}

\label{sec:OLS_MLE}

It is also common to consider a linear model for $\overline{\mu}%
^{Y}(a,s,X_{i})$ and a logistic model for $\overline{\mu}^{D}(a,s,X_{i})$,
i.e.,
\[
\overline{\mu}^{Y}(a,s,X_{i})=\mathring{\Psi}_{i,s}^{\top}t_{a,s}%
\quad\text{and}\quad\overline{\mu}^{D}(a,s,X_{i})=\lambda(\mathring{\Psi
}_{i,s}^{\top}b_{a,s}),
\]
where $\mathring{\Psi}_{i,s}=(1,\Psi_{i,s}^{\top})^{\top}$, $\Psi_{i,s}=\Psi_{s}(X_{i})$ and
$\lambda(u)=\exp(u)/(1+\exp(u))$ is the logistic CDF. As the model for
$\overline{\mu}^{D}(a,s,X_{i})$ is non-linear, the optimality result established in the previous section does not apply. We can consider fitting the
linear and logistic models by OLS and (quasi) MLE, respectively, and call this method
the nonlinear (logistic) adjustment. Specifically, define
\begin{align}
	\hat{\mu}^Y(a,s,X_i) = \mathring{\Psi}_{i,s}^{\top} \hat{\theta}_{a,s}^{OLS}     \quad \text{and} \quad \hat{\mu}^D(a,s,X_i) = \lambda(\mathring{\Psi}_{i,s}^{\top} \hat{\beta}_{a,s}^{MLE}),
	\label{eq:hatmu_OLS_MLE}
\end{align}
where
\begin{align}
	&  \hat{\theta}_{a,s}^{OLS}%
	=\del[3]{\frac{1}{n_a(s)} \sum_{i \in I_a(s)} \mathring{\Psi}_{i,s} \mathring{\Psi}_{i,s}^{\top}}^{-1}%
	\del[3]{\frac{1}{n_a(s)} \sum_{i \in I_a(s)} \mathring{\Psi}_{i,s} Y_i}\quad
	\text{and}\nonumber\\
	&  \hat{\beta}_{a,s}^{MLE}=\argmax_{b}\frac{1}{n_{a}(s)}\sum_{i\in I_{a}%
		(s)}\left[  D_{i}\log(\lambda(\mathring{\Psi}_{i,s}^{\top}b))+(1-D_{i}%
	)\log(1-\lambda(\mathring{\Psi}_{i,s}^{\top}b))\right]  .\label{eq:nonlinear (logistic)}%
\end{align}

It is clear that $\hat{\theta}_{a,s}^{OLS}$ and $\hat{\beta}_{a,s}^{MLE}$ are
the OLS and ML estimates of the following two stratum-specific (logistic)
regressions using observations in $I_{a}(s)$:
\begin{align}
	&  Y_{i} \sim \mathring{\Psi}_{i,s}^{\top}\theta_{a,s} \quad \text{and} \quad D_{i} \sim \lambda^{-1}(\mathring{\Psi}_{i,s}^{\top}\beta_{a,s}). 
\end{align}


In the logistic regression, we do allow the regressor $\mathring{\Psi}_{i,s}$
to contain the constant term. Suppose $\hat{\theta}_{a,s}^{OLS}=(\hat{h}%
_{a,s}^{OLS},\hat{\underline{\theta}}_{a,s}^{OLS,\top})^{\top}$, where
$\hat{h}_{a,s}^{OLS}$ is the intercept. Then, because our adjusted LATE estimator is
invariant to the stratum-specific location shift of the adjustment term, using
$\hat{\mu}^{Y}(a,s,X_{i})=\mathring{\Psi}_{i,s}^{\top}\hat{\theta}_{a,s}%
^{OLS}=\hat{h}_{a,s}^{OLS}+\Psi_{i,s}^{\top}\hat{\underline{\theta}}%
_{a,s}^{OLS}$ and $\hat{\mu}^{Y}(a,s,X_{i})=\Psi_{i,s}^{\top}\hat
{\underline{\theta}}_{a,s}^{OLS}$ produce the exact same LATE estimator. In
addition,  we have $\hat
{\underline{\theta}}_{a,s}^{OLS}=\hat{\theta}_{a,s}^{L}$ by construction. This means $\hat{\mu}^{Y}(a,s,X_{i})$ used here is the same as that for the optimal linear adjustment. In contrast,
because the logistic regression is nonlinear, the non-intercept part of
$\hat{\beta}_{a,s}^{MLE}$ does not equal $\hat{\beta}_{a,s}^{L}$. The limits of $\hat{\theta}_{a,s}^{OLS}$ and $\hat{\beta}_{a,s}^{MLE}$ are
defined as
\begin{align*}
	&  \theta_{a,s}^{OLS} = \left(  \mathbb{E}(\mathring{\Psi}_{i,s}
	\mathring{\Psi}_{i,s}^{\top} |S_{i}=s)\right) ^{-1}\left( \mathbb{E}%
	(\mathring{\Psi}_{i,s} Y_{i}(D_{i}(a))|S_{i}=s)\right)  \quad\text{and}\\
	&  \beta_{a,s}^{MLE} = \argmax_{b} \mathbb{E}\left( \left[ D_{i}(a)
	\log(\lambda(\mathring{\Psi}_{i,s}^{\top} b)) + (1-D_{i}(a)) \log
	(1-\lambda(\mathring{\Psi}_{i,s}^{\top} b)) \right] |S_{i}=s\right) ,
\end{align*}
which imply that the working models are 
\begin{align}
	\overline{\mu}^{Y}(a,s,X_{i}) = \mathring{\Psi}_{i,s}^{\top} \theta
	_{a,s}^{OLS} \quad\text{and} \quad\overline{\mu}^{D}(a,s,X_{i}) =
	\lambda(\mathring{\Psi}_{i,s}^{\top} \beta_{a,s}^{MLE}%
	).\label{eq:overlinemu_OLS_MLE}%
\end{align}

\begin{ass}
	\begin{enumerate}
		[label=(\roman*)]
		
		\item For $a=0,1$ and $s\in\mathcal{S}$, suppose $\mathbb{E}(\mathring{\Psi
		}_{i,s} \mathring{\Psi}_{i,s}^{\top}|S_{i}=s)$ is invertible and
		\[
		\mathbb{E}\left( \left[ D_{i}(a) \log(\lambda(\mathring{\Psi}_{i,s}^{\top} b))
		+ (1-D_{i}(a)) \log(1-\lambda(\mathring{\Psi}_{i,s}^{\top} b)) \right]
		|S_{i}=s\right)
		\]
		has $\beta_{a,s}^{MLE}$ as its unique maximizer.
		
		\item There exists a constant $C<\infty$ such that $\max_{a = 0,1,s
			\in\mathcal{S}}\mathbb{E}||\mathring{\Psi}_{i,s}||_{2}^{q} \leq C<\infty$ for
		some $q> 2$.
	\end{enumerate}
	
	\label{ass:OLS_MLE}
\end{ass}

\begin{thm} \label{thm:OLS_MLE}
	Suppose Assumptions \ref{ass:assignment1} and \ref{ass:OLS_MLE} hold. Then,
	\[
	\{\overline{\mu}^{b}(a,s,X_{i})\}_{b = D,Y, a=0,1, s \in\mathcal{S}}
	\quad\text{and} \quad\{\hat{\mu}^{b}(a,s,X_{i})\}_{b = D,Y, a=0,1, s
		\in\mathcal{S}}%
	\]
	defined in \eqref{eq:overlinemu_OLS_MLE} and \eqref{eq:hatmu_OLS_MLE},
	respectively, satisfy Assumption \ref{ass:Delta}. Denote the adjusted LATE
	estimator with adjustment $\{\hat{\mu}^{b}(a,s,X_{i})\}_{b = D,Y, a=0,1, s
		\in\mathcal{S}}$ defined in \eqref{eq:hatmu_OLS_MLE} as $\hat{\tau}_{NL}$.
	Then, all the results in Theorem \ref{thm:est}(i)-(ii) hold for $\hat{\tau}_{NL}$.
\end{thm}

Several remarks are in order. First, the nonlinear (logistic) adjustment is not optimal in
the sense that it does not necessarily minimize the asymptotic variance of the
corresponding LATE estimator over the class of linear/logistic adjustments. Second, the nonlinear (logistic) adjustment is not
necessarily less efficient than the optimal linear adjustment studied in
Section \ref{sec:linear} as $\mu^{D}(a,s,X_{i})$ could be nonlinear. In fact,
as Theorem \ref{thm:est} shows, if the adjustments are correctly specified,
then $\hat{\tau}_{NL}$ can achieve the semiparametric efficiency
bound. Compared with the linear probability model considered in Section
\ref{sec:linear}, the logistic model is expected to be less misspecified,
especially when the regressor $\Psi_{i,s}$ contains nonlinear transformations of $X_{i}$
such as interactions and quadratic terms. Third, we will further justify the
intuition above in Section \ref{sec:np}, in which we let $\Psi_{i,s}$ be
the sieve basis functions with an increasing dimension and show that the
nonlinear (logistic) method can consistently estimate the correct specification under some
regularity conditions. Fourth, one theoretical shortcoming of the nonlinear (logistic)
adjustment is that, unlike the optimal linear adjustment, it is not guaranteed
to be more efficient than no adjustment. We address this issue in Section
\ref{sec:imp} below.

\subsubsection{Further Efficiency Improvement}
\label{sec:imp} 

Following the lead of \cite{CF21}, we can treat the nonlinear (logistic) adjustments as regressors and obtain the optimal linear coefficients as proposed in Section \ref{sec:linear}. Let $\theta_{a,s}^{OLS} = (h_{a,s}^{OLS},\underline{\theta
}_{a,s}^{OLS})$ be the probability limit of $\hat \theta_{a,s}^{OLS}$ defined in \eqref{eq:nonlinear (logistic)}. If $\beta_{a,s}%
^{MLE}$ were known, the nonlinear (logistic) adjustment can be viewed as a linear
adjustment. Specifically, denote
\begin{align}
	&  \Phi_{i,s} := (\Psi_{i,s}^{\top}, \lambda(\mathring{\Psi}_{i,s}^{\top}
	\beta_{1,s}^{MLE}),\lambda(\mathring{\Psi}_{i,s}^{\top} \beta_{0,s}%
	^{MLE}))^{\top}\label{eq:Psi_new}\\
	&  t_{a,s}^{NL} := a
	\begin{pmatrix}
		\underline{\theta}_{1,s}^{OLS}\\
		0\\
		0
	\end{pmatrix}
	+ (1-a)%
	\begin{pmatrix}
		\underline{\theta}_{0,s}^{OLS}\\
		0\\
		0
	\end{pmatrix}
	, \quad b_{a,s}^{NL} := a
	\begin{pmatrix}
		0_{d_{\Psi}}\\
		1\\
		0
	\end{pmatrix}
	+ (1-a)%
	\begin{pmatrix}
		0_{d_{\Psi}}\\
		0\\
		1
	\end{pmatrix}
	,\nonumber
\end{align}
where $d_{\Psi}$ is the dimension of $\Psi_{i,s}$. Then, the nonlinear (logistic)
adjustment can be written as
\begin{align*}
	\overline{\mu}^{Y}(a,s,X_{i}) = \Phi_{i,s}^{\top}t_{a,s}^{NL} \quad\text{and}
	\quad\overline{\mu}^{D}(a,s,X_{i}) = \Phi_{i,s}^{\top}b_{a,s}^{NL}.
\end{align*}

Similarly, we can replicate no adjustments and the
optimal linear adjustments with $\Phi_{i,s}$ defined in \eqref{eq:Psi_new} as
regressors by letting
\begin{align*}
	\overline{\mu}^{Y}(a,s,X_{i}) = \Phi_{i,s}^{\top}t_{a,s} \quad\text{and}
	\quad\overline{\mu}^{D}(a,s,X_{i}) = \Phi_{i,s}^{\top}b_{a,s}%
\end{align*}
with $(t_{a,s},b_{a,s}) = 0$ and $(t_{a,s},b_{a,s}) = (t_{a,s}^{L}%
,b_{a,s}^{L})$, respectively, where
\begin{align*}
	&  t_{a,s}^{L} := a
	\begin{pmatrix}
		\theta_{1,s}^{L}\\
		0\\
		0
	\end{pmatrix}
	+ (1-a)%
	\begin{pmatrix}
		\theta_{0,s}^{L}\\
		0\\
		0
	\end{pmatrix}
	, \quad b_{a,s}^{L} := a
	\begin{pmatrix}
		\beta_{1,s}^{L}\\
		0\\
		0
	\end{pmatrix}
	+ (1-a)%
	\begin{pmatrix}
		\beta_{0,s}^{L}\\
		0\\
		0
	\end{pmatrix}
	.
\end{align*}


Based on Theorem \ref{thm:linear}, we can further improve all three types of
adjustments by setting the linear coefficients of $\Phi_{i,s}$ as
\begin{align*}
	&  \theta_{a,s}^{F} : =\left(  \mathbb{E} [\tilde{\Phi}_{i,s}\tilde{\Phi
	}_{i,s}^{\top}|S_{i}=s]\right) ^{-1}\left( [ \mathbb{E} \tilde{\Phi}%
	_{i,s}Y_{i}(D_{i}(a))|S_{i}=s]\right) ,\\
	&  \beta_{a,s}^{F} : = \left(  \mathbb{E}[ \tilde{\Phi}_{i,s}\tilde{\Phi
	}_{i,s}^{\top}|S_{i}=s]\right) ^{-1}\left(  [\mathbb{E} \tilde{\Phi}%
	_{i,s}D_{i}(a)|S_{i}=s]\right) ,
\end{align*}
where $\tilde{\Phi}_{i,s} = \Phi_{i,s} - \mathbb{E}(\Phi_{i,s}|S_{i}=s)$. The
final linear adjustments with $\theta_{a,s}^{F}$ and $\beta_{a,s}^{F}$ are
\begin{align}
	\overline{\mu}^{Y}(a,s,X_{i}) = \Phi_{i,s}^{\top}\theta_{a,s}^{F}
	\quad\text{and} \quad\overline{\mu}^{D}(a,s,X_{i}) = \Phi_{i,s}^{\top}%
	\beta_{a,s}^{F}.\label{eq:adj_OLS_MLE_2}%
\end{align}

Because $\beta_{a,s}^{MLE}$ is unknown, we can replace it by its estimate
proposed in Section \ref{sec:OLS_MLE}, i.e., define
\begin{align*}
	\hat{\Phi}_{i,s} := (\Psi_{i,s}, \lambda(\mathring{\Psi}_{i,s}^{\top}
	\hat{\beta}_{1,s}^{MLE}),\lambda(\mathring{\Psi}_{i,s}^{\top} \hat{\beta
	}_{0,s}^{MLE}))^{\top}\quad\text{and} \quad\breve{\Phi}_{i,a,s}: = \hat{\Phi
	}_{i,s} - \frac{1}{n_{a}(s)}\sum_{j\in I_{a}(s)}\hat{\Phi}_{j,s}.
\end{align*}
Then, we define the estimators of $\theta_{a,s}^{F}$ and $\beta_{a,s}^{F}$ as
\begin{align}
	&  \hat{\theta}_{a,s}^{F} : =
	\del[3]{\frac{1}{n_a(s)} \sum_{i \in I_a(s)}\breve{\Phi}_{i,a,s}\breve{\Phi}_{i,a,s}^\top}^{-1}%
	\del[3]{\frac{1}{n_a(s)} \sum_{i \in I_a(s)}\breve{\Phi}_{i,a,s}Y_i},\nonumber\\
	&  \hat{\beta}_{a,s}^{F} : =
	\del[3]{\frac{1}{n_a(s)} \sum_{i \in I_a(s)}\breve{\Phi}_{i,a,s}\breve{\Phi}_{i,a,s}^\top}^{-1}%
	\del[3]{\frac{1}{n_a(s)} \sum_{i \in I_a(s)}\breve{\Phi}_{i,a,s}D_i}.\label{rr2}%
\end{align}
The corresponding feasible adjustments are
\begin{align}
	\hat{\mu}^{Y}(a,s,X_{i}) = \hat{\Phi}_{i,s}^{\top}\hat{\theta}_{a,s}^{F}
	\quad\text{and} \quad\hat{\mu}^{D}(a,s,X_{i}) = \hat{\Phi}_{i,s}^{\top}%
	\hat{\beta}_{a,s}^{F}.\label{eq:adj_OLS_MLE_2_hat}%
\end{align}

\begin{ass}
	Suppose Assumption \ref{ass:psi} holds for $\Phi_{i,s}$ defined in
	\eqref{eq:Psi_new}. \label{ass:psi_new}
\end{ass}
\begin{thm}\label{thm:linear3}
	Suppose that Assumptions \ref{ass:assignment1}, \ref{ass:OLS_MLE}, and
	\ref{ass:psi_new} hold. Then,
	\[
	\{\overline{\mu}^{b}(a,s,X_{i})\}_{b=D,Y,a=0,1,s\in\mathcal{S}}\quad
	\text{and}\quad\{\hat{\mu}^{b}(a,s,X_{i})\}_{b=D,Y,a=0,1,s\in\mathcal{S}}%
	\]
	defined in \eqref{eq:adj_OLS_MLE_2} and \eqref{eq:adj_OLS_MLE_2_hat},
	respectively, satisfy Assumption \ref{ass:Delta}. Denote the LATE estimator
	with regression adjustments $\{\hat{\mu}^{b}(a,s,X_{i})\}_{b=D,Y,a=0,1,s\in
		\mathcal{S}}$ defined in \eqref{eq:adj_OLS_MLE_2_hat} as $\hat{\tau}_{F}$.
	Then, all the results in Theorem \ref{thm:est}(i)-(ii) hold for $\hat{\tau}_{F}$.
	In addition, $\hat{\tau}_{F}$ is weakly more efficient than $\hat{\tau}_{L}$,
	$\hat{\tau}_{NL}$ and $\hat{\tau}_{NA}$.
\end{thm}

Theorem \ref{thm:linear3} shows that by refitting nonlinear (logistic) adjustment in a
linear regression with optimal linear coefficients, we can further improve the
efficiency of the adjusted LATE estimator. As a by-product, $\hat{\tau}_{F}$ is
guaranteed to be weakly more efficient than the LATE estimator without any adjustments ($\hat{\tau}_{NA}$).

\subsection{Nonparametric Adjustments}

\label{sec:np}

In this section, we consider the nonparametric regression as the adjustments
for our LATE estimator. Specifically, we use linear and logistic sieve
regressions to estimate the true specifications $\mu^{Y}(a,s,X_{i})$ and
$\mu^{D}(a,s,X_{i})$, respectively. For implementation, the nonparametric
adjustment is exactly the same as nonlinear (logistic) adjustment studied in Section
\ref{sec:OLS_MLE}. Theoretically, we will let the regressors $\mathring{\Psi
}_{i,s}$ in \eqref{eq:hatmu_OLS_MLE} be sieve basis functions whose dimensions
will diverge to infinity as the sample size increases. For notational
simplicity, we suppress the subscript $s$ and denote the sieve regressors as
$\mathring{\Psi}_{i,n}\in\Re^{h_{n}}$, where the dimension $h_{n}$ can diverge
with the sample size.
The corresponding feasible regression adjustments are
\begin{align}
	\hat{\mu}^Y(a,s,X_i) = \mathring{\Psi}_{i,n}^\top \hat{\theta}_{a,s}^{NP} \quad \text{and} \quad \hat{\mu}^D(a,s,X_i) = \lambda(\mathring{\Psi}_{i,n}^\top \hat{\beta}_{a,s}^{NP}), 
	\label{eq:muhat_np}
\end{align}
where
\begin{align*}
	&  \hat{\theta}_{a,s}^{NP}%
	=\del[3]{\frac{1}{n_a(s)} \sum_{i \in I_a(s)} \mathring{\Psi}_{i,n} \mathring{\Psi}_{i,n}^{\top}}^{-1}%
	\del[3]{\frac{1}{n_a(s)} \sum_{i \in I_a(s)} \mathring{\Psi}_{i,n} Y_i}\quad
	\text{and}\\
	&  \hat{\beta}_{a,s}^{NP}=\argmax_{b}\frac{1}{n_{a}(s)}\sum_{i\in I_{a}%
		(s)}\left[  D_{i}\log(\lambda(\mathring{\Psi}_{i,n}^{\top}b))+(1-D_{i}%
	)\log(1-\lambda(\mathring{\Psi}_{i,n}^{\top}b))\right]  .
\end{align*}
We finally denote the corresponding adjusted LATE estimator as $\hat{\tau}_{NP}$.

\begin{ass} \label{ass:np}
	\begin{enumerate}
		[label=(\roman*)]
		
		\item There exist constants $0<c<C < \infty$ such that with probability
		approaching one,
		\[
		c\leq\lambda_{\min}%
		\del[3]{ \frac{1}{n_{a}(s)}\sum_{i \in I_{a}(s)}\mathring{\Psi}_{i,n}\mathring{\Psi}_{i,n}^\top}\leq
		\lambda_{\max}%
		\del[3]{\frac{1}{n_{a}(s)}\sum_{i \in I_{a}(s)}\mathring{\Psi}_{i,n}\mathring{\Psi}_{i,n}^\top}\leq
		C \quad \text{and}
		\]
		\[
		c\leq\lambda_{\min}\left(  \mathbb{E}[\mathring{\Psi}_{i,n}\mathring{\Psi
		}_{i,n}^{\top}|S_{i}=s]\right) \leq\lambda_{\max}\left(  \mathbb{E}%
		[\mathring{\Psi}_{i,n}\mathring{\Psi}_{i,n}^{\top}|S_{i}=s]\right) \leq C.
		\]

		\item For $a=0,1$, there exist $h_{n} \times1$ vectors $\theta_{a,s}^{NP}$ and
		$\beta_{a,s}^{NP}$ such that for
		\begin{align*}
			R^{Y}(a,s,x)  &  := \mathbb{\mathbb{E}}\sbr[1]{Y_i(D_i(a))|S_i=s,X_i=x} -
			\mathring{\Psi}_{i,n}^{\top}\theta_{a,s}^{NP} \quad\text{and}\\
			R^{D}(a,s,x)  &  := \mathbb{P}\del[1]{D_i(a)=1|S_i=s,X_i=x} - \lambda
			(\mathring{\Psi}_{i,n}^{\top}\beta_{a,s}^{NP}),
		\end{align*}
		we have $\sup_{a = 0,1, b \in\{D,Y\}, s \in\mathcal{S}, x \in\text{Supp}%
			(X)}|R^{b}(a,s,x)| = o_{p}(1)$,
		\begin{align*}
			\sup_{a = 0,1, b \in\{D,Y\}, s \in\mathcal{S}, x \in\text{Supp}(X)}\frac
			{1}{n_{a}(s)}\sum_{i \in I_{a}(s)}\del[1]{R^{b}(a,s,X_i)}^{2} = O_{p}\left(
			\frac{h_{n} \log n}{n}\right) , \quad \text{and}
		\end{align*}
		\begin{align*}
			\sup_{a=0,1,b \in\{D,Y\}, s\in\mathcal{S}} \mathbb{E}%
			\sbr[2]{\del[1]{R^{b}(a,s,X_i)}^2|S_i=s} = O\left( \frac{h_{n} \log n}%
			{n}\right) .
		\end{align*}

		\item For $a=0,1$, there exists a constant $c \in(0,0.5)$ such that
		\begin{align*}
			c\leq &  \inf_{a = 0,1, s \in\mathcal{S}, x \in\text{Supp}(X)}\mathbb{P}%
			\del[1]{D_i(a)=1|S_i=s,X_i=x}\\
			\qquad\leq &  \sup_{a = 0,1, s \in\mathcal{S}, x \in\text{Supp}(X)}%
			\mathbb{P}\del[1]{D_i(a)=1|S_i=s,X_i=x} \leq1-c.
		\end{align*}

		\item Suppose that $\mathbb{E}[\mathring{\Psi}_{i,n,k}^{2}|S_{i}=s]\leq
		C$ for some constant $C>0$, where $\mathring{\Psi}_{i,n,k}$ denotes the
		$k$th element of $\mathring{\Psi}_{i,n}$. $\max_{i\in\lbrack n]}%
		||\mathring{\Psi}_{i,n}||_{2}\leq\zeta(h_{n})$ a.s., where $\zeta(\cdot)$ is a
		deterministic increasing function satisfying $\zeta^{2}(h_{n})h_{n}\log
		n=o(n)$. Also $h_{n}^{2}\log^{2}n=o(n)$. 
	\end{enumerate}
\end{ass}

Assumption \ref{ass:np} is standard for linear and logistic sieve regressions.
We refer to \cite{HIR03} and \cite{c07} for more discussions. The quantity
$\zeta(h_{n})$ in Assumption \ref{ass:np}(iv) depends on the choice of basis
functions. For example, $\zeta(h_{n}) = O(h_{n}^{1/2})$ for splines and
$\zeta(h_{n}) = O(h_{n})$ for power series.

\begin{thm} \label{thm:np}
	Suppose Assumptions \ref{ass:assignment1} and \ref{ass:np} hold. Then
	$\{\hat{\mu}^{b}(a,s,X_{i})\}_{b = D,Y, a=0,1, s \in\mathcal{S}}$ defined in
	\eqref{eq:muhat_np} with $\overline{\mu}^{b}(a,s,X) = \mu^{b}(a,s,X)$ satisfy
	Assumption \ref{ass:Delta}. All the results in Theorem \ref{thm:est}(i)-(ii) hold for $\hat{\tau}_{NP}$. In addition, $\hat{\tau}_{NP}$ achieves the SEB. 
\end{thm}

The nonlinear (logistic) and nonparametric adjustments are numerically identical if the
same set of regressors are used. Theorem \ref{thm:np} then shows that the
nonlinear (logistic) adjustment with technical regressors performs well because it can
closely approximate the correct specification. Under the asymptotic framework
that the dimension of the regressors diverges to infinity and the
approximation error converges to zero, the nonlinear (logistic) adjustment can be viewed as
the nonparametric adjustment, which achieves the SEB. In fact, if we estimate both $\mu^Y(a,s,X)$ and $\mu^D(a,s,X)$ by linear sieve regressions, under similar conditions to Assumption \ref{ass:np}, we can show that such an adjusted estimator also achieves the SEB. So does \citeauthor{anseletal2018}'s (\citeyear{anseletal2018}) S estimator when their $X_i$ is replaced by sieve bases of $X_i$ because it is asymptotically equivalent to our estimator L with optimal linear adjustment.

\subsection{Regularized Large Dimensional Regression}

\label{sec:hd}

In this section, we consider the case where the regressor $\mathring{\Psi}_{i,n}\in\Re^{p_{n}}$ has dimension $p_{n}$ that can be much higher than $n$.
In this case, we can no longer use the nonlinear (logistic) (nonparametric) adjustment
method. Instead, we need to regularize the least squares and logistic
regressions. Specifically, let
\begin{align}
	\hat{\mu}^Y(a,s,X_i) = \mathring{\Psi}_{i,n}^\top \hat{\theta}_{a,s}^{R} \quad \text{and} \quad \hat{\mu}^D(a,s,X_i) = \lambda(\mathring{\Psi}_{i,n}^\top \hat{\beta}_{a,s}^{R}), 
	\label{eq:muhat_hd}
\end{align}
and the corresponding adjusted LATE estimator is denoted as $\hat{\tau}_{R}$,
where
\begin{align*}
	\hat{\theta}_{a,s}^{R}= &  \argmin_{t}\frac{-1}{n_{a}(s)}\sum_{i\in I_{a}%
		(s)}\del[1]{Y_i - \mathring{\Psi}_{i,n}^\top t}^{2}+\frac{\varrho_{n,a}%
		(s)}{n_{a}(s)}||\hat{\Omega}^{Y}t||_{1},\\
	\hat{\beta}_{a,s}^{R}= &  \argmin_{b}\frac{-1}{n_{a}(s)}\sum_{i\in I_{a}%
		(s)}%
	\sbr[2]{D_i\log\del[1]{\lambda(\mathring{\Psi}_{i,n}^\top b)} + (1-D_i)\log(1-\lambda\del[1]{\mathring{\Psi}_{i,n}^\top b)}}+\frac
	{\varrho_{n,a}(s)}{n_{a}(s)}||\hat{\Omega}^{D}b||_{1},
\end{align*}
where $\{\varrho_{n,a}(s)\}_{a=0,1,s\in\mathcal{S}}$ are tuning parameters, $\hat{\Omega}^{b}=\text{diag}(\hat{\omega}_{1}^{b},\cdots,\hat{\omega
}_{p_{n}}^{b})$ is a diagonal matrix of data-dependent penalty loadings for
$b=D,Y$, and $\|\cdot\|_1$ is the $\ell_1$ norm.\footnote{ We provide more details about  $\hat{\Omega}^{b}$ in Section \ref{sec:aux_imp} of the Online Supplement.} 

We maintain the following assumptions for Lasso and logistic Lasso
regressions. 

\begin{ass}
	\begin{enumerate}
		[label=(\roman*)] 
		
		\item For $a=0,1$. Suppose that%
		\begin{align*}
			&  \mathbb{E}\sbr[1]{Y_i(D_i(a))|X_i,S_i=s}=\mathring{\Psi}_{i,n}^{\top}%
			\theta_{a,s}^{R}+R^{Y}(a,s,X_{i})\quad\text{and}\\
			&  \mathbb{P}(D_{i}(a)=1|X_{i},S_{i}=s)=\lambda(\mathring{\Psi}_{i,n}^{\top
			}\beta_{a,s}^{R})+R^{D}(a,s,X_{i})
		\end{align*}
		such that $\max_{a=0,1,s\in\mathcal{S}}\max(||\theta_{a,s}^{R}||_{0}%
		,||\beta_{a,s}^{R}||_{0})\leq h_{n}$, where $||a||_{0}$ denotes the number of
		nonzero components in $a$.
		
		\item Suppose that for $q>2$,
		\begin{align*}
			\sup_{i\in\lbrack n]}||\mathring{\Psi}_{i,n}||_{\infty
			}\leq\zeta_{n}~a.s. \quad \text{and} \quad  \sup_{h\in\lbrack p_{n}]}\mathbb{E}\sbr[1]{|\mathring{\Psi}_{i,n,h}^q||S_i=s}<\infty,   
		\end{align*}
		where $\|\cdot\|_{\infty}$ is the $\ell_{\infty}$ norm. 
		
		\item Suppose that
		\[
		\max_{a=0,1,b=D,Y,s\in\mathcal{S}}\frac{1}{n_{a}(s)}\sum_{i\in I_{a}(s)}%
		(R^{b}(a,s,X_{i}))^{2}=O_{p}(h_{n}\log p_{n}/n),
		\]%
		\[
		\max_{a=0,1,b=D,Y,s\in\mathcal{S}}\mathbb{E}%
		\sbr[1]{(R^b(a,s,X_i))^2|S_i=s}=O(h_{n}\log p_{n}/n),
		\]
		and
		\[
		\sup_{a=0,1,b=D,Y,s\in\mathcal{S},x\in\mathcal{X}}|R^{b}(a,s,X)|=O(\sqrt
		{\zeta_{n}^{2}h_{n}^{2}\log p_{n}/n}).
		\]

		\item Suppose that $\frac{\log(p_{n})\zeta_{n}^{2}h_{n}^{2}}{n}\rightarrow0$
		and $\frac{\log^{2}(p_{n})\log^{2}(n)h_{n}^{2}}{n}\rightarrow0$. 
		
		\item There exists a constant $c \in(0,0.5)$ such that
		\begin{align*}
			c\leq &  \inf_{a = 0,1, s \in\mathcal{S}, x \in\text{Supp}(X)}\mathbb{P}%
			(D_{i}(a)=1|S_{i}=s,X_{i}=x)\\
			\leq &  \sup_{a = 0,1, s \in\mathcal{S}, x \in\text{Supp}(X)}\mathbb{P}%
			(D_{i}(a)=1|S_{i}=s,X_{i}=x) \leq1-c.
		\end{align*}

		\item Let $\ell_{n}$ be a sequence that diverges to infinity. Then there exist
		two constants $\kappa_{1}$ and $\kappa_{2}$ such that with probability
		approaching one,
		\begin{align*}
			0< \kappa_{1} \leq &  \inf_{a = 0,1, s\in\mathcal{S}, ||v||_{0} \leq h_{n}
				\ell_{n}} \frac{v^{\top}\left( \frac{1}{n_{a}(s)}\sum_{i \in I_{a}%
					(s)}\mathring{\Psi}_{i,n}\mathring{\Psi}_{i,n}^{\top}\right) v }{||v||_{2}%
				^{2}}\\
			\leq &  \sup_{a = 0,1, s\in\mathcal{S}, ||v||_{0} \leq h_{n} \ell_{n}}
			\frac{v^{\top}\left( \frac{1}{n_{a}(s)}\sum_{i \in I_{a}(s)}\mathring{\Psi
				}_{i,n}\mathring{\Psi}_{i,n}^{\top}\right) v }{||v||_{2}^{2}} \leq\kappa_{2} <
			\infty,
		\end{align*}
		and
		\begin{align*}
			0< \kappa_{1} \leq &  \inf_{a = 0,1, s \in\mathcal{S}, ||v||_{0} \leq h_{n}
				\ell_{n}} \frac{v^{\top}\mathbb{E}%
				\sbr[1]{\mathring{\Psi}_{i,n}\mathring{\Psi}_{i,n}^\top |S_i=s}v }%
			{||v||_{2}^{2}}\\
			\leq &  \sup_{a = 0,1, s\in\mathcal{S}, ||v||_{0} \leq h_{n} \ell_{n}}
			\frac{v^{\top}\mathbb{E}%
				\sbr[1]{\mathring{\Psi}_{i,n}\mathring{\Psi}_{i,n}^\top |S_i=s}v }%
			{||v||_{2}^{2}} \leq\kappa_{2} < \infty.
		\end{align*}

		\item For $a=0,1$, let $\varrho_{n,a}(s) = c \sqrt{n_{a}(s)}F_{N}^{-1}\del[1]{ 1-1/\sbr[1]{p_n\log(n_a(s))}}$ where $F_{N}(\cdot)$ is the standard normal CDF and $c>0$ is a constant. 
	\end{enumerate}
	
	\label{ass:hd} 
\end{ass}

Assumption \ref{ass:hd} is standard in the literature and we refer interested
readers to \cite{BCFH13} for more discussion.

\begin{thm}
	Suppose Assumptions \ref{ass:assignment1} and \ref{ass:hd} hold. Then
	$\{\hat{\mu}^{b}(a,s,X_{i})\}_{b = D,Y, a=0,1, s \in\mathcal{S}}$ defined in
	\eqref{eq:muhat_hd} and $\overline{\mu}^{b}(a,s,X) = \mu^{b}(a,s,X)$ satisfy
	Assumption \ref{ass:Delta}. All the results in Theorem \ref{thm:est}(i)-(ii) hold for $\hat{\tau}_{R}$. In addition, $\hat{\tau}_{R}$ achieves the SEB. \label{thm:hd}
\end{thm}

Due to the approximate sparsity, the Lasso method consistently estimate
the correct specification, which explains why the corresponding estimator
can achieve the SEB.

\section{Simulations}
\label{sec:sim}

\subsection{Data Generating Processes}

Three data generating processes (DGPs) are used to assess the finite sample
performance of the estimation and inference methods introduced in the paper.
Suppose that
\begin{align*}
	Y_{i}(d) &  = a_{d}+\alpha(X_{i}, Z_{i}) + \varepsilon_{d+1,i},~d=0,1,\quad D_{i}(0)   =1\{b_{0}+ \gamma(X_{i}, Z_{i}) >c_{0}\varepsilon_{3,i}\},\\
	D_{i}(1)  &  =\left\lbrace
	\begin{array}
		[c]{cc}%
		1\{b_{1} + \gamma(X_{i}, Z_{i}) >c_{1}\varepsilon_{4,i}\} & \text{if }
		D_{i}(0)=0,\\
		1 & \text{otherwise},%
	\end{array}
	\right. \\
	D_{i} & =D_{i}(1)A_{i}+D_{i}(0)(1-A_{i}), \quad \text{and} \quad Y_{i} =Y_{i}(1)D_{i}+Y_{i}(0)(1-D_{i}),
\end{align*}
where $\{X_{i}, Z_{i}\}_{i\in[n]}, \alpha(\cdot, \cdot)$, $\{a_{i}, b_{i},
c_{i}\}_{i=0,1}$ and $\{\varepsilon_{j,i}\}_{j\in[4], i\in[n]}$ are
specified as follows.

\begin{enumerate}
	[label=(\roman*)]
	
	\item Let $Z_{i}$ be i.i.d. according to standardized Beta$(2,2)$, $S_{i} =
	\sum_{j = 1}^{4} 1\{Z_{i} \leq g_{j}\}$, and $(g_{1}, g_{2}, g_{3}, g_{4}) =
	(-0.25\sqrt{20}, 0, 0.25\sqrt{20}, 0.5\sqrt{20})$. $X_{i}:=(X_{1,i},
	X_{2,i})^{\top}$, where $X_{1,i}$ follows a uniform distribution on $[-2,2]$,
	$X_{2,i}:=Z_{i}+N(0,1)$, and $X_{1,i}$ and $X_{2,i}$ are independent. Further
	define
	\begin{align*}
		\alpha(X_{i}, Z_{i}) & =0.7X_{1,i}^{2}+X_{2,i}+4Z_{i}, \quad \gamma(X_{i}, Z_{i})   = 0.5X_{1,i}^{2}-0.5X_{2,i}^{2}-0.5Z_{i}^{2},
	\end{align*}
	$a_{1}=2, a_{0}=1, b_{1}=1.3, b_{0}=-1, c_{1}=c_{0}=3$, and $(\varepsilon
	_{1,i}, \varepsilon_{2,i},\varepsilon_{3,i},\varepsilon_{4,i})^{\top
	}\overset{i.i.d}{\sim}N(0, \Sigma)$, where
	\begin{align*}
		\Sigma=
		\begin{pmatrix}
			1 & 0.5 & 0.5^{2} & 0.5^{3}\\
			0.5 & 1 & 0.5 & 0.5^{2}\\
			0.5^{2} & 0.5 & 1 & 0.5\\
			0.5^{3} & 0.5^{2} & 0.5 & 1
		\end{pmatrix}.
	\end{align*}

	\item Let $Z$ be i.i.d. according to uniform$[-2,2]$, $S_{i} = \sum_{j =
		1}^{4} 1\{Z_{i} \leq g_{j}\}$, and $(g_{1}, g_{2}, g_{3}, g_{4}) = (-1, 0, 1,
	2)$. Let $X_{i}:=(X_{1,i}, X_{2,i})^{\top}$, where $X_{1,i}$ follows a uniform
	distribution on $[-2,2]$, $X_{2,i}$ follows a standard normal distribution,
	and $X_{1,i}$ and $X_{2,i}$ are independent. Further, define
	\begin{align*}
		\alpha(X_{i}, Z_{i}) & = -0.8X_{1,i}\cdot X_{2,i}+Z_{i}^{2}+Z_{i}\cdot
		X_{1,i}, \quad \gamma(X_{i}, Z_{i})   = 0.5X_{1,i}^{2}-0.5X_{2,i}^{2}-0.5Z_{i}^{2},
	\end{align*}
	$a_{1}=2, a_{0}=1, b_{1}=1, b_{0}=-1, c_{1}=c_{0}=3$, and $(\varepsilon_{1,i},
	\varepsilon_{2,i},\varepsilon_{3,i},\varepsilon_{4,i})^{\top}$ are defined in DGP(i).
	
	\item Let $Z$ be i.i.d. according to standardized Beta(2, 2), $S_{i} = \sum_{j
		= 1}^{4} 1\{Z_{i} \leq g_{j}\}$, and $(g_{1}, g_{2}, g_{3}, g_{4}) =
	(-0.25\sqrt{20}, 0, 0.25\sqrt{20}, 0.5\sqrt{20})$. Let $X_{i}:=(X_{1,i},
	\cdots, X_{20,i})^{\top}$, where $X_{i} \overset{i.i.d}{\sim} N(0_{20 \times
		1}, \Omega)$ where $\Omega$ is the Toeplitz matrix
	\begin{align*}
		\Omega=
		\begin{pmatrix}
			1 & 0.5 & 0.5^{2} & \cdots & 0.5^{19}\\
			0.5 & 1 & 0.5 & \cdots & 0.5^{18}\\
			0.5^{2} & 0.5 & 1 & \cdots & 0.5^{17}\\
			\vdots & \vdots & \vdots & \ddots & \vdots\\
			0.5^{19} & 0.5^{18} & 0.5^{17} & \cdots & 1
		\end{pmatrix}.
	\end{align*}
	Further define $\alpha(X_{i}, Z_{i})= \sum_{k=1}^{20}X_{k,i}\beta_{k}+Z_{i}$,
	$\gamma(X_{i}, Z_{i}) = \sum_{k=1}^{20}X_{k,i}^{\top}\gamma_{k}-Z_{i}$, with
	$\beta_{k} = \sqrt{6}/k^{2}$ and $\gamma_{k} = -2/k^{2}$. Moreover, $a_{1}=2,
	a_{0}=1, b_{1}=2, b_{0}=-1, c_{1}=c_{0}=\sqrt{7}$, and $(\varepsilon_{1,i},
	\varepsilon_{2,i},\varepsilon_{3,i},\varepsilon_{4,i})^{\top}$ are defined in DGP(i).
\end{enumerate}

\bigskip

For each data generating process, we consider the four randomization
schemes (SRS, WEI, BCD, SBR) defined as in Examples \ref{ex:srs}--\ref{ex:sbr} in  Appendix \ref{sec:car_descri}, respectively. Specifically, for WEI and BCD, we set $f(x) = (1-x)/2$ and $\lambda = 0.75$, respectively. 

We compute the true LATE effect $\tau_{0}$ using Monte Carlo simulations, with sample size being 10,000 and the number of Monte Carlo simulations being 1,000. We gauge the size and power of various tests by testing the hypotheses $H_{0}: \tau=\tau_{0}$ and $H_{0}: \tau=\tau_{0}+1$, respectively. All the tests are carried out at 5\% level of significance, and with the number of Monte Carlo simulations being 10,000. 

\subsection{Estimators for Comparison}
\label{sec:sim_est}
For DGPs(i)-(ii), we consider the following estimators.

\begin{enumerate}
	[label=(\roman*)]
	
	\item NA: the fully saturated estimator by \cite{BG21}, which is equivalent to setting $\bar{\mu}%
	^{b}(a,s,x) = \hat{\mu}^{b}(a,s,x) = 0$ for $b=D,Y$, $a=0,1$, all $s$ and all
	$x$.
	
	\item TSLS: $\hat{\tau}_{TSLS}$ defined in Section \ref{sec:TSLS}. We use the usual IV heteroskedasticity-robust standard error (i.e., $\hat{\sigma}_{TSLS,naive}/\sqrt{n}$) for inference.

	\item L: the optimal linear estimator with $\Psi_{i,s}=X_{i}$ and the pseudo
	true values being estimated by $\hat{\theta}_{a,s}^{L}$ and $\hat{\beta
	}_{a,s}^{L}$ defined in (\ref{rr1}). 
	
	\item S: \citeauthor{anseletal2018}'s (\citeyear{anseletal2018}) S estimator with $X_i$ as regressor. We use the standard error of the S estimator (i.e., $\hat{\sigma}_S/\sqrt{n}$; see Section \ref{sec:s} of the Online Supplement for details) for inference.
	
	\item NL: the nonlinear (logistic) estimator with $\Psi_{i,s}=X_{i}$, and the pseudo true
	values being estimated by $\hat{\theta}_{a,s}^{OLS}$ and $\hat{\beta}%
	_{a,s}^{MLE}$ defined in (\ref{eq:nonlinear (logistic)}).
	
	\item F: the further efficiency improving estimator with $\Psi_{i,s}=X_{i}$,
	and the pseudo true values being estimated by $\hat{\theta}_{a,s}^{F}$ and
	$\hat{\beta}_{a,s}^{F}$ defined in (\ref{rr2}).
	
	\item NP: the nonparametric estimator outlined in Section \ref{sec:np}. The following 9 bases of a spline of order 3 are chosen as the sieve regressors:
	\begin{align}
		& \mathring{\Psi}_{i,n} = \left(  1, X_{1,i}, X_{2,i}, X_{1,i}^{2},
		X_{2,i}^{2}, X_{1,i}1\{X_{1,i}>t_{1}\}, X_{2,i}1\{X_{2,i}>t_{2}\},
		X_{1,i}X_{2,i}, \right. \notag \\
		& \qquad\qquad\left.  X_{1,i}1\{X_{1,i}>t_{1}\}X_{2,i}1\{X_{2,i}>t_{2}\}
		\right)^{\top}, \label{r1}
	\end{align}
	where $t_{1}$ and $t_{2}$ are the sample medians of $\{X_{1,i}\}_{i \in[n]}$ and $\{X_{2,i}\}_{i \in[n]}$, respectively.\footnote{The formal definition of spline is given in Section \ref{sec:aux_imp} of the Online Supplement.} 
	The adjustments are computed as in
	(\ref{eq:muhat_np}).
	
	\item SNP: \citeauthor{anseletal2018}'s (\citeyear{anseletal2018}) S estimator with $\mathring{\Psi}_{i,n}$ defined in (\ref{r1}) as regressor. We use the standard error of the S estimator (i.e., $\hat{\sigma}_S/\sqrt{n}$; see Section \ref{sec:s} of the Online Supplement for details) for inference.
	
	\item R: a regularized estimator. The nonparametric estimator outlined in Section \ref{sec:np} might not have a good size when the sample size is
	small, so we propose to use Lasso to select the sieve regressors. 
	The sieves regressors $\mathring{\Psi}_{i,n}$ are the same as in (\ref{r1}). 
	The adjustments are computed as in
	(\ref{eq:muhat_hd}). The tuning parameter is chosen as: $\varrho_{n,a}(s) = 1.1 \sqrt{n_{a}(s)}F_{N}^{-1}( 1-1/(p_{n}\log(n_{a}(s))))$. We compute the data-driven penalty loading matrices $\hat \Omega^Y$ and $\hat \Omega^D$ following the iterative procedure proposed by \cite{BCFH13}.\footnote{{\tt Matlab} code provided by \cite{BCFH13} and the {\tt R} package ``hdm" provide a built-in option for this iterative procedure.}
	
\end{enumerate}

For DGP(iii), we consider the estimator with no adjustments (NA), and the
lasso estimators $\hat{\theta}_{a,s}^{R}$ and $\hat{\beta}_{a,s}^{R}$ defined
in (\ref{eq:muhat_hd}) with $\mathring{\Psi}_{i,n}=(1,\Psi_{i,n}^{\top})^{\top}=(1,X_{i}^{\top})^{\top}$. The tuning parameters are choosing as: $\varrho_{n,a}(s) = 1.1 \sqrt{n_{a}(s)}F_{N}^{-1}( 1-1/(p_{n}\log(n_{a}(s))))$. 

\subsection{Simulation Results}

Tables \ref{tab:Simulation}-\ref{tab:Simulation3} present the empirical sizes and powers of the true null $H_{0}: \tau=\tau_{0}$ and false null $H_{0}: \tau=\tau_{0}+1$ under DGPs (i)-(iii), respectively. We also report the ratio of the median length of the confidence intervals of a particular estimator to that of the NA estimator is in the corresponding bracket. Note that none of the working models in DGPs (i)-(iii) is correctly specified. Consider DGP (i). When $n=200$, both the NA and TSLS estimators are slightly under-sized. Both the NP and SNP estimators are oversized because the numbers of sieve regressors are relatively large compared to the sample size, while the R estimator has the correct size thanks to the Lasso selection of the sieve regressors. The L estimator performs the same as the S estimator. All other estimators have sizes close to the nominal level of 5\%. This confirms that our estimation and inference procedures are robust to misspecification.

\renewcommand{\arraystretch}{0.5}
\begin{table}[H]
	\caption{{\protect\small Size and Power for DGP(i) }}%
	\label{tab:Simulation}%
	\centering
	\smallskip\begin{tabularx}{\linewidth}{@{\extracolsep{\fill}}lcccccccc}
		\toprule
		& \multicolumn{4}{c}{$n = 200$} & \multicolumn{4}{c}{$n = 400$} \\ \cmidrule{2-5}\cmidrule{6-9}
		Methods  & \multicolumn{1}{c}{SRS} & \multicolumn{1}{c}{WEI} & \multicolumn{1}{c}{BCD} & \multicolumn{1}{c}{SBR} & \multicolumn{1}{c}{SRS} & \multicolumn{1}{c}{WEI} & \multicolumn{1}{c}{BCD} & \multicolumn{1}{c}{SBR} \\ 
		\midrule
		\multicolumn{9}{l}{\textit{Size}} \\  \midrule
		\quad NA & 0.035 & 0.031 & 0.031 & 0.034  &  0.046 & 0.043 & 0.042 & 0.039 \\
		\quad  &  &  &  &   &   &  &  &  \\
		\quad TSLS &  0.036 & 0.034 & 0.032 & 0.038  & 0.045 & 0.040 & 0.044 & 0.042 \\
		\quad & [77.8\%] & [78.0\%] & [77.6\%] & [77.8\%] & [78.0\%] & [77.9\%] & [77.8\%] & [78.0\%]\\
		\quad L & 0.044 & 0.041 & 0.041 & 0.045  & 0.048 & 0.044 & 0.047 & 0.047 \\
		\quad & [76.6\%] & [76.6\%] & [76.5\%] & [76.5\%] & [77.3\%] & [77.1\%] & [77.2\%] & [77.4\%]\\
		\quad S &  0.044 & 0.041 & 0.041 & 0.045  & 0.048 & 0.044 & 0.047 & 0.047 \\
		\quad & [76.6\%] & [76.6\%] & [76.5\%] & [76.5\%] & [77.3\%] & [77.1\%] & [77.2\%] & [77.4\%] \\
		\quad NL & 0.044 & 0.042 & 0.040 & 0.045 & 0.049 & 0.045 & 0.047 & 0.047\\
		\quad & [77.5\%] & [77.3\%] & [77.2\%] & [77.2\%] & [77.6\%] & [77.5\%] & [77.5\%] & [77.6\%] \\
		\quad F & 0.054 & 0.052 & 0.049 & 0.053  & 0.054 & 0.048 & 0.052 & 0.050\\
		\quad & [74.7\%] & [74.9\%] & [74.6\%] & [74.5\%] & [75.4\%] & [75.3\%] & [75.3\%] & [75.6\%] \\
		\quad NP &  0.109 & 0.094 & 0.091 & 0.090 & 0.073 & 0.062 & 0.067 & 0.062\\
		\quad & [81.6\%] & [80.5\%] & [79.5\%] & [79.0\%] & [69.3\%] & [69.4\%] & [69.5\%] & [69.4\%]\\
		\quad SNP & 0.100 & 0.091 & 0.090 & 0.085 & 0.070 & 0.061 & 0.063 & 0.060  \\
		\quad & [73.2\%] & [72.3\%] & [72.0\%] & [71.8\%] & [68.0\%] & [67.9\%] & [68.1\%] & [68.0\%] \\
		\quad R & 0.053 & 0.050 & 0.049 & 0.055 & 0.057 & 0.049 & 0.051 & 0.047 \\
		\quad & [70.6\%] & [70.3\%] & [70.1\%] & [70.1\%] & [69.5\%] & [69.5\%] & [69.5\%] & [69.6\%] \\  \midrule
		\multicolumn{9}{l}{\textit{Power}} \\  \midrule 
		\quad NA & 0.170 & 0.169 & 0.170 & 0.170  & 0.293 & 0.289 & 0.291 & 0.294 \\
		\quad  &  &  &  &   &   &  &  &  \\
		\quad TSLS & 0.260 & 0.254 & 0.260 & 0.255 & 0.430 & 0.433 & 0.443 & 0.436 \\
		\quad & [77.8\%] & [78.0\%] & [77.6\%] & [77.8\%] & [78.0\%] & [77.9\%] & [77.8\%] & [78.0\%] \\
		\quad L & 0.274 & 0.264 & 0.273 & 0.268 & 0.439 & 0.440 & 0.447 & 0.444\\
		\quad & [76.6\%] & [76.6\%] & [76.5\%] & [76.5\%] & [77.3\%] & [77.1\%] & [77.2\%] & [77.4\%] \\
		\quad S & 0.274 & 0.264 & 0.273 & 0.268  &  0.439 & 0.440 & 0.447 & 0.444 \\
		\quad & [76.6\%] & [76.6\%] & [76.5\%] & [76.5\%] & [77.3\%] & [77.1\%] & [77.2\%] & [77.4\%]\\
		\quad NL & 0.268 & 0.257 & 0.267 & 0.261  & 0.434 & 0.435 & 0.443 & 0.439  \\
		\quad & [77.5\%] & [77.3\%] & [77.2\%] & [77.2\%] & [77.6\%] & [77.5\%] & [77.5\%] & [77.6\%]\\
		\quad F &  0.299 & 0.292 & 0.296 & 0.293 & 0.460 & 0.454 & 0.466 & 0.463  \\
		\quad & [74.7\%] & [74.9\%] & [74.6\%] & [74.5\%] & [75.4\%] & [75.3\%] & [75.3\%] & [75.6\%] \\
		\quad NP & 0.299 & 0.284 & 0.289 & 0.280  & 0.509 & 0.506 & 0.510 & 0.509 \\
		\quad & [81.6\%] & [80.5\%] & [79.5\%] & [79.0\%] & [69.3\%] & [69.4\%] & [69.5\%] & [69.4\%] \\
		\quad SNP & 0.344 & 0.331 & 0.340 & 0.333  &  0.532 & 0.526 & 0.533 & 0.532  \\
		\quad & [73.2\%] & [72.3\%] & [72.0\%] & [71.8\%] & [68.0\%] & [67.9\%] & [68.1\%] & [68.0\%]\\
		\quad R &  0.325 & 0.315 & 0.325 & 0.321  &  0.516 & 0.517 & 0.514 & 0.516  \\
		\quad & [70.6\%] & [70.3\%] & [70.1\%] & [70.1\%] & [69.5\%] & [69.5\%] & [69.5\%] & [69.6\%] \\
		\bottomrule
	\end{tabularx}
\end{table}

In terms of power, the NA estimator has the lowest power, corroborating the belief that one should carry out the regression adjustment whenever covariates correlate with the potential outcomes. The powers of the other estimators are much higher. In particular, the power of the F estimator is higher than those of the NA, TSLS, L, and NL estimators, which is consistent with our theory that the F estimator is weakly more efficient than those estimators. The NP, SNP, and R estimators enjoy the highest powers as a nonparametric model could approximate the true specification very well. The NP and SNP estimators have more size distortions than the R estimator when the sample size is 200. When the sample size is increased to 400, virtually all the sizes and powers of the estimators improve, and all the observations continue to hold.

\renewcommand{\arraystretch}{0.5}
\begin{table}[H]
	\caption{{\protect\small Size and Power for DGP(ii)}}%
	\label{tab:Simulation2}%
	\centering
	\smallskip\begin{tabularx}{\linewidth}{@{\extracolsep{\fill}}lcccccccc}
		\toprule
		& \multicolumn{4}{c}{$n = 200$} & \multicolumn{4}{c}{$n = 400$} \\ \cmidrule{2-5}\cmidrule{6-9}
		Methods  & \multicolumn{1}{c}{SRS} & \multicolumn{1}{c}{WEI} & \multicolumn{1}{c}{BCD} & \multicolumn{1}{c}{SBR} & \multicolumn{1}{c}{SRS} & \multicolumn{1}{c}{WEI} & \multicolumn{1}{c}{BCD} & \multicolumn{1}{c}{SBR} \\
		\midrule
		\multicolumn{9}{l}{\textit{Size}} \\ \midrule
		\quad NA & 0.033 & 0.031 & 0.029 & 0.030 & 0.045 & 0.042 & 0.043 & 0.041  \\
		\quad  &  &  &  &   &   &  &  &  \\
		\quad TSLS & 0.035 & 0.033 & 0.031 & 0.033  & 0.045 & 0.044 & 0.045 & 0.040 \\
		\quad & [99.5\%] & [99.4\%] & [99.6\%] & [99.4\%] & [99.8\%] & [99.8\%] & [99.8\%] & [99.7\%]\\
		\quad L & 0.044 & 0.040 & 0.044 & 0.038 & 0.049 & 0.047 & 0.046 & 0.046 \\
		\quad & [74.7\%] & [74.5\%] & [74.8\%] & [74.6\%] &  [75.5\%] & [75.6\%] & [75.6\%] & [75.5\%] \\
		\quad S & 0.044 & 0.040 & 0.044 & 0.038  & 0.049 & 0.047 & 0.046 & 0.046 \\
		\quad & [74.7\%] & [74.5\%] & [74.8\%] & [74.6\%] & [75.5\%] & [75.6\%] & [75.6\%] & [75.5\%]\\
		\quad NL & 0.043 & 0.039 & 0.042 & 0.037 & 0.049 & 0.047 & 0.046 & 0.046 \\
		\quad & [75.3\%] & [75.1\%] & [75.5\%] & [75.2\%] & [75.7\%] & [75.8\%] & [75.8\%] & [75.6\%]\\
		\quad F & 0.052 & 0.047 & 0.048 & 0.043  &  0.050 & 0.049 & 0.051 & 0.049  \\
		\quad & [70.2\%] & [69.7\%] & [70.3\%] & [70.4\%] & [70.8\%] & [70.8\%] & [70.9\%] & [70.7\%]\\
		\quad NP & 0.100 & 0.084 & 0.087 & 0.079   & 0.062 & 0.063 & 0.065 & 0.062 \\
		\quad & [69.2\%] & [67.7\%] & [67.3\%] & [67.7\%] &  [60.4\%] & [60.3\%] & [60.4\%] & [60.5\%]\\
		\quad SNP & 0.098 & 0.084 & 0.085 & 0.079 & 0.061 & 0.064 & 0.064 & 0.063 \\
		\quad & [63.7\%] & [62.8\%] & [63.0\%] & [62.8\%] & [59.7\%] & [59.8\%] & [59.7\%] & [59.8\%]\\
		\quad R &  0.055 & 0.051 & 0.049 & 0.049 & 0.052 & 0.051 & 0.048 & 0.045 \\
		\quad & [63.3\%] & [62.8\%] & [63.2\%] & [63.0\%] & [62.1\%] & [62.1\%] & [62.2\%] & [62.0\%]\\ \midrule
		\multicolumn{9}{l}{\textit{Power}} \\ \midrule
		\quad NA & 0.202 & 0.208 & 0.208 & 0.206  & 0.350 & 0.351 & 0.351 & 0.345\\
		\quad  &  &  &  &   &   &  &  &  \\
		\quad TSLS &  0.204 & 0.212 & 0.211 & 0.210 & 0.353 & 0.352 & 0.354 & 0.346 \\
		\quad & [99.5\%] & [99.4\%] & [99.6\%] & [99.4\%] & [99.8\%] & [99.8\%] & [99.8\%] & [99.7\%] \\
		\quad L & 0.334 & 0.331 & 0.342 & 0.340  & 0.512 & 0.526 & 0.524 & 0.516  \\
		\quad & [74.7\%] & [74.5\%] & [74.8\%] & [74.6\%] & [75.5\%] & [75.6\%] & [75.6\%] & [75.5\%] \\
		\quad S & 0.334 & 0.331 & 0.342 & 0.340  &  0.512 & 0.526 & 0.524 & 0.516 \\
		\quad & [74.7\%] & [74.5\%] & [74.8\%] & [74.6\%] & [75.5\%] & [75.6\%] & [75.6\%] & [75.5\%]\\
		\quad NL &  0.327 & 0.324 & 0.335 & 0.333 &  0.510 & 0.523 & 0.523 & 0.515 \\
		\quad & [75.3\%] & [75.1\%] & [75.5\%] & [75.2\%] & [75.7\%] & [75.8\%] & [75.8\%] & [75.6\%] \\
		\quad F &  0.372 & 0.374 & 0.379 & 0.375 & 0.562 & 0.568 & 0.566 & 0.561\\
		\quad & [70.2\%] & [69.7\%] & [70.3\%] & [70.4\%] & [70.8\%] & [70.8\%] & [70.9\%] & [70.7\%]\\
		\quad NP & 0.378 & 0.381 & 0.387 & 0.386  &  0.649 & 0.663 & 0.653 & 0.655\\
		\quad & [69.2\%] & [67.7\%] & [67.3\%] & [67.7\%] & [60.4\%] & [60.3\%] & [60.4\%] & [60.5\%]\\
		\quad SNP & 0.431 & 0.443 & 0.442 & 0.440  & 0.663 & 0.676 & 0.668 & 0.668 \\
		\quad &  [63.7\%] & [62.8\%] & [63.0\%] & [62.8\%] & [59.7\%] & [59.8\%] & [59.7\%] & [59.8\%]\\
		\quad R & 0.419 & 0.429 & 0.431 & 0.432 &  0.644 & 0.661 & 0.657 & 0.648  \\
		\quad & [63.3\%] & [62.8\%] & [63.2\%] & [63.0\%] & [62.1\%] & [62.1\%] & [62.2\%] & [62.0\%]\\
		\bottomrule
	\end{tabularx}
\end{table}

We also report the ratio of the median length of the confidence intervals of a particular estimator to that of the NA estimator in the corresponding parentheses. Generally speaking, the confidence intervals of the TSLS and adjusted estimators (L, NL, F, NP, and R) are 20\%-30\% shorter, in terms of the median, than that of the NA estimator. 

Most observations uncovered in DGP (i) carry forward to DGP (ii). Two new patterns emerge. First, the powers of the L, S, NL, F, NP, SNP, and R estimators are much higher than those of the NA and TSLS estimators. Second, the ratio of the median length of the confidence intervals of the TSLS estimator is as wide as that of the NA estimator, whereas the confidence intervals of the adjusted estimators (L, NL, F, NP, and R) become 25\%-40\% shorter, in terms of the median, than that of the NA estimator. This is probably because the true specifications for $Y_{i}(a)$ become more nonlinear.

We now consider DGP (iii). In this setting, only the NA and R estimators are feasible. When $n=200$, both estimators have the correct sizes but the R estimator has considerably higher power. When $n=400$, the sizes of these two estimators remain relatively unchanged, while their powers improve with a diverging gap. The confidence intervals of the R estimator are 60\%-65\% shorter, in terms of the median, than that of the NA estimator.

\begin{table}[H]
	\caption{{\protect\small Size and Power for DGP(iii) }}%
	\label{tab:Simulation3}%
	\centering
	\smallskip\begin{tabularx}{\linewidth}{@{\extracolsep{\fill}}lcccccccc}
		\toprule
		& \multicolumn{4}{c}{$n = 200$} & \multicolumn{4}{c}{$n = 400$} \\ \cmidrule{2-5}\cmidrule{6-9}
		Methods  & \multicolumn{1}{c}{SRS} & \multicolumn{1}{c}{WEI} & \multicolumn{1}{c}{BCD} & \multicolumn{1}{c}{SBR} & \multicolumn{1}{c}{SRS} & \multicolumn{1}{c}{WEI} & \multicolumn{1}{c}{BCD} & \multicolumn{1}{c}{SBR} \\
		\midrule
		\multicolumn{9}{l}{\textit{Size}} \\
		\quad NA & 0.046 & 0.043 & 0.046 & 0.048  &  0.046 & 0.047 & 0.045 & 0.047   \\
		\quad  &  &  &  &  &  &  &  &    \\
		\quad R & 0.064 & 0.058 & 0.061 & 0.060  &  0.057 & 0.061 & 0.058 & 0.060 \\
		\quad  & [37.2\%] & [36.9\%] & [36.7\%] & [36.8\%] & [34.4\%] & [34.4\%] & [34.5\%] & [34.5\%]   \\
		\multicolumn{9}{l}{\textit{Power}} \\
		\quad NA & 0.173 & 0.170 & 0.171 & 0.177 &  0.233 & 0.238 & 0.235 & 0.239\\
		\quad  &  &  &  &  &  &  &  &    \\
		\quad R & 0.516 & 0.524 & 0.533 & 0.534  &  0.811 & 0.815 & 0.817 & 0.815  \\
		\quad  & [37.2\%] & [36.9\%] & [36.7\%] & [36.8\%] & [34.4\%] & [34.4\%] & [34.5\%] & [34.5\%]  \\
		\bottomrule
	\end{tabularx}
\end{table}

In Section \ref{sec:additional simulation} of the Online Supplement, we simulate data with heterogeneous $\pi(s)$. We find that all estimators except TSLS have their empirical rejection rates close to the nominal size of 5\% under the null. TSLS, on the other hand, has around 15\% rejection rate when $n=1200$. This indicates the TSLS estimator is inconsistent when $\pi(s)$ is heterogeneous, in line with Theorem \ref{thm:TSLS}.

\subsection{Practical Recommendation}

If researchers want to use parametric adjustments without tuning parameters, we recommend the F estimator, which is guaranteed to be weakly more efficient than TSLS, L, and NL estimators. Regressors $\Psi_{i,s}$ can include linear, quadratic and interaction terms of the original covariates. If researchers want to achieve the SEB by using sieve bases and/or the dimension of covariates is high relative to the sample size, we recommend the R estimator.


\section{Empirical Application}

\label{sec:app}

Banking the unbanked is considered to be the first step toward broader
financial inclusion -- the focus of the World Bank's Universal Financial
Access 2020
initiative.\footnote{https://www.worldbank.org/en/topic/financialinclusion/brief/achieving-universal-financial-access-by-2020}
In a field experiment with a CAR design, \cite{dupasetal2018} examined the
impact of expanding access to basic saving accounts for rural households
living in three countries: Uganda, Malawi, and Chile. In particular, apart
from the intent-to-treat effects for the whole sample, they also studied the
local average treatment effects for the households who actively used the accounts. This section presents an application of our regression adjusted estimators to the same dataset to examine the LATEs of opening bank accounts on savings balance-- a central outcome of interest in their study.

We focus on the experiment conducted in Uganda. The sample consists of 2,160
households who were randomized with a CAR design. Specifically, within each of
41 strata formed by gender, occupation, and bank branch, half of households
were randomly allocated to the treatment group, the other half to the control
one. Households in the treatment group were then offered a voucher to open
bank accounts with no financial costs. However, not every treated household
ever opened and used the saving accounts for deposit. In fact, among those households with treatment assignment, only 41.87\% of them opened the accounts and made at least
one deposit within 2 years. Subject compliance is therefore imperfect in this experiment.

The randomization design apparently satisfies statements (i), (ii) and (iii) of Assumption \ref{ass:assignment1}. The target fraction of treatment assignment is 1/2. Because $\max_{s\in\mathcal{S}}|\frac{B_{n}(s)}{n(s)}|\approx0.056$,
it is plausible to claim that Assumption \ref{ass:assignment1}(iv) is also
satisfied. Since households in the control group need to pay for the fees of
opening accounts while the treated ones bear no financial costs, no-defiers
statement in Assumption \ref{ass:assignment1}(v) holds plausibly in this case.

One of the key analyses in \cite{dupasetal2018} is to estimate the treatment
effects on savings for active users -- households who actually opened the
accounts and made at least one deposit within 2 years. We follow their
footprints to estimate the same LATEs at savings balance.\footnote{Savings
	balance includes savings in formal financial intuitions, mobile money, cash at
	home or in secret place, savings in ROSCA/VSLA, savings with friends/family,
	other cash savings, total formal savings, total informal savings, and total
	savings (See \cite{dupasetal2018} for details). We use data from the first follow-up survey and exclude other cash savings because only 2\% of the households in the sample reported having it.} To maintain
comparability, for each outcome variable, we also keep $X_{i}$ similar to those
used in \cite{dupasetal2018} for our adjusted estimators.\footnote{The
	description of these estimators is similar to that in Section \ref{sec:sim}.
	Except for savings in formal financial institutions, mobile money, and total
	formal savings, $X_{i}$ includes baseline value for the outcome of interest, baseline value of total income, and a dummy for missing observations. For savings in formal financial institutions,
	mobile money, and total formal savings, since their baseline values are all
	zero, we set $X_{i}$ as the baseline value of total savings, baseline value of total income, and a dummy for missing observations.} Due to the low dimension of covariates used in the regression adjustments, we focus on the performance of the methods ``NA", ``TSLS", ``L", ``NL", and ``F".

Table \ref{tab:emp_ate} presents the LATE estimates and their standard errors
(in parentheses) estimated by these methods.\footnote{For each outcome variable, we filter out the observations with missing values of outcome variables or the strata with less than 10 observations. The total trimmed observations are less than 10\% of the whole sample in \cite{dupasetal2018}).} These results lead to four observations. First, consistent with the theoretical and simulation results, the standard errors for the LATE estimates with regression adjustments are
lower than those without adjustments. This observation holds for all the
outcome variables and all the regression adjustment methods. Over the eight
outcome variables, the standard errors estimated by regression adjustments are
on average around 8\% lower than those without adjustment. In particular, when
the outcome variable is total informal savings, the standard errors obtained
via the further improvement adjustment -- ``F" method is about 18\% lower
than those without adjustment. This means that regression adjustments,
with the similar covariates used in \cite{dupasetal2018}, can achieve sizable
efficiency gains in estimating the LATEs.

\newcolumntype{L}{>{\raggedright\arraybackslash}X} \newcolumntype{C}{>{\centering\arraybackslash}X}

\renewcommand{\arraystretch}{0.5}
\begin{table}[H]
	\caption{Impacts on Saving Stocks in 2010 US Dollars}%
	\label{tab:emp_ate}%
	\centering
	\begin{tabularx}{1\textwidth}{LCCCCCC}
		\toprule
		\qquad $Y$ & $n$ & NA & TSLS & L & NL & F \\
		\midrule
		Formal                  & 1968      & 20.558    & 21.154    & 22.160    & 22.196    & 22.743      \\ 
		fin. inst.              &           & (3.067)   & (3.015)   & (2.965)   & (2.976)   & (2.942)     \\ 
		\\
		Mobile                  & 1972      & -0.208    & -0.174    & -0.291    & -0.292    & -0.302     \\ 
		money                        &           & (0.223)   & (0.224)   & (0.212)   & (0.213)   & (0.208)  \\ 
		\\
		Total                   & 1966      & 20.399    & 21.097    & 21.924    & 21.986    & 22.335     \\ 
		formal                  &           & (3.089)   & (3.034)   & (2.979)   & (2.994)   & (2.956)    \\ 
		\\
		Cash at                 & 1971      & -10.826   & -7.456    & -9.004    & -8.904    & -8.373     \\ 
		home                    &           & (5.003)   & (4.404)   & (4.401)   & (4.355)   & (4.354)     \\ 
		\\
		ROSCA/                  & 1975      & -1.933    & -2.333    & -1.242    & -1.255    & 0.651     \\ 
		VSLA                    &           & (1.971)   & (1.858)   & (1.794)   & (1.812)   & (1.940)    \\ 
		\\
		Friends/                & 1974      & -3.621    & -3.346    & -1.428    & -1.536    & -2.067     \\ 
		family                  &           & (2.040)   & (1.999)   & (1.866)   & (2.015)   & (2.042)    \\ 
		\\
		Total                   & 1960      & -17.643   & -14.317   & -15.665   & -15.693   & -14.137     \\ 
		informal                &           & (6.200)   & (5.351)   & (5.185)   & (5.196)   & (5.082)    \\ 
		\\
		Total                   & 1952      & 2.787     & 7.153     & 7.169     & 7.193     & 8.962    \\ 
		savings                &           & (7.290)   & (6.368)   & (6.197)   & (6.218)   & (6.142)    \\ 
		\bottomrule
	\end{tabularx} \newline\vspace{-1ex} \justify
	Notes: The table reports the LATE estimates of opening bank accounts on saving
	stocks. NA, TSLS, L, NL, and F stand for the no-adjustment, TSLS, optimal
	linear, nonlinear (logistic), further efficiency improving,
	respectively. $n$ is the number of households. Standard errors are in parentheses.\end{table}

\renewcommand{\arraystretch}{1}

Second, the standard errors for the regression-adjusted LATE estimates are
mostly lower than those obtained by the usual TSLS procedure. Especially, when the outcome variables are mobile money and total informal savings, the standard errors obtained via ``F" method are about 7.1\% and 5\%, respectively, lower than those by TSLS. When the outcome variable is savings in friends/family, the standard error estimated by the optimal linear adjustment -- ``L" method is around 6.7\% lower than that obtained by TSLS. This means that, compared with our regression-adjusted methods, TSLS is generally less efficient to estimate the LATEs under CAR.

Third, the standard errors for the LATE estimates with regression adjustments
are similar in terms of magnitude. This implies that all the regression
adjustments achieve similar efficiency gain in this case.

Finally, as in \cite{dupasetal2018}, for the households who actively use bank
accounts, we find that reducing the cost of opening a bank account can
significantly increase their savings in formal institutions. We also observe
the evidence of crowd-out -- mainly moving cash from saving at home to saving
in bank.




\appendix

\section{Covariate-Adaptive Treatment Assignment Rules}
\label{sec:car_descri}

\begin{ex}
	[SRS] \label{ex:srs} Let $A_k$ be a Bernoulli random variable, independent of $\{S_i\}_{i=1,\neq k}^n$ and $\{A_i\}_{i=1}^{k-1}$, with success rate $\pi(s)$ when $S_k = s$ for $k=1,\ldots,n$. That is,
	\[
	\mathbb{P}\left(  A_{k}=1\big|\{S_{i}\}_{i=1}^{n},\{A_{i}\}_{i=1}^{k-1}\right)  =\mathbb{P}(A_{k}=1|S_k)=\pi(S_k).
	\]
\end{ex}

\begin{ex}
	[WEI] \label{ex:wei}  This design was first proposed by \cite{W78}. Let
	$n_{k-1}(S_{k}) = \sum_{i=1}^{k-1}1\{S_{i} = S_{k}\}$, $B_{k-1}(S_{k}) =
	\sum_{i=1}^{k-1}\left( A_{i} - \frac{1}{2} \right)  1\{S_{i} = S_{k}\}$, and
	\begin{align*}
		\mathbb{P}\left( A_{k} = 1\big| \{S_{i}\}_{i=1}^{k},\{A_{i}\}_{i=1}%
		^{k-1}\right)  = f\biggl(\frac{2B_{k-1}(S_{k})}{n_{k-1}(S_{k})}\biggr),
	\end{align*}
	where $f(\cdot):[-1,1] \mapsto[0,1]$ is a pre-specified non-increasing
	function satisfying $f(-x) = 1- f(x)$. Here, $\frac{B_{0}(S_{1})}{n_{0}%
		(S_{1})}$ and $B_{0}(S_{1})$ are understood to be zero.
\end{ex}

\begin{ex}
	[BCD] \label{ex:bcd}  The treatment status is determined sequentially for $1
	\leq k \leq n$ as
	\begin{align*}
		\mathbb{P}\left( A_{k} = 1| \{S_{i}\}_{i=1}^{k},\{A_{i}\}_{i=1}^{k-1}\right)
		=
		\begin{cases}
			\frac{1}{2} & \text{if }B_{k-1}(S_{k}) = 0\\
			\lambda & \text{if }B_{k-1}(S_{k}) < 0\\
			1-\lambda & \text{if }B_{k-1}(S_{k}) > 0,
		\end{cases}
	\end{align*}
	where $B_{k-1}(s)$ is defined as above and $\frac{1}{2}< \lambda\leq1$.
\end{ex}

\begin{ex}
	[SBR] \label{ex:sbr}  For each stratum, $\lfloor\pi(s) n(s) \rfloor$ units are
	assigned to treatment and the rest are assigned to control.
\end{ex}

\section{The S Estimator in \cite{anseletal2018}}
\label{sec:s}

\cite{anseletal2018} propose a LATE estimator adjusted with extra covariates.
It takes the form%

\begin{align*}
	\hat{\tau}_{S} := \frac{\sum_{s \in\mathcal{S}}\hat{p}(s)(\hat{\gamma}_{1s}^{Y} - \hat{\gamma}_{0s}^{Y} + (\hat{\nu}_{1s}^{Y} - \hat{\nu}_{0s}^{Y})^{\top}\bar{X}_{s}) }{\sum_{s \in\mathcal{S}}\hat{p}(s)(\hat{\gamma}_{1s}^{D} - \hat{\gamma}_{0s}^{D} + (\hat{\nu}_{1s}^{D} - \hat{\nu}_{0s}^{D})^{\top}\bar{X}_{s})},
\end{align*}
where $\hat{p}(s) := n(s)/n$, $\bar{X}_{s} := \frac{1}{n\hat{p}(s)}\sum_{i
	\in[n]}X_{i}1\{S_{i}=s\}$, and $(\hat{\gamma}_{as}^{Y},\hat{\gamma}_{as}^{D},\hat{\nu}_{as}^{Y},\hat{\nu}_{as}^{D})$ for $a=0,1$ are the estimated coefficients of the four sets of stratum-specific regressions using only the $s$ stratum:
\begin{align*}
	&  (1-A_{i})Y_{i} = (1-A_{i})(\gamma_{0s}^{Y}+X_{i}^{\top}\nu_{0s}^{Y} +
	e_{0i}^{Y}), \quad A_{i}Y_{i} = A_{i}(\gamma_{1s}^{Y}+X_{i}^{\top}\nu_{1s}^{Y}
	+ e_{1i}^{Y})\\
	&  (1-A_{i})D_{i} = (1-A_{i})(\gamma_{0s}^{D}+X_{i}^{\top}\nu_{0s}^{D} +
	e_{0i}^{D}), \quad A_{i}D_{i} = A_{i}(\gamma_{1s}^{D}+X_{i}^{\top}\nu_{1s}^{D}
	+ e_{1i}^{D}).
\end{align*}
Interpret $(\hat{\gamma}_{aS_i}^{Y},\hat{\gamma}_{aS_i}^{D},\hat{\nu}_{aS_i}^{Y},\hat{\nu}_{aS_i}^{D})$ for $a=0,1$ as the estimated coefficients of the four sets of stratum-specific regressions using only the $S_i$ stratum.

Under Assumption \ref{ass:assignment1} and Assumption \ref{ass:1iv} of our paper, \cite{anseletal2018} show that $\hat{\tau}_{S}$ is a consistent estimator of $\tau$, asymptotically normal, and the most efficient among the estimators studied in their paper ($\pi(s)$ can be heterogenous across strata). To define the explicit expression for the asymptotic variance of $\hat \tau_S$, denoted as $\sigma_{S}^{2}$, we need to introduce addition notation. For $s \in\mathcal{S}$, let $\tilde{X}_{is} :=
X_{i} - \mathbb{E}(X_{i}|S_{i}=s)$,
\begin{align*}
	\rho_{iS_i}(1)&:=\frac{Y_{i}(D_{i}(1)) - D_{i}(1)\tau- X_i^{\top}\nu_{1S_i}^{YD}}{\pi(S_i)}+X_{i}^{\top}(\nu_{1S_i}^{YD}-\nu_{0S_i}^{YD})\\
	\rho_{iS_i}(0)&:=\frac{Y_{i}(D_{i}(0)) - D_{i}(0)\tau- X_i^{\top}\nu_{0S_i}^{YD}}{1-\pi(S_i)}-X_{i}^{\top}(\nu_{1S_i}^{YD}-\nu_{0S_i}^{YD})\\
	\nu_{1s}^{YD} &: = \left[ \mathbb{E}(\tilde{X}_{is}\tilde{X}_{is}^{\top}|S_{i}=s)\right] ^{-1} \mathbb{E}\del[3]{\tilde{X}_{is}\sbr[2]{Y_i\del[1]{D_i(1)} - D_i(1)\tau }|S_i=s},\\
	\nu_{0s}^{YD} &:= \left[ \mathbb{E}(\tilde{X}_{is}\tilde{X}_{is}^{\top}|S_{i}=s)\right] ^{-1} \mathbb{E}\del[3]{\tilde{X}_{is}\sbr[2]{Y_i\del[1]{D_i(0)} - D_i(0)\tau }|S_i=s}.
\end{align*}
\begin{align*}
	\sigma_{S1}^{2}  &  := \mathbb{E}\sbr[3]{\pi(S_i)\cbr[2]{\rho_{iS_i}(1)-\mathbb{E}[\rho_{iS_i}(1)|S_i]}^2}\\
	\sigma_{S0}^{2}  &  := \mathbb{E}\sbr[3]{\del[1]{1-\pi(S_i)}\cbr[2]{\rho_{iS_i}(0)-\mathbb{E}[\rho_{iS_i}(0)|S_i]}^2}\\
	\sigma_{S2}^{2} &:=\mathbb{E}%
	\sbr[3]{\del[2]{\mathbb{E}\sbr[1]{Y_i(D_i(1)) - Y_i(D_i(0)) - \tau(D_i(1)-D_i(0))|S_i}}^2}.
\end{align*}
In addition, define
\begin{align*}
	\hat{\rho}_{iS_i}(1)&:=\frac{Y_{i} - D_{i}\hat{\tau}_S- X_i^{\top}\hat{\nu}_{1S_i}^{YD}}{\hat{\pi}(S_i)}+X_{i}^{\top}(\hat{\nu}_{1S_i}^{YD}-\hat{\nu}_{0S_i}^{YD})\\
	\hat{\rho}_{iS_i}(0)&:=\frac{Y_{i} - D_{i}\hat{\tau}_S- X_i^{\top}\hat{\nu}_{0S_i}^{YD}}{1-\hat{\pi}(S_i)}-X_{i}^{\top}(\hat{\nu}_{1S_i}^{YD}-\hat{\nu}_{0S_i}^{YD})\\
	\hat{\sigma}_{S1}^2&:=\frac{1}{n}\sum_{i\in [n]}A_i\sbr[3]{\hat{\rho}_{iS_i}(1)-\frac{1}{n_1(S_i)}\sum_{j\in I_1(S_i)}\hat{\rho}_{jS_j}(1)}^2\\
	\hat{\sigma}_{S0}^2&:=\frac{1}{n}\sum_{i\in [n]}(1-A_i)\sbr[3]{\hat{\rho}_{iS_i}(0)-\frac{1}{n_0(S_i)}\sum_{j\in I_0(S_i)}\hat{\rho}_{jS_j}(0)}^2\\
	\hat{\sigma}_{S2}^2&:=\frac{1}{n}\sum_{i\in [n]}\del[3]{\frac{1}{n_1(S_i)}\sum_{j\in I_1(S_i)}(Y_j-\hat{\tau}_SD_j)-\frac{1}{n_0(S_i)}\sum_{j\in I_0(S_i)}(Y_j-\hat{\tau}_SD_j)}^2\\
	\hat{\sigma}_S^2&:=\frac{\hat{\sigma}_{S1}^2+\hat{\sigma}_{S0}^2+\hat{\sigma}_{S2}^2}{\del[1]{\sum_{s \in\mathcal{S}}\hat{p}(s)(\hat{\gamma}_{1s}^{D} - \hat{\gamma}_{0s}^{D} + (\hat{\nu}_{1s}^{D} - \hat{\nu}_{0s}^{D})^{\top}\bar{X}_{s})}^2}
\end{align*}
where $\hat{\nu}_{aS_i}^{YD}:= \hat{\nu}_{aS_i}^{Y}-\hat{\tau}_S\hat{\nu}_{aS_i}^{D}$ for $a=0,1$.  

\begin{thm}
	Suppose Assumptions \ref{ass:assignment1} and \ref{ass:1iv} hold. Then,
	\begin{enumerate}[label=(\roman*)]
		\item \begin{align}
			\sigma_{S}^{2} = \frac{\sigma_{S1}^{2} + \sigma_{S0}^{2}+\sigma_{S2}^{2}}{(\mathbb{E}[D_i(1) - D_i(0)])^{2}}.\label{eq:ansel_sigma}
		\end{align}
		
		\item \[\hat{\sigma}_S^2 \xrightarrow{p} \sigma_{S}^{2}.\]
	\end{enumerate}
	\label{thm:sigma_S}
\end{thm}

It can be shown that $\sigma_{Sa}^{2} \geq\underline{\sigma}_{a}^{2}$ for $a
= 0,1$ and $\sigma_{S2}^{2} = \underline{\sigma}_{2}^{2}$, where the
inequalities are strict except special cases such as $\mathbb{E}(Y_{i}(D_{i}(a)) - D_{i}(a)\tau|X_{i},S_{i}=s)$ is linear in $X_{i}$, and $\underline{\sigma}_{a}^{2}$ for $a=0,1,2$ are defined in Theorem \ref{thm:eff}. This implies
in general, the S estimator is not semiparametrically most efficient.



\begin{thm}
	\label{thm:tau_TSLS} Suppose that Assumptions \ref{ass:assignment1} and \ref{ass:1iv} hold. Moreover, suppose that
	$\pi(s)$ is the same across $s\in\mathcal{S}$. Then $\hat{\tau}_{S}$ is more
	efficient than $\hat{\tau}_{TSLS}$ in the sense that $\sigma_{S}^{2}\leq
	\sigma_{TSLS}^{2}$.
\end{thm}

Theorem \ref{thm:tau_TSLS} could be deduced from Theorem \ref{thm:linear2}. Both $\hat{\tau}_{TSLS}$ and $\hat{\tau}_{S}$ use linear adjustments of $X_{i}$, but Theorem \ref{thm:tau_TSLS} states that $\hat{\tau}_{S}$ is more efficient than $\hat{\tau}_{TSLS}$. In the discussion following Theorem \ref{thm:linear2}, we further show that $\hat{\tau}_{S}$ achieves the minimum asymptotic variance among the class of estimators with linear adjustments. On the other hand, nonlinear adjustments may be more efficient than the optimal linear adjustment. 


\section{Implementation Details for Sieve and Lasso Regressions}

\label{sec:aux_imp}

\paragraph{Sieve regressions.}

We provide more details on the sieve basis. Recall $\mathring{\Psi}_{i,n}
\equiv(b_{1,n}(x), \cdots, b_{h_{n}, n}(x))^{\top}$, where $\{b_{h,n }%
(\cdot)\}_{h \in[h_{n}]}$ are $h_{n}$ basis functions of a linear sieve space,
denoted as $\mathcal{B}$. Given that all the elements of vector $X$ are
continuously distributed, the sieve space $\mathcal{B}$ can be constructed as follows.

\begin{enumerate}
	\item For each element $X^{(l)}$ of $X$, $l=1,\cdots,d_{x}$, where $d_{x}$
	denotes the dimension of vector $X$, let $\mathcal{B}_{l}$ be the univariate
	sieve space of dimension $J_{n}$. One example of $\mathcal{B}_{l}$ is the
	linear span of the $J_{n}$ dimensional polynomials given by
	\[
	\mathcal{B}_{l} = \biggl\{\sum_{k=0}^{J_{n}}\alpha_{k} x^{k}, x \in
	\text{Supp}(X^{(l)}), \alpha_{k} \in\mathbb{R} \biggr\};
	\]

	Another example is the linear span of $r$-order splines with $J_{n}$ nodes
	given by
	\[
	\mathcal{B}_{l} = \biggl\{\sum_{k=0}^{r-1}\alpha_{k} x^{k} + \sum_{j=1}%
	^{J_{n}}b_{j}[\max(x-t_{j},0)]^{r-1}, x \in\text{Supp}(X^{(l)}), \alpha_{k},
	b_{j} \in\mathbb{R} \biggr\},
	\]
	where the grid $-\infty=t_{0} \leq t_{1} \leq\cdots\leq t_{J_{n}} \leq
	t_{J_{n}+1} = \infty$ partitions $\text{Supp}(X^{(l)})$ into $J_{n}+1$ subsets
	$I_{j} = [t_{j},t_{j+1}) \cap\text{Supp}(X^{(l)})$, $j=1,\cdots,J_{n}-1$,
	$I_{0} = (t_{0},t_{1}) \cap\text{Supp}(X^{(l)})$, and $I_{J_{n}} = (t_{J_{n}%
	},t_{J_{n}+1}) \cap\text{Supp}(X^{(l)})$.
	
	\item Let $\mathcal{B}$ be the tensor product of $\{\mathcal{B}_{l}%
	\}_{l=1}^{d_{x}}$, which is defined as a linear space spanned by the functions
	$\prod_{l=1}^{d_{x}} g_{l}$, where $g_{l} \in\mathcal{B}_{l}$. The dimension
	of $\mathcal{B}$ is then $K \equiv d_{x} J_{n}$ if $\mathcal{B}_{l}$ is
	spanned by $J_{n}$ dimensional polynomials.
\end{enumerate}

We refer interested readers to \cite{HIR03} and \cite{c07} for more details
about the implementation of sieve estimation. Given the sieve basis, we can
compute the $\{\hat{\mu}^{b}(a,s,X_{i})\}_{a = 0,1, b = D,Y, s \in\mathcal{S}%
}$ following \eqref{eq:muhat_np}.

\paragraph{Lasso regressions.}

We follow the estimation procedure and the choice of tuning parameter proposed
by \cite{BCFH13}. We provide details below for completeness. Recall
$\varrho_{n,a}(s) = c \sqrt{n_{a}(s)}F_{N}^{-1}( 1-1/(p_{n}\log(n_{a}(s))))$.
We set $c=1.1$ following \cite{BCFH13}. We then implement the following
algorithm to estimate $\hat{\theta}_{a,s}^{R}$ and $\hat{\beta}_{a,s}^{R}$: 

\begin{enumerate}
	[label=(\roman*)] 
	
	\item Let $\hat{\sigma}_{h}^{Y,(0)} = \frac{1}{n_{a}(s)}\sum_{i \in I_{a}%
		(s)}(Y_{i} - \bar{Y}_{a,s})^{2}\mathring{\Psi}_{i,n,h}^{2}$ and $\hat{\sigma
	}_{h}^{D,(0)} = \frac{1}{n_{a}(s)}\sum_{i \in I_{a}(s)}(D_{i} - \bar{D}%
	_{a,s})^{2}\mathring{\Psi}_{i,n,h}^{2}$ for $h \in[p_{n}]$, where $\bar
	{Y}_{a,s} = \frac{1}{n_{a}(s)}\sum_{i \in I_{a}(s)}Y_{i}$ and $\bar{D}_{a,s} =
	\frac{1}{n_{a}(s)}\sum_{i \in I_{a}(s)}D_{i}$. Estimate
	\begin{align*}
		\hat{\theta}_{a,s}^{R,0} =  &  \argmin_{t} \frac{-1}{n_{a}(s)}\sum_{i \in
			I_{a}(s)} \biggl(Y_{i} - \mathring{\Psi}_{i,n}^{\top}t\biggr)^{2} +
		\frac{\varrho_{n,a}(s)}{n_{a}(s)} \sum_{h \in[p_{n}]} \hat{\sigma}_{h}%
		^{Y,(0)}|t_{h}|,\\
		\hat{\beta}_{a,s}^{R,0} =  &  \argmin_{b} \frac{-1}{n_{a}(s)}\sum_{i \in
			I_{a}(s)} \biggl[D_{i}\log(\lambda(\mathring{\Psi}_{i,n}^{\top}b)) +
		(1-D_{i})\log(1-\lambda(\mathring{\Psi}_{i,n}^{\top}b))\biggr]\\
		+  &  \frac{\varrho_{n,a}(s)}{n_{a}(s)} \sum_{h \in[p_{n}]} \hat{\sigma}%
		_{h}^{D,(0)}|b_{h}|.
	\end{align*}

	\item For $k = 1,\cdots,K$, obtain $\hat{\sigma}_{h}^{Y,(k)} = \sqrt{\frac
		{1}{n}\sum_{i \in[n]} (\mathring{\Psi}_{i,n,h}\hat{\varepsilon}_{i}%
		^{Y,(k)})^{2}}$, where $\hat{\varepsilon}_{i}^{Y,(k)} = Y_{i} - \mathring
	{\Psi}_{i,n}^{\top} \hat{\theta}_{a,s}^{R,k-1}$ and $\hat{\sigma}_{h}^{D,(k)}
	= \sqrt{\frac{1}{n}\sum_{i \in[n]} (\mathring{\Psi}_{i,n,h}\hat{\varepsilon
		}_{i}^{D,(k)})^{2}}$, where $\hat{\varepsilon}_{i}^{D,(k)} = D_{i} -
	\lambda(\mathring{\Psi}_{i,n}^{\top}\hat{\beta}_{a,s}^{R,k-1})$. Estimate
	\begin{align*}
		\hat{\theta}_{a,s}^{R,k} =  &  \argmin_{t} \frac{-1}{n_{a}(s)}\sum_{i \in
			I_{a}(s)} \del[2]{Y_i - \mathring{\Psi}_{i,n}^\top t}^{2} + \frac
		{\varrho_{n,a}(s)}{n_{a}(s)} \sum_{h \in[p_{n}]} \hat{\sigma}_{h}%
		^{Y,(k-1)}|t_{h}|,\\
		\hat{\beta}_{a,s}^{R,k} =  &  \argmin_{b} \frac{-1}{n_{a}(s)}\sum_{i \in
			I_{a}(s)}
		\sbr[2]{D_i\log(\lambda(\mathring{\Psi}_{i,n}^\top b)) + (1-D_i)\log(1-\lambda(\mathring{\Psi}_{i,n}^\top b))}\\
		\qquad+  & \frac{\varrho_{n,a}(s)}{n_{a}(s)} \sum_{h \in[p_{n}]} \hat{\sigma
		}_{h}^{D,(k-1)}|b_{h}|.
	\end{align*}

	\item Let $\hat{\theta}_{a,s}^{R} = \hat{\theta}_{a,s}^{R,K}$ and $\hat{\beta
	}_{a,s}^{R} = \hat{\beta}_{a,s}^{R,K}$. 
\end{enumerate}


\section{Regression Adjustment under Full Compliance}
\label{sec:full_compliance}
In this section, we briefly discuss the regression adjustment under full compliance. We aim to construct consistent and efficient estimators for the average treatment effect (ATE).  Under full compliance, we have $D(a) = a$ for $a = 0,1$ so that $D = A$. The estimator $\hat \mu^D(a,s,x) = a$ is correctly specified. Then, our proposed estimator of ATE is%
\begin{align}
	\hat{\tau}_{ATE}  &  := \frac{1}{n}\sum_{i \in [n]} \left[ \frac{A_{i}(Y_{i} - \hat{\mu}^{Y}(1,S_{i},X_{i}))}{\hat{\pi
		}(S_{i})} - \frac{(1-A_{i})(Y_{i}-\hat{\mu}^{Y}(0,S_{i},X_{i}))}{1-\hat{\pi
		}(S_{i})} + \hat{\mu}^{Y}(1,S_{i},X_{i})-\hat{\mu}^{Y}(0,S_{i},X_{i}) \right], \label{eq:ate}%
\end{align}
where $\hat \mu^Y(a,s,x)$ is an estimator of the working model $\overline \mu^Y(a,s,x)$. 

The optimal linear adjustment is $\hat \mu^Y(a,s,X_i) = \Psi_{i,s}^{\top}\hat \theta_{a,s}^L$, where 
\begin{align}
	\dot{\Psi}_{i,a,s}  &  := \Psi_{i,s} - \frac{1}{n_{a}(s)}\sum_{i\in I_{a}%
		(s)}\Psi_{i,s}\nonumber\\
	\hat{\theta}_{a,s}^{L}  &  :=
	\del[3]{\frac{1}{n_a(s)}\sum_{i \in I_a(s)}\dot{\Psi}_{i,a,s}\dot{\Psi}_{i,a,s}^\top }^{-1}%
	\del[3]{\frac{1}{n_a(s)}\sum_{i \in I_a(s)}\dot{\Psi}_{i,a,s}Y_i }\nonumber.
\end{align}
We can show that such an adjustment achieves the minimal variance of the ATE estimator that is adjusted by linear functions of $\Psi_{i,s}$. 

Let $\mathring{\Psi}_{i,n}$ contains sieve bases of $X_i$. Then, the nonparametric adjustment can be written as $\hat \mu^Y(a,s,X_i) = \mathring{\Psi}_{i,n}^\top \hat{\theta}_{a,s}^{NP}$, where
\begin{align*}
	&  \hat{\theta}_{a,s}^{NP}%
	=\del[3]{\frac{1}{n_a(s)} \sum_{i \in I_a(s)} \mathring{\Psi}_{i,n} \mathring{\Psi}_{i,n}^{\top}}^{-1}%
	\del[3]{\frac{1}{n_a(s)} \sum_{i \in I_a(s)} \mathring{\Psi}_{i,n} Y_i}.
\end{align*}

Last, suppose $\mathring{\Psi}_{i,n}$ contains high-dimensional regressors of $X_i$. Then, the regularized adjustment can be written as $\hat \mu^Y(a,s,X_i) = \mathring{\Psi}_{i,n}^\top \hat{\theta}_{a,s}^{R}$,
where
\begin{align*}
	\hat{\theta}_{a,s}^{R}= &  \argmin_{t}\frac{-1}{n_{a}(s)}\sum_{i\in I_{a}%
		(s)}\del[1]{Y_i - \mathring{\Psi}_{i,n}^\top t}^{2}+\frac{\varrho_{n,a}%
		(s)}{n_{a}(s)}||\hat{\Omega}^{Y}t||_{1},
\end{align*}
$\{\varrho_{n,a}(s)\}_{a=0,1,s\in\mathcal{S}}$ are tuning parameters, $\hat{\Omega}^{Y}=\text{diag}(\hat{\omega}_{1}^{Y},\cdots,\hat{\omega
}_{p_{n}}^{Y})$ is a diagonal matrix of data-dependent penalty loadings as defined in Section
\ref{sec:aux_imp}. Under similar conditions as in Assumptions \ref{ass:np} and \ref{ass:hd}, we can show that the ATE estimator with both the nonparametric and regularized adjustments achieves the semiparametric efficiency bound.

\section{Proof of Theorem \ref{thm:TSLS}}

\label{sec:aux_TSLS_proof}

	We define $\hat    \sigma_{TSLS,naive}^2$ as 
	\begin{align*}
		\hat    \sigma_{TSLS,naive}^2 = e_1^\top[S_{\bar Z,\bar X}^\top S_{\bar Z,\bar Z}^{-1}S_{\bar Z,\bar X}]^{-1} \left[S_{\bar Z,\bar X}^\top S_{\bar Z,\bar Z}^{-1} \left(\frac{1}{n}\sum_{i=1}^n(\bar Z_i \bar Z_i^\top \hat \eps_i^2)\right) S_{\bar Z,\bar Z}^{-1} S_{\bar Z,\bar X}\right] [S_{\bar Z,\bar X}^\top S_{\bar Z,\bar Z}^{-1} S_{\bar Z,\bar X}]^{-1}e_1,
	\end{align*}
	where $\bar X_i = (D_i,\{1\{S_i=s\}\}_{s \in \mathcal{S}},X_i^\top)^\top$ $\bar Z_i = (A_i,\{1\{S_i=s\}\}_{s \in \mathcal{S}},X_i^\top)^\top$, $S_{\bar Z, \bar Z} = \frac{1}{n}\sum_{i \in [n]}\bar Z_i \bar Z_i^\top $,
	$S_{\bar Z, \bar X} = \frac{1}{n}\sum_{i \in [n]}\bar Z_i \bar X_i^\top $, $e_1$ is a vector with its first element being one and the rest being zero, $\hat \eps_i = Y_i - \hat \tau_{TSLS} D_i - \sum_{s\in \mathcal{S}}\hat \alpha_{s,TSLS}1\{S_i=s\} - X_i^\top \hat \delta_{TSLS}$, and $(\hat \tau_{TSLS},\hat \alpha_{s,TSLS},\hat \delta_{TSLS})$ are the usual TSLS estimators. 
	
	Next, we define $\sigma_{TSLS}^{2} $ and $\sigma_{TSLS,naive}^{2}$. Let $\mathbb{X}_{i} = (X_{i}^{\top}, \{1\{S_{i}=s\}\}_{s \in\mathcal{S}%
	})^{\top}$,
	\begin{align*}
		\sigma_{TSLS}^{2}  &  = \frac{ \sigma_{TSLS,0}^{2}+\sigma_{TSLS,1}^{2} +
			\sigma_{TSLS,2}^{2} + \sigma_{TSLS,3}^{2}}{ (\mathbb{E}(D_{i}(1) -
			D_{i}(0)))^{2} },\\
		\sigma_{TSLS,1}^{2}  &  = \frac{\mathbb{E}%
			\sbr[1]{Y_i(D_i(1)) - D_i(1)\tau - \mathbb{X}_i^\top \lambda^* - \mathbb{E}[Y_i(D_i(1)) - D_i(1)\tau - \mathbb{X}_i^\top \lambda^*|S_i]}^{2}%
		}{\pi}\\
		\sigma_{TSLS,0}^{2}  &  = \frac{\mathbb{E}%
			\sbr[1]{Y_i(D_i(0)) - D_i(0)\tau - \mathbb{X}_i^\top \lambda^* - \mathbb{E}[Y_i(D_i(0)) - D_i(0)\tau - \mathbb{X}_i^\top \lambda^*|S_i]}^{2}%
		}{1-\pi},\\
		\sigma_{TSLS,2}^{2}  &  = \mathbb{E}\left[ \mathbb{E}%
		\sbr[1]{Y(D(1)) - Y(D(0)) - (D(1)-D(0))\tau|S_i} \right] ^{2},\\
		\sigma_{TSLS,3}^{2}  &  = \mathbb{E}\left\lbrace \gamma(S_{i}%
		)\del[3]{\mathbb{E}\sbr[3]{\frac{Y_i(D_i(1)) - D_i(1)\tau - \mathbb{X}_i^\top \lambda^*}{\pi}+\frac{Y_i(D_i(0)) - D_i(0)\tau - \mathbb{X}_i^\top \lambda^*}{1-\pi}\bigg|S_i} }^{2}%
		\right\rbrace , \\
		\lambda^{*}  &  = \left( \mathbb{E}\mathbb{X}_{i}\mathbb{X}_{i}^{\top}\right)
		^{-1}\mathbb{E}\mathbb{X}_{i}\left[ \pi(Y_{i}(D_{i}(1)) - D_{i}(1)\tau) +
		(1-\pi)(Y_{i}(D_{i}(0)) - D_{i}(0)\tau)\right] .
	\end{align*}
	Furthermore, define
	\begin{align*}
		& \sigma_{TSLS,naive}^{2}   = \frac{ \sigma_{TSLS,0}^{2}+\sigma_{TSLS,1}^{2} +
			\sigma_{TSLS,2}^{2} + \tilde \sigma_{TSLS,3}^{2}}{ (\mathbb{E}(D_{i}(1) -
			D_{i}(0)))^{2} },\quad \text{where}\\
		& \tilde \sigma_{TSLS,3}^{2} = \mathbb{E}\left\lbrace \pi(1-\pi)\del[3]{\mathbb{E}\sbr[3]{\frac{Y_i(D_i(1)) - D_i(1)\tau - \mathbb{X}_i^\top \lambda^*}{\pi}+\frac{Y_i(D_i(0)) - D_i(0)\tau - \mathbb{X}_i^\top \lambda^*}{1-\pi}\bigg|S_i} }^{2}%
		\right\rbrace. 
	\end{align*}
	By definition, $\sigma_{TSLS}^2 \leq \sigma_{TSLS,naive}^2$. The inequality is strict if $\gamma(s) < \pi(1-\pi)$.

	Define $\tilde{A}_{i}$ as the residual from the regression of $A_{i}$ on
	$X_{i}$ and $\{1\{S_{i}=s\}\}_{s \in\mathcal{S}}$. Then, we have
	\begin{align*}
		\hat{\tau}_{TSLS} = \frac{\sum_{i \in[n]}\tilde{A}_{i}Y_{i}}{\sum_{i \in
				[n]}\tilde{A}_{i}D_{i}}=\frac{\sum_{i\in[n]}\del[1]{A_i-\pi(S_i)}Y_{i}%
			+\sum_{i\in[n]}R_{i}Y_{i}}{\sum_{i\in[n]}\del[1]{A_i-\pi(S_i)}D_{i}+\sum
			_{i\in[n]}R_{i}D_{i}},
	\end{align*}
	where $R_{i} = \tilde{A_{i}} - (A_{i} - \pi(S_{i}))$. We first suppose that
	\begin{align}
		\frac{1}{n}\sum_{i \in[n]} R_{i}Y_{i} = o_{p}(1) \quad\text{and} \quad\frac
		{1}{n}\sum_{i \in[n]} R_{i}D_{i} = o_{p}(1).\label{eq:Ri}%
	\end{align}
	In addition, we note that
	\begin{align*}
		\frac{1}{n}\sum_{i \in[n]} (A_{i}-\pi(S_{i}))Y_{i}  &  = \frac{1}{n}\sum_{i
			\in[n]} A_{i}(1-\pi(S_{i}))Y_{i}(D_{i}(1)) - \frac{1}{n}\sum_{i \in[n]}
		(1-A_{i})\pi(S_{i})Y_{i}(D_{i}(0)).
	\end{align*}
	For the first term on the RHS of the above display, we have
	\begin{align}
		&  \frac{1}{n}\sum_{i \in[n]} A_{i}(1-\pi(S_{i}))Y_{i}(D_{i}(1))\nonumber\\
		&  = \frac{1}{n}\sum_{i \in[n]} A_{i}(1-\pi(S_{i}))(Y_{i}(D_{i}(1)) -
		\mathbb{E}(Y_{i}(D_{i}(1))|S_{i})) + \frac{1}{n}\sum_{i \in[n]} A_{i}%
		(1-\pi(S_{i}))\mathbb{E}(Y_{i}(D_{i}(1))|S_{i})\nonumber\\
		&  = o_{p}(1) + \frac{1}{n}\sum_{i \in[n]} \pi(S_{i})(1-\pi(S_{i}%
		))\mathbb{E}(Y_{i}(D_{i}(1))|S_{i}) + \frac{1}{n}\sum_{s \in\mathcal{S}}%
		B_{n}(s)(1-\pi(s))\mathbb{E}(Y_{i}(D_{i}(1))|S_{i}=s)\nonumber\\
		&  = \mathbb{E}\pi(S_{i})(1-\pi(S_{i}))\mathbb{E}(Y_{i}(D_{i}(1))|S_{i}) +
		o_{p}(1),\label{eq:AY}%
	\end{align}
	where the second equality is by conditional Chebyshev's inequality using the
	facts that
	\begin{align*}
		&  \mathbb{E}
		\sbr[3]{\frac{1}{n}\sum_{i \in [n]} A_i(1-\pi(S_i))(Y_i(D_i(1)) - \mathbb{E}(Y_i(D_i(1))|S_i))\bigg| \{A_i,S_i\}_{i\in [n]} }
		= 0\\
		&  \mathbb{E}
		\sbr[3]{ \del[3]{\frac{1}{n}\sum_{i \in [n]} A_i(1-\pi(S_i))(Y_i(D_i(1)) - \mathbb{E}(Y_i(D_i(1))|S_i))}^2\bigg| \{A_i,S_i\}_{i\in [n]} }\\
		&  \leq\sum_{s \in\mathcal{S}} \frac{n_{1}(s)(1-\pi(s))^{2} \mathbb{E}%
			(Y^{2}(D(1)) |S_{i}=s)}{n^{2}} = o_{p}(1),
	\end{align*}
	and the third equality is by Assumption \ref{ass:assignment1}(iv) and the
	usual LLN. For the same reason, we have
	\begin{align*}
		&  \frac{1}{n}\sum_{i \in[n]} (1-A_{i})\pi(S_{i})Y_{i}(D_{i}(0))
		\overset{p}{\longrightarrow} \mathbb{E}\pi(S_{i})(1-\pi(S_{i}))\mathbb{E}%
		(Y_{i}(D_{i}(0))|S_{i}),\\
		&  \frac{1}{n}\sum_{i \in[n]} A_{i}(1-\pi(S_{i}))D_{i}(1)
		\overset{p}{\longrightarrow} \mathbb{E}\pi(S_{i})(1-\pi(S_{i}))\mathbb{E}%
		(D_{i}(1)|S_{i}),\\
		&  \frac{1}{n}\sum_{i \in[n]} (1-A_{i})\pi(S_{i})D_{i}(0)
		\overset{p}{\longrightarrow} \mathbb{E}\pi(S_{i})(1-\pi(S_{i}))\mathbb{E}%
		D_{i}(0)|S_{i}),
	\end{align*}
	and
	\begin{align*}
		\hat{\tau}_{TSLS} \overset{p}{\longrightarrow} \frac{\mathbb{E}\pi(S_{i}%
			)(1-\pi(S_{i}))(\mathbb{E}(Y_{i}(D_{i}(1))|S_{i})-\mathbb{E}(Y_{i}%
			(D_{i}(0))|S_{i}) ) }{\mathbb{E}\pi(S_{i})(1-\pi(S_{i}))(\mathbb{E}%
			(D_{i}(1)|S_{i})-\mathbb{E}(D_{i}(0)|S_{i}) )}.
	\end{align*}
	Therefore, it is only left to show \eqref{eq:Ri}. Let $\mathbb{X}_{i} =
	(X_{i}^{\top},\{1\{S_{i}=s\}\}_{s \in\mathcal{S}})^{\top}$, $\hat{\theta}$ be
	the OLS coefficient of regressing $A_{i}$ on $\mathbb{X}_{i}$, and $\theta=
	(0_{d_{x}}^{\top}, \{\pi(s)\}_{s \in\mathcal{S}})^{\top}$, where $d_{x}$ is
	the dimension of $X_{i}$. Then, we have $R_{i} = -\mathbb{X}_{i}^{\top}%
	(\hat{\theta} - \theta)$. In order to show \eqref{eq:Ri}, it suffices to show
	$\hat{\theta} \overset{p}{\longrightarrow} \theta$, or equivalently, $\frac
	{1}{n} \sum_{i \in[n]}\mathbb{X}_{i} (A_{i} - \pi(S_{i}))
	\overset{p}{\longrightarrow} 0$. We note that
	\begin{align}
		&  \frac{1}{n} \sum_{i \in[n]}\mathbb{X}_{i} (A_{i} - \pi(S_{i})) = \frac
		{1}{n} \sum_{i \in[n]}(\mathbb{X}_{i} - \mathbb{E}(\mathbb{X}_{i}|S_{i}))
		(A_{i} - \pi(S_{i})) + \frac{1}{n} \sum_{i \in[n]} \mathbb{E}(\mathbb{X}%
		_{i}|S_{i}) (A_{i} - \pi(S_{i}))\nonumber\\
		&  = \frac{1}{n} \sum_{i \in[n]}(\mathbb{X}_{i} - \mathbb{E}(\mathbb{X}%
		_{i}|S_{i})) A_{i}(1 - \pi(S_{i})) - \frac{1}{n} \sum_{i \in[n]}%
		(\mathbb{X}_{i} - \mathbb{E}(\mathbb{X}_{i}|S_{i})) (1-A_{i})\pi(S_{i}) +
		\frac{1}{n} \sum_{s \in\mathcal{S}} \mathbb{E}(\mathbb{X}_{i}|S_{i}%
		=s)B_{n}(s)\nonumber\\
		& = o_{p}(1),\label{eq:tau41}%
	\end{align}
	where the last equality holds following the similar argument in \eqref{eq:AY}.
	This concludes the proof of the first statement.
	
	For the second statement, let $\mathbb{X}_{i} = (X_{i}^{\top}, \{1\{S_{i}%
	=s\}\}_{s \in\mathcal{S}})^{\top}$,
	\begin{align*}
		\hat{\theta} =
		\del[3]{\frac{1}{n}\sum_{i \in [n]}\mathbb{X}_i\mathbb{X}_i^\top }^{-1}
		\del[3]{\frac{1}{n}\sum_{i \in [n]}\mathbb{X}_iA_i},
	\end{align*}
	$\widetilde{A}_{i} = A_{i} -\mathbb{X}_{i}^{\top}\hat{\theta}$, and $\theta=
	(0_{d_{x}}^{\top},\pi,\cdots,\pi)^{\top}$. Then, we have
	\begin{align*}
		\sqrt{n}(\hat{\tau}_{TSLS}-\tau) = \frac{\frac{1}{\sqrt{n}}\sum_{i \in
				[n]}\widetilde{A}_{i}(Y_{i} - D_{i} \tau)}{\frac{1}{n}\sum_{i \in
				[n]}\widetilde{A}_{i}D_{i}}.
	\end{align*}
	By the same argument in the proof of the first statement of Theorem
	\ref{thm:TSLS}, we have
	\begin{align*}
		\frac{1}{n}\sum_{i \in[n]}\widetilde{A}_{i}D_{i} \overset{p}{\longrightarrow}
		\pi(1-\pi)\mathbb{E}(D(1)-D(0)).
	\end{align*}

	Next, we turn to the numerator. We have
	\begin{align*}
		& \frac{1}{\sqrt{n}}\sum_{i \in[n]}\widetilde{A}_{i}(Y_{i} - D_{i} \tau) =
		\frac{1}{\sqrt{n}}\sum_{i \in[n]}(A_{i} - \mathbb{X}_{i}^{\top}\theta-
		\mathbb{X}_{i}^{\top}(\hat{\theta} - \theta))(Y_{i} - D_{i} \tau)\\
		& =\frac{1}{\sqrt{n}}\sum_{i \in[n]}(A_{i} - \pi)(Y_{i} - D_{i} \tau) -
		\frac{1}{n}\sum_{i \in[n]} \mathbb{X}_{i}^{\top}(Y_{i}-D_{i}\tau)
		\del[3]{\frac{1}{n}\sum_{i \in [n]}\mathbb{X}_i\mathbb{X}_i^\top }^{-1}
		\del[3]{\frac{1}{\sqrt{n}}\sum_{i \in [n]}\mathbb{X}_i(A_i - \pi)}.
	\end{align*}
	where the second equality uses the facts that $\mathbb{X}_{i}^{\top}\theta=
	\pi$ and
	\begin{align*}
		\hat{\theta} - \theta &  =
		\del[3]{\frac{1}{n}\sum_{i \in [n]}\mathbb{X}_i\mathbb{X}_i^\top }^{-1}
		\del[3]{\frac{1}{n}\sum_{i \in [n]}\mathbb{X}_i(A_i - \mathbb{X}_i^\top \theta)}
		= \del[3]{\frac{1}{n}\sum_{i \in [n]}\mathbb{X}_i\mathbb{X}_i^\top }^{-1}
		\del[3]{\frac{1}{n}\sum_{i \in [n]}\mathbb{X}_i(A_i - \pi)}.
	\end{align*}

	We first consider the joint convergence of $\frac{1}{\sqrt{n}}\sum_{i \in
		[n]}(A_{i} - \pi)(Y_{i} - D_{i} \tau) $ and $\frac{1}{\sqrt{n}}\sum_{i \in
		[n]}\mathbb{X}_{i}(A_{i} - \pi)$. Let $\lambda_{1}$ be a scalar and
	$\lambda_{2} \in\Re^{d_{x}}$. Then, it suffices to consider the weak
	convergence of $\frac{1}{\sqrt{n}}\sum_{i \in[n]}(A_{i} - \pi)(\lambda
	_{1}(Y_{i} - D_{i}\tau) + \lambda_{2}^{\top}\mathbb{X}_{i})$. Let $\varpi_{i}=\lambda_{1}(Y_{i} - D_{i}\tau) + \lambda_{2}^{\top}\mathbb{X}_{i}$ and
	$\varpi_{i}(a) = \lambda_{1}(Y_{i}(D_{i}(a)) - D_{i}(a)\tau) + \lambda_{2}%
	^{\top}\mathbb{X}_{i}$. Note that $\varpi_{i}=A_{i}\varpi_{i}(1)+(1-A_{i})\varpi_{i}(0)$. We have 
	\begin{align}
		& \frac{1}{\sqrt{n}}\sum_{i \in[n]}(A_{i} - \pi)\varpi_{i} = \frac{1}{\sqrt{n}%
		}\sum_{i \in[n]}\left[ A_{i}(1 - \pi)\varpi_{i}(1) - (1-A_{i})\pi\varpi
		_{i}(0)\right] \nonumber\\
		&  = \frac{1}{\sqrt{n}}\sum_{i \in[n]}\left[ A_{i}(1 - \pi)(\varpi
		_{i}(1)-\mathbb{E}(\varpi_{i}(1)|S_{i})) - (1-A_{i})\pi(\varpi_{i}(0)-\mathbb{E}%
		(\varpi_{i}(0)|S_{i}))\right] \nonumber\\
		& \qquad+ \frac{1}{\sqrt{n}}\sum_{i \in[n]}\left[ A_{i}(1 - \pi)\mathbb{E}%
		(\varpi_{i}(1)|S_{i}) - (1-A_{i})\pi\mathbb{E}(\varpi_{i}(0)|S_{i})\right]
		\nonumber\\
		&  = \frac{1}{\sqrt{n}}\sum_{i \in[n]}\left[ A_{i}(1 - \pi)(\varpi
		_{i}(1)-\mathbb{E}(\varpi_{i}(1)|S_{i})) - (1-A_{i})\pi(\varpi_{i}(0)-\mathbb{E}%
		(\varpi_{i}(0)|S_{i}))\right] \nonumber\\
		&  \qquad+ \frac{1}{\sqrt{n}} \sum_{s \in\mathcal{S}} B_{n}(s)\left[ (1 -
		\pi)\mathbb{E}(\varpi_{i}(1)|S_{i}=s)+\pi\mathbb{E}(\varpi_{i}(0)|S_{i}=s)\right]
		+ \frac{\pi(1-\pi)}{\sqrt{n}}\sum_{i \in[n]}\mathbb{E}(\varpi_{i}(1)-\varpi
		_{i}(0)|S_{i})\nonumber\\
		&  = \frac{1}{\sqrt{n}}\sum_{i \in[n]}\left[ A_{i}(1 - \pi)(\varpi
		_{i}(1)-\mathbb{E}(\varpi_{i}(1)|S_{i})) - (1-A_{i})\pi(\varpi_{i}(0)-\mathbb{E}%
		(\varpi_{i}(0)|S_{i}))\right] \nonumber\\
		&  \qquad+ \frac{1}{\sqrt{n}} \sum_{s \in\mathcal{S}} B_{n}(s)\left[ (1 -
		\pi)\mathbb{E}(\varpi_{i}(1)|S_{i}=s)+\pi\mathbb{E}(\varpi_{i}(0)|S_{i}=s)\right]
		\nonumber\\
		&  \qquad+ \frac{\pi(1-\pi)}{\sqrt{n}}\sum_{i \in[n]}\left( \mathbb{E}%
		(\varpi_{i}(1)-\varpi_{i}(0)|S_{i}) - \mathbb{E}(\varpi_{i}(1) - \varpi_{i}(0))
		\right) \nonumber\\
		&  \rightsquigarrow\mathcal{N}(0,\Sigma^{2}),\label{r1s}%
	\end{align}
	where
	\begin{align*}
		\Sigma^{2}  &  = (1-\pi)\pi\left[  (1-\pi)\mathbb{E}%
		\sbr[1]{\varpi_i(1) - \mathbb{E}(\varpi_i(1)|S_i)}^{2} + \pi\mathbb{E}%
		\sbr[1]{\varpi_i(0) - \mathbb{E}(\varpi_i(0)|S_i)}^{2}\right] \\
		& \qquad+ \mathbb{E}%
		\sbr[3]{\gamma(S_i)\del[2]{\mathbb{E}\sbr[1]{(1-\pi)\varpi_i(1) + \pi\varpi_i(0)|S_i} }^2 }+
		\pi^{2}(1-\pi)^{2} \mathbb{E}%
		\del[2]{\mathbb{E}\sbr[1]{\varpi_i(1) - \varpi_i(0)|S_i} }^{2},
	\end{align*}
	the last convergence in distribution is by a similar argument in the proof of
	\citet[Lemma B.2]{BCS17} and the fact that
	\begin{align*}
		\mathbb{E}(\varpi_{i}(1) - \varpi_{i}(0)) = \lambda_{1}\mathbb{E}( Y_{i}(D_{i}(1))
		- Y_{i}(D_{i}(1)) - (D_{i}(1) - D_{i}(0))\tau) = 0.
	\end{align*}
	Thus (\ref{r1s}) implies both $\frac{1}{\sqrt{n}}\sum_{i \in[n]}(A_{i} -
	\pi)(Y_{i} - D_{i} \tau) $ and $\frac{1}{\sqrt{n}}\sum_{i \in[n]}%
	\mathbb{X}_{i}(A_{i} - \pi)$ are $O_{p}(1)$. In addition, let $\hat{\lambda} =
	\left( \frac{1}{n}\sum_{i \in[n]}\mathbb{X}_{i}\mathbb{X}_{i}^{\top}\right)
	^{-1}\frac{1}{n}\sum_{i \in[n]} \mathbb{X}_{i} (Y_{i}-D_{i}\tau)$. We can
	show
	\begin{align*}
		\hat{\lambda} \overset{p}{\longrightarrow} \lambda^{*} : = \left(
		\mathbb{E}\mathbb{X}_{i}\mathbb{X}_{i}^{\top}\right) ^{-1}\mathbb{E}%
		\mathbb{X}_{i}\left[ \pi(Y_{i}(D_{i}(1)) - D_{i}(1)\tau) + (1-\pi)(Y_{i}%
		(D_{i}(0)) - D_{i}(0)\tau)\right] .
	\end{align*}
	Therefore, by letting $\lambda_{1} = 1$ and $\lambda_{2} = \lambda^{*}$, we
	have
	\begin{align*}
		\sqrt{n}(\hat{\tau}_{TSLS} - \tau) \rightsquigarrow\mathcal{N}(0,\sigma_{TSLS}^{2}),
	\end{align*}
	where
	\begin{align*}
		\sigma_{TSLS}^{2}  &  = \frac{\sigma_{TSLS,0}^{2}+\sigma_{TSLS,1}^{2}+\sigma_{TSLS,2}^{2}+\sigma_{TSLS,3}^{2}}{ (\mathbb{E}(D_{i}(1) - D_{i}(0)))^{2} },\\
		\sigma_{TSLS,0}^{2}  &  = \frac{\mathbb{E}%
			\sbr[1]{Y_i(D_i(0)) - D_i(0)\tau - \mathbb{X}_i^\top \lambda^* - \mathbb{E}[Y_i(D_i(0)) - D_i(0)\tau - \mathbb{X}_i^\top \lambda^*|S_i]}^{2}%
		}{1-\pi},\\
		\sigma_{TSLS,1}^{2}  &  = \frac{\mathbb{E}%
			\sbr[1]{Y_i(D_i(1)) - D_i(1)\tau - \mathbb{X}_i^\top \lambda^* - \mathbb{E}[Y_i(D_i(1)) - D_i(1)\tau - \mathbb{X}_i^\top \lambda^*|S_i]}^{2}%
		}{\pi},\\
		\sigma_{TSLS,2}^{2}  &  = \mathbb{E}\left[ \mathbb{E}%
		\sbr[1]{Y(D(1)) - Y(D(0)) - (D(1)-D(0))\tau|S_i} \right] ^{2},\\
		\sigma_{TSLS,3}^{2}  &  = \mathbb{E}\left\lbrace \gamma(S_{i}%
		)\del[3]{\mathbb{E}\sbr[3]{\frac{Y_i(D_i(1)) - D_i(1)\tau - \mathbb{X}_i^\top \lambda^*}{\pi}+\frac{Y_i(D_i(0)) - D_i(0)\tau - \mathbb{X}_i^\top \lambda^*}{1-\pi}\bigg|S_i} }^{2}
		\right\rbrace .
	\end{align*}
	
	To see the second result, we note that $\bar X_i = (D_i,\mathbb{X}_{i}^\top)^\top$ and $\bar Z_i = (A_i,\mathbb{X}_{i}^\top)^\top$. Denote $\breve Z_i = (\tilde A_i, \mathbb X_i^\top)^\top$. Then, we have 
	\begin{align*}
		& e_1^\top [S_{\bar X, \bar Z} S_{\bar Z, \bar Z}^{-1} S_{\bar Z, \bar X}]^{-1} \\
		& = [S_{\bar X, \breve Z} S_{\breve Z, \breve Z}^{-1} S_{\breve Z, \bar X}]^{-1}\\
		& = e_1^\top\biggl\{ \begin{pmatrix}
			\sum_{i \in [n]}D_i\tilde A_i/n & \sum_{i \in [n]}D_i \mathbb X_i^\top/n \\
			0 & \sum_{i \in [n]} \mathbb X_i \mathbb X_i^\top/n
		\end{pmatrix} \begin{pmatrix}
			\sum_{i \in [n]}\tilde A_i^2/n & 0 \\
			0 & \sum_{i \in [n]} \mathbb X_i \mathbb X_i^\top/n
		\end{pmatrix}^{-1} \\
		& \times \begin{pmatrix}
			\sum_{i \in [n]}D_i\tilde A_i/n & 0 \\
			\sum_{i \in [n]}D_i \mathbb X_i/n & \sum_{i \in [n]} \mathbb X_i \mathbb X_i^\top/n
		\end{pmatrix}\biggr\}^{-1} \\
		& = e_1^\top\begin{pmatrix}
			\frac{(\sum_{i \in [n]}D_i\tilde A_i/n)^2}{(\sum_{i \in [n]}\tilde A_i^2/n)}+ \sum_{i \in [n]}D_i \mathbb X_i^\top/n \left[\sum_{i \in [n]} \mathbb X_i \mathbb X_i^\top/n\right]^{-1}\sum_{i \in [n]}D_i \mathbb X_i/n& \sum_{i \in [n]}D_i \mathbb X_i^\top/n \\
			\sum_{i \in [n]}D_i \mathbb X_i/n & \sum_{i \in [n]} \mathbb X_i \mathbb X_i^\top/n
		\end{pmatrix}^{-1}\\
		& \convP [\pi(1-\pi)]^{-1}(\mathbb{E}(D(1) - D(0)))^{-2} \begin{pmatrix}
			1 & -\gamma_D^\top 
		\end{pmatrix}.
	\end{align*}
	and
	\begin{align*}
		S_{\bar X,\bar{Z}}S_{\bar Z,\bar Z} \bar Z_i & = S_{\bar X,\breve{Z}}S_{\breve Z,\breve Z} \breve Z_i  \\
		& = \begin{pmatrix}
			\sum_{i \in [n]}D_i\tilde A_i/n & \sum_{i \in [n]}D_i \mathbb X_i^\top/n \\
			0 & \sum_{i \in [n]} \mathbb X_i \mathbb X_i^\top/n
		\end{pmatrix} \begin{pmatrix}
			\sum_{i \in [n]}\tilde A_i^2/n & 0 \\
			0 & \sum_{i \in [n]} \mathbb X_i \mathbb X_i^\top/n
		\end{pmatrix}^{-1} \begin{pmatrix}
			\tilde A_i \\
			\mathbb X_i
		\end{pmatrix} \\
		& = \begin{pmatrix}
			(\sum_{i \in [n]}D_i\tilde A_i/n)(\sum_{i \in [n]}\tilde A_i^2/n)^{-1} & (\sum_{i \in [n]}D_i \mathbb X_i^\top/n)(\sum_{i \in [n]} \mathbb X_i \mathbb X_i^\top/n)^{-1} \\
			0 & I
		\end{pmatrix}\begin{pmatrix}
			\tilde A_i \\
			\mathbb X_i
		\end{pmatrix},
	\end{align*}
	where $\gamma_D = (\mathbb E \mathbb X_i \mathbb X_i^\top)^{-1}\mathbb{E}(\mathbb X_i (\pi D_i(1) + (1-\pi)D_i(0)))$. Further note that
	\begin{align*}
		\begin{pmatrix}
			(\sum_{i \in [n]}D_i\tilde A_i/n)(\sum_{i \in [n]}\tilde A_i^2/n)^{-1} & (\sum_{i \in [n]}D_i \mathbb X_i^\top/n)(\sum_{i \in [n]} \mathbb X_i \mathbb X_i^\top/n)^{-1} \\
			0 & I
		\end{pmatrix} \convP \begin{pmatrix}
			\mathbb{E}(D(1) - D(0)) & \gamma_D^\top \\
			0 & I
		\end{pmatrix}
	\end{align*}
	and
	\begin{align*}
		\hat \lambda_{TSLS} \equiv \begin{pmatrix}
			\hat \alpha_{1,TSLS} \\
			\vdots \\
			\hat \alpha_{S,TSLS} \\
			\hat \theta_{TSLS}
		\end{pmatrix} & = (\sum_{i \in [n]}\mathbb X_i \mathbb X_i^\top/n)^{-1} (\sum_{i \in [n]}\mathbb X_i (Y_i - D_i \hat \tau_{TSLS})/n) \\
		& = \hat \lambda + (\sum_{i \in [n]}\mathbb X_i \mathbb X_i^\top/n)^{-1} (\sum_{i \in [n]}\mathbb X_i D_i/n) (\tau-\hat \tau_{TSLS}) \convP \lambda^*.
	\end{align*}
	Then, we have
	\begin{align*}
		\hat e_i = e_i - D_i(\hat \tau_{TSLS} - \tau) - \mathbb X_i^\top (\hat \lambda_{TSLS} - \lambda^*),
	\end{align*}
	where $e_i = Y_i - D_i\tau - \mathbb X_i^\top \lambda^*$. In addition, as shown above, we have $\tilde A_i  = A_i - \pi - \mathbb X_i^\top (\hat \theta - \theta)$ and $\hat \theta \convP \theta$. This implies, 
	\begin{align*}
		& \frac{1}{n}  \sum_{i \in [n]}\hat e_i^2 \begin{pmatrix}
			\tilde A_i^2 & \tilde A_i \mathbb X_i^\top  \\
			\tilde A_i \mathbb X_i & \mathbb X_i \mathbb X_i^\top  
		\end{pmatrix} \\
		& = \frac{1}{n}  \sum_{i \in [n]} e_i^2 \begin{pmatrix}
			(A_i-\pi)^2 & (A_i-\pi) \mathbb X_i^\top  \\
			(A_i-\pi) \mathbb X_i & \mathbb X_i \mathbb X_i^\top  
		\end{pmatrix} + o_P(1) \\
		& = \frac{1}{n}  \sum_{i \in [n]} \left[A_i(Y_i(D_i(1)) - D_i(1)\tau - \mathbb X_i^\top \lambda^*)^2 + (1-A_i)(Y_i(D_i(0)) - D_i(0)\tau - \mathbb X_i^\top \lambda^*)^2 \right] \\
		& \times \begin{pmatrix}
			(A_i-\pi)^2 & (A_i-\pi) \mathbb X_i^\top  \\
			(A_i-\pi) \mathbb X_i & \mathbb X_i \mathbb X_i^\top  
		\end{pmatrix} + o_P(1)\\
		& = \frac{1}{n}\sum_{i \in [n]}A_i\begin{pmatrix}
			(1-\pi)^2(Y_i(D_i(1)) - D_i(1)\tau - \mathbb X_i^\top \lambda^*)^2 & (1-\pi)(Y_i(D_i(1)) - D_i(1)\tau - \mathbb X_i^\top \lambda^*)^2\mathbb X_i^\top \\ 
			(1-\pi)(Y_i(D_i(1)) - D_i(1)\tau - \mathbb X_i^\top \lambda^*)^2\mathbb X_i & (Y_i(D_i(1)) - D_i(1)\tau - \mathbb X_i^\top \lambda^*)^2\mathbb X_i \mathbb X_i^\top
		\end{pmatrix} \\
		& + \frac{1}{n}\sum_{i \in [n]}(1-A_i)\begin{pmatrix}
			\pi^2(Y_i(D_i(0)) - D_i(0)\tau - \mathbb X_i^\top \lambda^*)^2 & -\pi(Y_i(D_i(0)) - D_i(0)\tau - \mathbb X_i^\top \lambda^*)^2\mathbb X_i^\top \\ 
			-\pi(Y_i(D_i(0)) - D_i(0)\tau - \mathbb X_i^\top \lambda^*)^2\mathbb X_i & (Y_i(D_i(0)) - D_i(0)\tau - \mathbb X_i^\top \lambda^*)^2\mathbb X_i \mathbb X_i^\top
		\end{pmatrix} + o_P(1).
	\end{align*}
	For the first term on the RHS of the above display, we have
	\begin{align*}
		& \frac{1}{n}\sum_{i \in [n]}A_i\begin{pmatrix}
			(1-\pi)^2(Y_i(D_i(1)) - D_i(1)\tau - \mathbb X_i^\top \lambda^*)^2 & (1-\pi)(Y_i(D_i(1)) - D_i(1)\tau - \mathbb X_i^\top \lambda^*)^2\mathbb X_i^\top \\ 
			(1-\pi)(Y_i(D_i(1)) - D_i(1)\tau - \mathbb X_i^\top \lambda^*)^2\mathbb X_i & (Y_i(D_i(1)) - D_i(1)\tau - \mathbb X_i^\top \lambda^*)^2\mathbb X_i \mathbb X_i^\top
		\end{pmatrix} \\
		& = \frac{1}{n}\sum_{i \in [n]}A_i\biggl\{\begin{pmatrix}
			(1-\pi)^2(Y_i(D_i(1)) - D_i(1)\tau - \mathbb X_i^\top \lambda^*)^2 & (1-\pi)(Y_i(D_i(1)) - D_i(1)\tau - \mathbb X_i^\top \lambda^*)^2\mathbb X_i^\top \\ 
			(1-\pi)(Y_i(D_i(1)) - D_i(1)\tau - \mathbb X_i^\top \lambda^*)^2\mathbb X_i & (Y_i(D_i(1)) - D_i(1)\tau - \mathbb X_i^\top \lambda^*)^2\mathbb X_i \mathbb X_i^\top
		\end{pmatrix}  \\
		&-    \mathbb{E}\left[\begin{pmatrix}
			(1-\pi)^2(Y_i(D_i(1)) - D_i(1)\tau - \mathbb X_i^\top \lambda^*)^2 & (1-\pi)(Y_i(D_i(1)) - D_i(1)\tau - \mathbb X_i^\top \lambda^*)^2\mathbb X_i^\top \\ 
			(1-\pi)(Y_i(D_i(1)) - D_i(1)\tau - \mathbb X_i^\top \lambda^*)^2\mathbb X_i & (Y_i(D_i(1)) - D_i(1)\tau - \mathbb X_i^\top \lambda^*)^2\mathbb X_i \mathbb X_i^\top
		\end{pmatrix}    \bigg| S_i\right] \biggr\} \\
		& +     \frac{1}{n}\sum_{i \in [n]}(A_i-\pi)\mathbb{E}\left[\begin{pmatrix}
			(1-\pi)^2(Y_i(D_i(1)) - D_i(1)\tau - \mathbb X_i^\top \lambda^*)^2 & (1-\pi)(Y_i(D_i(1)) - D_i(1)\tau - \mathbb X_i^\top \lambda^*)^2\mathbb X_i^\top \\ 
			(1-\pi)(Y_i(D_i(1)) - D_i(1)\tau - \mathbb X_i^\top \lambda^*)^2\mathbb X_i & (Y_i(D_i(1)) - D_i(1)\tau - \mathbb X_i^\top \lambda^*)^2\mathbb X_i \mathbb X_i^\top
		\end{pmatrix}    \bigg| S_i\right] \\
		& + \frac{1}{n}\sum_{i \in [n]}\pi\mathbb{E}\left[\begin{pmatrix}
			(1-\pi)^2(Y_i(D_i(1)) - D_i(1)\tau - \mathbb X_i^\top \lambda^*)^2 & (1-\pi)(Y_i(D_i(1)) - D_i(1)\tau - \mathbb X_i^\top \lambda^*)^2\mathbb X_i^\top \\ 
			(1-\pi)(Y_i(D_i(1)) - D_i(1)\tau - \mathbb X_i^\top \lambda^*)^2\mathbb X_i & (Y_i(D_i(1)) - D_i(1)\tau - \mathbb X_i^\top \lambda^*)^2\mathbb X_i \mathbb X_i^\top
		\end{pmatrix}    \bigg| S_i\right] \\
		& \convP \pi\mathbb{E}\left[\begin{pmatrix}
			(1-\pi)^2(Y_i(D_i(1)) - D_i(1)\tau - \mathbb X_i^\top \lambda^*)^2 & (1-\pi)(Y_i(D_i(1)) - D_i(1)\tau - \mathbb X_i^\top \lambda^*)^2\mathbb X_i^\top \\ 
			(1-\pi)(Y_i(D_i(1)) - D_i(1)\tau - \mathbb X_i^\top \lambda^*)^2\mathbb X_i & (Y_i(D_i(1)) - D_i(1)\tau - \mathbb X_i^\top \lambda^*)^2\mathbb X_i \mathbb X_i^\top
		\end{pmatrix}\right]\\
		& \equiv \Omega_1.
	\end{align*}
	To see the convergence in probability in the above display, we note that 
	\begin{align*}
		& A_i\biggl\{\begin{pmatrix}
			(1-\pi)^2(Y_i(D_i(1)) - D_i(1)\tau - \mathbb X_i^\top \lambda^*)^2 & (1-\pi)(Y_i(D_i(1)) - D_i(1)\tau - \mathbb X_i^\top \lambda^*)^2\mathbb X_i^\top \\ 
			(1-\pi)(Y_i(D_i(1)) - D_i(1)\tau - \mathbb X_i^\top \lambda^*)^2\mathbb X_i & (Y_i(D_i(1)) - D_i(1)\tau - \mathbb X_i^\top \lambda^*)^2\mathbb X_i \mathbb X_i^\top
		\end{pmatrix}  \\
		&-    \mathbb{E}\left[\begin{pmatrix}
			(1-\pi)^2(Y_i(D_i(1)) - D_i(1)\tau - \mathbb X_i^\top \lambda^*)^2 & (1-\pi)(Y_i(D_i(1)) - D_i(1)\tau - \mathbb X_i^\top \lambda^*)^2\mathbb X_i^\top \\ 
			(1-\pi)(Y_i(D_i(1)) - D_i(1)\tau - \mathbb X_i^\top \lambda^*)^2\mathbb X_i & (Y_i(D_i(1)) - D_i(1)\tau - \mathbb X_i^\top \lambda^*)^2\mathbb X_i \mathbb X_i^\top
		\end{pmatrix}    \bigg| S_i\right] \biggr\} 
	\end{align*}
	is independent and conditionally mean zero given $(A^{(n)},S^{(n)})$. Therefore, by the conditional Chebyshev's inequality, we have
	\begin{align*}
		& \frac{1}{n}\sum_{i \in [n]}A_i\biggl\{\begin{pmatrix}
			(1-\pi)^2(Y_i(D_i(1)) - D_i(1)\tau - \mathbb X_i^\top \lambda^*)^2 & (1-\pi)(Y_i(D_i(1)) - D_i(1)\tau - \mathbb X_i^\top \lambda^*)^2\mathbb X_i^\top \\ 
			(1-\pi)(Y_i(D_i(1)) - D_i(1)\tau - \mathbb X_i^\top \lambda^*)^2\mathbb X_i & (Y_i(D_i(1)) - D_i(1)\tau - \mathbb X_i^\top \lambda^*)^2\mathbb X_i \mathbb X_i^\top
		\end{pmatrix}  \\
		&-    \mathbb{E}\left[\begin{pmatrix}
			(1-\pi)^2(Y_i(D_i(1)) - D_i(1)\tau - \mathbb X_i^\top \lambda^*)^2 & (1-\pi)(Y_i(D_i(1)) - D_i(1)\tau - \mathbb X_i^\top \lambda^*)^2\mathbb X_i^\top \\ 
			(1-\pi)(Y_i(D_i(1)) - D_i(1)\tau - \mathbb X_i^\top \lambda^*)^2\mathbb X_i & (Y_i(D_i(1)) - D_i(1)\tau - \mathbb X_i^\top \lambda^*)^2\mathbb X_i \mathbb X_i^\top
		\end{pmatrix}    \bigg| S_i\right] \biggr\} \\
		& = o_P(1).
	\end{align*}
	
	Also, by Assumption \ref{ass:1iv}, we have
	\begin{align*}
		& \frac{1}{n}\sum_{i \in [n]}(A_i-\pi)\mathbb{E}\left[\begin{pmatrix}
			(1-\pi)^2(Y_i(D_i(1)) - D_i(1)\tau - \mathbb X_i^\top \lambda^*)^2 & (1-\pi)(Y_i(D_i(1)) - D_i(1)\tau - \mathbb X_i^\top \lambda^*)^2\mathbb X_i^\top \\ 
			(1-\pi)(Y_i(D_i(1)) - D_i(1)\tau - \mathbb X_i^\top \lambda^*)^2\mathbb X_i & (Y_i(D_i(1)) - D_i(1)\tau - \mathbb X_i^\top \lambda^*)^2\mathbb X_i \mathbb X_i^\top
		\end{pmatrix}    \bigg| S_i\right]\\
		& = o_P(1).
	\end{align*}
	Last, by the usual Law of Large numbers for i.i.d. data, we have
	\begin{align*}
		& \frac{1}{n}\sum_{i \in [n]}\pi\mathbb{E}\left[\begin{pmatrix}
			(1-\pi)^2(Y_i(D_i(1)) - D_i(1)\tau - \mathbb X_i^\top \lambda^*)^2 & (1-\pi)(Y_i(D_i(1)) - D_i(1)\tau - \mathbb X_i^\top \lambda^*)^2\mathbb X_i^\top \\ 
			(1-\pi)(Y_i(D_i(1)) - D_i(1)\tau - \mathbb X_i^\top \lambda^*)^2\mathbb X_i & (Y_i(D_i(1)) - D_i(1)\tau - \mathbb X_i^\top \lambda^*)^2\mathbb X_i \mathbb X_i^\top
		\end{pmatrix}    \bigg| S_i\right] \\
		& \convP \Omega_1.
	\end{align*}
	
	Similarly, we have
	\begin{align*}
		& \frac{1}{n}\sum_{i \in [n]}(1-A_i)\begin{pmatrix}
			\pi^2(Y_i(D_i(0)) - D_i(0)\tau - \mathbb X_i^\top \lambda^*)^2 & -\pi(Y_i(D_i(0)) - D_i(0)\tau - \mathbb X_i^\top \lambda^*)^2\mathbb X_i^\top \\ 
			-\pi(Y_i(D_i(0)) - D_i(0)\tau - \mathbb X_i^\top \lambda^*)^2\mathbb X_i & (Y_i(D_i(0)) - D_i(0)\tau - \mathbb X_i^\top \lambda^*)^2\mathbb X_i \mathbb X_i^\top
		\end{pmatrix} \\
		& \convP (1-\pi)\mathbb{E}\left[ \begin{pmatrix}
			\pi^2(Y_i(D_i(0)) - D_i(0)\tau - \mathbb X_i^\top \lambda^*)^2 & -\pi(Y_i(D_i(0)) - D_i(0)\tau - \mathbb X_i^\top \lambda^*)^2\mathbb X_i^\top \\ 
			-\pi(Y_i(D_i(0)) - D_i(0)\tau - \mathbb X_i^\top \lambda^*)^2\mathbb X_i & (Y_i(D_i(0)) - D_i(0)\tau - \mathbb X_i^\top \lambda^*)^2\mathbb X_i \mathbb X_i^\top
		\end{pmatrix}\right] \\
		& \equiv \Omega_0.
	\end{align*}
	
	Consequently, we have
	\begin{align*}
		\hat \sigma_{TSLS,naive}^2 & \convP \frac{\begin{pmatrix}
				1 & -\gamma_D^\top \end{pmatrix} \begin{pmatrix}
				\mathbb{E}(D(1) - D(0)) & \gamma_D^\top \\
				0 & I
			\end{pmatrix} (\Omega_1 + \Omega_0) \begin{pmatrix}
				\mathbb{E}(D(1) - D(0)) & \gamma_D^\top \\
				0 & I
			\end{pmatrix}^\top \begin{pmatrix}
				1 \\
				-\gamma_D
		\end{pmatrix}}{[\pi(1-\pi)]^2[\mathbb E(D(1) - D(0))]^4} \\
		& = \frac{\pi^{-1} \mathbb{E}(Y_i(D_i(1)) - D_i(1)\tau - \mathbb X_i^\top \lambda^*)^2 + (1-\pi)^{-1} \mathbb{E}(Y_i(D_i(1)) - D_i(1)\tau - \mathbb X_i^\top \lambda^*)^2 }{[\mathbb E(D(1) - D(0))]^2} \\
		& = \frac{\sigma_{TSLS,0}^{2}+\sigma_{TSLS,1}^{2}+\sigma_{TSLS,2}^{2}+\tilde \sigma_{TSLS,3}^{2}}{ (\mathbb{E}(D_{i}(1) - D_{i}(0)))^{2}}.
	\end{align*}
	
	For the last result, by the proof Theorem \ref{thm:est} with $\overline{\mu}^b(a,s,x) = 0$ for $a = 0,1$ and $b = D,Y$ and $\pi(s) = \pi$, we have
	\begin{align*}
		\sigma_{NA}^{2} & = \frac{\sum_{s \in S} \frac{p(s)}{\pi}Var(Y(D(1))- \tau D(1)|S=s) +
			\sum_{s \in S}\frac{p(s)}{1-\pi}Var(Y(D(0)) - \tau D(0)|S=s)}{\mathbb{P}(D(1)>D(0))^{2}} \\
		& + \frac{ Var(\mathbb{E}[W_i-Z_i|S_i] -
			\tau\left( \mathbb{E}[D_{i}(1)-D_{i}(0)|S_{i}]
			\right))%
		}{\mathbb{P}(D(1)>D(0))^{2}} \\
		& =  \frac{\mathbb E\frac{1}{\pi}Var(Y(D(1))- \tau D(1)|S) +\frac{1}{1-\pi}Var(Y(D(0)) - \tau D(0)|S)}{\mathbb{P}(D(1)>D(0))^{2}} + \frac{\sigma_{TSLS,2}^2}{\mathbb{P}(D(1)>D(0))^{2}}.
	\end{align*}
	
	Then, we have $\sigma_{NA}^2 < \sigma_{TSLS}^2$ if and only if
	$$\mathbb E\left[\frac{1}{\pi}Var(Y(D(1))- \tau D(1)|S) +\frac{1}{1-\pi}Var(Y(D(0)) - \tau D(0)|S)\right]<\sigma_{TSLS,0}^2 + \sigma_{TSLS,1}^2 + \sigma_{TSLS,3}^2,$$
	which is equivalent to 
	\begin{align*}
		2 \left[\frac{\mathbb E cov(Y_i(D_i(1)) - D_i(1)\tau, \mathbb X_i^\top\lambda^*|S)}{\pi} + \frac{\mathbb E cov(Y_i(D_i(0)) - D_i(0)\tau, \mathbb X_i^\top\lambda^*|S)}{1-\pi} \right] \leq \frac{\mathbb E Var(\mathbb X_i^\top \lambda^*|S)}{\pi(1-\pi)} + \sigma_{TSLS,3}^2.
	\end{align*}

	\section{Proof of Theorem \ref{thm:est}}
	
	Let
	\begin{align*}
		G  &  :=\mathbb{E}\sbr[2]{\del[1]{Y(1)-Y(0)}\del[1]{D(1)-D(0)}},\\
		H  &  := \mathbb{E}\sbr[1]{D(1)-D(0)},\\
		\hat{G} & := \frac{1}{n}\sum_{i \in[n]}\left[ \frac{A_{i}(Y_{i} - \hat{\mu
			}^{Y}(1,S_{i},X_{i}))}{\hat{\pi}(S_{i})} - \frac{(1-A_{i})(Y_{i}-\hat{\mu}%
			^{Y}(0,S_{i},X_{i}))}{1-\hat{\pi}(S_{i})} + \hat{\mu}^{Y}(1,S_{i},X_{i}%
		)-\hat{\mu}^{Y}(0,S_{i},X_{i}) \right] ,\\
		\hat{H}  & := \frac{1}{n}\sum_{i \in[n]}\left[ \frac{A_{i}(D_{i} - \hat{\mu
			}^{D}(1,S_{i},X_{i}))}{\hat{\pi}(S_{i})} - \frac{(1-A_{i})(D_{i}-\hat{\mu}%
			^{D}(0,S_{i},X_{i}))}{1-\hat{\pi}(S_{i})} + \hat{\mu}^{D}(1,S_{i},X_{i}%
		)-\hat{\mu}^{D}(0,S_{i},X_{i}) \right]  .
	\end{align*}
	Then, we have
	\begin{align}
		\sqrt{n}(\hat{\tau}-\tau) & =\sqrt{n}%
		\del[3]{\frac{\hat{G}}{\hat{H}}-\frac{G}{H}}\nonumber\\
		& =\frac{1}{\hat{H}}\sqrt{n}(\hat{G}-G)-\frac{G}{\hat{H}H}\sqrt{n}(\hat
		{H}-H)\nonumber\\
		& =\frac{1}{\hat{H}}%
		\sbr[3]{\sqrt{n}(\hat{G}-G)-\tau\sqrt{n}(\hat{H}-H)}.\label{eq:tau}%
	\end{align}
	Next, we divide the proof into three steps. In the first step, we obtain the
	linear expansion of $\sqrt{n}(\hat{G}-G)$. Based on the same argument, we can
	obtain the linear expansion of $\sqrt{n}(\hat{H}-H)$. In the second step, we
	obtain the linear expansion of $\sqrt{n}(\hat{\tau} - \tau)$ and then prove
	the asymptotic normality. In the third step, we show the consistency of
	$\hat{\sigma}$. The second result in the Theorem is obvious given the
	semiparametric efficiency bound derived in Theorem \ref{thm:eff}.
	

	\textbf{Step 1.} We have
	\begin{align*}
		\sqrt{n}(\hat{G}-G) & =\frac{1}{\sqrt{n}}\sum_{i \in[n]}\biggl[\frac
		{A_{i}(Y_{i} - \hat{\mu}^{Y}(1,S_{i},X_{i}))}{\hat{\pi}(S_{i})} -
		\frac{(1-A_{i})(Y_{i}-\hat{\mu}^{Y}(0,S_{i},X_{i}))}{1-\hat{\pi}(S_{i})}\\
		&  + \hat{\mu}^{Y}(1,S_{i},X_{i})-\hat{\mu}^{Y}(0,S_{i},X_{i}) \biggr]-\sqrt
		{n}G\\
		& =\frac{1}{\sqrt{n}}\sum_{i=1}^{n}%
		\sbr[3]{\hat{\mu}^Y(1, S_i, X_i)-\frac{A_i\hat{\mu}^Y(1, S_i, X_i)}{\hat{\pi}(S_i)}}\\
		& +\frac{1}{\sqrt{n}}\sum_{i=1}^{n}%
		\sbr[3]{\frac{(1-A_i)\hat{\mu}^Y(0, S_i, X_i)}{1-\hat{\pi}(S_i)}-\hat{\mu}^Y(0, S_i, X_i)}\\
		& +\frac{1}{\sqrt{n}}\sum_{i=1}^{n}\frac{A_{i}Y_{i}}{\hat{\pi}(S_{i})}%
		-\frac{1}{\sqrt{n}}\sum_{i=1}^{n}\frac{(1-A_{i})Y_{i}}{1-\hat{\pi}(S_{i}%
			)}-\sqrt{n}G\\
		&  =:R_{n,1}+R_{n,2}+R_{n,3},
	\end{align*}
	where
	\begin{align*}
		R_{n,1} & :=\frac{1}{\sqrt{n}}\sum_{i=1}^{n}%
		\sbr[3]{\hat{\mu}^Y(1, S_i, X_i)-\frac{A_i\hat{\mu}^Y(1, S_i, X_i)}{\hat{\pi}(S_i)}},\\
		R_{n,2} & :=\frac{1}{\sqrt{n}}\sum_{i=1}^{n}%
		\sbr[3]{\frac{(1-A_i)\hat{\mu}^Y(0, S_i, X_i)}{1-\hat{\pi}(S_i)}-\hat{\mu}^Y(0, S_i, X_i)},\\
		R_{n,3} & :=\frac{1}{\sqrt{n}}\sum_{i=1}^{n}\frac{A_{i}Y_{i}}{\hat{\pi}%
			(S_{i})}-\frac{1}{\sqrt{n}}\sum_{i=1}^{n}\frac{(1-A_{i})Y_{i}}{1-\hat{\pi
			}(S_{i})}-\sqrt{n}G.
	\end{align*}

	Lemma \ref{lem:G} shows that
	\begin{align*}
		R_{n,1} & =\frac{1}{\sqrt{n}}\sum_{i=1}^{n}\del[3]{1-\frac{1}{\pi(S_i)}}A_{i}%
		\tilde{\mu}^{Y}(1, S_{i}, X_{i})+\frac{1}{\sqrt{n}}\sum_{i=1}^{n}%
		(1-A_{i})\tilde{\mu}^{Y}(1, S_{i}, X_{i})+o_{p}(1),\\
		R_{n,2} & =\frac{1}{\sqrt{n}}\sum_{i=1}^{n}%
		\del[3]{\frac{1}{1-\pi(S_i)}-1}(1-A_{i})\tilde{\mu}^{Y}(0, S_{i}, X_{i}%
		)-\frac{1}{\sqrt{n}}\sum_{i=1}^{n}A_{i}\tilde{\mu}^{Y}(0, S_{i}, X_{i}%
		)+o_{p}(1),\\
		R_{n,3}  &  = \frac{1}{\sqrt{n}}\sum_{i=1}^{n}\frac{1}{\pi(S_{i})}\tilde
		{W}_{i}A_{i}-\frac{1}{\sqrt{n}}\sum_{i=1}^{n}\frac{1-A_{i}}{1-\pi(S_{i}%
			)}\tilde{Z}_{i} + \frac{1}{\sqrt{n}}\sum_{i=1}^{n}%
		\del[1]{\mathbb{E}[W_i-Z_i|S_i]-\mathbb{E}[W_i-Z_i]}.
	\end{align*}

	This implies
	\begin{align}
		\sqrt{n}(\hat{G} - G)  &  = \biggl\{\frac{1}{\sqrt{n}}\sum_{i=1}^{n} \left[
		\left( 1- \frac{1}{\pi(S_{i})} \right) \tilde{\mu}^{Y}(1,S_{i},X_{i}) -
		\tilde{\mu}^{Y}(0,S_{i},X_{i}) + \frac{\tilde{W}_{i}}{\pi(S_{i})}\right]
		A_{i}\nonumber\\
		&  + \frac{1}{\sqrt{n}}\sum_{i=1}^{n} \left[ \left( \frac{1}{1-\pi(S_{i})}-1
		\right) \tilde{\mu}^{Y}(0,S_{i},X_{i}) + \tilde{\mu}^{Y}(1,S_{i},X_{i}) -
		\frac{\tilde{Z}_{i}}{1-\pi(S_{i})}\right] (1-A_{i})\biggr\}\nonumber\\
		&  + \left\{ \frac{1}{\sqrt{n}}\sum_{i=1}^{n}%
		\del[1]{\mathbb{E}[W_i-Z_i|S_i]-\mathbb{E}[W_i-Z_i]} \right\} +o_{p}%
		(1).\label{eq:G}%
	\end{align}

	Similarly, we can show that
	\begin{align}
		\sqrt{n}(\hat{H} - H)  &  = \biggl\{\frac{1}{\sqrt{n}}\sum_{i=1}^{n} \left[
		\left( 1- \frac{1}{\pi(S_{i})} \right) \tilde{\mu}^{D}(1,S_{i},X_{i}) -
		\tilde{\mu}^{D}(0,S_{i},X_{i}) + \frac{ \tilde{D}_{i}(1)}{\pi(S_{i})}\right]
		A_{i}\nonumber\\
		&  + \frac{1}{\sqrt{n}}\sum_{i=1}^{n} \left[ \left( \frac{1}{1-\pi(S_{i})}-1
		\right) \tilde{\mu}^{D}(0,S_{i},X_{i}) + \tilde{\mu}^{D}(1,S_{i},X_{i}) -
		\frac{\tilde{D}_{i}(0)}{1-\pi(S_{i}) }\right] (1-A_{i})\biggr\}\nonumber\\
		&  + \left\{ \frac{1}{\sqrt{n}}\sum_{i=1}^{n}%
		\del[1]{\mathbb{E}[D_i(1)-D_i(0)|S_i]-\mathbb{E}[D_i(1)-D_i(0)]} \right\}
		+o_{p}(1),\label{eq:H}%
	\end{align}
	where $\tilde{D}_{i}(a) = D_{i}(a) - \mathbb{E}(D_{i}(a)|S_{i})$ for $a=0,1$
	and $\tilde{\mu}^{D}(0,s,X_{i}) = \overline{\mu}^{D}(0,s,X_{i}) -
	\mathbb{E}(\overline{\mu}^{D}(0,S_{i},X_{i})|S_{i}=s).$
	
	Combining \eqref{eq:tau}, \eqref{eq:G}, and \eqref{eq:H}, we obtain the linear
	expansion for $\hat{\tau}$ as
	\begin{align*}
		\sqrt{n}(\hat{\tau} - \tau)  &  = \frac{1}{\hat{H}} \left[ \sqrt{n}(\hat{G} -
		G) - \tau\sqrt{n}(\hat{H} - H) \right] \\
		&  = \frac{1}{\hat{H}} \left[ \frac{1}{\sqrt{n}}\sum_{i=1}^{n} \Xi
		_{1}(\mathcal{D}_{i}, S_{i})A_{i} + \frac{1}{\sqrt{n}}\sum_{i=1}^{n} \Xi
		_{0}(\mathcal{D}_{i},S_{i})(1-A_{i}) + \frac{1}{\sqrt{n}}\sum_{i=1}^{n}\Xi
		_{2}(S_{i})\right]  + o_{p}(1),
	\end{align*}
	where $\mathcal{D}_{i} = \{Y_{i}(1), Y_{i}(0), D_{i}(1), D_{i}(0), X_{i}\}$,
	\begin{align*}
		\Xi_{1}(\mathcal{D}_{i}, S_{i})  &  = \left[ \left( 1- \frac{1}{\pi(S_{i})}
		\right) \tilde{\mu}^{Y}(1,S_{i},X_{i}) - \tilde{\mu}^{Y}(0,S_{i},X_{i}) +
		\frac{\tilde{W}_{i}}{\pi(S_{i})}\right] \\
		&  - \tau\left[ \left( 1- \frac{1}{\pi(S_{i})} \right) \tilde{\mu}^{D}%
		(1,S_{i},X_{i}) - \tilde{\mu}^{D}(0,S_{i},X_{i}) + \frac{ \tilde{D}_{i}%
			(1)}{\pi(S_{i})}\right] ,\\
		\Xi_{0}(\mathcal{D}_{i}, S_{i})  &  = \left[ \left( \frac{1}{1-\pi(S_{i})}-1
		\right) \tilde{\mu}^{Y}(0,S_{i},X_{i}) + \tilde{\mu}^{Y}(1,S_{i},X_{i}) -
		\frac{\tilde{Z}_{i}}{1-\pi(S_{i})}\right] \\
		& - \tau\left[ \left( \frac{1}{1-\pi(S_{i})}-1 \right) \tilde{\mu}^{D}%
		(0,S_{i},X_{i}) + \tilde{\mu}^{D}(1,S_{i},X_{i}) - \frac{\tilde{D}_{i}%
			(0)}{1-\pi(S_{i}) }\right] ,\\
		\Xi_{2}(S_{i})  &  = \del[1]{\mathbb{E}[W_i-Z_i|S_i]-\mathbb{E}[W_i-Z_i]} -
		\tau\left[ \mathbb{E}[D_{i}(1)-D_{i}(0)|S_{i}]-\mathbb{E}[D_{i}(1)-D_{i}(0)]
		\right] .
	\end{align*}

	\textbf{Step 2.} Lemma \ref{lem:clt} implies that
	\begin{align*}
		& \frac{1}{\sqrt{n}}\sum_{i=1}^{n} \Xi_{1}(\mathcal{D}_{i}, S_{i})A_{i}
		\rightsquigarrow\mathcal{N}(0, \sigma_{1}^{2}),\quad\frac{1}{\sqrt{n}}%
		\sum_{i=1}^{n} \Xi_{0}(\mathcal{D}_{i}, S_{i})(1-A_{i}) \rightsquigarrow
		\mathcal{N}(0, \sigma_{0}^{2}), \quad\text{and}\\
		& \frac{1}{\sqrt{n}}\sum_{i=1}^{n} \Xi_{2}(S_{i}) \rightsquigarrow
		\mathcal{N}(0, \sigma_{2}^{2}),
	\end{align*}
	and the three terms are asymptotically independent, where
	\begin{align*}
		&  \sigma_{1}^{2} = \mathbb{E}\pi(S_{i})\Xi_{1}^{2}(\mathcal{D}_{i},S_{i}),
		\quad\sigma_{0}^{2} = \mathbb{E}(1-\pi(S_{i}))\Xi_{0}^{2}(\mathcal{D}%
		_{i},S_{i}), \quad\text{and} \quad\sigma_{2}^{2} = \mathbb{E}\Xi_{2}^{2}%
		(S_{i}).
	\end{align*}
	This further implies $\hat{H} \overset{p}{\longrightarrow} H$ and
	\begin{align*}
		\sqrt{n}(\hat{\tau} - \tau) \rightsquigarrow\mathcal{N}\left( 0, \frac
		{\sigma_{1}^{2} + \sigma_{0}^{2} + \sigma_{2}^{2}}{H^{2}}\right) ,
	\end{align*}

	\textbf{Step 3.} We aim to show the consistency of $\hat{\sigma}^{2}$. First
	note that
	\[
	\frac{1}{n}\sum_{i=1}^{n} \Xi_{H,i} = \hat{H} \overset{p}{\longrightarrow} H =
	\mathbb{E}(D_{i}(1)-D_{i}(0)).
	\]
	In addition, Lemma \ref{lem:sigma} shows.
	\begin{align*}
		\frac{1}{n}\sum_{i=1}^{n} A_i\hat{\Xi}_{1}^{2}(\mathcal{D}_{i},S_{i})
		\overset{p}{\longrightarrow} \sigma_{1}^{2}, \quad\frac{1}{n}\sum_{i=1}^{n}(1-A_i)
		\hat{\Xi}_{0}^{2}(\mathcal{D}_{i},S_{i}) \overset{p}{\longrightarrow}
		\sigma_{0}^{2}, \quad\text{and} \quad\frac{1}{n}\sum_{i=1}^{n} \hat{\Xi}%
		_{2}^{2}(\mathcal{D}_{i},S_{i}) \overset{p}{\longrightarrow} \sigma_{2}^{2}.
	\end{align*}
	This implies $\hat{\sigma}^{2} \overset{p}{\longrightarrow} \sigma^{2}$.
	

	\section{Proof of Theorem \ref{thm:eff}}
	
	\label{sec:pf_eff}
	
	Without loss of generality, we assume $A_{i} = \phi_{i}(\{S_{i}\}_{i \in
		[n]},U)$, where $\phi_{i}(\cdot)$ is a deterministic function and $U$ is a
	random variable (vector) with density $P_{U}(\cdot)$ and is independent of
	everything else in the data. Further denote $\mathcal{Y}_{i}(a) =
	\{Y_{i}(D_{i}(a)),D_{i}(a),X_{i}\}$. We consider parametric submodels indexed
	by a generic parameter $\theta$. The likelihoods of $S_{i}$ evaluated at $s$
	and $\mathcal{Y}_{i}(a)$ given $S_{i}=s$ evaluated at $\overline{y}$ are
	written as $f_{S}(s;\theta)$ and $f_{\mathcal{Y}(a)|S}(\overline{y}|s;\theta)$
	for $a=0,1$, respectively. The density of $U$ does not depend on $\theta$. Let
	$\theta_{n} = \theta^{*} + h/\sqrt{n}$, where $\theta^{*}$ indexes the true
	underlying DGP.
	
	By Assumption \ref{ass:assignment1}, the joint likelihood of $\{Y_{i}%
	,X_{i},S_{i},A_{i}\}_{i \in[n]}$ under $\theta$ can be written as
	\begin{align*}
		P_{U}(u) \Pi_{i \in[n]}\left[ f_{S}(s_{i};\theta) \Pi_{a=0,1}f_{\mathcal{Y}%
			(a)|S}(\tilde{y}_{i}(a)|s_{i};\theta)^{1\{\phi_{i}(s_{1},\cdots,s_{n}%
			,u)=a\}}\right]
	\end{align*}
	where $(x_{i},y_{i}(d_{i}(a)), d_{i}(a), u, s_{i})$ are the realizations
	$(X_{i},Y_{i}(D_{i}(a)), D_{i}(a), U, S_{i})$ for $i \in[n]$ and $\tilde
	{y}_{i}(a) = \{y_{i}(d_{i}(a)),d_{i}(a),x_{i}\}$. We make the following
	regularity assumptions with respect to the submodel.
	
	\begin{ass}
		\begin{enumerate}
			[label=(\roman*)]
			
			\item Suppose $f_{S}(s;\theta)$ and $f_{\mathcal{Y}(a)|S}(\overline
			{y}|s;\theta)$ for $a=0,1$ are differentiable in quadratic mean at $\theta
			^{*}$ with score functions $g_{s}(S_{i})$ and $g_{a}(\mathcal{Y}_{i}%
			(a)|S_{i})$ for $a=0,1$, respectively, such that
			\begin{align*}
				&  \dot{f}_{\mathcal{Y}(a)|S}(\overline{y}|s;\theta) = \frac{\partial
					\log(f_{\mathcal{Y}(a)|S}(\overline{y}|s;\theta) )}{\partial\theta}, \quad
				\dot{f}_{S}(s;\theta) = \frac{\partial\log(f_{S}(s;\theta))}{\partial\theta
				},\\
				&  \dot{f}_{\mathcal{Y}(a)|S}(\mathcal{Y}_{i}(a)|S_{i};\theta^{*}) =
				g_{a}(\mathcal{Y}_{i}(a)|S_{i}), \quad\text{and} \quad\dot{f}_{S}(S_{i}%
				;\theta^{*}) = g_{s}(S_{i}).
			\end{align*}

			\item Suppose $\dot{f}_{\mathcal{Y}(a)|S}(\overline{y}|s;\theta)$ and $\dot
			{f}_{S}(s;\theta)$ are continuous at $\theta^{*}$ so that there exist a
			sequence $t_{n} = o(1)$ and a function $L_{a}(\mathcal{Y}_{i}(a),S_{i})$ such
			that
			\begin{align*}
				|\dot{f}_{\mathcal{Y}(a)|S}(\mathcal{Y}_{i}(a)|S_{i};\theta^{*}+h/\sqrt{n}) -
				g_{a}(\mathcal{Y}_{i}(a)|S_{i})| + |\dot{f}_{S}(S_{i};\theta^{*}+h/\sqrt
				{n})-g_{s}(S_{i})| \leq t_{n} L_{a}(\mathcal{Y}_{i}(a),S_{i})
			\end{align*}
			and $\mathbb{E} |Y_{i}(D_{i}(a))L_{a}(\mathcal{Y}_{i}(a),S_{i})| < \infty$ for
			$a=0,1$.
			
			\item Suppose there exists a constant $C>0$ such that
			\begin{align*}
				\max_{s \in\mathcal{S}} \mathbb{E}%
				\sbr[2]{ \left.\envert[1]{\underline{\Xi}_1(\mathcal{D}_i,S_i)g_1(\mathcal{Y}_i(1)|S_i)} +\envert[1]{ \underline{\Xi}_0(\mathcal{D}_i,S_i)g_0(\mathcal{Y}_i(0)|S_i)}\right|S_i=s}
				&  \leq C\\
				\max_{s \in\mathcal{S}} \mathbb{E}%
				\sbr[2]{ \left.\underline{\Xi}_1(\mathcal{D}_i,S_i)^2 + \underline{\Xi}_0(\mathcal{D}_i,S_i)^2\right|S_i=s}
				&  \leq C,
			\end{align*}
			where $\underline{\Xi}_{1}(\mathcal{D}%
			_{i},S_{i})$, $\underline{\Xi}_{0}(\mathcal{D}_{i},S_{i})$ and $\underline{\Xi
			}_{2}(S_{i})$ are defined as $\Xi_{1}(\mathcal{D}_{i},S_{i})$, $\Xi_{0}(\mathcal{D}%
			_{i},S_{i})$ and $\Xi_{2}(S_{i})$ in \eqref{eq:Xi1}--\eqref{eq:Xi2},
			respectively, with the researcher-specified working model $\overline{\mu}%
			^{b}(a,s,x)$ equal to the true specification $\mu^{b}(a,s,x)$ for all
			$(a,b,s,x) \in\{0,1\} \times\{D,Y\} \times\mathcal{SX}$.
		\end{enumerate}
		\label{ass:E}
	\end{ass}
	We denote $\tau(\theta) = \mathbb{E}_{\theta} (Y_{i}(1) - Y_{i}(0)
	|D_{i}(1)>D_{i}(0))$, where $\mathbb{E}_{\theta}(\cdot)$ means the expectation
	is taken with the parametric submodel indexed by $\theta$. We further denote
	$\mathbb{E}(\cdot) = \mathbb{E}_{\theta^{*}}(\cdot)$, which is the expectation
	with respect to the true DGP.

	\textbf{Proof of Theorem \ref{thm:eff}.} Following the same argument in
	\cite{A22}, in order to show the semiparametric efficiency bound, we only need
	to show (1) local asymptotic normality of the log likelihood ratio for the
	parametric submodel with tangent set of the form
	\begin{align}
		\mathbb{T} =
		\begin{pmatrix}
			& \Psi(\mathcal{D}_{i},S_{i},A_{i}) = g_{s}(S_{i}) + A_{i}g_{1}(\mathcal{Y}%
			_{i}(1)|S_{i}) + (1-A_{i})g_{0}(\mathcal{Y}_{i}(0)|S_{i}):\\
			& \mathbb{E}%
			\sbr[1]{g_s^2(S_i) + \sum_{a=0,1}g_a^2(\mathcal{Y}_i(a)|S_i)}<\infty,
			\mathbb{E}g_{s}(S_{i})=0, \mathbb{E}(g_{a}(\mathcal{Y}_{i}(a)|S_{i})|S_{i}) =
			0,\\
			& \mathbb{E}(g_{1}(\mathcal{Y}_{i}(1)|S_{i})|X_{i},S_{i}) = \mathbb{E}%
			(g_{0}(\mathcal{Y}_{i}(0)|S_{i})|X_{i},S_{i})
		\end{pmatrix}
		.\label{eq:T}%
	\end{align}
	and (2) $\sqrt{n}(\tau(\theta^{*}+h/\sqrt{n}) - \tau(\theta^{*})) =
	\langle\tilde{\Psi},\Psi\rangle_{\bar{\mathbb{P}}} h+ o(1)$, where
	$\tilde{\Psi}(\mathcal{D}_{i},S_{i},A_{i})$ is the efficient score defined as
	\begin{align}
		\tilde{\Psi}(\mathcal{D}_{i},S_{i},A_{i}) =
		\sbr[1]{\underline{\Xi}_2(S_i) + A_i\underline{\Xi}_1(\mathcal{D}_i,S_i) + (1-A_i)\underline{\Xi}_0(\mathcal{D}_i,S_i)}/\mathbb{E}%
		[D_{i}(1)-D_{i}(0)]
	\end{align}
	and $\langle\tilde{\Psi},\Psi\rangle_{\bar{\mathbb{P}}} = \frac{1}{n}\sum_{i
		\in[n]}\mathbb{E}\tilde{\Psi}(\mathcal{D}_{i},S_{i},A_{i})\Psi(\mathcal{D}%
	_{i},S_{i},A_{i})$ is the inner product w.r.t. measure $\bar{\mathbb{P}}
	:=\frac{1}{n}\sum_{i \in[n]} \mathbb{P}_{i}$. We establish these two results
	in two steps.
	
	\paragraph{Step 1.}
	
	Denote $\theta_{n} = \theta^{*} + h/\sqrt{n}$ where $\theta^{*}$ is fixed and
	$\mathbb{P}_{n,h}$ as the joint distribution of $\{Y_{i},X_{i},S_{i}%
	,A_{i}\}_{i \in[n]}$ under $\theta_{n}$. The log likelihood ratio for
	$\theta_{n}$ against $\theta^{*}$ is given by
	\begin{align*}
		\ell_{n,h} = \sum_{i \in[n]}\tilde{\ell}_{s}(S_{i};\theta_{n}) + \sum_{a =
			0,1}\sum_{i \in[n]}1\{A_{i}=a\} \tilde{\ell}_{\mathcal{Y}(a)|S}(\mathcal{Y}%
		_{i}|S_{i};\theta_{n}),
	\end{align*}
	where $\mathcal{Y}_{i} = (Y_{i},D_{i},X_{i})$, $\tilde{\ell}_{s}(S_{i}%
	;\theta_{n}) = \log\left( \frac{f_{S}(S_{i};\theta_{n})}{f_{S}(S_{i}%
		;\theta^{*})}\right) $, and $\tilde{\ell}_{\mathcal{Y}(a)|S}(\mathcal{Y}%
	_{i}|S_{i};\theta_{n}) = \log\left( \frac{f_{\mathcal{Y}(a)|S}(\mathcal{Y}%
		_{i}|S_{i};\theta_{n})}{f_{\mathcal{Y}(a)|S}(\mathcal{Y}_{i}|S_{i};\theta
		^{*})}\right) $ for $a=0,1$. Then, \citet[Corollary 3.1]{A22} shows
	$\ell_{n,h}$ converges in distribution to a $\mathcal{N}(-h^{\prime}\tilde
	{I}^{*}h/2,h^{\prime}\tilde{I}^{*}h)$ law under $\theta^{*}$ where $\tilde
	{I}^{*}$ is the limit of
	\begin{align*}
		\mathbb{E}_{\theta^{*}}g_{s}^{2}(S_{i}) + \frac{1}{n}\sum_{i \in[n]} \sum_{a =
			0,1}1\{A_{i}=a\}\mathbb{E}_{\theta^{*}}\left[ g_{a}^{2}(\mathcal{Y}%
		_{i}(a)|S_{i})|S_{i}\right] .
	\end{align*}
	and the score for this parametric submodel can be written as
	\begin{align}
		\Psi(\mathcal{D}_{i},S_{i},A_{i}) = g_{s}(S_{i}) + A_{i}g_{1}(\mathcal{Y}%
		_{i}(1)|S_{i}) + (1-A_{i})g_{0}(\mathcal{Y}_{i}(0)|S_{i}). \label{eq:score}%
	\end{align}

	We note that by definition, we have
	\begin{align*}
		\mathbb{E}g_{s}(S_{i})=0 \quad\text{and} \quad\mathbb{E}(g_{a}(\mathcal{Y}%
		_{i}(a)|S_{i})|S_{i}) = 0.
	\end{align*}
	In addition, we have the equality that, for an arbitrary function $h(\cdot)$
	of $X$ such that $\mathbb{E}h^{2}(X)<\infty$,
	\begin{align}
		&  \mathbb{E}_{\theta}(h(X)|S) = \int_{x} h(x)f_{X|S}(x|S;\theta)
		\text{d}x\nonumber\\
		& = \int_{x}
		h(x)\sbr[3]{\int_{y(d(a)),d(a)}f_{Y(D(a)), D(a)|X,S}(y(d(a)),d(a)|x,S;\theta)\text{d} y(d(a)) \text{d} d(a)}f_{X|S}%
		(x|S;\theta) \text{d}x\nonumber\\
		& =\int_{y(d(a)),d(a),x} h(x) f_{\mathcal{Y}(a)|S}(y(d(a)),d(a),x|S;\theta)
		\text{d} y(d(a)) \text{d} d(a) \text{d}x\label{eq:fs}%
	\end{align}
	for $a=0,1$, where $f_{\mathcal{Y}(a)|S}(y(d(a)),d(a),x|s;\theta)$ is the
	joint likelihood of $(Y(D(a)),D(a),X)$ given $S$ for $a =0,1$. We note that,
	for $a =0,1$,
	\begin{align*}
		\frac{\partial f_{\mathcal{Y}(a)|S}(y(d(a)),d(a),x|S;\theta^{*})}{
			\partial\theta}  &  = f_{\mathcal{Y}(a)|S}(y(d(a)),d(a),x|S;\theta^{*})
		g_{a}(\mathcal{Y}(a)|S).
	\end{align*}
	Therefore, taking derivatives of $\theta$ in \eqref{eq:fs} and evaluating the
	derivatives at $\theta^{*}$, we have
	\begin{align*}
		\mathbb{E} \sbr[1]{h(X) g_1(\mathcal{Y}(1)|S)|S} = \mathbb{E}%
		\sbr[1]{ h(X) g_0(\mathcal{Y}(0)|S)|S},
	\end{align*}
	which implies $\mathbb{E}%
	\sbr[1]{g_1(\mathcal{Y}(1)|S) - g_0(\mathcal{Y}(0)|S)|X,S} = 0$. Therefore,
	the tangent set can be written in \eqref{eq:T}.
	

	\paragraph{Step 2.}
	
	We have
	\begin{align*}
		\tau(\theta) = \frac{\mathbb{E}_{\theta}(Y_{i}(D_{i}(1)) - Y_{i}(D_{i}%
			(0)))}{\mathbb{E}_{\theta}(D_{i}(1) - D_{i}(0))}.
	\end{align*}
	By the mean-value theorem, we have
	\begin{align*}
		\tau(\theta^{*} + h/\sqrt{n}) - \tau(\theta^{*})  &  = \frac{\partial
			\tau(\theta)}{ \partial\theta}\bigg|_{ \theta= \tilde{\theta}} \frac{h}%
		{\sqrt{n}}\\
		&  = \frac{\partial\tau(\theta)}{ \partial\theta}\bigg|_{ \theta= \theta^{*}}
		\frac{h}{\sqrt{n}} + \left[ \frac{\partial\tau(\theta)}{ \partial\theta
		}\bigg|_{ \theta= \tilde{\theta}} -\frac{\partial\tau(\theta)}{ \partial
			\theta}\bigg|_{ \theta= \theta^{*}}\right] \frac{h}{\sqrt{n}}.
	\end{align*}
	Let $G(\theta) = \mathbb{E}_{\theta}\sbr[1]{Y(D(1)) -  Y(D(0))}$, $H(\theta) =
	\mathbb{E}_{\theta}\sbr[1]{D(1) -  D(0)}$, $G = G(\theta^{*})$, and $H =
	H(\theta^{*})$. Note that $\tau(\theta)=G(\theta)/H(\theta)$ and $\tau=G/H$.
	Then, we have
	\begin{align*}
		\frac{\partial G(\theta)}{\partial\theta}  &  = \mathbb{E}_{\theta}
		\sbr[1]{Y(D(1))(\dot{f}_{\mathcal{Y}(1)|S}(\mathcal{Y}(1)|S;\theta) + \dot{f}_S(S;\theta)}
		- \mathbb{E}_{\theta}%
		\sbr[1]{ Y(D(0))(\dot{f}_{\mathcal{Y}(0)|S}(\mathcal{Y}(0)|S;\theta)+ \dot{f}_S(S;\theta))}\\
		\frac{\partial H(\theta)}{\partial\theta}  &  = \mathbb{E}_{\theta
		}%
		\sbr[1]{ D(1)(\dot{f}_{\mathcal{Y}(1)|S}(\mathcal{Y}(1)|S;\theta)+ \dot{f}_S(S;\theta))}
		- \mathbb{E}_{\theta}%
		\sbr[1]{ D(0)(\dot{f}_{\mathcal{Y}(0)|S}(\mathcal{Y}(0)|S;\theta)+ \dot{f}_S(S;\theta))}.
	\end{align*}
	Therefore by Assumption \ref{ass:E} we can find a constant $L$ such that
	\begin{align*}
		\biggl|\frac{\partial\tau(\theta)}{ \partial\theta}\bigg|_{ \theta=
			\tilde{\theta}} -\frac{\partial\tau(\theta)}{ \partial\theta}\bigg|_{ \theta=
			\theta^{*}}\biggr|  &  = \biggl|\frac{H(\tilde{\theta}) \frac{\partial
				G(\tilde{\theta})}{\partial\theta} - G(\tilde{\theta})\frac{\partial
				H(\tilde{\theta})}{\partial\theta}}{H^{2}(\tilde{\theta})} - \frac
		{H(\theta^{*}) \frac{\partial G(\theta^{*})}{\partial\theta} - G(\theta
			^{*})\frac{\partial H(\theta^{*})}{\partial\theta}}{H^{2}(\theta^{*}%
			)}\biggr|\\
		&  \leq t_{n} L.
	\end{align*}
	This implies
	\begin{align}
		\sqrt{n}(\tau(\theta^{*} + h/\sqrt{n}) - \tau(\theta^{*})) = \frac
		{\partial\tau(\theta)}{ \partial\theta}\bigg|_{ \theta= \theta^{*}} h +
		o(1).\label{eq:derv1}%
	\end{align}

	In addition, following the calculation by \cite{F07late}, we have
	\begin{align*}
		& \frac{\partial\tau(\theta)}{ \partial\theta}\bigg|_{ \theta= \theta^{*}} =
		\frac{\left[ \frac{\partial G(\theta)}{\partial\theta} - \tau\frac{\partial
				H(\theta)}{\partial\theta}\right] \bigg|_{ \theta= \theta^{*}}}{H}\\
		&  = \frac{\mathbb{E}\left[ (Y(D(1))-\tau D(1))(g_{1}(\mathcal{Y}(1)|S) +
			g_{s}(S)) \right] }{H} - \frac{\mathbb{E}\left[ (Y(D(0))-\tau D(0))(g_{0}%
			(\mathcal{Y}(0)|S) + g_{s}(S)) \right] }{H},
	\end{align*}
	where for notation simplicity, we write $\mathbb{E}_{\theta^{*}}$ as
	$\mathbb{E}$. Let
	\begin{align*}
		\Gamma(X,S)  &  = \biggl[\pi(S)(\mathbb{E}(Z|X,S) - \mathbb{E}(Z|S)) + \left(
		1- \pi(S)\right) (\mathbb{E}(W|X,S)-\mathbb{E}(W|S))\\
		&  - \tau\left(  \pi(S)(\mathbb{E}(D(0)|X,S) - \mathbb{E}(D(0)|S)) + \left( 1-
		\pi(S)\right) (\mathbb{E}(D(1)|X,S) -\mathbb{E}(D(1)|S)) \right) \biggr].
	\end{align*}
	Then, we have
	\begin{align*}
		&  Y(D(1))-\tau D(1) = \pi(S) \underline{\Xi}_{1}(\mathcal{D},S) +
		\mathbb{E}(W - \tau D(1)|S) + \Gamma(X,S),\\
		&  Y(D(0))-\tau D(0) = -(1-\pi(S)) \underline{\Xi}_{0}(\mathcal{D},S) +
		\mathbb{E}(Z-\tau D(0)|S) + \Gamma(X,S).
	\end{align*}
	This implies
	\begin{align*}
		&  \mathbb{E}(Y(D(1))-\tau D(1))(g_{1}(\mathcal{Y}(1)|S) + g_{s}(S))\\
		&  = \mathbb{E}\pi(S)\underline{\Xi}_{1}(\mathcal{D},S)g_{1}(\mathcal{Y}%
		(1)|S)+\mathbb{E}\Gamma(X,S)g_{1}(\mathcal{Y}(1)|S) + \mathbb{E}%
		(\mathbb{E}(W-\tau D(1)|S)g_{s}(S)),
	\end{align*}
	\begin{align*}
		&  \mathbb{E}(Y(D(0))-\tau D(0))(g_{0}(\mathcal{Y}(0)|S) + g_{s}(S))\\
		&  = -\mathbb{E}(1-\pi(S))\Xi_{0}(\mathcal{D},S)g_{0}(\mathcal{Y}%
		(0)|S)+\mathbb{E}\Gamma(X,S)g_{0}(\mathcal{Y}(0)|S) + \mathbb{E}%
		(\mathbb{E}(Z-\tau D(0)|S)g_{s}(S)),
	\end{align*}
	where we have used $\mathbb{E}[\underline{\Xi}_{a}(\mathcal{D},S)|S]=0$,
	$\mathbb{E}[\Gamma(X,S)|S]=0$ and $\mathbb{E}[g_{a}(\mathcal{Y}(a)|S)|S]=0$
	for $a=0,1$. Then
	\begin{align}
		\frac{\partial\tau(\theta)}{ \partial\theta}\bigg|_{ \theta= \theta^{*}} &  =
		\frac{\mathbb{E}\pi(S)\underline{\Xi}_{1}(\mathcal{D},S)g_{1}(\mathcal{Y}%
			(1)|S)}{H} + \frac{\mathbb{E}(1-\pi(S))\underline{\Xi}_{0}(\mathcal{D}%
			,S)g_{0}(\mathcal{Y}(0)|S)}{H} + \frac{\mathbb{E}g_{s}(S)\underline{\Xi}%
			_{2}(S)}{H}\nonumber\\
		&  + \frac{\mathbb{E}\Gamma(X,S) (g_{1}(\mathcal{Y}(1)|S) - g_{0}%
			(\mathcal{Y}(0)|S))}{H}\nonumber\\
		&  = \frac{\mathbb{E}\pi(S)\underline{\Xi}_{1}(\mathcal{D},S)g_{1}%
			(\mathcal{Y}(1)|S)}{H} + \frac{\mathbb{E}(1-\pi(S))\underline{\Xi}%
			_{0}(\mathcal{D},S)g_{0}(\mathcal{Y}(0)|S)}{H} + \frac{\mathbb{E}%
			g_{s}(S)\underline{\Xi}_{2}(S)}{H}.\label{eq:derv2}%
	\end{align}
	where the last equality is due to (\ref{eq:T}).
	
	On the other hand, we note that
	\begin{align*}
		& \langle\tilde{\Psi},\Psi\rangle_{\bar{\mathbb{P}}} = \frac{1}{n}\sum_{i
			\in[n]}\left[ \frac{ \mathbb{E}\sbr[1]{g_s(S_i)\underline{\Xi}_2(S_i)}}{H} +
		\frac{\mathbb{E}%
			\sbr[1]{A_i\underline{\Xi}_1(\mathcal{D}_i,S_i)g_1(\mathcal{Y}_i(1)|S_i)}}{H}
		+ \frac{\mathbb{E}%
			\sbr[1]{(1-A_i)\underline{\Xi}_0(\mathcal{D}_i,S_i)g_0(\mathcal{Y}_i(0)|S_i)}}%
		{H}\right] \\
		&  = \frac{\partial\tau(\theta)}{ \partial\theta}\bigg|_{ \theta= \theta^{*}}
		+ \frac{1}{n}\sum_{i \in[n]}\left[  \frac{\mathbb{E}%
			\sbr[1]{(A_i-\pi(S_i))\underline{\Xi}_1(\mathcal{D}_i,S_i)g_1(\mathcal{Y}_i(1)|S_i)}}%
		{H} - \frac{\mathbb{E}%
			\sbr[1]{(A_i-\pi(S_i))\underline{\Xi}_0(\mathcal{D}_i,S_i)g_0(\mathcal{Y}_i(0)|S_i)}}%
		{H}\right] .
	\end{align*}
	In addition, by Assumption \ref{ass:E}, we have, for some constant $C>0$,
	that
	\begin{align*}
		&  \left| \frac{1}{n}\sum_{i \in[n]}\left[  \frac{\mathbb{E}(A_{i}-\pi
			(S_{i}))\underline{\Xi}_{1}(\mathcal{D}_{i},S_{i})g_{1}(\mathcal{Y}%
			_{i}(1)|S_{i})}{H} - \frac{\mathbb{E}(A_{i}-\pi(S_{i}))\underline{\Xi}%
			_{0}(\mathcal{D}_{i},S_{i})g_{0}(\mathcal{Y}_{i}(0)|S_{i})}{H}\right] \right|
		\\
		&  \leq\frac{C}{n}\sum_{s\in\mathcal{S}} \mathbb{E}|B_{n}(s)| = o(1),
	\end{align*}
	where the inequality is by law of iterated expectation and Assumption
	\ref{ass:E}(iii) and the last equality is due to $\mathbb{E}|B_{n}%
	(s)|/n=o(1)$.\footnote{Since $|B_{n}(s)/n| \leq1$, $\{B_{n}(s)/n\}$ is
		uniformly integrable. Then from $B_{n}(s)/n = o_{p}(1)$, we have
		$\mathbb{E}|B_{n}(s)|/n=o(1)$.} This implies
	\begin{align}
		\langle\tilde{\Psi},\Psi\rangle_{\bar{\mathbb{P}}} = \frac{\partial\tau
			(\theta)}{ \partial\theta}\bigg|_{ \theta= \theta^{*}}+o(1) .\label{eq:derv3}%
	\end{align}
	Combining \eqref{eq:derv1}, \eqref{eq:derv2} and \eqref{eq:derv3}, we obtained
	the desired result for Step 2. Last, it is obvious from the previous
	calculation that
	\begin{align*}
		\langle\tilde{\Psi},\tilde{\Psi} \rangle_{\bar{\mathbb{P}}} \rightarrow
		\underline{\sigma}^{2}.
	\end{align*}



	\section{Proof of Theorem \ref{thm:par}}
	
	The proof is divided into two steps. In the first step, we show Assumption
	\ref{ass:Delta}(i). In the second step, we establish Assumptions
	\ref{ass:Delta}(ii) and \ref{ass:Delta}(iii).
	
	\vspace{1.5mm} \noindent\textbf{Step 1.} Recall
	\begin{align*}
		\Delta^{Y}(a,s,X_{i}) = \hat{\mu}^{Y}(a,s,X_{i}) - \overline{\mu}%
		^{Y}(a,s,X_{i}) = \Lambda_{a,s}^{Y}(X_{i},\hat{\theta}_{a,s})-\Lambda
		_{a,s}^{Y}(X_{i},\theta_{a,s}),
	\end{align*}
	and $\{X_{i}^{s}\}_{i \in[n]}$ is generated independently from the
	distribution of $X_{i}$ given $S_{i}=s$, and so is independent of
	$\{A_{i},S_{i}\}_{i \in[n]}$. Let $M_{a,s}(\theta_{1},\theta_{2}): =
	\mathbb{E}[\Lambda_{a,s}^{Y}(X_{i},\theta_{1}) - \Lambda_{a,s}^{Y}%
	(X_{i},\theta_{2})|S_{i}=s] = \mathbb{E}[\Lambda_{a,s}^{Y}(X_{i}^{s}%
	,\theta_{1}) - \Lambda_{a,s}^{Y}(X_{i}^{s},\theta_{2})]$. We have
	\begin{align}
		&  \biggl|\frac{\sum_{i\in I_{1}(s)}\Delta^{Y}(a,s,X_{i})}{n_{1}(s)} -
		\frac{\sum_{i \in I_{0}(s)}\Delta^{Y}(a,s,X_{i})}{n_{0}(s)}\biggr|\nonumber\\
		& \leq
		\envert[3]{\frac{\sum_{i\in I_1(s)}[ \Delta^Y(a,s,X_i) - M_{a,s}(\hat{\theta}_{a,s}, \theta_{a,s})] }{n_1(s)} }
		+
		\envert[3]{\frac{\sum_{i\in I_0(s)}[ \Delta^Y(a,s,X_i) - M_{a,s}(\hat{\theta}_{a,s},\theta_{a,s})] }{n_0(s)} }\nonumber\\
		& = o_{p}(n^{-1/2}).\label{eq:Delta}%
	\end{align}
	To see the last equality, we note that, for any $\varepsilon>0$, with
	probability approaching one (w.p.a.1), we have
	\begin{align*}
		\max_{s \in\mathcal{S}}||\hat{\theta}_{a,s} - \theta_{a,s}||_{2}
		\leq\varepsilon.
	\end{align*}
	Therefore, on the event $\mathcal{A}_{n}(\varepsilon) := \{\max_{s
		\in\mathcal{S}}||\hat{\theta}_{a,s} - \theta_{a,s}||_{2} \leq\varepsilon,
	\min_{s \in\mathcal{S}}n_{1}(s) \geq\varepsilon n\}$ we have
	\begin{align*}
		&
		\envert[3]{\frac{\sum_{i\in I_1(s)}[ \Delta^Y(a,s,X_i) - M_{a,s}(\hat{\theta}_{a,s},\theta_{a,s})] }{n_1(s)} }
		\biggl|\{A_{i},S_{i}\}_{i \in[n]}\\
		& \overset{d}{=}
		\envert[3]{\frac{\sum_{i = N(s)+1}^{N(s)+n_1(s)}[ \Delta^Y(a,s,X_i^s) - M_{a,s}(\hat{\theta}_{a,s},\theta_{a,s})] }{n_1(s)} }
		\biggl|\{A_{i},S_{i}\}_{i \in[n]} \leq||\mathbb{P}_{n_{1}(s)} - \mathbb{P}%
		||_{\mathcal{F}}\biggl|\{A_{i},S_{i}\}_{i \in[n]},
	\end{align*}
	where
	\begin{align*}
		\mathcal{F} = \{ \Lambda_{a,s}^{Y}(X_{i}^{s},\theta_{1}) - \Lambda_{a,s}%
		^{Y}(X_{i}^{s},\theta_{2}) - M_{a,s}(\theta_{1},\theta_{2}): ||\theta
		_{1}-\theta_{2}||_{2}\leq\varepsilon\}.
	\end{align*}
	Therefore, for any $\delta>0$ we have
	\begin{align*}
		&  \mathbb{P}%
		\del[3]{ \envert[3]{\frac{\sum_{i\in I_1(s)}[ \Delta^Y(a,s,X_i) - M_{a,s}(\hat{\theta}_{a,s},\theta_{a,s})] }{n_1(s)}}\geq \delta n^{-1/2}}\\
		& \leq\mathbb{P}%
		\del[3]{\envert[3]{\frac{\sum_{i\in I_1(s)}[ \Delta^Y(a,s,X_i) - M_{a,s}(\hat{\theta}_{a,s},\theta_{a,s})] }{n_1(s)} } \geq \delta n^{-1/2}, \mathcal{A}_n(\eps) }
		+ \mathbb{P}(\mathcal{A}_{n}^{c}(\varepsilon))\\
		& \leq\mathbb{E}%
		\sbr[3]{\mathbb{P}\del[3]{ \envert[3]{\frac{\sum_{i\in I_1(s)}[ \Delta^Y(a,s,X_i) - M_{a,s}(\hat{\theta}_{a,s},\theta_{a,s})] }{n_1(s)} } \geq \delta n^{-1/2}, \mathcal{A}_n(\eps)\biggl|\{A_i,S_i\}_{i \in [n]}}}
		+ \mathbb{P}(\mathcal{A}_{n}^{c}(\varepsilon))\\
		& \leq\sum_{s \in\mathcal{S}}\mathbb{E}\left[ \mathbb{P}\left( ||\mathbb{P}%
		_{n_{1}(s)} - \mathbb{P}||_{\mathcal{F}} \geq\delta n^{-1/2} \biggl|\{A_{i}%
		,S_{i}\}_{i \in[n]}\right) 1\{n_{1}(s) \geq n\varepsilon\}\right]  +
		\mathbb{P}(\mathcal{A}_{n}^{c}(\varepsilon))\\
		& \leq\sum_{s \in\mathcal{S}}\mathbb{E}\left\{  \frac{n^{1/2}\mathbb{E}\left[
			||\mathbb{P}_{n_{1}(s)} - \mathbb{P}||_{\mathcal{F}}|\{A_{i},S_{i}\}_{i
				\in[n]}\right] 1\{n_{1}(s) \geq n\varepsilon\}}{\delta}\right\}  +
		\mathbb{P}(\mathcal{A}_{n}^{c}(\varepsilon)).
	\end{align*}
	By Assumption \ref{ass:par}, $\mathcal{F}$ is a VC-class with a fixed VC index
	and envelope $L_{i}$ such that $\mathbb{E}(L_{i}^{q}|\{A_{i},S_{i}\}_{i
		\in[n]}) \leq C<\infty$. This implies $\mathbb{E} \max_{i \in[n_{1}(s)]}%
	L_{i}^{2} \leq C n_{1}^{2/q}(s)$. In addition,
	\begin{align*}
		\sup_{f \in\mathcal{F}}\mathbb{P}f^{2} \leq\mathbb{E}L_{i}^{2}(\theta
		_{1}-\theta_{2})^{2} \leq C\varepsilon^{2}.
	\end{align*}
	Invoke \citet[Corollary 5.1]{CCK14} with $A$ and $\nu$ being fixed constants,
	and $\sigma^{2}$, $F$, $M$ being $C\varepsilon^{2}$, $L$, $\max_{1\leq i\leq
		n_{1}(s)} L_{i}$, respectively, in our setting. We have
	\begin{align*}
		&  n^{1/2}\mathbb{E}\left[ ||\mathbb{P}_{n_{1}(s)} - \mathbb{P}||_{\mathcal{F}%
		}|\{A_{i},S_{i}\}_{i \in[n]}\right] 1\{n_{1}(s) \geq n\varepsilon\}\\
		&  \leq
		C\del[3]{\sqrt{\frac{n}{n_1(s)}\eps^2\log(1/\eps)} + n^{1/2}n_1^{1/q-1}(s)\log(1/\eps)}1\{n_{1}%
		(s) \geq n\varepsilon\}\\
		& \leq C (\varepsilon^{1/2}\log^{1/2}(1/\varepsilon) + n^{1/q-1/2}%
		\varepsilon^{1/q-1} \log(1/\varepsilon)).
	\end{align*}
	Therefore,
	\begin{align*}
		\mathbb{E}\left\{  \frac{n^{1/2}\mathbb{E}%
			\sbr[1]{||\mathbb{P}_{n_1(s)} - \mathbb{P}||_{\mathcal{F}}|\{A_i,S_i\}_{i \in [n]}}1\{n_{1}%
			(s) \geq n\varepsilon\}}{\delta}\right\}  \leq C\mathbb{E}\left(
		\varepsilon^{1/2}\log^{1/2}(1/\varepsilon) + n^{1/q-1/2}\varepsilon
		^{1/q-1}\log(1/\varepsilon)\right) /\delta.
	\end{align*}
	By letting $n \rightarrow\infty$ followed by $\varepsilon\rightarrow0$, we
	have
	\begin{align*}
		\lim_{n\rightarrow\infty} \mathbb{P}%
		\del[3]{ \envert[3]{\frac{\sum_{i\in I_1(s)}[ \Delta^Y(a,s,X_i) - M_{a,s}(\hat{\theta}_{a,s}, \theta_{a,s})] }{n_1(s)} }\geq \delta n^{-1/2}}
		= 0,
	\end{align*}
	Therefore,
	\begin{align*}
		\envert[3]{\frac{\sum_{i\in I_1(s)}[ \Delta^Y(a,s,X_i) - M_{a,s}(\hat{\theta}_{a,s}, \theta_{a,s})] }{n_1(s)} }
		= o_{p}(n^{-1/2}).
	\end{align*}
	For the same reason, we have
	\begin{align*}
		\envert[3]{\frac{\sum_{i\in I_0(s)}[ \Delta^Y(a,s,X_i) - M_{a,s}(\hat{\theta}_{a,s},\theta_{a,s})] }{n_0(s)} }
		= o_{p}(n^{-1/2}),
	\end{align*}
	and \eqref{eq:Delta} holds.
	
	\textbf{Step 2}. We have
	\begin{align*}
		&  \frac{1}{n}\sum_{i=1}^{n} \Delta^{Y,2}(a,S_{i},X_{i}) = \frac{1}{n}%
		\sum_{i=1}^{n} \sum_{s \in\mathcal{S}}1\{S_{i}=s\} (\Lambda_{a,s}^{Y}%
		(X_{i},\hat{\theta}_{a,s}) - \Lambda_{a,s}^{Y}(X_{i},\theta_{a,s}))^{2}\\
		&  \leq\del[3]{\frac{1}{n}\sum_{i =1}^n L_i^2} C\max_{s \in\mathcal{S}}%
		||\hat{\theta}_{a,s}-\theta_{a,s}||_{2}^{2} = o_{p}(1).
	\end{align*}
	This verifies Assumption \ref{ass:Delta}(ii). Assumption \ref{ass:Delta}(iii)
	holds by Assumption \ref{ass:par}(ii).
	
	\section{Proof of Theorem \ref{thm:linear}}
	
	Let
	\begin{align}
		&  \nu^{Y}(a,S_{i},X_{i}) = \mathbb{E}(Y_{i}(D_{i}(a))|S_{i},X_{i}) -
		\mathbb{E}(Y_{i}(D_{i}(a))|S_{i})\quad\text{and}\nonumber\\
		&  \nu^{D}(a,S_{i},X_{i}) = \mathbb{E}(D_{i}(a)|S_{i},X_{i}) - \mathbb{E}%
		(D_{i}(a)|S_{i}) .\label{eq:nutilde}%
	\end{align}
	Also recall that $W_{i}= Y_{i}(D_{i}(1))$, $Z_{i} = Y_{i}(D_{i}(0))$, $\mu
	^{Y}(a,S_{i},X_{i}) = \mathbb{E}(Y_{i}(D_{i}(a))|S_{i},X_{i})$. Then, we have
	\begin{align*}
		\mathbb{E}\pi(S_{i}) \Xi_{1}^{2}(\mathcal{D}_{i},S_{i})  &  = \mathbb{E}%
		\left\{ \frac{\left( W_{i} - \mu^{Y}(1,S_{i},X_{i}) - \tau(D_{i}(1) - \mu
			^{D}(1,S_{i},X_{i}))\right) ^{2}}{\pi(S_{i})}\right\} \\
		&  + \mathbb{E}\biggl\{\pi(S_{i})\biggl[\frac{\nu^{Y}(1,S_{i},X_{i}) -
			\tilde{\mu}^{Y}(1,S_{i},X_{i}) - \tau(\nu^{D}(1,S_{i},X_{i}) - \tilde{\mu}%
			^{D}(1,S_{i},X_{i}) )}{\pi(S_{i})}\\
		&  + \tilde{\mu}^{Y}(1,S_{i},X_{i}) - \tilde{\mu}^{Y}(0,S_{i},X_{i}) -
		\tau(\tilde{\mu}^{D}(1,S_{i},X_{i}) - \tilde{\mu}^{D}(0,S_{i},X_{i}%
		))\biggr]^{2}\biggr\}.
	\end{align*}
	Similarly, we have
	\begin{align*}
		\mathbb{E}(1-\pi(S_{i})) \Xi_{0}^{2}(\mathcal{D}_{i},S_{i})  &  =
		\mathbb{E}\left\{ \frac{\left( Z_{i} - \mu^{Y}(0,S_{i},X_{i}) - \tau(D_{i}(0)
			- \mu^{D}(0,S_{i},X_{i}))\right) ^{2}}{1-\pi(S_{i})}\right\} \\
		&  + \mathbb{E}\biggl\{(1-\pi(S_{i}))\biggl[\frac{\nu^{Y}(0,S_{i},X_{i}) -
			\tilde{\mu}^{Y}(0,S_{i},X_{i}) - \tau(\nu^{D}(0,S_{i},X_{i}) - \tilde{\mu}%
			^{D}(0,S_{i},X_{i}) )}{1-\pi(S_{i})}\\
		&  -\left(  \tilde{\mu}^{Y}(1,S_{i},X_{i}) - \tilde{\mu}^{Y}(0,S_{i},X_{i}) -
		\tau(\tilde{\mu}^{D}(1,S_{i},X_{i}) - \tilde{\mu}^{D}(0,S_{i},X_{i}))\right)
		\biggr]^{2}\biggr\}.
	\end{align*}
	Last, we have
	\begin{align*}
		\mathbb{E}\Xi_{2}^{2}(S_{i})  &  = \mathbb{E}(\mu^{Y}(1,S_{i},X_{i}) - \mu
		^{Y}(0,S_{i},X_{i})-\tau(\mu^{D}(1,S_{i},X_{i}) - \mu^{D}(0,S_{i},X_{i}%
		)))^{2}\\
		&  - \mathbb{E}(\nu^{Y}(1,S_{i},X_{i}) - \nu^{Y}(0,S_{i},X_{i})-\tau(\nu
		^{D}(1,S_{i},X_{i}) - \nu^{D}(0,S_{i},X_{i})))^{2}%
	\end{align*}
	Let
	\begin{align*}
		\sigma_{\ast}^{2}  &  = (\mathbb{P}(D_{i}(1)>D_{i}(0)))^{-2}\biggl\{\mathbb{E}%
		\sbr[3]{\frac{\del[1]{W_i - \mu^Y(1,S_i,X_i) - \tau(D_i(1) - \mu^D(1,S_i,X_i))}^2}{\pi(S_i)}}\\
		&  + \mathbb{E}%
		\sbr[3]{\frac{\del[1]{Z_i - \mu^Y(0,S_i,X_i) - \tau[D_i(0) - \mu^D(0,S_i,X_i)]}^2}{1-\pi(S_i)}}\\
		&  + \mathbb{E}%
		\del[2]{\mu^Y(1,S_i,X_i) - \mu^Y(0,S_i,X_i)-\tau[\mu^D(1,S_i,X_i) - \mu^D(0,S_i,X_i)]}^{2}%
		\biggr\},
	\end{align*}
	which does not depend on the working models $\overline{\mu}^{b}(a,S_{i}%
	,X_{i})$ for $a=0,1$ and $b = D,Y$.
	Then, we have
	\begin{align*}
		\sigma^{2}((t_{a,s},b_{a,s})_{a=0,1,s\in\mathcal{S}}) = \frac{\sigma_{*}^{2} +
			V((t_{a,s},b_{a,s})_{a=0,1,s\in\mathcal{S}})}{\mathbb{P}(D_{i}(1) >
			D_{i}(0))^{2}},
	\end{align*}
	where $\sigma_{*}^{2}$ does not depend on $(t_{a,s},b_{a,s})_{a=0,1,s\in
		\mathcal{S}}$ and
	\begin{align*}
		V((t_{a,s},b_{a,s})_{a=0,1,s\in\mathcal{S}})  &  = \mathbb{E}
		\del[3]{\sqrt{\frac{\pi(S_i)}{1-\pi(S_i)}} A_{0}(S_i,X_i) + \sqrt{\frac{1-\pi(S_i)}{\pi(S_i)}} A_{1}(S_i,X_i)}^{2}%
		\\
		&  = \sum_{s\in\mathcal{S}} p(s)\mathbb{E}%
		\sbr[3]{ \del[3]{\sqrt{\frac{\pi(s)}{1-\pi(s)}} A_{0}(s,X_i) + \sqrt{\frac{1-\pi(s)}{\pi(s)}} A_{1}(s,X_i)}^2 \biggl|S_i=s}
	\end{align*}
	where for $a=0,1$,
	\begin{align*}
		A_{a}(s,x)  &  = \nu^{Y}(a,s,x) - \tilde{\mu}^{Y}(a,s,x) - \tau(\nu^{D}(a,s,x)
		- \tilde{\mu}^{D}(a,s,x) )\\
		&  = (\nu^{Y}(a,s,x) -\tau\nu^{D}(a,s,x)) - \tilde{\Psi}_{i,s}^{\top}(t_{a,s}
		- \tau b_{a,s}),
	\end{align*}
	and $(\tilde{\mu}^{Y}(a,s,x),\tilde{\mu}^{D}(a,s,x))$ and $(\nu^{Y}%
	(a,s,x),\nu^{D}(a,s,x))$ are defined in \eqref{eq:mutilde} and
	\eqref{eq:nutilde}, respectively. Specifically, we have
	\begin{align*}
		\tilde{\mu}^{Y}(a,s,x) = \tilde{\Psi}_{i,s}^{\top}t_{a,s}, \quad\tilde{\mu
		}^{D}(a,s,x) = \tilde{\Psi}_{i,s}^{\top}b_{a,s}, \quad\text{and} \quad
		\tilde{\Psi}_{i,s} = \Psi_{i,s} - \mathbb{E}(\Psi_{i,s}|S_{i}=s).
	\end{align*}

	In order to minimize $V((t_{a,s},b_{a,s})_{a=0,1,s\in\mathcal{S}})$, it
	suffices to minimize
	\[
	\mathbb{E}%
	\sbr[3]{ \del[3]{\sqrt{\frac{\pi(s)}{1-\pi(s)}} A_{0}(s,X_i) + \sqrt{\frac{1-\pi(s)}{\pi(s)}} A_{1}(s,X_i)}^2 \biggl|S_i=s}
	\]
	for each $s \in\mathcal{S}$. In addition, we have
	\begin{align*}
		&  \mathbb{E}%
		\sbr[3]{ \del[3]{\sqrt{\frac{\pi(s)}{1-\pi(s)}} A_{0}(s,X_i) + \sqrt{\frac{1-\pi(s)}{\pi(s)}} A_{1}(s,X_i)}^2 \biggl|S_i=s}
		= \mathbb{E}\left(  (\overline{y}_{i,s} - \tilde{\Psi}_{i,s}^{\top}\gamma
		_{s})^{2} \biggl|S_{i}=s\right) ,
	\end{align*}
	where
	\begin{align*}
		\overline{y}_{i,s} = \sqrt{\frac{1-\pi(s)}{\pi(s)}} (\nu^{Y}(1,s,X_{i})
		-\tau\nu^{D}(1,s,X_{i})) + \sqrt{\frac{\pi(s)}{1-\pi(s)}}(\nu^{Y}(0,s,X_{i})
		-\tau\nu^{D}(0,s,X_{i}))
	\end{align*}
	and
	\begin{align*}
		\gamma_{s} = \sqrt{\frac{1-\pi(s)}{\pi(s)}} (t_{1,s} -\tau b_{1,s}) +
		\sqrt{\frac{\pi(s)}{1-\pi(s)}}(t_{0,s} -\tau b_{0,s}).
	\end{align*}
	By solving the first order condition, we find that
	\begin{align*}
		\Theta^{*}  &  =
		\begin{pmatrix}
			& (\theta_{a,s}^{*},\beta_{a,s}^{*})_{a=0,1,s\in\mathcal{S}}:\\
			& \sqrt{\frac{1-\pi(s)}{\pi(s)}} (\theta_{1,s}^{*} -\tau\beta_{1,s}^{*}) +
			\sqrt{\frac{\pi(s)}{1-\pi(s)}}(\theta_{0,s}^{*} -\tau\beta_{0,s}^{*}) =
			\mathbb{E}(\tilde{\Psi}_{i,s}\tilde{\Psi}_{i,s}^{\top}|S_{i}=s)^{-1}%
			\mathbb{E}(\tilde{\Psi}_{i,s}\overline{y}_{i,s}|S_{i}=s)
		\end{pmatrix}
		\\
		&  =
		\begin{pmatrix}
			& (\theta_{a,s}^{*},\beta_{a,s}^{*})_{a=0,1,s\in\mathcal{S}}:\\
			& \sqrt{\frac{1-\pi(s)}{\pi(s)}} (\theta_{1,s}^{*} -\tau\beta_{1,s}^{*}) +
			\sqrt{\frac{\pi(s)}{1-\pi(s)}}(\theta_{0,s}^{*} -\tau\beta_{0,s}^{*})\\
			& = \sqrt{\frac{1-\pi(s)}{\pi(s)}} (\theta_{1,s}^{L} -\tau\beta_{1,s}^{L}) +
			\sqrt{\frac{\pi(s)}{1-\pi(s)}}(\theta_{0,s}^{L} -\tau\beta_{0,s}^{L}).
		\end{pmatrix}
		,
	\end{align*}
	where
	\begin{align*}
		\theta_{a,s}^{L}  &  = [\mathbb{E}(\tilde{\Psi}_{i,s}\tilde{\Psi}_{i,s}%
		^{\top}|S_{i}=s)]^{-1}[\mathbb{E}(\tilde{\Psi}_{i,s}\nu^{Y}(a,s,X_{i}%
		)|S_{i}=s)]\\
		&  = [\mathbb{E}(\tilde{\Psi}_{i,s}\tilde{\Psi}_{i,s}^{\top}|S_{i}%
		=s)]^{-1}[\mathbb{E}(\tilde{\Psi}_{i,s}\mathbb{E}(Y_{i}(D_{i}(a))|S_{i}%
		,X_{i})|S_{i}=s)]\\
		&  = [\mathbb{E}(\tilde{\Psi}_{i,s}\tilde{\Psi}_{i,s}^{\top}|S_{i}%
		=s)]^{-1}[\mathbb{E}(\tilde{\Psi}_{i,s}Y_{i}(D_{i}(a))|S_{i}=s)].
	\end{align*}
	Similarly, we have
	\begin{align*}
		\beta_{a,s}^{L}  &  = [\mathbb{E}(\tilde{\Psi}_{i,s}\tilde{\Psi}_{i,s}^{\top
		}|S_{i}=s)]^{-1}[\mathbb{E}(\tilde{\Psi}_{i,s}D_{i}(a)|S_{i}=s)].
	\end{align*}
	This concludes the proof.
	
	\section{Proof of Theorem \ref{thm:linear2}}
	
	In order to verify Assumption \ref{ass:Delta}, by Theorem \ref{thm:par}, it
	suffices to show that $\hat{\theta}_{a,s}^{L} \overset{p}{\longrightarrow}
	\theta_{a,s}^{L}$ and $\hat{\beta}_{a,s}^{L} \overset{p}{\longrightarrow}
	\beta_{a,s}^{L}$. We focus on the former with $a=1$. Let $\{W_{i}^{s}%
	,X_{i}^{s}\}_{i \in[n]}$ be generated independently from the joint
	distribution of $(Y_{i}(D_{i}(1)),X_{i})$ given $S_{i}=s$ and denote
	$\Psi_{i,s}^{s} = \Psi_{s}(X_{i}^{s})$, $\tilde{\Psi}_{i,s}^{s} = \Psi
	_{s}(X_{i}^{s}) - \mathbb{E}\Psi_{s}(X_{i}^{s})$, $\dot{\Psi}_{i,1,s}^{s} =
	\Psi_{s}(X_{i}^{s}) - \frac{1}{n_{1}(s)}\sum_{i=N(s)+1}^{N(s)+n_{1}(s)}%
	\Psi_{s}(X_{i}^{s})$, and $\dot{\Psi}_{i,0,s}^{s} = \Psi_{s}(X_{i}^{s}) -
	\frac{1}{n_{0}(s)}\sum_{i=N(s)+n_{1}(s)+1}^{N(s)+n(s)}\Psi_{s}(X_{i}^{s})$.
	Then, we have
	\begin{align*}
		\hat{\theta}_{1,s}^{L} \overset{d}{=}
		\del[3]{\frac{1}{n_1(s)} \sum_{i=N(s)+1}^{N(s)+n_1(s)} \dot{\Psi}_{i,1,s}^s\dot{\Psi}_{i,1,s}^{s,\top} }^{-1}%
		\del[3]{\frac{1}{n_1(s)} \sum_{i=N(s)+1}^{N(s)+n_1(s)} \dot{\Psi}_{i,1,s}^sW_i^s }.
	\end{align*}
	As $\frac{1}{n_{1}(s)}\sum_{i=N(s)+1}^{N(s)+n_{1}(s)}\Psi_{i,s}^{s}
	\overset{p}{\longrightarrow} \mathbb{E}\Psi_{i,s}^{s} = \mathbb{E}(\Psi
	_{s}(X_{i}^{s})) = \mathbb{E}(\Psi_{s}(X_{i})|S_{i}=s)$ by the standard LLN,
	we have
	\begin{align*}
		&
		\del[3]{\frac{1}{n_1(s)} \sum_{i=N(s)+1}^{N(s)+n_1(s)} \dot{\Psi}_{i,1,s}^s\dot{\Psi}_{i,1,s}^{s,\top} }
		=
		\del[3]{\frac{1}{n_1(s)} \sum_{i=N(s)+1}^{N(s)+n_1(s)} \tilde{\Psi}_{i,s}^s\tilde{\Psi}_{i,s}^{s,\top} }
		+o_{p}(1),\\
		&
		\del[3]{\frac{1}{n_1(s)} \sum_{i=N(s)+1}^{N(s)+n_1(s)} \dot{\Psi}_{i,1,s}^sW_i^s }
		=
		\del[3]{\frac{1}{n_1(s)} \sum_{i=N(s)+1}^{N(s)+n_1(s)} \tilde{\Psi}_{i,s}^s W_i^s }
		+o_{p}(1).
	\end{align*}
	In addition, by the standard LLN,
	\begin{align*}
		&  \frac{1}{n_{1}(s)} \sum_{i=N(s)+1}^{N(s)+n_{1}(s)} \tilde{\Psi}_{i,s}%
		^{s}\tilde{\Psi}_{i,s}^{s,\top} \overset{p}{\longrightarrow} \mathbb{E}%
		\tilde{\Psi}_{i,s}^{s}\tilde{\Psi}_{i,s}^{s,\top} = \mathbb{E}(\tilde{\Psi
		}_{i,s}\tilde{\Psi}_{i,s}^{\top}|S_{i}=s),\\
		&  \frac{1}{n_{1}(s)} \sum_{i=N(s)+1}^{N(s)+n_{1}(s)} \tilde{\Psi}_{i,s}%
		^{s}W_{i}^{s} \overset{p}{\longrightarrow} \mathbb{E}\tilde{\Psi}_{i,s}%
		^{s}W_{i}^{s} = \mathbb{E}(\tilde{\Psi}_{i,s}Y_{i}(D_{i}(1))|S_{i}=s).
	\end{align*}
	Last, Assumption \ref{ass:psi} implies $\mathbb{E}(\tilde{\Psi}_{i,s}%
	\tilde{\Psi}_{i,s}^{\top}|S_{i}=s)$ is invertible, this means
	\begin{align*}
		\hat{\theta}_{1,s}^{L} \overset{p}{\longrightarrow} \left[ \mathbb{E}%
		(\tilde{\Psi}_{i,s}\tilde{\Psi}_{i,s}^{\top}|S_{i}=s)\right] ^{-1}%
		\mathbb{E}(\tilde{\Psi}_{i,s}Y_{i}(D_{i}(1))|S_{i}=s) = \theta_{1,s}^{L}.
	\end{align*}
	Similarly, we can show that $\hat{\theta}_{0,s}^{L}
	\overset{p}{\longrightarrow} \theta_{0,s}^{L}$ and $\hat{\beta}_{a,s}^{L}
	\overset{p}{\longrightarrow} \beta_{a,s}^{L}$ for $a=0,1$ and $s
	\in\mathcal{S}$. Therefore, Assumption \ref{ass:Delta} holds, and thus, all
	the results in Theorem \ref{thm:est} hold for $\hat{\tau}_{L}$. Then, the
	optimality result in the second half of Theorem \ref{thm:linear2} is a direct
	consequence of Theorem \ref{thm:linear}.
	
	Last, we compare the asymptotic variances of TSLS estimator and the estimator with the optimal linear adjustment with $\pi(s) = \pi$ for $s \in \mathcal{S}$ and $\Psi_{i,s} = X_i$. In this special case, we first note that the asymptotic variance of the estimator with the optimal linear adjustment is 
	\begin{align*}
		\frac{\sigma_{1}^{2}+\sigma_{0}^{2}+\sigma_{2}^{2}}{[\mathbb{E}(D(1)-D(0))]^2},
	\end{align*}
	where 
	\begin{align*}
		\sigma_{0}^{2}  &  = \mathbb{E} (1-\pi)\Xi_0^2(\mathcal{D}_i,S_i) \\
		\Xi_{0}(\mathcal{D
		}_{i}, S_{i})  &  := \left[ \left( \frac{1}{1-\pi}-1 \right) \tilde{%
			X}_{i}^{\top}\theta_{0s} + \tilde{X}_{i}^{\top}\theta_{1s} - \frac{\tilde{Z}
			_{i}}{1-\pi}\right]  - \tau\left[ \left( \frac{1}{1-\pi}-1 \right) \tilde{X}_{i}^{\top}\beta_{0s} + \tilde{X}_{i}
		^{\top}\beta_{1s} - \frac{\tilde{D}_{i}(0)}{1-\pi }\right] \\
		\sigma_{1}^{2}  &  =  \mathbb{E} \pi\Xi_1^2(\mathcal{D}_i,S_i), \\
		\Xi_{1}(\mathcal{D}_{i}, S_{i})  &  := \left[ \left( 1- \frac{1}{\pi }
		\right) \tilde{X}_{i}^{\top}\theta_{1s} - \tilde{X}_{i}^{\top}\theta_{0s} +
		\frac{\tilde{W}_{i}}{\pi }\right]   - \tau\left[ \left( 1- \frac{1}{\pi} \right) \tilde{X}_{i}^{\top}\beta_{1
			s} - \tilde{X}_{i}^{\top}\beta_{0s} + \frac{ \tilde{D}_{i}(1)}{\pi}\right] ,\\
		\sigma_{2}^{2}  &  = \mathbb{E}\left[ \mathbb{E}%
		\sbr[1]{Y(D(1)) - Y(D(0)) - (D(1)-D(0))\tau|S_i} \right] ^{2},
	\end{align*}
	with
	\begin{align}
		\theta_{as}  &  = \left[ \mathbb{E}(\tilde{X}_{is}\tilde{X}_{is}^{\top}%
		|S_{i}=s)\right] ^{-1} \left[  \mathbb{E}(\tilde{X}_{is}Y_{i}(D_{i}%
		(a))|S_{i}=s)\right] \quad\text{and} \notag \\
		\beta_{as}  &  = \left[ \mathbb{E}(\tilde{X}_{is}\tilde{X}_{is}^{\top}%
		|S_{i}=s)\right] ^{-1} \left[  \mathbb{E}(\tilde{X}_{is}D_{i}(a)|S_{i}%
		=s)\right] ,\quad a=0,1.
		\label{eq:optimal1}
	\end{align}

	Observe that $\sigma_{TSLS,1}^{2}%
	+\sigma_{TSLS,0}^{2}$ can also be written as
	\begin{align}
		\mathbb{E}\left[ \pi(S_{i})\Xi_{1}(\mathcal{D}_{i}, S_{i})^{2} + (1-\pi
		(S_{i}))\Xi_{0}(\mathcal{D}_{i}, S_{i})^{2}\right] ,\label{eq:Xi_linear}%
	\end{align}
	where
	\begin{align*}
		\Xi_{1}(\mathcal{D}_{i}, S_{i})  &  := \left[ \left( 1- \frac{1}{\pi(S_{i})}
		\right) \tilde{X}_{i}^{\top}\theta_{1s} - \tilde{X}_{i}^{\top}\theta_{0s} +
		\frac{\tilde{W}_{i}}{\pi(S_{i})}\right] \\ &  - \tau\left[ \left( 1- \frac{1}{\pi(S_{i})} \right) \tilde{X}_{i}^{\top}\beta_{1
			s} - \tilde{X}_{i}^{\top}\beta_{0s} + \frac{ \tilde{D}_{i}(1)}{\pi(S_{i%
			})}\right] ,\\
		\Xi_{0}(\mathcal{D
		}_{i}, S_{i})  &  := \left[ \left( \frac{1}{1-\pi(S_{i})}-1 \right) \tilde{%
			X}_{i}^{\top}\theta_{0s} + \tilde{X}_{i}^{\top}\theta_{1s} - \frac{\tilde{Z}
			_{i}}{1-\pi(S_{i})}\right] \\ & - \tau\left[ \left( \frac{1}{1-\pi(S_{i})}-1 \right) \tilde{X}_{i}^{\top}\beta_{0s} + \tilde{X}_{i}
		^{\top}\beta_{1s} - \frac{\tilde{D}_{i}(0)}{1-\pi(S_{i}) }\right] ,
	\end{align*}
	with 
	$\theta_{1s} = \theta_{0s}$, $\beta_{1s} =
	\beta_{0s}$ and $\theta_{1s} - \tau\beta_{1s} = \lambda_{x}^{*}$, where $\lambda_x^*$ is the first $d_x$ coefficients of $\lambda^*$ defined in Theorem \ref{thm:TSLS} where $d_x$ is the dimension of $X_i$. By Theorem \ref{thm:linear}, we achieve the optimal linear adjustment when $\theta_{a,s}$ and $\beta_{a,s}$ satisfy \eqref{eq:optimal1},  which implies
	\begin{align*}
		\sigma_{1}^{2}+\sigma_{0}^{2} \leq\sigma_{TSLS,1}^{2}+\sigma_{TSLS,0}^{2}.
	\end{align*}
	In addition, we have $\sigma_{2}^{2} = \sigma_{TSLS,2}^{2}$ and $0 \leq
	\sigma_{TSLS,3}^{2}$, 
	which implies the desired result.

	\section{Proof of Theorem \ref{thm:OLS_MLE}}
	
	Let $\{D_{i}^{s}(1),X_{i}^{s}\}_{i \in[n]}$ be generated independently from
	the joint distribution of $(D_{i}(1),X_{i})$ given $S_{i}=s$, $\Psi_{i,s}^{s}
	= \Psi_{s}(X_{i}^{s})$, and $\mathring{\Psi}_{i,s}^{s}=(1,\Psi_{i,s}^{s,\top
	})^{\top}$. Then, we have, pointwise in $b$,
	\begin{align*}
		&  \frac{1}{n_{1}(s)} \sum_{i \in I_{1}(s)} \left[ D_{i} \log(\lambda
		(\mathring{\Psi}_{i,s}^{\top} b)) + (1-D_{i}) \log(1-\lambda(\mathring{\Psi
		}_{i,s}^{\top} b)) \right] \\
		&  \overset{d}{=} \frac{1}{n_{1}(s)} \sum_{i = N(s)+1}^{N(s)+n_{1}(s)} \left[
		D_{i}^{s}(1) \log(\lambda(\mathring{\Psi}_{i,s}^{s,\top} b)) + (1-D_{i}%
		^{s}(1)) \log(1-\lambda(\mathring{\Psi}_{i,s}^{s, \top} b)) \right] \\
		&  \overset{p}{\longrightarrow} \mathbb{E}\left[ D_{i}^{s}(1) \log
		(\lambda(\mathring{\Psi}_{i,s}^{s, \top} b)) + (1-D_{i}^{s}(1)) \log
		(1-\lambda(\mathring{\Psi}_{i,s}^{s,\top} b)) \right] \\
		&  = \mathbb{E}\left[ D_{i}(1) \log(\lambda(\mathring{\Psi}_{i,s}^{ \top} b))
		+ (1-D_{i}(1)) \log(1-\lambda(\mathring{\Psi}_{i,s}^{\top} b)) |S_{i}=s\right]
		.
	\end{align*}
	As the logistic likelihood function is concave in $b$, the pointwise
	convergence in $b$ implies uniform convergence, i.e.,
	\begin{align*}
		\sup_{b} \biggl\vert  &  \frac{1}{n_{1}(s)} \sum_{i \in I_{1}(s)} \left[ D_{i}
		\log(\lambda(\mathring{\Psi}_{i,s}^{\top} b)) + (1-D_{i}) \log(1-\lambda
		(\mathring{\Psi}_{i,s}^{\top} b)) \right] \\
		&  - \mathbb{E}\left[ D_{i}(1) \log(\lambda(\mathring{\Psi}_{i,s}^{\top} b)) +
		(1-D_{i}(1)) \log(1-\lambda(\mathring{\Psi}_{i,s}^{\top} b)) |S_{i}=s\right]
		\biggr\vert \overset{p}{\longrightarrow} 0.
	\end{align*}
	Then, by the standard proof for the extremum estimation, we have $\hat{\beta
	}_{a,s}^{MLE} \overset{p}{\longrightarrow} \beta_{a,s}^{MLE}$. Similarly, we
	can show that $\hat{\theta}_{a,s}^{OLS} \overset{p}{\longrightarrow}
	\theta_{a,s}^{OLS}$. The verifies Assumption \ref{ass:par}(i). Assumptions
	\ref{ass:par}(ii) and \ref{ass:par}(iii) follow from Assumption
	\ref{ass:OLS_MLE}(ii). Then, the desired results hold due to Theorem
	\ref{thm:par}.
	
	\section{Proof of Theorem \ref{thm:linear3}}
	
	We note that the adjustments proposed in Theorem \ref{thm:linear3} are still
	parametric. Specifically, we have
	\begin{align*}
		&  \overline{\mu}^{Y}(a,s,X_{i}) = \Lambda_{a,s}^{Y}(X_{i},\{\beta_{1,s}%
		^{MLE},\beta_{0,s}^{MLE},\theta_{a,s}^{F}\}),\\
		&  \overline{\mu}^{D}(a,s,X_{i}) = \Lambda_{a,s}^{D}(X_{i},\{\beta_{1,s}%
		^{MLE},\beta_{0,s}^{MLE},\beta_{a,s}^{F}\}),\\
		&  \hat{\mu}^{Y}(a,s,X_{i}) = \Lambda_{a,s}^{Y}(X_{i},\{\hat{\beta}%
		_{1,s}^{MLE},\hat{\beta}_{0,s}^{MLE},\hat{\theta}_{a,s}^{F}\}), \quad
		\text{and}\\
		&  \hat{\mu}^{D}(a,s,X_{i}) = \Lambda_{a,s}^{D}(X_{i},\{\hat{\beta}%
		_{1,s}^{MLE},\hat{\beta}_{0,s}^{MLE},\hat{\beta}_{a,s}^{F}\}),
	\end{align*}
	where
	\begin{align*}
		\Lambda_{a,s}^{Y}(X_{i},\{b_{1},b_{0},t_{a}^{*}\}) =
		\begin{pmatrix}
			\Psi_{i,s}^{ \top}\\
			\lambda(\mathring{\Psi}_{i,s}^{\top} b_{1})\\
			\lambda(\mathring{\Psi}_{i,s}^{\top} b_{0})
		\end{pmatrix}
		^{\top}t_{a}^{*} \quad\text{and} \quad\Lambda_{a,s}^{D}(X_{i},\{b_{1}%
		,b_{0},b_{a}^{*}\}) =
		\begin{pmatrix}
			\Psi_{i,s}^{ \top}\\
			\lambda(\mathring{\Psi}_{i,s}^{\top} b_{1})\\
			\lambda(\mathring{\Psi}_{i,s}^{\top} b_{0})
		\end{pmatrix}
		^{\top}b_{a}^{*}.
	\end{align*}
	Therefore, in view of Theorem \ref{thm:par}, to verify Assumption
	\ref{ass:Delta}, it suffices to show that $\hat{\theta}_{a,s}^{F}
	\overset{p}{\longrightarrow} \theta_{a,s}^{F}$ and $\hat{\beta}_{a,s}^{F}
	\overset{p}{\longrightarrow} \beta_{a,s}^{F}$, as we have already shown the
	consistency of $\hat{\beta}_{a,s}^{MLE}$ in the proof of Theorem
	\ref{thm:OLS_MLE}. We focus on $\hat{\theta}_{a,s}^{F}$. Define $\dot{\Phi
	}_{i,a,s}: = \Phi_{i,s} - \frac{1}{n_{a}(s)}\sum_{i\in I_{a}(s)}\Phi_{i,s}$,
	where
	\begin{align*}
		\Phi_{i,s} =
		\begin{pmatrix}
			\Psi_{i,s}\\
			\lambda(\mathring{\Psi}_{i,s}^{\top} \beta_{1,s}^{MLE})\\
			\lambda(\mathring{\Psi}_{i,s}^{\top} \beta_{0,s}^{MLE})
		\end{pmatrix}
		.
	\end{align*}
	We first show that
	\begin{align}
		\frac{1}{n_{a}(s)} \sum_{i \in I_{a}(s)}\breve{\Phi}_{i,a,s}\breve{\Phi
		}_{i,a,s}^{\top}= \frac{1}{n_{a}(s)} \sum_{i \in I_{a}(s)}\dot{\Phi}%
		_{i,a,s}\dot{\Phi}_{i,a,s}^{\top}+ o_{p}(1).\label{eq:1}%
	\end{align}
	Let $v,u \in\Re^{d_{\Psi}+2}$ be two arbitrary vectors such that
	$||u||_{2}=||v||_{2} = 1$. Then, we have
	\begin{align}
		&
		\envert[3]{ v^\top \sbr[3]{\frac{1}{n_a(s)} \sum_{i \in I_a(s)}\left(\breve{\Phi}_{i,a,s}\breve{\Phi}_{i,a,s}^\top -\dot{\Phi}_{i,a,s}\dot{\Phi}_{i,a,s}^\top \right)}u }\nonumber\\
		&  =
		\envert[3]{ \frac{1}{n_a(s)} \sum_{i \in I_a(s)} \left[(v^\top\breve{\Phi}_{i,a,s})(u^\top\breve{\Phi}_{i,a,s}) - (v^\top\dot{\Phi}_{i,a,s})(u^\top\dot{\Phi}_{i,a,s})\right] }\nonumber\\
		&  =
		\envert[3]{ \frac{1}{n_a(s)} \sum_{i \in I_a(s)} \left[v^\top(\breve{\Phi}_{i,a,s}-\dot{\Phi}_{i,a,s})(u^\top\breve{\Phi}_{i,a,s}) + (v^\top\dot{\Phi}_{i,a,s})u^\top( \breve{\Phi}_{i,a,s}-\dot{\Phi}_{i,a,s} )\right] }\nonumber\\
		&  \leq\frac{1}{n_{a}(s)}\sum_{i \in I_{a}(s)} ||\breve{\Phi}_{i,a,s}%
		-\dot{\Phi}_{i,a,s}||_{2} (||\breve{\Phi}_{i,a,s}||_{2}+||\dot{\Phi}%
		_{i,a,s}||_{2}) \label{r12}%
	\end{align}
	where the first inequality is due to Cauchy-Schwarz inequality. We now show
	(\ref{r12}) is $o_{p}(1)$. First note that
	\begin{align*}
		&  ||\breve{\Phi}_{i,a,s}-\dot{\Phi}_{i,a,s}||_{2}\leq\sum_{a^{\prime}=0,1}
		\|B_{a^{\prime}}\|_{2}%
	\end{align*}
	where
	\begin{align*}
		B_{a^{\prime}}:= \lambda(\mathring{\Psi}_{i,s}^{\top} \hat{\beta}_{a^{\prime
			},s}^{MLE})-\lambda(\mathring{\Psi}_{i,s}^{\top}\beta_{a^{\prime},s}^{MLE}
		)-\frac{1}{n_{a}(s)}\sum_{i\in I_{a}(s)}
		\sbr[1]{\lambda(\mathring{\Psi}_{i,s}^{\top} \hat{\beta}_{a',s}^{MLE})-\lambda(\mathring{\Psi}_{i,s}^{\top}\beta_{a',s}^{MLE} )}.
	\end{align*}
	Note that
	\begin{align*}
		\lambda(\mathring{\Psi}_{i,s}^{\top} \hat{\beta}_{a^{\prime},s}^{MLE}%
		)-\lambda(\mathring{\Psi}_{i,s}^{\top}\beta_{a^{\prime},s}^{MLE} ) &
		=\frac{\partial\lambda(\mathring{\Psi}_{i,s}^{\top} \tilde{\beta}_{a^{\prime
				},s}^{MLE})}{\partial\beta_{a^{\prime},s}}(\hat{\beta}_{a^{\prime},s}%
		^{MLE}-\beta_{a^{\prime},s}^{MLE})\\
		\frac{1}{n_{a}(s)}\sum_{i\in I_{a}(s)}\lambda(\mathring{\Psi}_{i,s}^{\top}
		\hat{\beta}_{a^{\prime},s}^{MLE})-\lambda(\mathring{\Psi}_{i,s}^{\top}%
		\beta_{a^{\prime},s}^{MLE} ) &
		=\sbr[3]{\frac{1}{n_a(s)}\sum_{i\in I_a(s)}\frac{\partial \lambda(\mathring{\Psi}_{i,s}^{\top} \tilde{\beta}_{a',s}^{MLE})}{\partial \beta_{a',s}}}(\hat
		{\beta}_{a^{\prime},s}^{MLE}-\beta_{a^{\prime},s}^{MLE})
	\end{align*}
	where $\tilde{\beta}_{a^{\prime},s}^{MLE}$ is a mid-point of $\hat{\beta
	}_{a^{\prime},s}^{MLE}$ and $\beta_{a^{\prime},s}^{MLE}$. Hence
	\begin{align*}
		\|B_{a^{\prime}}\|_{2}%
		=\enVert[3]{\frac{\partial \lambda(\mathring{\Psi}_{i,s}^{\top} \tilde{\beta}_{a',s}^{MLE})}{\partial \beta_{a',s}}-\frac{1}{n_a(s)}\sum_{i\in I_a(s)}\frac{\partial \lambda(\mathring{\Psi}_{i,s}^{\top} \tilde{\beta}_{a',s}^{MLE})}{\partial \beta_{a',s}}}_{2}%
		\|\hat{\beta}_{a^{\prime},s}^{MLE}-\beta_{a^{\prime},s}^{MLE}\|_{2}.
	\end{align*}
	Since $\partial\lambda(u)/\partial u\leq1$,
	\begin{align*}
		&
		\enVert[3]{\frac{\partial \lambda(\mathring{\Psi}_{i,s}^{\top} \tilde{\beta}_{a',s}^{MLE})}{\partial \beta_{a',s}}}_{2}%
		=\enVert[3]{\left. \frac{\partial \lambda(u)}{\partial u}\right|_{u=\mathring{\Psi}_{i,s}^{\top} \tilde{\beta}_{a',s}^{MLE}}\cdot \mathring{\Psi}_{i,s}^{\top} }_{2}%
		\leq\|\mathring{\Psi}_{i,s}\|_{2}.
	\end{align*}
	Thus,
	\begin{align}
		\|B_{a^{\prime}}\|_{2} & \leq
		\del[3]{\|\mathring{\Psi}_{i,s}\|_2+\frac{1}{n_a(s)}\sum_{i\in I_a(s)}\|\mathring{\Psi}_{i,s}\|_2}\|\hat
		{\beta}_{a^{\prime},s}^{MLE}-\beta_{a^{\prime},s}^{MLE}\|_{2}\nonumber\\
		&  \leq
		\del[3]{2+\|\Psi_{i,s}\|_2+\frac{1}{n_a(s)}\sum_{i\in I_a(s)}\|\Psi_{i,s}\|_2}\|\hat
		{\beta}_{a^{\prime},s}^{MLE}-\beta_{a^{\prime},s}^{MLE}\|_{2},\nonumber\\
		||\breve{\Phi}_{i,a,s}-\dot{\Phi}_{i,a,s}||_{2} & \leq
		\del[3]{2+\|\Psi_{i,s}\|_2+\frac{1}{n_a(s)}\cdot\sum_{i\in I_a(s)}\|\Psi_{i,s}\|_2}\sum
		_{a^{\prime}=0,1}\|\hat{\beta}_{a^{\prime},s}^{MLE}-\beta_{a^{\prime},s}%
		^{MLE}\|_{2}.\label{r15}%
	\end{align}

	Moreover, we can show
	\begin{align}
		&  ||\breve{\Phi}_{i,a,s}||_{2}+||\dot{\Phi}_{i,a,s}||_{2} \leq
		2\del[3]{4 + || \Psi_{i,s}||_2+\frac{1}{n_a(s)}\sum_{i \in I_a(s)} \|\Psi_{i,s}\|_2}.\label{r17}%
	\end{align}
	Substituting (\ref{r15}), (\ref{r17}) and the fact that $||\hat{\beta}%
	_{a,s}^{MLE} - \beta_{a,s}^{MLE}||_{2} = o_{p}(1)$ into (\ref{r12}), we show
	that (\ref{r12}) is $o_{p}(1)$. As it holds for arbitrary $u,v$, it implies
	\eqref{eq:1}. Similarly, we can show that
	\begin{align}
		\frac{1}{n_{a}(s)} \sum_{i \in I_{a}(s)}\breve{\Phi}_{i,a,s}Y_{i} = \frac
		{1}{n_{a}(s)} \sum_{i \in I_{a}(s)}\dot{\Phi}_{i,a,s}Y_{i} + o_{p}%
		(1).\label{eq:2}%
	\end{align}
	Following the same argument in the proof of Theorem \ref{thm:linear2}, we can
	show that
	\begin{align*}
		\sbr[3]{\frac{1}{n_a(s)} \sum_{i \in I_a(s)}\dot{\Phi}_{i,a,s}\dot{\Phi}_{i,a,s}^\top}^{-1}%
		\sbr[3]{\frac{1}{n_a(s)} \sum_{i \in I_a(s)}\dot{\Phi}_{i,a,s}Y_i }
		\overset{p}{\longrightarrow} \theta_{a,s}^{F}.
	\end{align*}
	In addition, by Assumption \ref{ass:psi_new}, with probability approaching
	one, there exists a constant $c>0$ such that
	\begin{align}
		\lambda_{\min}%
		\del[3]{\frac{1}{n_a(s)} \sum_{i \in I_a(s)}\dot{\Phi}_{i,a,s}\dot{\Phi}_{i,a,s}^\top}
		\geq c.\label{eq:3}%
	\end{align}

	Combining \eqref{eq:1}, \eqref{eq:2}, and \eqref{eq:3}, we can show that
	\begin{align*}
		\hat{\theta}_{a,s}^{F}  &  =
		\sbr[3]{\frac{1}{n_a(s)} \sum_{i \in I_a(s)}\breve{\Phi}_{i,a,s}\breve{\Phi}_{i,a,s}^\top }^{-1}
		\sbr[3]{\frac{1}{n_a(s)} \sum_{i \in I_a(s)}\breve{\Phi}_{i,a,s}Y_i}\\
		&  =
		\sbr[3]{\frac{1}{n_a(s)} \sum_{i \in I_a(s)}\dot{\Phi}_{i,a,s}\dot{\Phi}_{i,a,s}^\top}^{-1}%
		\sbr[3]{\frac{1}{n_a(s)} \sum_{i \in I_a(s)}\dot{\Phi}_{i,a,s}Y_i }+o_{p}(1)
		\overset{p}{\longrightarrow} \theta_{a,s}^{F}.
	\end{align*}
	Similarly, we have $\hat{\beta}_{a,s}^{F} \overset{p}{\longrightarrow}
	\beta_{a,s}^{F}$, which implies all the results in Theorem \ref{thm:est} hold
	for $\hat{\tau}_{F}$. The optimality result in the second half of the theorem
	is a direct consequence of Theorem \ref{thm:linear}.
	
	\section{Proof of Theorem \ref{thm:np}}
	
	We focus on verifying Assumption \ref{ass:Delta} for $\hat{\mu}^{D}%
	(a,s,X_{i})$. The proof for $\hat{\mu}^{Y}(a,s,X_{i})$ is similar and hence
	omitted. Following the proof of Theorem \ref{thm:OLS_MLE}, we note that, for
	each $a =0,1$ and $s \in\mathcal{S}$, the data in cell $I_{a}(s)$, denoted
	$\{D_{i}^{s}(a),X_{i}^{s}\}_{i \in[n]}$, can be viewed as i.i.d. following the
	joint distribution of $(D_{i}(a),X_{i})$ given $S_{i}=s$ conditionally on
	$\{A_{i},S_{i}\}_{i \in[n]}$. Then following the standard logistic sieve
	regression in \cite{HIR03}, we have
	\begin{align*}
		\max_{a=0,1, s \in\mathcal{S}} ||\hat{\beta}_{a,s}^{NP} - \beta_{a,s}%
		^{NP}||_{2} = O_{p}\left( \sqrt{ h_{n}/n_{a}(s)}\right) .
	\end{align*}
	Then we have
	\begin{align}
		&  \biggl|\frac{\sum_{i\in I_{1}(s)}\Delta^{D}(a,s,X_{i})}{n_{1}(s)} -
		\frac{\sum_{i \in I_{0}(s)}\Delta^{D}(a,s,X_{i})}{n_{0}(s)}\biggr|\nonumber\\
		& \leq\biggl|\frac{\sum_{i\in I_{1}(s)}
			\del[1]{\lambda(\mathring{\Psi}_{i,n}^\top \hat{\beta}_{a,s}^{NP}) -\lambda(\mathring{\Psi}_{i,n}^\top \beta_{a,s}^{NP})}
		}{n_{1}(s)} - \frac{\sum_{i \in I_{0}(s)}%
			\del[1]{\lambda(\mathring{\Psi}_{i,n}^\top \hat{\beta}_{a,s}^{NP})-\lambda(\mathring{\Psi}_{i,n}^\top \beta_{a,s}^{NP})}}%
		{n_{0}(s)}\biggr|\nonumber\\
		& \qquad+
		\envert[3]{\frac{1}{n_1(s)}\sum_{i \in I_1(s)}\del[2]{ R^D(a,s,X_i) - \mathbb{E}[R^D(a,s,X_i)|S_i=s]}}\nonumber\\
		& \qquad+
		\envert[3]{\frac{1}{n_0(s)}\sum_{i \in I_0(s)}\del[2]{ R^D(a,s,X_i) - \mathbb{E}[R^D(a,s,X_i)|S_i=s]}}
		= :I + II + III.\label{eq:Delta_np}%
	\end{align}
	To bound term $I$ in \eqref{eq:Delta_np}, we define $M_{a,s}(\beta_{1}%
	,\beta_{2}) := \mathbb{E}%
	\sbr[1]{\lambda(\mathring{\Psi}_{i,n}^\top \beta_1) - \lambda(\mathring{\Psi}_{i,n}^\top \beta_2)|S_i=s}
	= \mathbb{E}[\lambda(\mathring{\Psi}_{i,n}^{s,\top} \beta_{1}) -
	\lambda(\mathring{\Psi}_{i,n}^{s,\top} \beta_{2})]$, where $\mathring{\Psi
	}_{i,n}^{s} = \mathring{\Psi}(X_{i}^{s})$. Then we have
	\begin{align*}
		I  &  \leq\biggl|\frac{\sum_{i\in I_{1}(s)}
			\sbr[1]{\lambda(\mathring{\Psi}_{i,n}^\top \hat{\beta}_{a,s}^{NP}) -\lambda(\mathring{\Psi}_{i,n}^\top \beta_{a,s}^{NP}) - M_{a,s}(\hat{\beta}_{a,s}^{NP},\beta_{a,s}^{NP})}
		}{n_{1}(s)}\biggr|\\
		&  \qquad+ \biggl|\frac{\sum_{i \in I_{0}(s)}%
			\sbr[1]{\lambda(\mathring{\Psi}_{i,n}^\top \hat{\beta}_{a,s}^{NP})-\lambda(\mathring{\Psi}_{i,n}^\top \beta_{a,s}^{NP})-M_{a,s}(\hat{\beta}_{a,s}^{NP},\beta_{a,s}^{NP}) }}%
		{n_{0}(s)}\biggr| =: I_{1}+I_{2}.
	\end{align*}
	Following the argument in the proof of Theorem \ref{thm:par}, in order to show
	$I_{1} = o_{p}(n^{-1/2})$, we only need to show
	\[
	n^{1/2}\mathbb{E}%
	\sbr[1]{||\mathbb{P}_{n_1(s)} - \mathbb{P}||_{\mathcal{F}}|\{A_i,S_i\}_{i \in [n]}}1\{n_{1}%
	(s) \geq n\varepsilon, n_{0}(s)\geq n\varepsilon\} = o(1),
	\]
	where $\varepsilon$ is an arbitrary but fixed constant, and
	\[
	\mathcal{F} :=
	\cbr[2]{\lambda(\mathring{\Psi}_{i,n}^\top \beta_1)-\lambda(\mathring{\Psi}_{i,n}^\top \beta_{a,s}^{NP}): \beta_1 \in \Re^{h_n}, ||\beta_1 - \beta_{a,s}^{NP}||_2 \leq C \sqrt{h_n/n_a(s)}},
	\]
	for some constant $C>0$. Furthermore, we note that $\mathcal{F}$ has a bounded
	envelope, is of the VC-type with VC-index upper bounded by $Ch_{n}%
	$,\footnote{See \citet[Section 2.6.5]{VW96} for the calculation of the VC
		index.} and has
	\begin{align*}
		\sup_{f \in\mathcal{F}}\mathbb{E}\sbr[1]{f^2|\{A_i,S_i\}_{i \in [n]}}
		\leq\frac{Ch_{n}}{n_{a}(s)}.
	\end{align*}

	Invoking \citet[Corollary 5.1]{CCK14} with $A$ being a constant, $\nu=Ch_{n}$,
	$\sigma^{2} = Ch_{n}/n_{a}(s)$, and $F$ and $M$ being $2h_{n}$, we have
	\begin{align*}
		&  n^{1/2}\mathbb{E}\left[ ||\mathbb{P}_{n_{1}(s)} - \mathbb{P}||_{\mathcal{F}%
		}|\{A_{i},S_{i}\}_{i \in[n]}\right] 1\{n_{1}(s) \geq n\varepsilon, n_{0}(s)
		\geq n\varepsilon\}\\
		& \leq C\sqrt{\frac{n}{n_{1}(s)}}
		\del[3]{\sqrt{ \frac{h_n^2 \log n}{n_a(s)}} + \frac{h_n \log n}{\sqrt{n_1(s)}}}1\{n_{1}%
		(s) \geq n\varepsilon, n_{0}(s) \geq n\varepsilon\}\\
		& \leq C\sqrt{\frac{1}{\varepsilon}}%
		\del[3]{\sqrt{\frac{h_n^2\log n}{n\varepsilon}}+\frac{h\log n}{\sqrt{n\varepsilon}}}
		\rightarrow0,
	\end{align*}
	as $n\to\infty$.
	
	Similarly, we can show $I_{2} = o_{p}(n^{-1/2})$. In addition, we note that
	\begin{align*}
		II \overset{d}{=}
		\envert[3]{\frac{1}{n_1(s)}\sum_{i =N(s)+1}^{N(s)+n_1(s)} \del[2]{R^D(a,s,X_i^s) - \mathbb{E}[R^D(a,s,X_i^s)]}}
		= o_{p}(n^{-1/2})
	\end{align*}
	by the Chebyshev's inequality as $\mathbb{E}R^{D,2}(a,s,X_{i}^{s}) =
	\mathbb{E}[R^{D,2}(a,s,X_{i})|S_{i}=s] = o(1)$ by Assumption \ref{ass:np}(ii).
	Similarly we have $III = o_{p}(n^{-1/2})$. Combining the bounds of $I$, $II$,
	$III$ with \eqref{eq:Delta_np}, we have
	\begin{align*}
		\biggl|\frac{\sum_{i\in I_{1}(s)}\Delta^{D}(a,s,X_{i})}{n_{1}(s)} - \frac
		{\sum_{i \in I_{0}(s)}\Delta^{D}(a,s,X_{i})}{n_{0}(s)}\biggr| = o_{p}%
		(n^{-1/2}),
	\end{align*}
	which verifies Assumption \ref{ass:Delta}(i).
	
	To verify Assumption \ref{ass:Delta}(ii), we note that
	\begin{align*}
		& \frac{1}{n}\sum_{i=1}^{n} \Delta^{D,2}(a,s,X_{i}) \leq\frac{2}{n}\sum
		_{i=1}^{n}
		\del[1]{\lambda(\mathring{\Psi}_{i,n}^\top\hat{\beta}_{a,s}^{NP}) - \lambda(\mathring{\Psi}_{i,n}^\top \beta_{a,s}^{NP})}^{2}
		+ \frac{2}{n}\sum_{i=1}^{n} R^{D,2}(a,S_{i},X_{i})\\
		&  \leq\frac{2}{n}\sum_{i=1}^{n} \|\mathring{\Psi}_{i,n}\|_{2}^{2}\|\hat
		{\beta}_{a,s}^{NP}-\beta_{a,s}^{NP}\|_{2}^{2}+ \frac{2}{n}\sum_{i=1}^{n}
		R^{D,2}(a,S_{i},X_{i})\\
		& = \frac{2}{n}\sum_{i=1}^{n} \|\mathring{\Psi}_{i,n}\|_{2}^{2}\|\hat{\beta
		}_{a,s}^{NP}-\beta_{a,s}^{NP}\|_{2}^{2}+ o_{p}(1)\leq2\max_{i}\|\mathring
		{\Psi}_{i,n}\|_{2}^{2}\max_{s} \|\hat{\beta}_{a,s}^{NP}-\beta_{a,s}^{NP}%
		\|_{2}^{2}+ o_{p}(1)\\
		&  = O_{p}\del[1]{\zeta^2(h_n) h_n/n_a(s)}+ o_{p}(1)=o_{p}(1)
	\end{align*}
	where the first equality is due to Assumption \ref{ass:np}(ii), and the second
	equality is due to Assumption \ref{ass:np}(iv).
	
	Last, Assumption \ref{ass:Delta}(iii) is implied by Assumption
	\ref{ass:assignment1}(vi) via Jensen's inequality.
	
	\section{Proof of Theorem \ref{thm:hd}}
	
	We focus on verifying Assumption \ref{ass:Delta} for $\hat{\mu}^{D}%
	(a,s,X_{i})$. The proof for $\hat{\mu}^{Y}(a,s,X_{i})$ is similar and hence
	omitted. Following the proof of Theorem \ref{thm:OLS_MLE}, we note that, for
	each $a =0,1$ and $s \in\mathcal{S}$, the data in cell $I_{a}(s)$, denoted
	$\{D_{i}^{s}(a),X_{i}^{s}\}_{i \in[n]}$, can be viewed as i.i.d. following the
	joint distribution of $(D_{i}(a),X_{i})$ given $S_{i}=s$ conditionally on
	$\{A_{i},S_{i}\}_{i \in[n]}$. Then following the standard logistic Lasso
	regression in \cite{BCFH13}, we have
	\begin{align*}
		\max_{a=0,1, s \in\mathcal{S}} ||\hat{\beta}_{a,s}^{R} - \beta_{a,s}^{R}||_{2}
		= O_{p}\left( \sqrt{ h_{n} \log p_{n}/n_{a}(s)}\right)  \quad\text{and}
		\quad\max_{a=0,1, s \in\mathcal{S}} ||\hat{\beta}_{a,s}^{R} ||_{0} =
		O_{p}(h_{n}).
	\end{align*}
	Then, we have
	\begin{align}
		&  \biggl|\frac{\sum_{i\in I_{1}(s)}\Delta^{D}(a,s,X_{i})}{n_{1}(s)} -
		\frac{\sum_{i \in I_{0}(s)}\Delta^{D}(a,s,X_{i})}{n_{0}(s)}\biggr|\nonumber\\
		& \leq\biggl|\frac{\sum_{i\in I_{1}(s)}
			\del[1]{\lambda(\mathring{\Psi}_{i,n}^\top \hat{\beta}_{a,s}^{R}) -\lambda(\mathring{\Psi}_{i,n}^\top \beta_{a,s}^{R})}
		}{n_{1}(s)} - \frac{\sum_{i \in I_{0}(s)}%
			\del[1]{\lambda(\mathring{\Psi}_{i,n}^\top \hat{\beta}_{a,s}^{R})-\lambda(\mathring{\Psi}_{i,n}^\top \beta_{a,s}^{R})}}%
		{n_{0}(s)}\biggr|\nonumber\\
		& \qquad+
		\envert[3]{\frac{1}{n_1(s)}\sum_{i \in I_1(s)}\del[2]{ R^D(a,s,X_i) - \mathbb{E}[R^D(a,s,X_i)|S_i=s]}}\nonumber\\
		& \qquad+
		\envert[3]{\frac{1}{n_0(s)}\sum_{i \in I_0(s)}\del[2]{ R^D(a,s,X_i) - \mathbb{E}[R^D(a,s,X_i)|S_i=s]}}
		:= I + II + III.\label{eq:Delta_hd}%
	\end{align}
	To bound term $I$ in \eqref{eq:Delta_np}, we define $M_{a,s}(\beta_{1}%
	,\beta_{2}) := \mathbb{E}%
	\sbr[1]{\lambda(\mathring{\Psi}_{i,n}^\top \beta_1) - \lambda(\mathring{\Psi}_{i,n}^\top \beta_2)|S_i=s}
	= \mathbb{E}[\lambda(\mathring{\Psi}_{i,n}^{s,\top} \beta_{1}) -
	\lambda(\mathring{\Psi}_{i,n}^{s,\top} \beta_{2})]$, where $\mathring{\Psi
	}_{i,n}^{s} = \mathring{\Psi}(X_{i}^{s})$. Then we have
	\begin{align*}
		I  &  \leq\biggl|\frac{\sum_{i\in I_{1}(s)}
			\sbr[1]{\lambda(\mathring{\Psi}_{i,n}^\top \hat{\beta}_{a,s}^{R}) -\lambda(\mathring{\Psi}_{i,n}^\top \beta_{a,s}^{R}) - M_{a,s}(\hat{\beta}_{a,s}^{R},\beta_{a,s}^{R})}
		}{n_{1}(s)}\biggr|\\
		&  \qquad+ \biggl|\frac{\sum_{i \in I_{0}(s)}%
			\sbr[1]{\lambda(\mathring{\Psi}_{i,n}^\top \hat{\beta}_{a,s}^{R})-\lambda(\mathring{\Psi}_{i,n}^\top \beta_{a,s}^{R})-M_{a,s}(\hat{\beta}_{a,s}^{R},\beta_{a,s}^{R}) }}%
		{n_{0}(s)}\biggr| =: I_{1}+I_{2}.
	\end{align*}
	Following the argument in the proof of Theorems \ref{thm:par} and
	\ref{thm:np}, in order to show $I_{1} = o_{p}(n^{-1/2})$, we only need to
	show
	\[
	n^{1/2}\mathbb{E}%
	\sbr[1]{||\mathbb{P}_{n_1(s)} - \mathbb{P}||_{\mathcal{F}}|\{A_i,S_i\}_{i \in [n]}}1\{n_{1}%
	(s) \geq n\varepsilon, n_{0}(s)\geq n\varepsilon\} = o(1),
	\]
	where $\varepsilon$ is an arbitrary but fixed constant, and
	\[
	\mathcal{F} :=
	\cbr[2]{\lambda(\mathring{\Psi}_{i,n}^\top \beta_1)-\lambda(\mathring{\Psi}_{i,n}^\top \beta_{a,s}^{R}): \beta_1 \in \Re^{h_n}, ||\beta_1 - \beta_{a,s}^{R}||_2 \leq C \sqrt{h_n \log(p_n)/n_a(s)}, ||\beta_1||_0 \leq C h_n},
	\]
	for some constant $C>0$. Furthermore, we note that $\mathcal{F}$ has a bounded
	envelope and
	\begin{align*}
		\sup_{Q} N(\mathcal{F}, e_{Q}, \varepsilon||F||_{Q,2}) \leq\left( \frac{c_{1}
			p_{n}}{\varepsilon}\right) ^{c_{2}h_{n}},
	\end{align*}
	where $c_{1},c_{2}$ are two fixed constants, $N(\cdot)$ is the covering
	number, $e_{Q}(f,g) = \sqrt{Q|f-g|^{2}}$, and the supremum is taken over all
	discrete probability measures $Q$. Last, we have
	\begin{align*}
		\sup_{f \in\mathcal{F}}\mathbb{E}\sbr[1]{f^2|\{A_i,S_i\}_{i \in [n]}}\leq
		\frac{Ch_{n} \log p_{n}}{n_{a}(s)}.
	\end{align*}

	Invoking \citet[Corollary 5.1]{CCK14} with $A = Cp_{n}$, $\nu=Ch_{n}$,
	$\sigma^{2} = Ch_{n} \log(p_{n})/n_{a}(s)$, and $F$ and $M$ being $2$, we
	have
	\begin{align*}
		&  n^{1/2}\mathbb{E}\left[ ||\mathbb{P}_{n_{1}(s)} - \mathbb{P}||_{\mathcal{F}%
		}|\{A_{i},S_{i}\}_{i \in[n]}\right] 1\{n_{1}(s) \geq n\varepsilon, n_{0}(s)
		\geq n\varepsilon\}\\
		& \leq C \sqrt{\frac{n}{n_{1}(s)}}\left(  \sqrt{h_{n}\frac{h_{n}\log p_{n}%
			}{n_{a}(s)}\log\del[3]{\frac{p_n}{\sqrt{\frac{h_n\log p_n}{n_a(s)}}}}}%
		+\frac{h_{n}}{\sqrt{n_{1}(s)}}\log
		\del[3]{\frac{p_n}{\sqrt{\frac{h_n\log p_n}{n_a(s)}}}}\right)  1\{n_{1}(s)
		\geq n\varepsilon, n_{0}(s) \geq n\varepsilon\}\\
		& \leq C\left( \sqrt{\frac{n}{n_{1}(s)}}\right)  \left( \frac{h_{n} \log
			(p_{n})}{\sqrt{n_{1}(s)\wedge n_{0}(s)}}\right) 1\{n_{1}(s) \geq
		n\varepsilon,n_{0}(s) \geq n\varepsilon\} \rightarrow0.
	\end{align*}
	The bounds for $I_{2}$, $II$ and $III$ can be established following the same
	argument as in the proof of Theorem \ref{thm:np}. We omit the detail for
	brevity. This leads to Assumption \ref{ass:Delta}(i).
	
	To verify Assumption \ref{ass:Delta}(ii), we note that
	\begin{align*}
		& \frac{1}{n}\sum_{i=1}^{n} \Delta^{D,2}(a,s,X_{i}) \leq\frac{2}{n}\sum
		_{i=1}^{n}
		\del[1]{\lambda(\mathring{\Psi}_{i,n}^\top\hat{\beta}_{a,s}^{R}) - \lambda(\mathring{\Psi}_{i,n}^\top \beta_{a,s}^{R})}^{2}
		+ \frac{2}{n}\sum_{i=1}^{n} R^{D,2}(a,S_{i},X_{i})\\
		& =\frac{2}{n}\sum_{i=1}^{n}
		\del[1]{\lambda(\mathring{\Psi}_{i,n}^\top\hat{\beta}_{a,s}^{R}) - \lambda(\mathring{\Psi}_{i,n}^\top \beta_{a,s}^{R})}^{2}
		+ o_{p}(1)= o_{p}(1),
	\end{align*}
	where the first equality is due to Assumption \ref{ass:hd}(iii) and the second
	equality is by Assumption \ref{ass:hd}(vi) and the fact that
	\begin{align*}
		\frac{2}{n}\sum_{i=1}^{n}
		\del[1]{\lambda(\mathring{\Psi}_{i,n}^\top\hat{\beta}_{a,s}^{R}) - \lambda(\mathring{\Psi}_{i,n}^\top \beta_{a,s}^{R})}^{2}
		\lesssim\frac{(\hat{\beta}_{a,s}^{R} - \beta_{a,s}^{R})^{\top}}{n} \sum
		_{i=1}^{n} \mathring{\Psi}_{i,n}\mathring{\Psi}_{i,n}^{\top}(\hat{\beta}%
		_{a,s}^{R} - \beta_{a,s}^{R}) \lesssim||\hat{\beta}_{a,s}^{R} - \beta
		_{a,s}^{R}||_{2}^{2} = o_{p}(1),
	\end{align*}
	where the first probability inequality is due to the fact that $\lambda
	(\cdot)$ is Lipschitz continuous with Lipschitz constant 1. Last, Assumption
	\ref{ass:Delta}(iii) is implied by Assumption \ref{ass:assignment1}(vi) via
	Jensen's inequality.
	
	\section{Proof of Theorem \ref{thm:sigma_S}}
	\label{sec:sigma_S}
	
	Some part of the proof of part (i) is due to \cite{anseletal2018} while some part of the
	proof is original. 
	Let $U_{i}:=(1, X_{i}^{\top})^{\top}$ and $\hat{\lambda
	}_{as}:=(\hat{\gamma}_{as}^{b}, \hat{\nu}_{as}^{b,\top})^{\top}$ for $a=0,1$
	and $b=Y,D$. Consider $\hat{\lambda}_{0s}^{D}$ as an example; note that
	\[
	\hat{\lambda}_{0s}^{D}=
	\del[3]{\frac{1}{n}\sum_{i=1}^{n}(1-A_i)1\{S_i=s\}U_iU_i^{\top}}^{-1}\frac
	{1}{n}\sum_{i=1}^{n}(1-A_{i})1\{S_{i}=s\}U_{i}D_{i}.
	\]
	Consider the denominator of $\hat{\lambda}_{0s}^{D}$:
	\begin{align}
		& \frac{1}{n}\sum_{i=1}^{n}(1-A_{i})1\{S_{i}=s\}U_{i}U_{i}^{\top}=\frac{1}%
		{n}\sum_{i=1}^{n}(\pi(s)-A_{i})1\{S_{i}=s\}U_{i}U_{i}^{\top}+\frac{1}{n}%
		\sum_{i=1}^{n}(1-\pi(s))1\{S_{i}=s\}U_{i}U_{i}^{\top}\nonumber\\
		& =\frac{1}{n}\sum_{i=1}^{n}(\pi(s)-A_{i})1\{S_{i}%
		=s\}\del[1]{U_iU_i^{\top}-\mathbb{E}[U_iU_i^{\top}|S_i]}+\frac{1}{n}\sum
		_{i=1}^{n}(\pi(s)-A_{i})1\{S_{i}=s\}\mathbb{E}[U_{i}U_{i}^{\top}%
		|S_{i}]\nonumber\\
		& \qquad+\frac{1}{n}\sum_{i=1}^{n}(1-\pi(s))1\{S_{i}=s\}U_{i}U_{i}^{\top}
		.\label{r3}%
	\end{align}
	Consider the first term of (\ref{r3}). Note that
	\begin{align*}
		\mathbb{E}%
		\sbr[3]{ \frac{1}{n}\sum_{i=1}^{n}(\pi(s)-A_i)1\{S_i=s\}\del[1]{U_iU_i^{\top}-\mathbb{E}[U_iU_i^{\top}|S_i]}\vert A^{(n)}, S^{(n)}}=0.
	\end{align*}
	Invoking the conditional Chebyshev's inequality, we have, for any $a>0$,
	$1\leq k,\ell\leq\text{dim}(U_{i})$, {\footnotesize
		\begin{align}
			& \mathbb{P}%
			\del[3]{\envert[3]{\frac{1}{n}\sum_{i=1}^{n}(\pi(s)-A_i)1\{S_i=s\}\del[1]{U_{ik}U_{i\ell}-\mathbb{E}[U_{ik}U_{i\ell}|S_i]}}\geq a\vert A^{(n)}, S^{(n)}}\nonumber\\
			&  \leq\frac{1}{a^{2}}%
			\var \del[3]{\frac{1}{n}\sum_{i=1}^{n}(\pi(s)-A_i)1\{S_i=s\}\del[1]{U_{ik}U_{i\ell}-\mathbb{E}[U_{ik}U_{i\ell}|S_i]}\vert A^{(n)}, S^{(n)}}\nonumber\\
			&  = \frac{\sum_{i,j\in[n]}(\pi(s)-A_{i})(\pi(s)-A_{j})1\{S_{i}=s\}1\{S_{j}%
				=s\}\mathbb{E}%
				\sbr[2]{\del[1]{U_{ik}U_{i\ell}-\mathbb{E}[U_{ik}U_{i\ell}|S_i]}\del[1]{U_{jk}U_{j\ell}-\mathbb{E}[U_{jk}U_{j\ell}|S_j]}\vert A^{(n)}, S^{(n)} }}%
			{a^{2}n^{2}}\nonumber\\
			&  = \frac{\sum_{i\in[n]}(\pi(s)-A_{i})^{2}1\{S_{i}=s\}\mathbb{E}%
				\sbr[2]{\del[1]{U_{ik}U_{i\ell}-\mathbb{E}[U_{ik}U_{i\ell}|S_i]}^2\vert A^{(n)}, S^{(n)} }}%
			{a^{2}n^{2}}\nonumber\\
			&  \leq\frac{\sum_{i\in[n]}(\pi(s)-A_{i})^{2}1\{S_{i}=s\}\mathbb{E}%
				\sbr[2]{U_{ik}^2U_{i\ell}^2\vert S_i }}{a^{2}n^{2}}\leq\frac{\sum_{i\in
					[n]}\mathbb{E}\sbr[2]{U_{ik}^2U_{i\ell}^2\vert S_i=s }}{a^{2}n^{2}%
			}=o(1)\label{r4}%
		\end{align}
	} where the second equality is due to
	\begin{align*}
		& \mathbb{E}\sbr[2]{\del[1]{U_{ik}U_{i\ell}-\mathbb{E}[U_{ik}U_{i\ell}|S_i]}\del[1]{U_{jk}U_{j\ell}-\mathbb{E}[U_{jk}U_{j\ell}|S_j]}\vert A^{(n)}, S^{(n)} }\\
		& =\mathbb{E}\sbr[2]{\del[1]{U_{ik}U_{i\ell}-\mathbb{E}[U_{ik}U_{i\ell}|S_i]}\del[1]{U_{jk}U_{j\ell}-\mathbb{E}[U_{jk}U_{j\ell}|S_j]}\vert S^{(n)} }\\
		& =\mathbb{E}\sbr[2]{U_{ik}U_{i\ell}-\mathbb{E}[U_{ik}U_{i\ell}|S_i]\vert S^{(n)} }\mathbb{E}%
		\sbr[2]{U_{jk}U_{j\ell}-\mathbb{E}[U_{jk}U_{j\ell}|S_j]\vert S^{(n)} }\\
		& =\mathbb{E}%
		\sbr[2]{U_{ik}U_{i\ell}-\mathbb{E}[U_{ik}U_{i\ell}|S_i]\vert S_i }\mathbb{E}%
		\sbr[2]{U_{jk}U_{j\ell}-\mathbb{E}[U_{jk}U_{j\ell}|S_j]\vert S_j }=0
	\end{align*}
	for $i\neq j$, where the second equality is due to that $U_{ik}U_{i\ell}-\mathbb{E}[U_{ik}U_{i\ell}|S_i]$ and $U_{jk}U_{j\ell}-\mathbb{E}[U_{jk}U_{j\ell}|S_j]$ are independent conditional on $S^{(n)}$. From (\ref{r4}), we deduce that the first term of (\ref{r3}) is
	$o_{p}(1)$. Consider the second term of (\ref{r3}).
	\begin{align*}
		& \frac{1}{n}\sum_{i=1}^{n}(\pi(s)-A_{i})1\{S_{i}=s\}\mathbb{E}[U_{i}%
		U_{i}^{\top}|S_{i}]=\mathbb{E}[UU^{\top}|S=s]\frac{1}{n}\sum_{i=1}^{n}%
		(\pi(s)-A_{i})1\{S_{i}=s\}\\
		& =\mathbb{E}[UU^{\top}|S=s]\frac{1}{n}B_{n}(s)=o_{p}(1).
	\end{align*}
	Consider the third term of (\ref{r3}).
	\begin{align*}
		& \frac{1}{n}\sum_{i=1}^{n}(1-\pi(s))1\{S_{i}=s\}U_{i}U_{i}^{\top}%
		=(1-\pi(s))\frac{n(s)}{n}\frac{1}{n(s)}\sum_{i=1}^{n}1\{S_{i}=s\}U_{i}%
		U_{i}^{\top}\\
		& \qquad\xrightarrow{p}(1-\pi(s))p(s)\mathbb{E}[UU^{\top}|S=s].
	\end{align*}

	We hence have
	\begin{align*}
		\frac{1}{n}\sum_{i=1}^{n}(1-A_{i})1\{S_{i}=s\}U_{i}U_{i}^{\top}%
		\xrightarrow{p}(1-\pi(s))p(s)\mathbb{E}[UU^{\top}|S=s].
	\end{align*}
	Similarly, we have
	\begin{align*}
		& \frac{1}{n}\sum_{i=1}^{n}(1-A_{i})1\{S_{i}=s\}U_{i}D_{i}%
		\xrightarrow{p}(1-\pi(s))\hat{p}(s)\mathbb{E}[UD(0)|S=s]\\
		&  \hat{\lambda}_{0s}^{D}
		\xrightarrow{p}\del[2]{\mathbb{E}[UU^{\top}|S=s]}^{-1}\mathbb{E}[UD(0)|S=s]\\
		&  \hat{\lambda}_{1s}^{D}
		\xrightarrow{p}\del[2]{\mathbb{E}[UU^{\top}|S=s]}^{-1}\mathbb{E}[UD(1)|S=s].
	\end{align*}
	Thus, we have
	\begin{align*}
		& \sum_{s \in\mathcal{S}}\hat{p}(s)(\hat{\gamma}_{1s}^{D} - \hat{\gamma}%
		_{0s}^{D} + (\hat{\nu}_{1s}^{D} - \hat{\nu}_{0s}^{D})^{\top}\bar{X}_{s}%
		)=\sum_{s \in\mathcal{S}}(\hat{\lambda}_{1s}^{D}-\hat{\lambda}_{0s}^{D}%
		)^{\top}\left(
		\begin{array}
			[c]{c}%
			\frac{1}{n}\sum_{i\in[n]}1\{S_{i}=s\}\\
			\frac{1}{n}\sum_{i\in[n]}X_{i}1\{S_{i}=s\}
		\end{array}
		\right) \\
		& \qquad=\sum_{s \in\mathcal{S}} \frac{n(s)}{n}\frac{1}{n(s)}\sum_{i\in
			[n]}1\{S_{i}=s\}U_{i}^{\top}(\hat{\lambda}_{1s}^{D}-\hat{\lambda}_{0s}^{D})\\
		& \qquad\xrightarrow{p}\sum_{s \in\mathcal{S}} p(s)\mathbb{E}[U^{\top}|S=s]
		\del[2]{\mathbb{E}[UU^{\top}|S=s]}^{-1}\mathbb{E}%
		\sbr[2]{U\del[1]{D(1)-D(0)}|S=s}\\
		& \qquad=\sum_{s \in\mathcal{S}} p(s)\mathbb{E}%
		\sbr[1]{D(1)-D(0)|S=s}=\mathbb{E}\sbr[1]{D(1)-D(0)}
	\end{align*}
	where the second last equality is due to $\mathbb{E}[U^{\top}|S=s]
	\del[1]{\mathbb{E}[UU^{\top}|S=s]}^{-1}=(1,0,\ldots, 0)$ (\cite{anseletal2018}
	p290). Thus, the denominator of $\sqrt{n}(\hat{\tau}_{S}-\tau)$ converges in
	probability to $\mathbb{E}[D(1)-D(0)]$.
	
	\bigskip
	
	We now consider the numerator of $\sqrt{n}(\hat{\tau}_{S}-\tau)$. Relying on a
	similar argument, we have
	\begin{align*}
		&  \hat{\lambda}_{1s}^{Y}%
		=\del[3]{\frac{1}{n}\sum_{i=1}^{n}A_i1\{S_i=s\}U_iU_i^{\top}}^{-1}\frac{1}%
		{n}\sum_{i=1}^{n}A_{i}1\{S_{i}=s\}U_{i}Y_{i}(D_{i}(1))\\
		& \qquad\xrightarrow{p}\del[2]{\mathbb{E}[UU^{\top}|S=s]}^{-1}\mathbb{E}%
		[UY(D(1))|S=s]\\
		&  \hat{\lambda}_{0s}^{Y}%
		=\del[3]{\frac{1}{n}\sum_{i=1}^{n}(1-A_i)1\{S_i=s\}U_iU_i^{\top}}^{-1}\frac
		{1}{n}\sum_{i=1}^{n}(1-A_{i})1\{S_{i}=s\}U_{i}Y_{i}(D_{i}(0))\\
		& \qquad\xrightarrow{p}\del[2]{\mathbb{E}[UU^{\top}|S=s]}^{-1}\mathbb{E}%
		[UY(D(0))|S=s]
	\end{align*}
	\begin{align*}
		& \hat{\eta}_{1s}:=\hat{\lambda}_{1s}^{Y}-\tau\hat{\lambda}_{1s}%
		^{D}\xrightarrow{p}\del[2]{\mathbb{E}[UU^{\top}|S=s]}^{-1}\mathbb{E}%
		\sbr[2]{U\sbr[1]{Y(D(1))-\tau D(1)}|S=s}=:\eta_{1s}\\
		& \hat{\eta}_{0s}:=\hat{\lambda}_{0s}^{Y}-\tau\hat{\lambda}_{0s}%
		^{D}\xrightarrow{p}\del[2]{\mathbb{E}[UU^{\top}|S=s]}^{-1}\mathbb{E}%
		\sbr[2]{U\sbr[1]{Y(D(0))-\tau D(0)}|S=s}=:\eta_{0s}.\notag
	\end{align*}
	The numerator of $\sqrt{n}(\hat{\tau}_{S}-\tau)$ could be written as
	\begin{align}
		& \sqrt{n}\sum_{s \in\mathcal{S}}\hat{p}(s)(\hat{\gamma}_{1s}^{Y} -
		\hat{\gamma}_{0s}^{Y} + (\hat{\nu}_{1s}^{Y} - \hat{\nu}_{0s}^{Y})^{\top}%
		\bar{X}_{s}) -\sqrt{n}\sum_{s \in\mathcal{S}}\hat{p}(s)(\hat{\gamma}_{1s}^{D}
		- \hat{\gamma}_{0s}^{D} + (\hat{\nu}_{1s}^{D} - \hat{\nu}_{0s}^{D})^{\top}%
		\bar{X}_{s}) \tau\nonumber\\
		& =\sqrt{n} \sum_{s \in\mathcal{S}}\hat{p}(s)\frac{1}{n(s)}\sum_{i\in[n]}
		1\{S_{i}=s\}U_{i}^{\top}%
		\sbr[1]{\hat{\lambda}_{1s}^Y-\tau \hat{\lambda}_{1s}^D-(\hat{\lambda}_{0s}^Y-\tau \hat{\lambda}_{0s}^D)}\nonumber\\
		& =\sqrt{n} \sum_{s \in\mathcal{S}}\hat{p}(s)\bar{U}_{s}^{\top}(\hat{\eta
		}_{1s}-\eta_{1s})-\sqrt{n} \sum_{s \in\mathcal{S}}\hat{p}(s)\bar{U}_{s}^{\top
		}(\hat{\eta}_{0s}-\eta_{0s})+\sqrt{n} \sum_{s \in\mathcal{S}}\hat{p}(s)\bar
		{U}_{s}^{\top}(\eta_{1s}-\eta_{0s})\label{r6}%
	\end{align}
	where $\bar{U}_{s}:=\frac{1}{n(s)}\sum_{i\in[n]} 1\{S_{i}=s\}U_{i}%
	\xrightarrow{p}\mathbb{E}[U|S=s]$. Consider the first term of (\ref{r6}).
	{\footnotesize
		\begin{align}
			& \sqrt{n} \sum_{s \in\mathcal{S}}\hat{p}(s)\bar{U}_{s}^{\top}(\hat{\eta}%
			_{1s}-\eta_{1s})\nonumber\\
			& = \sum_{s \in\mathcal{S}}\hat{p}(s)\bar{U}_{s}^{\top}%
			\del[3]{\frac{1}{n}\sum_{i=1}^{n}A_i1\{S_i=s\}U_iU_i^{\top}}^{-1}\frac
			{1}{\sqrt{n}}\sum_{i=1}^{n}A_{i}1\{S_{i}=s\}U_{i}%
			\sbr[1]{Y_i(D_i(1))-\tau D_i(1)-U_i^{\top}\eta_{1s}}\nonumber\\
			& = \sum_{s \in\mathcal{S}}\hat{p}(s)\mathbb{E}[U^{\top}%
			|S=s]\del[3]{\pi(s)\hat{p}(s)\mathbb{E}[UU^{\top}|S=s]}^{-1}\frac{1}{\sqrt{n}%
			}\sum_{i=1}^{n}A_{i}1\{S_{i}=s\}U_{i}%
			\sbr[1]{Y_i(D_i(1))-\tau D_i(1)-U_i^{\top}\eta_{1s}}+o_{p}(1)\nonumber\\
			& = \sum_{s \in\mathcal{S}}\frac{1}{\pi(s)}\mathbb{E}[U^{\top}%
			|S=s]\del[3]{\mathbb{E}[UU^{\top}|S=s]}^{-1}\frac{1}{\sqrt{n}}\sum_{i=1}%
			^{n}A_{i}1\{S_{i}=s\}U_{i}%
			\sbr[1]{Y_i(D_i(1))-\tau D_i(1)-U_i^{\top}\eta_{1s}}+o_{p}(1)\nonumber\\
			& = \sum_{s \in\mathcal{S}}\frac{1}{\pi(s)}\frac{1}{\sqrt{n}}\sum_{i=1}%
			^{n}A_{i}1\{S_{i}%
			=s\}\sbr[1]{Y_i(D_i(1))-\tau D_i(1)-U_i^{\top}\eta_{1s}}+o_{p}(1)\label{r7}%
		\end{align}
	} where the second equality is based on that
	\[
	n^{-1/2}\sum_{i=1}^{n}A_{i}1\{S_{i}=s\}U_{i}%
	\sbr[1]{Y_i(D_i(1))-\tau D_i(1)-U_i^{\top}\eta_{1s}}=O_{p}(1),
	\]
	which is implied by the asymptotic normality of (\ref{r10a}), which we will prove shortly, and the last equality is due to $\mathbb{E}%
	[U^{\top}|S=s] \del[1]{\mathbb{E}[UU^{\top}|S=s]}^{-1}=(1,0,\ldots, 0)$
	(\cite{anseletal2018} p290). Likewise, the second term of (\ref{r6})
	\begin{align}
		& \sqrt{n} \sum_{s \in\mathcal{S}}\hat{p}(s)\bar{U}_{s}^{\top}(\hat{\eta}%
		_{0s}-\eta_{0s})\nonumber\\
		& =\sum_{s \in\mathcal{S}}\frac{1}{1-\pi(s)}\frac{1}{\sqrt{n}}\sum_{i=1}%
		^{n}(1-A_{i})1\{S_{i}%
		=s\}\sbr[1]{Y_i(D_i(0))-\tau D_i(0)-U_i^{\top}\eta_{0s}}+o_{p}(1).\label{r8}%
	\end{align}

	Note that
	\begin{align*}
		\eta_{as}=\left(
		\begin{array}
			[c]{c}%
			\mathbb{E}\sbr[1]{Y(D(a))-\tau D(a)|S=s}-\mathbb{E}[X^{\top}\nu_{as}%
			^{YD}|S=s]\\
			\nu_{as}^{YD}%
		\end{array}
		\right)
	\end{align*}
	for $a=0,1$ via the Frisch-Waugh Theorem. Hence
	\begin{align}
		U_{i}^{\top}\eta_{as}=\mathbb{E}\sbr[1]{Y(D(a))-\tau D(a)|S=s}+X_{i}^{\top}%
		\nu_{as}^{YD}-\mathbb{E}[X^{\top}\nu_{as}^{YD}|S=s].\label{r9}%
	\end{align}
	Substituting (\ref{r7}), (\ref{r8}) and (\ref{r9}) into (\ref{r6}), we could
	write the numerator of $\sqrt{n}(\hat{\tau}_{S}-\tau)$ as {\footnotesize
		\begin{align}
			& \sum_{s \in\mathcal{S}}\frac{1}{\pi(s)}\frac{1}{\sqrt{n}}\sum_{i=1}^{n}%
			A_{i}1\{S_{i}%
			=s\}\sbr[1]{Y_i(D_i(1))-\tau D_i(1)-U_i^{\top}\eta_{1s}}\nonumber\\
			& \qquad- \sum_{s \in\mathcal{S}}\frac{1}{1-\pi(s)}\frac{1}{\sqrt{n}}%
			\sum_{i=1}^{n}(1-A_{i})1\{S_{i}%
			=s\}\sbr[1]{Y_i(D_i(0))-\tau D_i(0)-U_i^{\top}\eta_{0s}}\label{r10a}\\
			& \qquad+ \sum_{s \in\mathcal{S}}\frac{1}{\sqrt{n}}\sum_{i\in[n]}
			1\{S_{i}=s\}U_{i}^{\top}(\eta_{1s}-\eta_{0s})+o_{p}(1)\nonumber\\
			& =\sum_{s \in\mathcal{S}}\frac{1}{\pi(s)}\frac{1}{\sqrt{n}}\sum_{i=1}%
			^{n}A_{i}1\{S_{i}%
			=s\}\sbr[3]{Y_i(D_i(1))-\tau D_i(1)- \mathbb{E}\sbr[1]{Y(D(1))-\tau D(1)|S=s}-\del[1]{X_i^{\top}\nu_{1s}^{YD}-\mathbb{E}[X^{\top}\nu_{1s}^{YD}|S=s]} }\nonumber\\
			& \quad- \sum_{s \in\mathcal{S}}\frac{1}{1-\pi(s)}\frac{1}{\sqrt{n}}\sum
			_{i=1}^{n}(1-A_{i})1\{S_{i}%
			=s\}\sbr[3]{Y_i(D_i(0))-\tau D_i(0)- \mathbb{E}\sbr[1]{Y(D(0))-\tau D(0)|S=s}-\del[1]{X_i^{\top}\nu_{0s}^{YD}-\mathbb{E}[X^{\top}\nu_{0s}^{YD}|S=s]}}\nonumber\\
			& \quad+ \sum_{s \in\mathcal{S}}\frac{1}{\sqrt{n}}\sum_{i\in[n]}
			1\{S_{i}=s\}\mathbb{E}\sbr[2]{Y(D(1))-Y(D(0))-\tau( D(1)-D(0))|S=s}\nonumber\\
			& \quad+\sum_{s \in\mathcal{S}}\frac{1}{\sqrt{n}}\sum_{i\in[n]} 1\{S_{i}%
			=s\}\del[2]{X_i^{\top}(\nu_{1s}^{YD}-\nu_{0s}^{YD})-\mathbb{E}[X^{\top}(\nu_{1s}^{YD}-\nu_{0s}^{YD})|S=s]}+o_{p}%
			(1)\nonumber\\
			& =\sum_{s \in\mathcal{S}}\frac{1}{\sqrt{n}}\sum_{i=1}^{n}A_{i}1\{S_{i}%
			=s\}\sbr[3]{\frac{Y_i(D_i(1))-\tau D_i(1)-X_i^{\top}\nu_{1s}^{YD}- \mathbb{E}\sbr[1]{Y(D(1))-\tau D(1)-X^{\top}\nu_{1s}^{YD}|S=s}}{\pi(s)} }\nonumber\\
			& \quad+\sum_{s \in\mathcal{S}}\frac{1}{\sqrt{n}}\sum_{i\in[n]} A_{i}%
			1\{S_{i}%
			=s\}\del[2]{X_i^{\top}(\nu_{1s}^{YD}-\nu_{0s}^{YD})-\mathbb{E}[X^{\top}(\nu_{1s}^{YD}-\nu_{0s}^{YD})|S=s]}\nonumber\\
			& \quad- \sum_{s \in\mathcal{S}}\frac{1}{\sqrt{n}}\sum_{i=1}^{n}%
			(1-A_{i})1\{S_{i}%
			=s\}\sbr[3]{\frac{Y_i(D_i(0))-\tau D_i(0)-X_i^{\top}\nu_{0s}^{YD}- \mathbb{E}\sbr[1]{Y(D(0))-\tau D(0)-X^{\top}\nu_{0s}^{YD}|S=s}}{1-\pi(s)}}\nonumber\\
			& \quad+\sum_{s \in\mathcal{S}}\frac{1}{\sqrt{n}}\sum_{i\in[n]} (1-A_{i}%
			)1\{S_{i}%
			=s\}\del[2]{X_i^{\top}(\nu_{1s}^{YD}-\nu_{0s}^{YD})-\mathbb{E}[X^{\top}(\nu_{1s}^{YD}-\nu_{0s}^{YD})|S=s]}\nonumber\\
			& \quad+ \frac{1}{\sqrt{n}}\sum_{i\in[n]} \mathbb{E}%
			\sbr[2]{Y(D(1))-Y(D(0))-\tau( D(1)-D(0))|S}+o_{p}(1).\label{r10}%
		\end{align}
	}
	
	Define
	\begin{align*}  
		\rho_{is}(1) & := \frac{Y_{i}(D_{i}(1))-\tau D_{i}(1)-X_{i}^{\top}\nu_{1s}^{YD}}{\pi(s)}+X_{i}^{\top}(\nu_{1s}^{YD}-\nu_{0s}^{YD})\\
		\rho_{is}(0) & :=\frac{Y_{i}(D_{i}(0))-\tau D_{i}(0)-X_{i}^{\top}\nu_{0s}^{YD}%
		}{1-\pi(s)}-X_{i}^{\top}(\nu_{1s}^{YD}-\nu_{0s}^{YD}).
	\end{align*}
	Then the first four terms of (\ref{r10}) could be written compactly as
	\begin{align*}
		R_{n,1} & :=\sum_{s \in\mathcal{S}}\frac{1}{\sqrt{n}}\sum_{i=1}^{n}%
		A_{i}1\{S_{i}=s\}\sbr[1]{\rho_{is}(1)-\mathbb{E}[\rho_{is}(1)|S_i=s]}\\
		& \qquad-\sum_{s \in\mathcal{S}}\frac{1}{\sqrt{n}}\sum_{i=1}^{n}%
		(1-A_{i})1\{S_{i}=s\}\sbr[1]{\rho_{is}(0)-\mathbb{E}[\rho_{is}(0)|S_i=s]}.
	\end{align*}
	Define $R_{n,2}:=\frac{1}{\sqrt{n}}\sum_{i\in[n]} \mathbb{E}%
	\sbr[1]{Y(D(1))-Y(D(0))-\tau( D(1)-D(0))|S}$. To establish the asymptotic
	distribution of (\ref{r10}), we first argue that
	\begin{align*}
		(R_{n,1},R_{n,2})\overset{d}{=}(R_{n,1}^{*},R_{n,2})+o_{p}(1)
	\end{align*}
	for a random variable $R_{n,1}^{*}$ that satisfies $R_{n,1}^{*}\perp
	\!\!\!\perp R_{n,2}$. Conditional on $\{S^{(n)}, A^{(n)}\}$, the distribution
	of $R_{n,1}$ is the same as the distribution of the same quantity where units
	are ordered by strata and then ordered by $A_{i}=1$ first and $A_{i}=0$ second
	within strata. To this end, define $N(s):=\sum_{i=1}^{n}1\{S_{i}<s\}$ and
	$F(s):=\mathbb{P}(S_{i}<s)$. Furthermore, independently for each
	$s\in\mathcal{S}$ and independently of $\{S^{(n)}, A^{(n)}\}$, let
	$\cbr[1]{Y_i(1)^s,Y_i(0)^s,D_i(1)^s,D_i(0)^s, X_i^s: 1\leq i\leq n}$ be i.i.d.
	over $i$ with distribution equal to that of $(Y(1),Y(0),D(1),D(0), X)|S=s$.
	Define
	\begin{align*}
		\tilde{\rho}_{is}(a):=\rho_{is}(a)-\mathbb{E}[\rho_{is}(a)|S_{i}=s], \qquad
		\tilde{\rho}_{is}^{s}(a):=\rho_{is}^{s}(a)-\mathbb{E}[\rho_{is}^{s}(a)|S_{i}=s],
	\end{align*}
	where
	\begin{align*}
		\rho_{is}^{s}(1) & := \frac{Y_{i}^{s}(D_{i}^{s}(1))-\tau D_{i}^{s}(1)-X_{i}^{s,
				\top}\nu_{1s}^{YD}}{\pi(s)}+X_{i}^{s,\top}(\nu_{1s}^{YD}-\nu_{0s}^{YD})\\
		\rho_{is}^{s}(0) & :=\frac{Y_{i}^{s}(D_{i}^{s}(0))-\tau D_{i}^{s}(0)-X_{i}^{s,
				\top}\nu_{0s}^{YD}}{1-\pi(s)}-X_{i}^{s,\top}(\nu_{1s}^{YD}-\nu_{0s}^{YD}).
	\end{align*}
	Then we have
	\begin{align*}
		R_{n,1} & :=\sum_{s \in\mathcal{S}}\frac{1}{\sqrt{n}}\sum_{i=1}^{n}%
		1\{S_{i}=s\}\sbr[2]{A_i\tilde{\rho}_{is}(1)-(1-A_i)\tilde{\rho}_{is}(0)}.
	\end{align*}
	Define
	\begin{align*}
		\tilde{R}_{n,1} & :=\sum_{s \in\mathcal{S}}%
		\sbr[3]{\frac{1}{\sqrt{n}}\sum_{i=n\frac{N(s)}{n}+1}^{n\del[1]{\frac{N(s)}{n}+\frac{n_1(s)}{n}}}\tilde{\rho}_{is}^s(1)-\frac{1}{\sqrt{n}}\sum_{n\del[1]{\frac{N(s)}{n}+\frac{n_1(s)}{n}}+1}^{n\del[1]{\frac{N(s)}{n}+\frac{n(s)}{n}}}\tilde{\rho}_{is}^s(0)}\\
		R_{n,1}^{*} & :=\sum_{s\in\mathcal{S}}%
		\sbr[3]{\frac{1}{\sqrt{n}}\sum_{i=\lfloor nF(s)\rfloor+1}^{\lfloor n(F(s)+\pi(s) p(s))\rfloor}\tilde{\rho}_{is}^s(1)-\frac{1}{\sqrt{n}}\sum_{i=\lfloor n(F(s)+\pi(s) p(s))\rfloor+1}^{\lfloor n(F(s)+p(s))\rfloor}\tilde{\rho}_{is}^s(0)}.
	\end{align*}
	Thus $R_{n,1}|S^{(n)}, A^{(n)}\overset{d}{=}\tilde{R}_{n,1}|S^{(n)}, A^{(n)}$
	(and as a by-product $R_{n,1}\overset{d}{=}\tilde{R}_{n,1}$). Since $R_{n,2}$
	is a function of $\{S^{(n)}, A^{(n)}\}$, we have, arguing along the line of a
	joint distribution being the product of a conditional distribution and a
	marginal distribution, $(R_{n,1}, R_{n,2})\overset{d}{=}(\tilde{R}_{n,1},
	R_{n,2})$. Define the following partial sum process
	\begin{align*}
		g_{n}(u):=\frac{1}{\sqrt{n}}\sum_{i=1}^{\lfloor nu\rfloor}\tilde{\rho}^{s}_{is}(1).
	\end{align*}
	Under our assumptions, $g_{n}(u)$ converges weakly to a suitably scaled
	Brownian motion. Next, by elementary properties of Brownian motion, we have
	that
	\begin{align}
		g_{n}\del[1]{F(s)+\pi(s) p(s)}-g_{n}\del[1]{F(s)} & =\frac{1}{\sqrt{n}}%
		\sum_{i=\lfloor nF(s)\rfloor+1}^{\lfloor n(F(s)+\pi(s) p(s))\rfloor}%
		\tilde{\rho}^{s}_{is}(1)\rightsquigarrow\mathcal{N}%
		\del[3]{0, \pi(s)p(s)\var\del[1]{\rho_s(1)|S=s}}.\label{r11}%
	\end{align}
	Furthermore, since
	\begin{align*}
		\del[3]{\frac{N(s)}{n}, \frac{n_1(s)}{n}}\xrightarrow{p}
		\del[1]{F(s), \pi(s) p(s)},
	\end{align*}
	it follows that
	\begin{align}
		g_{n}\del[3]{\frac{N(s)+n_1(s)}{n}}-g_{n}%
		\del[3]{\frac{N(s)}{n}}-\sbr[3]{g_n\del[1]{F(s)+\pi(s) p(s)}-g_n\del[1]{F(s)}}\xrightarrow{p}0\label{r14}%
	\end{align}
	where the convergence follows from the stochastic equicontinuity of the
	partial sum process. Using (\ref{r11}) and (\ref{r14}), we have:
	\begin{align}
		(R_{n,1}, R_{n,2}) & \overset{d}{=}(\tilde{R}_{n,1}, R_{n,2})=(R^{*}_{n,1},
		R_{n,2})+o_{p}(1)\label{r16}\\
		R^{*}_{n,1} & \rightsquigarrow\mathcal{N}%
		\del[3]{0, \sum_{s\in \mathcal{S}}\sbr[3]{\pi(s)p(s)\var\del[1]{\rho_s(1)|S=s}+\sbr[1]{1-\pi(s)}p(s)\var\del[1]{\rho_s(0)|S=s}}}\nonumber\\
		& = \mathcal{N}%
		\del[3]{0, \mathbb{E}\sbr[2]{\pi(S)\del[1]{\rho_S(1)-\mathbb{E}[\rho_S(1)|S]}^2+(1-\pi(S))\del[1]{\rho_S(0)-\mathbb{E}[\rho_S(0)|S]}^2}}\nonumber\\
		& =:\zeta_{1}\nonumber
	\end{align}
	where the convergence in distribution is due to an analogous argument for
	$\tilde{\rho}^{s}_{is}(0)$ and the independence of
	$\cbr[1]{Y_i(1)^s,Y_i(0)^s,D_i(1)^s,D_i(0)^s, X_i^s: 1\leq i\leq n, s\in \mathcal{S}}$
	across both $i$ and $s$. Moreover, since $R_{n,1}^{*}$ is a function of
	$\cbr[1]{Y_i(1)^s,Y_i(0)^s,D_i(1)^s,D_i(0)^s, X_i^s: 1\leq i\leq n, s\in \mathcal{S}}$
	$\perp\!\!\!\perp S^{(n)}, A^{(n)}$, and $R_{n,2}$ is a function of
	$\{S^{(n)}, A^{(n)}\}$, we see that $R_{n,1}^{*}\perp\!\!\!\perp R_{n,2}$.
	Thus (\ref{r16}) implies
	\begin{align}
		(R_{n,1}, R_{n,2})\overset{d}{=}(R^{*}_{n,1}, R_{n,2})+o_{p}%
		(1)\rightsquigarrow\del[1]{\zeta_{1},\zeta_{2}}\nonumber
	\end{align}
	where $\zeta_{1}$ and $\zeta_{2}$ are independent, with
	\begin{align*}
		\zeta_{2}:=\mathcal{N}%
		\del[3]{0, \mathbb{E}\sbr[3]{\del[2]{\mathbb{E}\sbr[1]{Y(D(1)) - Y(D(0)) - \tau(D(1)-D(0))|S}}^2}}.
	\end{align*}
	We hence show that the asymptotic distribution of the numerator of $\sqrt{n}(\hat{\tau}_{S}-\tau)$ is $\zeta_{1}+\zeta_{2}$. This completes the proof of part (i). The proof of part (ii), available upon request, is omitted in the interest of space as it is quite similar to that of part (ii) of Theorem \ref{thm:est}.

	\section{Technical Lemmas Used in the Proof of Theorem \ref{thm:est}}
	\label{sec:lem}
	\begin{lem}
		Suppose assumptions in Theorem \ref{thm:est} hold. Then, we have
		\begin{align*}
			R_{n,1} & =\frac{1}{\sqrt{n}}\sum_{i=1}^{n}\del[3]{1-\frac{1}{\pi(S_i)}}A_{i}%
			\tilde{\mu}^{Y}(1, S_{i}, X_{i})+\frac{1}{\sqrt{n}}\sum_{i=1}^{n}%
			(1-A_{i})\tilde{\mu}^{Y}(1, S_{i}, X_{i})+o_{p}(1),\\
			R_{n,2} & =\frac{1}{\sqrt{n}}\sum_{i=1}^{n}%
			\del[3]{\frac{1}{1-\pi(S_i)}-1}(1-A_{i})\tilde{\mu}^{Y}(0, S_{i}, X_{i}%
			)-\frac{1}{\sqrt{n}}\sum_{i=1}^{n}A_{i}\tilde{\mu}^{Y}(0, S_{i}, X_{i}%
			)+o_{p}(1),\\
			R_{n,3}  &  = \frac{1}{\sqrt{n}}\sum_{i=1}^{n}\frac{1}{\pi(S_{i})}\tilde
			{W}_{i}A_{i}-\frac{1}{\sqrt{n}}\sum_{i=1}^{n}\frac{1-A_{i}}{1-\pi(S_{i}%
				)}\tilde{Z}_{i} + \frac{1}{\sqrt{n}}\sum_{i=1}^{n}%
			\del[1]{\mathbb{E}[W_i-Z_i|S_i]-\mathbb{E}[W_i-Z_i]}+o_{p}(1),
		\end{align*}
		where for $a=0,1$,
		\begin{align*}
			&  \tilde{\mu}^{Y}(a, S_{i}, X_{i}):=\overline{\mu}^{Y}(a, S_{i},
			X_{i})-\overline{\mu}^{Y}(a,S_{i}),\quad\overline{\mu}^{Y}(a, S_{i})
			:=\mathbb{E}\sbr[1]{\overline{\mu}^Y(a, S_i, X_i)|S_i},\\
			&  W_{i} :=Y_{i}(1)D_{i}(1)+Y_{i}(0)(1-D_{i}(1)), \quad Z_{i}:=Y_{i}%
			(1)D_{i}(0)+Y_{i}(0)(1-D_{i}(0)),\\
			&  \tilde{W}_{i}:=W_{i}-\mathbb{E}[W_{i}|S_{i}], \quad\text{and} \quad
			\tilde{Z}_{i}:=Z_{i}-\mathbb{E}[Z_{i}|S_{i}].
		\end{align*}
		\label{lem:G}
	\end{lem}
	
	\begin{proof}
		We have
		\begin{align}
			R_{n,1}&=\frac{1}{\sqrt{n}}\sum_{i=1}^{n}\sbr[3]{\hat{\mu}^Y(1, S_i, X_i)-\frac{A_i\hat{\mu}^Y(1, S_i, X_i)}{\hat{\pi}(S_i)}}\notag \\
			&=-\frac{1}{\sqrt{n}}\sum_{i=1}^{n}\frac{A_i-\hat{\pi}(S_i)}{\hat{\pi}(S_i)}\hat{\mu}^Y(1, S_i, X_i)\notag\\
			&=-\frac{1}{\sqrt{n}}\sum_{i=1}^{n}\frac{A_i-\hat{\pi}(S_i)}{\hat{\pi}(S_i)}\sbr[1]{\hat{\mu}^Y(1, S_i, X_i)-\overline{\mu}^Y(1, S_i, X_i)+\overline{\mu}^Y(1, S_i, X_i)}\notag\\
			&=-\frac{1}{\sqrt{n}}\sum_{i=1}^{n}\frac{A_i-\hat{\pi}(S_i)}{\hat{\pi}(S_i)}\Delta^Y(1, S_i, X_i)-\frac{1}{\sqrt{n}}\sum_{i=1}^{n}\frac{A_i}{\hat{\pi}(S_i)}\overline{\mu}^Y(1, S_i, X_i)+\frac{1}{\sqrt{n}}\sum_{i=1}^{n}\overline{\mu}^Y(1, S_i, X_i)\notag\\
			&=-\frac{1}{\sqrt{n}}\sum_{i=1}^{n}\frac{A_i-\hat{\pi}(S_i)}{\hat{\pi}(S_i)}\Delta^Y(1, S_i, X_i)-\frac{1}{\sqrt{n}}\sum_{i=1}^{n}\frac{A_i}{\hat{\pi}(S_i)}\tilde{\mu}^Y(1, S_i, X_i)+\frac{1}{\sqrt{n}}\sum_{i=1}^{n}\tilde{\mu}^Y(1, S_i, X_i),\label{r2}
		\end{align}
		where the last equality is due to
		\begin{align*}
			\frac{1}{\sqrt{n}}\sum_{i=1}^{n}\frac{A_i}{\hat{\pi}(S_i)}\overline{\mu}^Y(1,S_i)&=\frac{1}{\sqrt{n}}\sum_{i=1}^{n}\overline{\mu}^Y(1,S_i).
		\end{align*}
		Consider the first term of (\ref{r2}).
		\begin{align}
			&\envert[3]{\frac{1}{\sqrt{n}}\sum_{i=1}^{n}\frac{A_i-\hat{\pi}(S_i)}{\hat{\pi}(S_i)}\Delta^Y(1, S_i, X_i)}=\envert[3]{\frac{1}{\sqrt{n}}\sum_{s\in \mathcal{S}}\sum_{i=1}^{n}\frac{A_i-\hat{\pi}(s)}{\hat{\pi}(s)}\Delta^Y(1, s, X_i)1\{S_i=s\}}\notag\\
			&=\frac{1}{\sqrt{n}}\envert[3]{\sum_{s\in \mathcal{S}}\frac{1}{\hat{\pi}(s)}\sum_{i=1}^{n}A_i\Delta^Y(1, s, X_i)1\{S_i=s\}-\sum_{s\in \mathcal{S}}\sum_{i=1}^{n}\Delta^Y(1, s, X_i)1\{S_i=s\}}\notag\\
			&=\frac{1}{\sqrt{n}}\envert[3]{\sum_{s\in \mathcal{S}}\sum_{i\in I_1(s)}\Delta^Y(1, s, X_i)\frac{n(s)}{n_1(s)}-\sum_{s\in \mathcal{S}}\sum_{i\in I_0(s)\cup I_1(s)}\Delta^Y(1, s, X_i)}\notag\\
			&=\frac{1}{\sqrt{n}}\envert[3]{\sum_{s\in \mathcal{S}}\sum_{i\in I_1(s)}\Delta^Y(1, s, X_i)\frac{n_0(s)}{n_1(s)}-\sum_{s\in \mathcal{S}}\sum_{i\in I_0(s)}\Delta^Y(1, s, X_i)}\notag\\
			&=\frac{1}{\sqrt{n}} \envert[3]{\sum_{s\in \mathcal{S}}n_0(s)\sbr[3]{\frac{\sum_{i\in I_1(s)}\Delta^Y(1,s,X_i)}{n_1(s)}-\frac{\sum_{i\in I_0(s)}\Delta^Y(1,s,X_i)}{n_0(s)}} }\notag\\
			&\leq \frac{1}{\sqrt{n}} \sum_{s\in \mathcal{S}}n_0(s)\envert[3]{\frac{\sum_{i\in I_1(s)}\Delta^Y(1,s,X_i)}{n_1(s)}-\frac{\sum_{i\in I_0(s)}\Delta^Y(1,s,X_i)}{n_0(s)}} =o_p(1)\notag
		\end{align}
		where the last equality is due to Assumption \ref{ass:Delta}. Thus
		\begin{align}
			R_{n,1}&=-\frac{1}{\sqrt{n}}\sum_{i=1}^{n}\frac{A_i}{\hat{\pi}(S_i)}\tilde{\mu}^Y(1, S_i, X_i)+\frac{1}{\sqrt{n}}\sum_{i=1}^{n}\tilde{\mu}^Y(1, S_i, X_i)+o_p(1)\notag\\
			&=-\frac{1}{\sqrt{n}}\sum_{i=1}^{n}\frac{A_i}{\hat{\pi}(S_i)}\tilde{\mu}^Y(1, S_i, X_i)+\frac{1}{\sqrt{n}}\sum_{i=1}^{n}A_i\tilde{\mu}^Y(1, S_i, X_i)+\frac{1}{\sqrt{n}}\sum_{i=1}^{n}(1-A_i)\tilde{\mu}^Y(1, S_i, X_i)+o_p(1)\notag\\
			&=\frac{1}{\sqrt{n}}\sum_{i=1}^{n}\del[3]{1-\frac{1}{\hat{\pi}(S_i)}}A_i\tilde{\mu}^Y(1, S_i, X_i)+\frac{1}{\sqrt{n}}\sum_{i=1}^{n}(1-A_i)\tilde{\mu}^Y(1, S_i, X_i)+o_p(1).\label{r5}
		\end{align}
		In addition, we note that
		\begin{align*}
			\frac{1}{\sqrt{n}}\sum_{i=1}^{n}\del[3]{1-\frac{1}{\hat{\pi}(S_i)}}A_i\tilde{\mu}^Y(1, S_i, X_i) & = \frac{1}{\sqrt{n}}\sum_{i=1}^{n}\del[3]{1-\frac{1}{\pi(S_i)}}A_i\tilde{\mu}^Y(1, S_i, X_i) \\
			& + \sum_{s\in \mathcal{S}}\del[3]{\frac{1}{\pi(s)}-\frac{1}{\hat{\pi}(s)}}\frac{1}{\sqrt{n}}\sum_{i=1}^{n}A_i\tilde{\mu}^Y(1, s, X_i)1\{S_i=s\}.
		\end{align*}
		Note that under Assumption \ref{ass:assignment1}(i), conditional on $\{S^{(n)}, A^{(n)}\}$, the distribution of $$\frac{1}{\sqrt{n}}\sum_{i=1}^{n}A_i\tilde{\mu}^Y(1, s, X_i)1\{S_i=s\}$$ is the same as the distribution of the same quantity where units are ordered by strata and then ordered by $A_i=1$ first and $A_i=0$ second within strata. To this end, define $N(s):=\sum_{i=1}^{n}1\{S_i<s\}$ and $F(s):=\mathbb{P}(S_i<s)$. Furthermore, independently for each $s\in \mathcal{S}$ and independently of $\{S^{(n)}, A^{(n)}\}$, let $\cbr[1]{X_i^s: 1\leq i\leq n}$ be i.i.d with marginal distribution equal to the distribution of $X_i|S=s$. Define
		\begin{align*}
			\tilde{\mu}^b(a, s, X_i^s)&:=\overline{\mu}^b(a, s, X_i^s)-\mathbb{E}\sbr[1]{\overline{\mu}^b(a, s, X_i^s)|S_i=s}
		\end{align*}
		Then, we have, for $s \in \mathcal{S}$,
		\begin{align*}
			\frac{1}{\sqrt{n}}\sum_{i=1}^{n}A_i\tilde{\mu}^Y(1, s, X_i)1\{S_i=s\} \stackrel{d}{=}\frac{1}{\sqrt{n}}\sum_{i=N(s)+1}^{N(s)+n_1(s)}\tilde{\mu}^Y(1, s, X_i^s).
		\end{align*}
		In addition, we have
		\begin{align*}
			\mathbb{E}\left[\left(\frac{1}{\sqrt{n}}\sum_{i=N(s)+1}^{N(s)+n_1(s)}\tilde{\mu}^Y(1, s, X_i^s) \right)^2\biggl| S^{(n)}, A^{(n)}\right] & = \frac{n_1(s)}{n}\mathbb{E}\sbr[1]{\tilde{\mu}^{Y,2}(a, s, X_i^s)|S^{(n)}} \\
			& \leq  \frac{n_1(s)}{n}E\left[\overline{\mu}^{Y,2}(a, s, X_i)|S_i=s\right] = O_p(1),
		\end{align*}
		which implies
		\begin{align*}
			\max_{s \in \mathcal{S}}\biggl|\frac{1}{\sqrt{n}}\sum_{i=N(s)+1}^{N(s)+n_1(s)}\tilde{\mu}^Y(1, s, X_i^s)\biggr| = O_p(1).
		\end{align*}
		Combining this with the facts that $\max_{s \in \mathcal{S}}|\hat{\pi}(s) - \pi(s)| = o_p(1)$ and $\min_{s \in \mathcal{S}}\pi(s)>c>0$ for some constant $c$, we have
		\begin{align*}
			& \sum_{s\in \mathcal{S}}\del[3]{\frac{1}{\pi(s)}-\frac{1}{\hat{\pi}(s)}}\frac{1}{\sqrt{n}}\sum_{i=1}^{n}A_i\tilde{\mu}^Y(1, s, X_i)1\{S_i=s\} = o_p(1) \\
			& \frac{1}{\sqrt{n}}\sum_{i=1}^{n}\del[3]{1-\frac{1}{\hat{\pi}(S_i)}}A_i\tilde{\mu}^Y(1, S_i, X_i) = \frac{1}{\sqrt{n}}\sum_{i=1}^{n}\del[3]{1-\frac{1}{\pi(S_i)}}A_i\tilde{\mu}^Y(1, S_i, X_i) + o_p(1).
		\end{align*}
		Therefore, we have
		\begin{align*}
			R_{n,1}&=\frac{1}{\sqrt{n}}\sum_{i=1}^{n}\del[3]{1-\frac{1}{\pi(S_i)}}A_i\tilde{\mu}^Y(1, S_i, X_i)+\frac{1}{\sqrt{n}}\sum_{i=1}^{n}(1-A_i)\tilde{\mu}^Y(1, S_i, X_i)+o_p(1).
		\end{align*}
		The linear expansion of $R_{n,2}$ can be established in the same manner.
		For $R_{n,3}$, note that
		\begin{align*}
			Y_i&=Y_i(1)\sbr[1]{D_i(1)A_i+D_i(0)(1-A_i)}+Y_i(0)\sbr[1]{1-D_i(1)A_i-D_i(0)(1-A_i)}\\
			&=\sbr[1]{Y_i(1)D_i(1)-Y_i(0)D_i(1)}A_i+\sbr[1]{Y_i(1)D_i(0)-Y_i(0)D_i(0)}(1-A_i)+Y_i(0).
		\end{align*}
		Then
		\begin{align}
			A_iY_i&= \sbr[1]{Y_i(1)D_i(1)+Y_i(0)(1-D_i(1))}A_i,\notag\\
			(1-A_i)Y_i&=\sbr[1]{Y_i(1)D_i(0)+Y_i(0)(1-D_i(0))}(1-A_i),\notag\\
			\frac{1}{\sqrt{n}}\sum_{i=1}^{n}\frac{A_iY_i}{\hat{\pi}(S_i)}&=\frac{1}{\sqrt{n}}\sum_{i=1}^{n}\frac{1}{\hat{\pi}(S_i)}  \sbr[1]{Y_i(1)D_i(1)+Y_i(0)(1-D_i(1))}A_i=:\frac{1}{\sqrt{n}}\sum_{i=1}^{n}\frac{1}{\hat{\pi}(S_i)}W_iA_i,\notag\\
			\frac{1}{\sqrt{n}}\sum_{i=1}^{n}\frac{(1-A_i)Y_i}{1-\hat{\pi}(S_i)}&=\frac{1}{\sqrt{n}}\sum_{i=1}^{n}\frac{\sbr[1]{Y_i(1)D_i(0)+Y_i(0)(1-D_i(0))}(1-A_i)}{1-\hat{\pi}(S_i)}=:\frac{1}{\sqrt{n}}\sum_{i=1}^{n}\frac{Z_i(1-A_i)}{1-\hat{\pi}(S_i)}.\notag
		\end{align}
		Thus we have
		\begin{align}
			R_{n,3}&=\frac{1}{\sqrt{n}}\sum_{i=1}^{n}\frac{A_iY_i}{\hat{\pi}(S_i)}-\frac{1}{\sqrt{n}}\sum_{i=1}^{n}\frac{(1-A_i)Y_i}{1-\hat{\pi}(S_i)}-\sqrt{n}G\notag\\
			& = \left\{\frac{1}{\sqrt{n}}\sum_{i=1}^{n}\frac{1}{\hat{\pi}(S_i)}\tilde{W}_iA_i-\frac{1}{\sqrt{n}}\sum_{i=1}^{n}\frac{1-A_i}{1-\hat{\pi}(S_i)}\tilde{Z}_i\right\} \notag\\
			&+\left\{\frac{1}{\sqrt{n}}\sum_{i=1}^{n}\frac{1}{\hat{\pi}(S_i)}\mathbb{E}[W_i|S_i]A_i-\frac{1}{\sqrt{n}}\sum_{i=1}^{n}\frac{1-A_i}{1-\hat{\pi}(S_i)}\mathbb{E}[Z_i|S_i]-\sqrt{n}G\right\}. \label{r20}
		\end{align}
		We now consider the second term on the RHS of \eqref{r20}. First note that
		\begin{align*}
			\frac{1}{\sqrt{n}}\sum_{i=1}^{n}\frac{1}{\hat{\pi}(S_i)}\mathbb{E}[W_i|S_i]A_i =  \frac{1}{\sqrt{n}}\sum_{i=1}^{n}\frac{1}{\pi(S_i)}\mathbb{E}[W_i|S_i]A_i-\frac{1}{\sqrt{n}}\sum_{i=1}^{n}\frac{\hat{\pi}(S_i)-\pi(S_i)}{\hat{\pi}(S_i)\pi(S_i)}\mathbb{E}[W_i|S_i]A_i,
		\end{align*}
		\begin{align}
			&\frac{1}{\sqrt{n}}\sum_{i=1}^{n}\frac{1}{\pi(S_i)}\mathbb{E}[W_i|S_i]A_i=\sum_{s\in \mathcal{S}}\frac{1}{\sqrt{n}}\sum_{i=1}^{n}\frac{1}{\pi(s)}\mathbb{E}[W_i|S_i=s]A_i1\{S_i = s\}\notag\\
			&=\sum_{s\in \mathcal{S}}\frac{1}{\sqrt{n}}\sum_{i=1}^{n}\frac{\mathbb{E}[W_i|S_i=s]}{\pi(s)}(A_i-\pi(s))1\{S_i = s\}+\sum_{s\in \mathcal{S}}\frac{1}{\sqrt{n}}\sum_{i=1}^{n}\frac{1}{\pi(s)}\mathbb{E}[W_i|S_i=s]\pi(s)1\{S_i = s\}\notag\\
			&=\sum_{s\in \mathcal{S}}\frac{\mathbb{E}[W|S=s]}{\pi(s)\sqrt{n}}\sum_{i=1}^{n}(A_i-\pi(s))1\{S_i = s\}+\sum_{s\in \mathcal{S}}\frac{\mathbb{E}[W|S=s]}{\sqrt{n}}\sum_{i=1}^{n}1\{S_i = s\}\notag\\
			&=\sum_{s\in \mathcal{S}}\frac{\mathbb{E}[W|S=s]}{\pi(s)\sqrt{n}}B_n(s)+\sum_{s\in \mathcal{S}}\frac{\mathbb{E}[W|S=s]}{\sqrt{n}}n(s),\label{r13}
		\end{align}
		and
		\begin{align}
			&\frac{1}{\sqrt{n}}\sum_{i=1}^{n}\frac{\hat{\pi}(S_i)-\pi(S_i)}{\hat{\pi}(S_i)\pi(S_i)}\mathbb{E}[W_i|S_i]A_i=\sum_{s\in \mathcal{S}}\frac{1}{\sqrt{n}}\sum_{i=1}^{n}\frac{\hat{\pi}(s)-\pi(s)}{\hat{\pi}(s)\pi(s)}\mathbb{E}[W_i|S_i=s]A_i1\{S_i=s\}\notag\\
			&=\sum_{s\in \mathcal{S}}\frac{1}{\sqrt{n}}\sum_{i=1}^{n}\frac{B_n(s)}{n(s)\hat{\pi}(s)\pi(s)}\mathbb{E}[W_i|S_i=s]A_i1\{S_i=s\}\notag \\
			&=\sum_{s\in \mathcal{S}}\frac{B_n(s)\mathbb{E}[W|S=s]}{\sqrt{n}n(s)\hat{\pi}(s)\pi(s)}\sum_{i=1}^{n}A_i1\{S_i=s\}\notag=\sum_{s\in \mathcal{S}}\frac{B_n(s)\mathbb{E}[W|S=s]}{\sqrt{n}n(s)\hat{\pi}(s)\pi(s)}n_1(s) \notag \\
			&=\sum_{s\in \mathcal{S}}\frac{B_n(s)\mathbb{E}[W|S=s]}{\sqrt{n}\pi(s)}.\notag 
		\end{align}
		Therefore, we have
		\begin{align*}
			\frac{1}{\sqrt{n}}\sum_{i=1}^{n}\frac{1}{\hat{\pi}(S_i)}\mathbb{E}[W_i|S_i]A_i  = \sum_{s\in \mathcal{S}}\frac{\mathbb{E}[W|S=s]}{\sqrt{n}}n(s).
		\end{align*}
		Similarly, we have
		\begin{align*}
			\frac{1}{\sqrt{n}}\sum_{i=1}^{n}\frac{1-A_i}{1-\hat{\pi}(S_i)}\mathbb{E}[Z_i|S_i] = \sum_{s\in \mathcal{S}}\frac{\mathbb{E}[Z|S=s]}{\sqrt{n}}n(s)
		\end{align*}
		Then, we have
		\begin{align}
			& \frac{1}{\sqrt{n}}\sum_{i=1}^{n}\frac{1}{\hat{\pi}(S_i)}\mathbb{E}[W_i|S_i]A_i-\frac{1}{\sqrt{n}}\sum_{i=1}^{n}\frac{1-A_i}{1-\hat{\pi}(S_i)}\mathbb{E}[Z_i|S_i]-\sqrt{n}G \notag \\
			& = \sum_{s\in \mathcal{S}}\frac{\mathbb{E}[W|S=s]}{\sqrt{n}}n(s)-\sum_{s\in \mathcal{S}}\frac{\mathbb{E}[Z|S=s]}{\sqrt{n}}n(s) -\sqrt{n}G \notag \\
			&=\sum_{s\in \mathcal{S}}\sqrt{n}\del[3]{\frac{n(s)}{n}-p(s)}\mathbb{E}[W-Z|S=s]+\sum_{s\in \mathcal{S}}\sqrt{n}p(s)\mathbb{E}[W-Z|S=s]-\sqrt{n}G \notag \\
			&=\sum_{s\in \mathcal{S}}\sqrt{n}\del[3]{\frac{n(s)}{n}-p(s)}\mathbb{E}[W-Z|S=s]+\sqrt{n}\mathbb{E}[W-Z]-\sqrt{n}G \notag \\
			&=\sum_{s\in \mathcal{S}}\frac{n(s)}{\sqrt{n}}\mathbb{E}[W-Z|S=s]-\sqrt{n}\mathbb{E}[W-Z] \notag \\
			&=\frac{1}{\sqrt{n}}\sum_{s\in \mathcal{S}}\sum_{i=1}^{n}\del[2]{1\{S_i=s\}\mathbb{E}[W_i-Z_i|S_i=s]}-\sqrt{n}\mathbb{E}[W-Z]\notag \\
			&=\frac{1}{\sqrt{n}}\sum_{i=1}^{n}\mathbb{E}[W_i-Z_i|S_i]-\sqrt{n}\mathbb{E}[W-Z] \notag \\
			&=\frac{1}{\sqrt{n}}\sum_{i=1}^{n}\del[1]{\mathbb{E}[W_i-Z_i|S_i]-\mathbb{E}[W_i-Z_i]}. \label{eq:Rc2}
		\end{align}
		Combining \eqref{r20} and \eqref{eq:Rc2}, we have
		\begin{align*}
			R_{n,3}& = \left\{\frac{1}{\sqrt{n}}\sum_{i=1}^{n}\frac{1}{\hat{\pi}(S_i)}\tilde{W}_iA_i-\frac{1}{\sqrt{n}}\sum_{i=1}^{n}\frac{1-A_i}{1-\hat{\pi}(S_i)}\tilde{Z}_i\right\}+\left\{\frac{1}{\sqrt{n}}\sum_{i=1}^{n}\del[1]{\mathbb{E}[W_i-Z_i|S_i]-\mathbb{E}[W_i-Z_i]}\right\} \\
			& = \left\{\frac{1}{\sqrt{n}}\sum_{i=1}^{n}\frac{1}{\pi(S_i)}\tilde{W}_iA_i-\frac{1}{\sqrt{n}}\sum_{i=1}^{n}\frac{1-A_i}{1-\pi(S_i)}\tilde{Z}_i\right\}\\
			&+\left\{\frac{1}{\sqrt{n}}\sum_{i=1}^{n}\del[1]{\mathbb{E}[W_i-Z_i|S_i]-\mathbb{E}[W_i-Z_i]}\right\} + o_p(1),
		\end{align*}
		where the second equality holds because
		\begin{align*}
			& \left(\frac{1}{\pi(s)} - \frac{1}{\hat{\pi}(s)}\right)\frac{1}{\sqrt{n}}\sum_{i=1}^{n}\tilde{W}_iA_i1\{S_i=s\} = o_p(1) \quad \text{and} \\
			& \left(\frac{1}{\pi(s)} - \frac{1}{\hat{\pi}(s)}\right)\frac{1}{\sqrt{n}}\sum_{i=1}^{n}\tilde{Z}_i(1-A_i)1\{S_i=s\} = o_p(1)
		\end{align*}
		due to the same argument used in the proofs of $R_{n,1}$.
	\end{proof}

		\begin{lem}
			Under the assumptions in Theorem \ref{thm:est}, we have
			\begin{align*}
				& \frac{1}{\sqrt{n}}\sum_{i=1}^{n} \Xi_{1}(\mathcal{D}_{i}, S_{i})A_{i}
				\rightsquigarrow\mathcal{N}\left( 0, \mathbb{E}\pi(S_{i})\Xi_{1}%
				^{2}(\mathcal{D}_{i},S_{i})\right) ,\\
				&  \frac{1}{\sqrt{n}}\sum_{i=1}^{n} \Xi_{0}(\mathcal{D}_{i}, S_{i})(1-A_{i})
				\rightsquigarrow\mathcal{N}\left( 0, \mathbb{E}(1-\pi(S_{i}))\Xi_{0}%
				^{2}(\mathcal{D}_{i},S_{i})\right) , \quad\text{and}\\
				& \frac{1}{\sqrt{n}}\sum_{i=1}^{n} \Xi_{2}(S_{i}) \rightsquigarrow
				\mathcal{N}(0, \mathbb{E}\Xi_{2}^{2}(S_{i})),
			\end{align*}
			and the three terms are asymptotically independent. \label{lem:clt}
		\end{lem}
		
		\begin{proof}
			Note that under Assumption \ref{ass:assignment1}(i), conditional on $\{S^{(n)}, A^{(n)}\}$, the distribution of
			$$\left(\frac{1}{\sqrt{n}}\sum_{i=1}^n \Xi_{1}(\mathcal{D}_i, S_i)A_i,~\frac{1}{\sqrt{n}}\sum_{i=1}^n \Xi_{0}(\mathcal{D}_i, S_i)(1-A_i)\right)$$
			is the same as the distribution of the same quantity where units are ordered by strata and then ordered by $A_i=1$ first and $A_i=0$ second within strata. To this end, define $N(s):=\sum_{i=1}^{n}1\{S_i<s\}$ and $F(s):=\mathbb{P}(S_i<s)$. Furthermore, independently for each $s\in \mathcal{S}$ and independently of $\{S^{(n)}, A^{(n)}\}$, let $\cbr[1]{\mathcal{D}_i^s: 1\leq i\leq n}$ be i.i.d with marginal distribution equal to the distribution of $\mathcal{D}|S=s$. Then, we have
			\begin{align*}
				& \left(\frac{1}{\sqrt{n}}\sum_{i=1}^n \Xi_{1}(\mathcal{D}_i, S_i)A_i,~\frac{1}{\sqrt{n}}\sum_{i=1}^n \Xi_{0}(\mathcal{D}_i, S_i)(1-A_i)\right)\biggl|S^{(n)},A^{(n)} \\
				& \stackrel{d}{=} \left(\frac{1}{\sqrt{n}}\sum_{s \in \mathcal{S}}\sum_{i=N(s)+1}^{N(s)+n_1(s)} \Xi_{1}(\mathcal{D}_i^s, s),~\frac{1}{\sqrt{n}}\sum_{s \in \mathcal{S}}\sum_{N(s)+n_1(s)+1}^{N(s)+n(s)} \Xi_{0}(\mathcal{D}_i^s, s)\right)\biggl|S^{(n)},A^{(n)}.
			\end{align*}
			In addition, since $\Xi_{2}(S_i)$ is a function of $\{S^{(n)}, A^{(n)}\}$, we have, arguing along the line of a joint distribution being the product of a conditional distribution and a marginal distribution,
			\begin{align*}
				& \left(\frac{1}{\sqrt{n}}\sum_{i=1}^n \Xi_{1}(\mathcal{D}_i, S_i)A_i,~\frac{1}{\sqrt{n}}\sum_{i=1}^n \Xi_{0}(\mathcal{D}_i, S_i)(1-A_i),~\frac{1}{\sqrt{n}}\sum_{i=1}^n\Xi_{2}(S_i)\right) \\
				& \stackrel{d}{=} \left(\frac{1}{\sqrt{n}}\sum_{s \in \mathcal{S}}\sum_{i=N(s)+1}^{N(s)+n_1(s)} \Xi_{1}(\mathcal{D}_i^s, s),~\frac{1}{\sqrt{n}}\sum_{s \in \mathcal{S}}\sum_{N(s)+n_1(s)+1}^{N(s)+n(s)} \Xi_{0}(\mathcal{D}_i^s, s),~\frac{1}{\sqrt{n}}\sum_{i=1}^n\Xi_{2}(S_i)\right).
			\end{align*}
			Define $\Gamma_{a,n}(u,s) = \frac{1}{\sqrt{n}}\sum_{i=1}^{\lfloor u n \rfloor}\Xi_a(\mathcal{D}_i^s,s)$ for $a = 0,1, s \in \mathcal{S}$. We have
			\begin{align*}
				& \frac{1}{\sqrt{n}}\sum_{s \in \mathcal{S}}\sum_{i=N(s)+1}^{N(s)+n_1(s)} \Xi_{1}(\mathcal{D}_i^s, s) = \sum_{s \in \mathcal{S}}\left[\Gamma_{1,n}\left(\frac{N(s)+n_1(s)}{n},s\right) - \Gamma_{1,n}\left(\frac{N(s)}{n},s\right)\right], \\
				& \frac{1}{\sqrt{n}}\sum_{s \in \mathcal{S}}\sum_{N(s)+n_1(s)+1}^{N(s)+n(s)} \Xi_{0}(\mathcal{D}_i^s, s) = \sum_{s \in \mathcal{S}}\left[ \Gamma_{0,n}\left(\frac{N(s)+n(s)}{n},s\right) - \Gamma_{0,n}\left(\frac{N(s)+n_1(s)}{n},s\right)\right].
			\end{align*}
			In addition, the partial sum process (w.r.t. $u \in [0,1]$) is stochastic equicontinuous and
			\begin{align*}
				\del[3]{\frac{N(s)}{n}, \frac{n_1(s)}{n}}\xrightarrow{p}\del[1]{F(s), \pi(s) p(s)}.
			\end{align*}
			Therefore,
			\begin{align*}
				& \left(\frac{1}{\sqrt{n}}\sum_{s \in \mathcal{S}}\sum_{i=N(s)+1}^{N(s)+n_1(s)} \Xi_{1}(\mathcal{D}_i^s, s),~\frac{1}{\sqrt{n}}\sum_{s \in \mathcal{S}}\sum_{N(s)+n_1(s)+1}^{N(s)+n(s)} \Xi_{0}(\mathcal{D}_i^s, s),~\frac{1}{\sqrt{n}}\sum_{i=1}^n\Xi_{2}(S_i)\right) \\
				& = \begin{pmatrix}
					& \sum_{s \in \mathcal{S}}\left[\Gamma_{1,n}\left(F(s)+p(s)\pi(s),s\right) - \Gamma_{1,n}\left(F(s),s\right)\right],\\
					&\sum_{s \in \mathcal{S}}\left[\Gamma_{0,n}\left(F(s)+p(s),s\right) - \Gamma_{0,n}\left(F(s)+\pi(s)p(s),s\right)\right],\\
					&~\frac{1}{\sqrt{n}}\sum_{i=1}^n\Xi_{2}(S_i)
				\end{pmatrix} + o_p(1)
			\end{align*}
			and by construction,
			\begin{align*}
				& \sum_{s \in \mathcal{S}}\left[\Gamma_{1,n}\left(F(s)+p(s)\pi(s),s\right) - \Gamma_{1,n}\left(F(s),s\right)\right],\\
				& \sum_{s \in \mathcal{S}}\left[\Gamma_{0,n}\left(F(s)+p(s),s\right) - \Gamma_{0,n}\left(F(s)+p(s)\pi(s),s\right)\right],\\
				& \text{and} \quad \frac{1}{\sqrt{n}}\sum_{i=1}^n\Xi_{2}(S_i)
			\end{align*}
			are independent. Last, we have
			\begin{align*}
				& \sum_{s \in \mathcal{S}}\left[\Gamma_{1,n}\left(F(s)+p(s)\pi(s),s\right) - \Gamma_{1,n}\left(F(s),s\right)\right] \convD \N\left(0, \mathbb{E}\pi(S_i)\Xi_1^2(\mathcal{D}_i,S_i)\right) \\
				& \sum_{s \in \mathcal{S}}\left[\Gamma_{0,n}\left(F(s)+p(s),s\right) - \Gamma_{0,n}\left(F(s)+p(s)\pi(s),s\right)\right] \convD \N\left(0, \mathbb{E}(1-\pi(S_i))\Xi_0^2(\mathcal{D}_i,S_i)\right) \\
				& \frac{1}{\sqrt{n}}\sum_{i=1}^n\Xi_{2}(S_i) \convD \N\left(0, \mathbb{E}\Xi_{2}^2(S_i)\right).
			\end{align*}
			This implies the desired result.
		\end{proof}

		\begin{lem}
			Suppose assumptions in Theorem \ref{thm:est} hold. Then,
			\begin{align*}
				\frac{1}{n}\sum_{i=1}^{n} A_{i}\hat{\Xi}_{1}^{2}(\mathcal{D}_{i},S_{i})
				\overset{p}{\longrightarrow} \sigma_{1}^{2}, \quad\frac{1}{n}\sum_{i=1}%
				^{n}(1-A_{i}) \hat{\Xi}_{0}^{2}(\mathcal{D}_{i},S_{i})
				\overset{p}{\longrightarrow} \sigma_{0}^{2}, \quad\text{and} \quad\frac{1}%
				{n}\sum_{i=1}^{n} \hat{\Xi}_{2}^{2}(\mathcal{D}_{i},S_{i})
				\overset{p}{\longrightarrow} \sigma_{2}^{2}.
			\end{align*}
			\label{lem:sigma}
		\end{lem}
		
		\begin{proof}
			To derive the limit of $\frac{1}{n}\sum_{i=1}^nA_i \hat{\Xi}_1^2(\mathcal{D}_i,S_i)$, we first define
			\begin{align*}
				\tilde{\Xi}_1^*(\mathcal{D}_i,s) & = \left[\left(1- \frac{1}{\pi(s)} \right)\overline{\mu}^Y(1,s,X_i) - \overline{\mu}^Y(0,s,X_i) + \frac{Y_i}{\pi(s)}\right] \\
				& - \tau \left[\left(1- \frac{1}{\pi(s)} \right)\overline{\mu}^D(1,s,X_i) - \overline{\mu}^D(0,s,X_i) + \frac{D_i}{\pi(s)}\right] \quad \text{and}\\
				\breve{\Xi}_1(\mathcal{D}_i,s) & = \left[\left(1- \frac{1}{\hat{\pi}(s)} \right)\overline{\mu}^Y(1,s,X_i) - \overline{\mu}^Y(0,s,X_i) + \frac{Y_i}{\hat{\pi}(s)}\right] \\
				& - \hat{\tau} \left[\left(1- \frac{1}{\hat{\pi}(s)} \right)\overline{\mu}^D(1,s,X_i) - \overline{\mu}^D(0,s,X_i) + \frac{D_i}{\hat{\pi}(s)}\right]
			\end{align*}
			Then, we have
			\begin{align*}
				& \left[\frac{1}{n_1(s)}\sum_{i \in I_1(s)}(\tilde{\Xi}_1^*(\mathcal{D}_i,s) - \tilde{\Xi}_1(\mathcal{D}_i,s))^2 \right]^{1/2} \\
				& \leq \left[\frac{1}{n_1(s)}\sum_{i \in I_1(s)}(\tilde{\Xi}_1^*(\mathcal{D}_i,s) - \breve{\Xi}_1(\mathcal{D}_i,s))^2 \right]^{1/2} + \left[\frac{1}{n_1(s)}\sum_{i \in I_1(s)}(\tilde{\Xi}_1(\mathcal{D}_i,s) - \breve{\Xi}_1(\mathcal{D}_i,s))^2 \right]^{1/2}  \\
				& \leq \frac{|\hat{\pi}(s) - \pi(s)|}{\hat{\pi}(s)\pi(s)} \biggl\{ \sbr[3]{\frac{1}{n_1(s)}\sum_{i \in I_1(s)} \overline{\mu}^{Y,2}(1,s,X_i) }^{1/2} + \sbr[3]{\frac{1}{n_1(s)}\sum_{i \in I_1(s)} W_i^2}^{1/2}\biggr\} \\
				& \quad + \left(|\hat{\tau} -\tau|+ \frac{|\tau\hat{\pi}(s) - \hat{\tau}\pi(s)|}{\hat{\pi}(s)\pi(s)} \right) \biggl\{\sbr[3]{\frac{1}{n_1(s)}\sum_{i \in I_1(s)} \overline{\mu}^{D,2}(1,s,X_i)}^{1/2} +  \sbr[3]{\frac{1}{n_1(s)}\sum_{i \in I_1(s)} D_i^2(1) }^{1/2}\biggr\} \\
				&\quad+ |\hat{\tau}-\tau|\sbr[3]{\frac{1}{n_1(s)}\sum_{i \in I_1(s)} \overline{\mu}^{D,2}(0,s,X_i)}^{1/2}\\
				& \quad+ \left(\frac{1}{\hat{\pi}(s)}-1\right)  \biggl\{\sbr[3]{\frac{1}{n_1(s)}\sum_{i \in I_1(s)} \Delta^{Y,2}(1,s,X_i)}^{1/2}  +  |\hat{\tau}|\sbr[3]{\frac{1}{n_1(s)}\sum_{i \in I_1(s)} \Delta^{D,2}(1,s,X_i)}^{1/2} \biggr\} \\
				& \quad+ \sbr[3]{\frac{1}{n_1(s)}\sum_{i \in I_1(s)} \Delta^{Y,2}(0,s,X_i) }^{1/2} + |\hat{\tau}|\sbr[3]{\frac{1}{n_1(s)}\sum_{i \in I_1(s)} \Delta^{D,2}(0,s,X_i)}^{1/2} = o_p(1),
			\end{align*}
			where the second inequality holds by the triangle inequality and the fact that when $i \in I_{1}(s)$, $A_i=1$, $Y_i = W_i$, and $D_i = D_i(1)$, and the last equality is due to Assumption \ref{ass:Delta}(ii) and the facts that $\hat{\pi}(s) \convP \pi(s)$ and $\hat{\tau} \convP \tau$. This further implies
			\begin{align*}
				\frac{1}{n_1(s)}\sum_{i \in I_1(s)}(\tilde{\Xi}_1^*(\mathcal{D}_i,s) - \tilde{\Xi}_1(\mathcal{D}_i,s)) \convP 0,
			\end{align*}
			by the Cauchy-Schwarz inequality and thus,
			\begin{align*}
				\left[\frac{1}{n_1(s)}\sum_{i \in I_1(s)}\hat{\Xi}_1^2(\mathcal{D}_i,s) \right]^{1/2} \leq  \left[\frac{1}{n_1(s)}\sum_{i \in I_1(s)}\left(\tilde{\Xi}_1^*(\mathcal{D}_i,s) - \frac{1}{n_1}\sum_{i \in I_1(s)}\tilde{\Xi}_1^*(\mathcal{D}_i,s)\right)^2 \right]^{1/2} + o_p(1).
			\end{align*}
			Next, following the same argument in the proof of Lemma  \ref{lem:clt}, we have
			\begin{align*}
				\frac{1}{n_1(s)}\sum_{i \in I_1(s)}\tilde{\Xi}_1^*(\mathcal{D}_i,s) & \stackrel{d}{=} \frac{1}{n_1(s)}\sum_{i = N(s)+1}^{N(s)+n_1(s)}\biggl\{\left[\left(1- \frac{1}{\pi(s)} \right)\overline{\mu}^Y(1,s,X_i^s) - \overline{\mu}^Y(0,s,X_i^s) + \frac{W_i^s}{\pi(s)}\right] \\
				& \qquad - \tau \left[\left(1- \frac{1}{\pi(s)} \right)\overline{\mu}^D(1,s,X_i^s) - \overline{\mu}^D(0,s,X_i^s) + \frac{D_i^s(1)}{\pi(s)}\right] \biggr\} \\
				& \convP \mathbb{E}\biggl\{\left[\left(1- \frac{1}{\pi(S_i)} \right)\overline{\mu}^Y(1,S_i,X_i) - \overline{\mu}^Y(0,S_i,X_i) + \frac{W_i}{\pi(S_i)}\right] \\
				& \qquad - \tau \left[\left(1- \frac{1}{\pi(S_i)} \right)\overline{\mu}^D(1,S_i,X_i) - \overline{\mu}^D(0,S_i,X_i) + \frac{D_i(1)}{\pi(S_i)}\right] |S_i=s\biggr\},
			\end{align*}
			This implies
			\begin{align*}
				& \left[\frac{1}{n_1(s)}\sum_{i \in I_1(s)}\left(\tilde{\Xi}_1^*(\mathcal{D}_i,s) - \frac{1}{n_1}\sum_{i \in I_1(s)}\tilde{\Xi}_1^*(\mathcal{D}_i,s)\right)^2 \right]^{1/2} \\
				& = \biggl[\frac{1}{n_1(s)}\sum_{i \in I_1(s)}\biggl(\tilde{\Xi}_1^*(\mathcal{D}_i,s) -  \mathbb{E}\biggl\{\left[\left(1- \frac{1}{\pi(S_i)} \right)\overline{\mu}^Y(1,S_i,X_i) - \overline{\mu}^Y(0,S_i,X_i) + \frac{W_i}{\pi(S_i)}\right] \\
				& \qquad - \tau \left[\left(1- \frac{1}{\pi(S_i)} \right)\overline{\mu}^D(1,S_i,X_i) - \overline{\mu}^D(0,S_i,X_i) + \frac{D_i(1)}{\pi(S_i)}\right] |S_i=s\biggr\} \biggr)^2 \biggr]^{1/2} + o_p(1) \\
				& = \biggl[\frac{1}{n_1(s)}\sum_{i \in I_1(s)}\biggl(\left[\left(1- \frac{1}{\pi(S_i)} \right)\tilde{\mu}^Y(1,S_i,X_i) - \tilde{\mu}^Y(0,S_i,X_i) + \frac{\tilde{W}_i}{\pi(S_i)}\right] \\
				& \qquad - \tau \left[\left(1- \frac{1}{\pi(S_i)} \right)\tilde{\mu}^D(1,S_i,X_i) - \tilde{\mu}^D(0,S_i,X_i) + \frac{\tilde{D}_i(1)}{\pi(S_i)}\right] \biggr)^2 \biggr]^{1/2} + o_p(1).
			\end{align*}
			Last, following the same argument in the proof of Lemma  \ref{lem:clt}, we have
			\begin{align*}
				& \biggl[\frac{1}{n_1(s)}\sum_{i \in I_1(s)}\biggl(\left[\left(1- \frac{1}{\pi(S_i)} \right)\tilde{\mu}^Y(1,S_i,X_i) - \tilde{\mu}^Y(0,S_i,X_i) + \frac{\tilde{W}_i}{\pi(S_i)}\right] \\
				& - \tau \left[\left(1- \frac{1}{\pi(S_i)} \right)\tilde{\mu}^D(1,S_i,X_i) - \tilde{\mu}^D(0,S_i,X_i) + \frac{\tilde{D}_i(1)}{\pi(S_i)}\right] \biggr)^2 \biggr]^{1/2}  \\
				& \stackrel{d}{=}  \biggl[\frac{1}{n_1(s)}\sum_{i =N(s)+1}^{N(s)+n_1(s)}\biggl(\left[\left(1- \frac{1}{\pi(s)} \right)\tilde{\mu}^Y(1,s,X_i^s) - \tilde{\mu}^Y(0,s,X_i^s) + \frac{\tilde{W}_i^s}{\pi(s)}\right] \\
				& - \tau \left[\left(1- \frac{1}{\pi(s)} \right)\tilde{\mu}^D(1,s,X_i^s) - \tilde{\mu}^D(0,s,X_i^s) + \frac{\tilde{D}_i^s(1)}{\pi(s)}\right] \biggr)^2 \biggr]^{1/2} \\
				& \convP \left[\mathbb{E}(\Xi_1^2(\mathcal{D}_i,S_i)|S_i=s)\right]^{1/2},
			\end{align*}
			where $\tilde{W}_i^s = W_i^s - \mathbb{E}(W_i|S_i=s)$ and $\tilde{D}_i^s(1) = D_i^s(1) - \mathbb{E}(D_i(1)|S_i=s)$ and the last convergence is due to the fact that conditionally on $S^{(n)},A^{(n)}$, $\{X_i^s, \tilde{W}_i^s, \tilde{D}_i^s(1)\}_{i \in I_1(s)}$ is a sequence of i.i.d. random variables so that the standard LLN is applicable. Combining all the results above, we have shown that
			\begin{align*}
				\frac{1}{n_1(s)}\sum_{i\in I_1(s)} \hat{\Xi}_1^2(\mathcal{D}_i,S_i)&\xrightarrow{p} \mathbb{E}(\Xi_1^2(\mathcal{D}_i,S_i)|S_i=s)\\
				\frac{1}{n}\sum_{i=1}^n A_i \hat{\Xi}_1^2(\mathcal{D}_i,S_i) & = \sum_{s \in S} \frac{n_1(s)}{n}  \del[3]{\frac{1}{n_1(s)}\sum_{i \in I_1(s)}\hat{\Xi}_1^2(\mathcal{D}_i,S_i)}\\
				& \convP \sum_{s\in \mathcal{S}} p(s)\pi(s)      \mathbb{E}(\Xi_1^2(\mathcal{D}_i,S_i)|S_i=s) = \mathbb{E}\sbr[2]{\pi(S_i)\mathbb{E}(\Xi_1^2(\mathcal{D}_i,S_i)|S_i)}  = \sigma_1^2.
			\end{align*}
			For the same reason, we can show that
			\begin{align*}
				\frac{1}{n}\sum_{i=1}^n (1-A_i)\hat{\Xi}_0^2(\mathcal{D}_i,S_i) \convP \sigma_0^2.
			\end{align*}
			Last, by the similar argument, we have
			\begin{align*}
				\frac{1}{n}\sum_{i=1}^n \hat{\Xi}_2^2(S_i) & = \sum_{s \in \mathcal{S}} \frac{n(s)}{n} \hat{\Xi}_2^2(s)\\
				& = \sum_{s \in \mathcal{S}} \frac{n(s)}{n} ( \mathbb{E}(W_i - \tau D_i(1)|S_i=s) - \mathbb{E}(Z_i - \tau D_i(0)|S_i=s))^2 + o_p(1)\\
				& = \sum_{s \in \mathcal{S}} \frac{n(s)}{n}\Xi_2^2(s) + o_p(1)\\
				& \convP \sum_{s \in \mathcal{S}} p(s)\Xi_2^2(s) = \mathbb{E}\Xi_2^2(S_i) = \sigma_2^2.
			\end{align*}
		\end{proof}
		
		\section{An Additional Simulation}
		\label{sec:additional simulation}
		
		In this section, we use an additional simulation to demonstrate that when probabilities of treatment assignment $\{\pi(s)\}$ are heterogeneous across strata, the TSLS estimator could be inconsistent. The data generating process we consider here, denoted DGP(iv), is almost the same as DGP(i) in Section \ref{sec:sim}; the only difference is in $Y_i(a)$:
		\begin{align*}
			Y_i(1) = &2+S_i^2+0.7X_{1,i}^2+X_{2,i}+4Z_i+\varepsilon_{1,i}\\
			Y_i(0) = &1+0.7X_{1,i}^2+X_{2,i}+4Z_i+\varepsilon_{2,i}.
		\end{align*}
		The rationale for specifying this DGP is to allow a difference between the probabilistic limit of $\hat{\tau}_{TSLS}$ and $\tau$. We consider randomization schemes SRS and SBR with $(\pi(1),\pi(2),\pi(3),\pi(4))=(0.2,0.2,0.2,0.5)$. We do not consider randomization scheme WEI or BCD because for these two, $\pi(s)=0.5$ for all $s\in \mathcal{S}$. The rest of the simulation setting is the same as DGP(i) in Section \ref{sec:sim}. Table \ref{tab:Simulation extra} presents the empirical sizes. We see that all estimators, except the TSLS estimator, have the empirical sizes converging to 0.05 as sample size increases.
		
		\begin{table}[pth]
			\caption{{\protect\small NA, TSLS, L, S, NL, F, NP, SNP, R stand for the unadjusted estimator, TSLS estimator, optimally linearly adjusted estimator, \cite{anseletal2018}'s S estimator with $X_i$ as regressor, nonlinearly (logistic) adjusted estiamtor, further efficiency improving estimator, nonparametrically adjusted estimator, \cite{anseletal2018}'s S estimator with $\mathring{\Psi}_{i,n}$ defined in (\ref{r1}) as regressor, and estimator with regularized adjustments, respectively. }}%
			\label{tab:Simulation extra}%
			\centering
			\smallskip\begin{tabularx}{\linewidth}{@{\extracolsep{\fill}}lcccccc}
				\toprule
				& \multicolumn{2}{c}{$n = 400$} & \multicolumn{2}{c}{$n = 800$} & \multicolumn{2}{c}{$n = 1200$}\\ \cmidrule{2-3}\cmidrule{4-5}\cmidrule{6-7}
				Methods  & \multicolumn{1}{c}{SRS} & \multicolumn{1}{c}{SBR} & \multicolumn{1}{c}{SRS} &  \multicolumn{1}{c}{SBR}& \multicolumn{1}{c}{SRS} &  \multicolumn{1}{c}{SBR} \\
				\midrule
				\multicolumn{7}{c}{DGP(iv)} \\
				\multicolumn{7}{l}{\textit{Size}} \\
				\quad NA & 0.043 & 0.036 & 0.048 & 0.047  &  0.044 & 0.050  \\
				\quad TSLS & 0.069 & 0.072 & 0.114 & 0.113  &  0.149 & 0.142 \\
				\quad L &  0.064 & 0.060 & 0.056 & 0.055  &  0.053 & 0.054 \\
				\quad S & 0.064 & 0.060 & 0.056 & 0.055  &  0.053 & 0.054 \\
				\quad NL & 0.060 & 0.058 & 0.056 & 0.054  &  0.053 & 0.054\\
				\quad F & 0.083 & 0.077 & 0.064 & 0.062  &  0.058 & 0.057\\
				\quad NP & 0.201 & 0.174 & 0.107 & 0.096  &  0.080 & 0.078\\
				\quad SNP & 0.201 & 0.190 & 0.102 & 0.099  &  0.076 & 0.077  \\
				\quad R & 0.078 & 0.079 & 0.063 & 0.061  &  0.060 & 0.061 \\
				\bottomrule
		\end{tabularx}\end{table}
		
		\bibliographystyle{chicago}
		\bibliography{BCAR}

\bibliographystyle{chicago}
\bibliography{BCAR}

\end{document}